\documentclass[preprint,12pt]{elsarticle}
\usepackage{amsmath,amssymb,amsthm}
\usepackage[margin=2.5cm]{geometry}
\usepackage{setspace,booktabs,graphicx,subcaption,placeins,morefloats,hyperref}
\usepackage{tikz}
\usetikzlibrary{positioning,fit,backgrounds,calc,shapes.misc,arrows.meta}
\hypersetup{colorlinks=true,linkcolor=blue,citecolor=blue}
\journal{Journal of Computational Physics}
\biboptions{sort&compress}
\newtheorem{proposition}{Proposition}
\newtheorem{theorem}{Theorem}
\newcounter{algcounter}
\newlength{\figW}

\newcommand{\dd}{\mathrm{d}}
\newcommand{\bx}{\boldsymbol{x}}
\newcommand{\bxi}{\boldsymbol{\xi}}
\newcommand{\bu}{\boldsymbol{u}}
\newcommand{\bg}{\boldsymbol{g}}
\newcommand{\bn}{\boldsymbol{n}}
\newcommand{\bc}{\boldsymbol{c}}
\newcommand{\bq}{\boldsymbol{q}}
\newcommand{\bpsi}{\boldsymbol{\psi}}
\newcommand{\bU}{\boldsymbol{U}}
\newcommand{\bF}{\boldsymbol{F}}
\newcommand{\bj}{\boldsymbol{j}}
\newcommand{\bt}{\boldsymbol{t}}
\newcommand{\avg}[1]{\left\langle #1\right\rangle}

\begin{document}
\begin{frontmatter}
\title{Wave--Particle Decomposition for Kinetic Equations II:\\ Full Boltzmann equation}
\author[iapcm]{Chang Liu\corref{cor1}}
\ead{liuchang@iapcm.ac.cn}
\cortext[cor1]{Corresponding author}
\author[hkustmath,hkustmae,hkustsri]{Kun Xu}
\ead{makxu@ust.hk}

\address[iapcm]{Institute of Applied Physics and Computational Mathematics, Beijing, P.R. China}
\address[hkustmath]{Department of Mathematics, Hong Kong University of Science and Technology, Clear Water Bay, Kowloon, Hong Kong}
\address[hkustmae]{Department of Mechanical and Aerospace Engineering, Hong Kong University of Science and Technology, Clear Water Bay, Kowloon, Hong Kong}
\address[hkustsri]{Shenzhen Research Institute, Hong Kong University of Science and Technology, Shenzhen, China}

\begin{abstract}
The wave--particle decomposition of kinetic relaxation equations is
extended to the full Boltzmann equation. An identity transformation writes
the collision term as a relaxation toward the local Maxwellian plus a
conservative remainder. Along characteristics, the equilibrium integral over a
local kinetic horizon defines the wave, and its
exact complement, the particle, carries the remainder. The
moments of the wave and particle equations are extended Navier--Stokes
equations, and the total conservation law and the particle equation form the
wave--particle multiscale equations, a closed system, exact for
every horizon, that spans a continuous spectrum from the Boltzmann equation to
Navier--Stokes hydrodynamics, parametrized by the horizon-to-relaxation
ratio. A coupled implicit wave--particle iteration solves
this system, updating the conservative variables with the particle fixed and
the particle with the predicted wave and an endpoint traction from the
macroscopic prediction. Under discrete compatibility assumptions the
iteration preserves the discrete Boltzmann solution as its fixed point, and
its near-continuum convergence factor is bounded by a particle transport
share decaying exponentially with this ratio. Tests on a normal shock,
Couette and cavity flows, and hypersonic cylinder and Apollo flows agree with
the conventional iterative scheme in rarefied and transition regimes and
with Navier--Stokes solutions near the continuum limit. There the wave--particle
iteration needs about three to more than four orders of magnitude fewer
iterations than this scheme or, where it did not converge, its Shakhov
version, and the Boltzmann solution costs one to four times the wall time of a
Navier--Stokes solver on the same mesh.
\end{abstract}
\end{frontmatter}

\section{Introduction}\label{sec:intro}
The Boltzmann equation~\cite{boltzmann1872} describes a dilute gas through the
distribution of its molecular velocities and remains valid from the
free-molecular regime to the continuum limit. Its mathematical theory and
applications were systematized by Cercignani~\cite{cercignani1988}, and its
asymptotic theory for small Knudsen numbers, the Knudsen number being the
ratio of the mean free path to a macroscopic length, was developed by
Sone~\cite{sone2007}. The kinetic and the hydrodynamic descriptions are
connected by the Chapman--Enskog expansion~\cite{chapman1970}, which recovers
the Euler equations at leading order and the Navier--Stokes equations at first
order. In this regime collisions drive the distribution towards a local
Maxwellian while conserving mass, momentum and energy, so that the slow
evolution of the conservative variables is governed by transport. An
efficient multiscale method should therefore keep the kinetic description
where non-equilibrium effects are essential and, near equilibrium, treat the
slow conservative modes directly.

The direct simulation Monte Carlo (DSMC) method of Bird~\cite{bird1994} is the
standard particle method for rarefied flows. It represents the gas by
simulation particles with stochastic collisions, and since its cell size and
time step must be smaller than the mean free path and the collision time, it
becomes costly near the continuum limit. Deterministic discrete-velocity
methods solve for the distribution on a velocity grid without statistical
noise. Yang and Huang~\cite{yang1995} computed rarefied flows with nonlinear
model Boltzmann equations, and Mieussens~\cite{mieussens2000} constructed
discrete-velocity models and schemes for the Bhatnagar--Gross--Krook (BGK)
equation. The BGK model~\cite{bgk1954} replaces the collision integral by a
relaxation toward the local Maxwellian, and the Shakhov
model~\cite{shakhov1968} adds a heat-flux correction that gives the correct
Prandtl number. For the
full collision operator, Mouhot and Pareschi~\cite{mouhot2006} developed fast
spectral algorithms, with which Wu et al.~\cite{wu2013} obtained
deterministic Boltzmann solutions of rarefied gas flows.

For steady problems these kinetic solvers commonly rely on the conventional
iterative scheme (CIS)~\cite{xu2022implicit,su2020}, which treats the
transport and the loss term implicitly and takes the gain term from the
previous iterate. The same iteration is known as source iteration in
discrete-ordinates particle transport, where its acceleration was reviewed by
Adams and Larsen~\cite{adams2002}. Because the conservative modes relax only
through transport, the scheme converges slowly as the Knudsen number
decreases. Asymptotic-preserving (AP) schemes address the stiffness of the
continuum limit. Jin~\cite{jin1999} formulated the AP concept for multiscale
kinetic equations, Pareschi and Russo~\cite{pareschi2005} developed
implicit--explicit Runge--Kutta schemes for relaxation systems, Bennoune,
Lemou and Mieussens~\cite{bennoune2008} constructed a micro--macro scheme that
preserves the compressible Navier--Stokes asymptotics of the Boltzmann
equation, and Filbet and Jin~\cite{filbet2010} penalized the stiff collision
operator by a BGK operator; these developments were reviewed by Dimarco and
Pareschi~\cite{dimarco2014}. For steady rarefied flows, the general synthetic
iterative scheme (GSIS) of Su, Wu and co-workers~\cite{su2020} couples the
kinetic equation with synthetic macroscopic equations and reaches steady
solutions within dozens of iterations.

The unified gas-kinetic scheme (UGKS) of Xu and Huang~\cite{xu2010} builds
the multiscale behaviour into the numerical flux, following the direct
modeling methodology. It updates the conservative variables and the discrete
distribution together within a time step, and its interface flux is obtained
by integrating over the step the time-dependent integral solution of the
kinetic model equation along characteristics. In this solution the
equilibrium integral supplies the hydrodynamic part of the flux and the free
transport of the initial distribution supplies the kinetic part, with weights
set by the ratio of the time step to the local collision time, so that the
flow physics is recovered on the cell scale without restricting the cell size
and the time step to the mean free path and the collision time. Huang, Xu and
Yu~\cite{huang2012ugks} extended the scheme to multidimensional flows,
Xu~\cite{xu2014direct} formulated the underlying direct modeling methodology,
and Liu et al.~\cite{liu2016boltzmann} incorporated the full Boltzmann
collision operator. For steady flows, Zhu, Zhong and Xu~\cite{zhu2016implicit}
developed an implicit UGKS in which an implicit solution of the macroscopic
equations predicts the equilibrium state and the kinetic equation is then
solved implicitly with this prediction, so that the macroscopic and
microscopic equations are coupled within one iteration. They accelerated the
method further by multigrid~\cite{zhu2017multigrid}, and Xu et
al.~\cite{xu2022implicit} formulated a UGKS-based implicit iterative method
for multiscale non-equilibrium flows. The discrete unified gas-kinetic scheme
of Guo, Xu and Wang~\cite{guo2013} provides a related finite-volume
discrete-velocity framework, and the UGKS has been extended to radiative
transfer~\cite{sun2015ugks}, multicomponent plasma
transport~\cite{liu2017plasma} and gas--particle multiphase
flows~\cite{liu2019multiphase}; its first decade was reviewed by Zhu and
Xu~\cite{zhu2021decade}.

The unified gas-kinetic wave--particle (UGKWP) method represents the same
integral solution within each time step by an analytically evolved wave and
by stochastic particles. The fraction $e^{-\Delta t/\tau}$ of the gas that
travels the time step $\Delta t$ without collision is represented by sampled
particles, and the remainder is evolved analytically as a wave, so that the
particle population decays exponentially in the continuum regime. Liu, Zhu
and Xu~\cite{liu2020ugkwp} constructed the method for continuum and rarefied
gas flows and Zhu et al.~\cite{zhu2019ugkwp} for unstructured meshes.
Subsequent work developed it for photon transport~\cite{li2020ugkwp}, gas
mixtures and plasma~\cite{liu2021plasma}, diatomic gases with
rotational~\cite{xu2021diatomic} and vibrational~\cite{wei2024ugkwp}
non-equilibrium, gas--particle flows~\cite{yang2022multiphase}, hypersonic
non-equilibrium flows~\cite{long2024nonequilibrium} and non-equilibrium
turbulence~\cite{yang2026turbulence}, and Guo, Zhu and
Xu~\cite{guo2026representation} analysed its kinetic representation. For
steady flows, local time stepping was introduced into the simplified unified
wave--particle method by Yang et al.~\cite{yang2026lts} and, with a rigorously
justified balance of the interfacial particle flux, into the UGKWP method by
Guo et al.~\cite{guo2026lts}.

Two structural ideas run through these developments: the wave--particle
modeling of the integral solution, which comes from the UGKS and UGKWP line,
and the two-way coupling between a macroscopic conservation law and a kinetic
component, which comes from the micro--macro decomposition
line~\cite{bennoune2008}. Building on both, the first paper of this
series~\cite{liu_wpd_I}, referred to as Part~I, formulated the
wave--particle decomposition (WPD) at the level of the kinetic equation for
kinetic relaxation equations. In that formulation the characteristic integral
solution is truncated at a local kinetic horizon: the relaxation contribution
accumulated over the horizon defines the wave, and the collisionless factor
beyond the horizon defines the particle. The proportion of wave and particle
is therefore set by the local horizon-to-relaxation ratio rather than by the
global time step of UGKWP, and the particle is a collisionless fraction of the
distribution rather than the signed residual of a micro--macro decomposition.
The wave equation, the particle equation and the source-free total
conservation law form a unified wave--particle system across the full
Knudsen spectrum, which Part~I realized by an explicit conservative
macro--micro finite-volume scheme with a Navier--Stokes gas-kinetic wave flux
and a discrete-ordinate or Monte Carlo particle solver, establishing its
asymptotic-preserving continuum limit and its regime-adaptive kinetic
representation.

The present work extends the decomposition to the full Boltzmann equation, and
its contribution is both theoretical and numerical. On the theoretical side,
an identity transformation splits the collision term into a relaxation
toward the local Maxwellian and a conservative remainder, leaving the
Boltzmann equation unchanged. Integration along characteristics then yields
the characteristic integral solution, whose equilibrium integral defines the
wave and whose complement defines a particle carrying the collision
remainder. The particle is thus no longer the pure collisionless fraction of
Part~I: it adds to that fraction the accumulated remainder, which is
conservative but signed, and reduces to the collisionless population when the
remainder vanishes. The wave and the particle satisfy exact microscopic
equations and extended Navier--Stokes equations, and the total conservation
law and the particle equation form the wave--particle multiscale equations, a
closed system that is equivalent to the Boltzmann equation for every horizon
and spans a continuous spectrum, parametrized by the horizon-to-relaxation
ratio, from the Boltzmann equation to its Euler and Navier--Stokes limits.

On the numerical side, the system is solved by a coupled implicit
wave--particle (WP) iteration built on second-order gas-kinetic wave fluxes and MUSCL-type upwind particle
fluxes, which inherits the prediction--evolution structure of the implicit
UGKS~\cite{zhu2016implicit}. The W iteration advances the conservative
variables with the particle fixed, and the P iteration advances the particle
with the wave fixed at the predicted state and with an endpoint traction
supplied by the macroscopic prediction. The fixed point of the iteration is
the discrete Boltzmann solution, and a linear analysis shows how the two-way
coupling removes the slow relaxation of the conservative modes that limits
the conventional iterative scheme near equilibrium. In the continuum tests
the coupled WP iteration needs about three to more than four orders of magnitude
fewer iterations than the conventional iterative scheme, applied to the
Boltzmann equation or, where that could not be converged, to the cheaper and
faster Shakhov model, and its wall time lies between one and about four times
that of a Navier--Stokes solver sharing its macroscopic discretization.

Section~\ref{sec:theory} presents the Boltzmann equation, the wave--particle
equations and their continuous spectrum, and Section~\ref{sec:algorithm}
constructs the coupled WP iteration. Section~\ref{sec:analysis} establishes the
equivalence of the coupled system with the Boltzmann equation, the
asymptotic-preserving continuum limit of the discrete fixed point and a bound
on the convergence factor of the iteration. The numerical tests are reported
in Section~\ref{sec:results}, and Section~\ref{sec:conclusion} concludes the
paper.

\section{Wave--Particle Multiscale Equations}\label{sec:theory}
\subsection{Boltzmann equation and fluid limits}\label{sec:gas_model}
We consider a monatomic dilute gas described by its mass distribution
$f(\bx,\bxi,t)$, a function of position $\bx$, molecular velocity $\bxi$ and
time $t$. Lengths, densities and speeds are scaled by reference values
$L_0$, $\rho_0$ and $c_0$, and time, temperature, viscosity and distribution
by $L_0/c_0$, $c_0^2/R_g$, $\rho_0c_0L_0$ and $\rho_0/c_0^3$, where $R_g$ is
the specific gas constant. With the Knudsen number
$\varepsilon\equiv\mathrm{Kn}>0$ and the normalized Boltzmann collision
operator $Q$, the governing equation reads
\begin{equation}\label{eq:boltzmann}
 \partial_t f+\bxi\cdot\nabla_{\bx}f
 =\frac{1}{\varepsilon}Q(f,f),
\end{equation}
where $\nabla_{\bx}$ is the spatial gradient. Velocity moments are written
$\avg{\cdot}=\int_{\mathbb R^3}\cdot\,\dd\bxi$, and
$\bpsi=(1,\bxi^T,|\bxi|^2/2)^T$ are the collision invariants, the superscript
$T$ denoting transpose. The conservative state is
\begin{equation}\label{eq:conservative_state}
 \bU=\avg{\bpsi f}=(\rho,\rho\bu^T,\rho E)^T,
\end{equation}
in which $\rho$, $\bu$ and $T$ are the density, bulk velocity and
temperature, $E=|\bu|^2/2+3T/2$ is the specific total energy and $p=\rho T$ is
the pressure. In terms of the peculiar velocity $\bc=\bxi-\bu$, the local
Maxwellian
\begin{equation}\label{eq:maxwellian}
 M=\frac{\rho}{(2\pi T)^{3/2}}\exp\!\left(-\frac{|\bc|^2}{2T}\right)
\end{equation}
has the same conservative moments as $f$.

For an elastic collision between molecules with velocities $\bxi$ and
$\bxi_*$, let $\bg=\bxi-\bxi_*$ be the relative velocity, $\bn$ the unit
vector on the sphere $S^2$ along the post-collision relative velocity, and
$\theta\in[0,\pi]$ the scattering angle, $\cos\theta=\bg\cdot\bn/|\bg|$.
The post-collision velocity of the first molecule is
\begin{equation}\label{eq:post_collision}
 \bxi'=\frac{\bxi+\bxi_*}{2}+\frac{|\bg|}{2}\bn,
\end{equation}
and momentum conservation gives that of its partner,
$\bxi'_*=\bxi+\bxi_*-\bxi'$. With the collision kernel $B$, the collision
operator is
\begin{equation}\label{eq:collision_integral}
 Q(f,f)=\int_{\mathbb R^3}\int_{S^2}B(|\bg|,\theta)
   \big(f'f'_*-ff_*\big)\,\dd\bn\,\dd\bxi_*,
\end{equation}
where $f$, $f_*$, $f'$ and $f'_*$ denote the distribution at $\bxi$,
$\bxi_*$, $\bxi'$ and $\bxi'_*$, all at the same position and time, and the
associated symmetric bilinear operator is its polarization
$Q(g,h)=\frac12\big[Q(g+h,g+h)-Q(g,g)-Q(h,h)\big]$. Collisions conserve mass,
momentum and energy,
$\avg{\bpsi Q(f,f)}=\boldsymbol{0}$. The kernel encodes the law of molecular
interaction; the decomposition developed below does not depend on it, and
the kernel used in the computations is specified in
Section~\ref{sec:test_gas_model}. Since the kernel does not depend on
$\varepsilon$, the factor $\varepsilon^{-1}$ in \eqref{eq:boltzmann} sets the
kinetic scale.

For a smooth flow close to local equilibrium, the Chapman--Enskog
expansion~\cite{chapman1970} writes
\begin{equation}\label{eq:boltzmann_ce_expansion}
 f=M+\varepsilon\,f^{(1)}+\mathcal O(\varepsilon^2),
\end{equation}
where the first-order correction carries no conservative moments,
$\avg{\bpsi f^{(1)}}=\boldsymbol 0$, and $\mathcal O(\varepsilon^2)$ denotes
terms bounded by a constant times $\varepsilon^2$ in the velocity-weighted
norms required for the moments below. Since $Q(M,M)=0$, the first order of
\eqref{eq:boltzmann} gives
\begin{equation}\label{eq:boltzmann_first_correction}
 L_Mf^{(1)}=\partial_t^{(0)}M+\bxi\cdot\nabla_{\bx}M,
\end{equation}
where $L_M\,\delta f=2Q(M,\delta f)$ is the linearization of $Q$ at $M$ acting
on a perturbation $\delta f$, and $\partial_t^{(0)}M$ is the time derivative
of $M$ evaluated with the leading-order conservation laws. Solvability
requires the right-hand side to have zero conservative moments, which yields
the Euler equations, and the moment constraint then selects $f^{(1)}$ on the
complement of the collision invariants.

The deviatoric stress of a distribution $Z$ is
\begin{equation}\label{eq:transport_moments}
 \boldsymbol\sigma_Z=
 \avg{\left(\bc\bc^T-\frac{|\bc|^2}{3}\mathbf I\right)Z},
\end{equation}
and its heat flux is
\begin{equation}\label{eq:heat_flux_moment}
 \bq_Z=\frac12\avg{|\bc|^2\bc\,Z},
\end{equation}
where $\mathbf I$ is the identity tensor and $\bc$ is formed with the bulk
velocity of $f$ for every distribution. For a rotationally invariant kernel,
the correction $f^{(1)}$ gives $\boldsymbol\sigma_f=\boldsymbol\sigma_{\rm NS}+\mathcal
O(\varepsilon^2)$ and $\bq_f=\bq_{\rm NS}+\mathcal O(\varepsilon^2)$, with
Newton's law
\begin{equation}\label{eq:boltzmann_transport_laws}
 \boldsymbol\sigma_{\rm NS}=
 -\mu\left[\nabla_{\bx}\bu+(\nabla_{\bx}\bu)^T
       -\frac23(\nabla_{\bx}\cdot\bu)\mathbf I\right]
\end{equation}
and Fourier's law
\begin{equation}\label{eq:fourier_law}
 \bq_{\rm NS}=-\kappa\nabla_{\bx}T.
\end{equation}
The viscosity $\mu=\varepsilon\mu_1$ and the heat conductivity
$\kappa=\varepsilon\kappa_1$ contain order-one coefficients $\mu_1$ and
$\kappa_1$ determined by $L_M$, and their ratio defines the Prandtl number
$\mathrm{Pr}=c_p\mu/\kappa$ with the specific heat $c_p=5/2$. Taking moments
of \eqref{eq:boltzmann} and dropping terms of order $\varepsilon^2$ yields the
Navier--Stokes system
\begin{equation}\label{eq:ns_system}
 \left\{
 \begin{aligned}
 &\partial_t\rho+\nabla_{\bx}\cdot(\rho\bu)=0,\\[4pt]
 &\partial_t(\rho\bu)+\nabla_{\bx}\cdot
       (\rho\bu\bu^T+p\mathbf I+\boldsymbol\sigma_{\rm NS})=\boldsymbol 0,\\[4pt]
 &\partial_t(\rho E)+\nabla_{\bx}\cdot
       \big[(\rho E+p)\bu+\boldsymbol\sigma_{\rm NS}\bu+\bq_{\rm NS}\big]=0,
 \end{aligned}
 \right.
\end{equation}
and setting $\boldsymbol\sigma_{\rm NS}$ and $\bq_{\rm NS}$ to zero gives the
Euler system. Both limits hold formally in smooth flow, under moment bounds
and the solvability of the linearized collision problem.

\subsection{Wave operator and wave--particle equations}\label{sec:wave_theory}
Let $\tau(\bx,t)>0$ be a velocity-independent relaxation time, whose choice
for the computations is given in Section~\ref{sec:test_gas_model}. The
identity transformation adds and subtracts the BGK relaxation
term~\cite{bgk1954} and writes the Boltzmann equation as
\begin{equation}\label{eq:collision_remainder}
 \partial_t f+\bxi\cdot\nabla_{\bx}f=\frac{M-f}{\tau}+R,
\end{equation}
with the collision remainder
\begin{equation}\label{eq:remainder_definition}
 R=\frac{1}{\varepsilon}Q(f,f)-\frac{M-f}{\tau}.
\end{equation}
Because $M$ and $f$ share their conservative moments and $\tau$ does not
depend on $\bxi$, the remainder satisfies $\avg{\bpsi R}=\boldsymbol{0}$, so that the relaxation term
and the remainder are separately conservative. Equation
\eqref{eq:collision_remainder} is an identity rather than a model: the
relaxation time parametrizes the decomposition below, but it does not enter
the Boltzmann equation or its solutions.

Along the characteristic through $(\bx,t)$, let $s\geq0$ be the flight time,
$(\bx_s,t_s)=(\bx-\bxi s,t-s)$ the characteristic point and $a_s$ the value
of any field $a$ at that point, so that $M_s$, $f_s$, $R_s$ and $\tau_s$
need no further arguments. Let $\ell(\bx,t)\geq0$ be the local kinetic
horizon, measured in time, with the subscript $\ell$ denoting evaluation at
$s=\ell$. The cumulative collision frequency
$\nu(s)=\int_0^s\tau_r^{-1}\,\dd r$, in which $r$ is a flight-time integration
variable, defines the collisionless factor $e^{-\nu(s)}$, the probability
that a molecule travels the flight time $s$ without a collision. Integrating
\eqref{eq:collision_remainder} along the characteristic with the integrating
factor $e^{\nu(s)}$ gives the characteristic integral solution
\begin{equation}\label{eq:integral_solution}
 f=\int_0^{\ell}\frac{M_s}{\tau_s}e^{-\nu(s)}\,\dd s
   +e^{-\nu(\ell)}f_{\ell}
   +\int_0^{\ell}R_s e^{-\nu(s)}\,\dd s .
\end{equation}
For a relaxation model, $R=0$, with a locally frozen relaxation time and a
horizon equal to the time elapsed within a step, this is the integral solution
from which the UGKS constructs its multiscale flux~\cite{xu2010,xu2014direct}
and on which the wave--particle split of the UGKWP method is
based~\cite{liu2020ugkwp}. Its first term, the equilibrium integral,
accumulates the local Maxwellian along the characteristic with the collision
probability density $e^{-\nu(s)}/\tau_s$ and dominates when the horizon is
long compared with the relaxation time. Its second term is the free transport
of the upstream distribution, the free-transport memory that survives beyond
the horizon, and dominates when the horizon is short. The wave is defined as
the equilibrium integral, which is the wave operator of Part~I with the local
Maxwellian as relaxation target,
\begin{equation}\label{eq:wave_operator}
 W=\int_0^{\ell}\frac{M_s}{\tau_s}e^{-\nu(s)}\,\dd s,
\end{equation}
and the particle is its exact complement,
\begin{equation}\label{eq:particle_operator}
 P=f-W=e^{-\nu(\ell)}f_{\ell}+\int_0^{\ell}R_s e^{-\nu(s)}\,\dd s .
\end{equation}
For a prescribed horizon and relaxation time, the wave depends only on the
macroscopic state along the backward characteristic segment of duration
$\ell$, truncated at the initial time or at an inflow boundary where data are
prescribed. The particle, in turn, contains the surviving upstream information
and the accumulated collision remainder. Unlike the particle of Part~I, which
is the collisionless fraction of a non-negative distribution, it carries the
signed remainder and need not be non-negative; non-negativity is required of
$f$ only.

For a locally frozen relaxation time let $\eta=\ell/\tau$ be the local
horizon-to-relaxation ratio, and let the second derivative of $M$ along each
characteristic be bounded within the horizon. A Taylor expansion of $M_s$ in
the wave operator then gives, uniformly in $\eta\geq0$,
\begin{equation}\label{eq:wave_expansion}
 W=A_0(\eta)M-A_1(\eta)\,\tau
       \left(\partial_t M+\bxi\cdot\nabla_{\bx}M\right)
       +\mathcal{O}(\tau^2),
\end{equation}
where $\mathcal O(\tau^2)$ denotes terms bounded by a constant times $\tau^2$
and the horizon weights are $A_0(\eta)=1-e^{-\eta}$ and
$A_1(\eta)=1-(1+\eta)e^{-\eta}$. These weights are the shares of the
Maxwellian and of its first-order relaxation response that the wave carries.
For $\eta\ll1$ they vanish as $A_0\sim\eta$ and $A_1\sim\eta^2/2$, so that
the wave carries little of the distribution and the transport is left to the
particle. For $\eta\gg1$ both weights approach one, and the particle retains
the fraction $1-A_0=e^{-\eta}$ of the Maxwellian, the fraction
$1-A_1=(1+\eta)e^{-\eta}$ of the first-order response and the accumulated
collision remainder.

Differentiating the wave operator along the characteristic
(Proposition~\ref{prop:wp_equivalence}) shows that the wave obeys a
relaxation equation with the horizon exchange
\begin{equation}\label{eq:endpoint_exchange}
 \mathcal S_{\ell}=\frac{M_{\ell}}{\tau_{\ell}}e^{-\nu(\ell)}
       \left(1-\partial_t\ell-\bxi\cdot\nabla_{\bx}\ell\right),
\end{equation}
which is the distribution transferred per unit time across the horizon and
corresponds to the local-horizon source of Part~I. Subtracting the wave
equation from the Boltzmann equation gives the particle equation, and the two
together form the microscopic wave--particle equations
\begin{equation}\label{eq:wp_kinetic_system}
 \left\{
 \begin{aligned}
 &\partial_tW+\bxi\cdot\nabla_{\bx}W
   =\frac{M-W}{\tau}-\mathcal S_{\ell},\\[6pt]
 &\partial_tP+\bxi\cdot\nabla_{\bx}P
   =\frac{1}{\varepsilon}Q(W+P,W+P)
    -\frac{M-W}{\tau}+\mathcal S_{\ell}.
 \end{aligned}
 \right.
\end{equation}
They are posed in phase space and coupled through $f=W+P$,
$\bU=\avg{\bpsi f}$ and $M$. Their exchange terms cancel in the sum, so that
the full collision operator is retained for every horizon, including the
wave--wave, cross and particle--particle collisions in
$Q(W+P,W+P)=Q(W,W)+2Q(W,P)+Q(P,P)$.

The macroscopic equations are the velocity moments of
\eqref{eq:wp_kinetic_system}, and for each component they take the form of
extended Navier--Stokes equations. For $Z=W,P$, let $\rho_Z=\avg{Z}$ be the
density, $\bj_Z=\avg{\bc\,Z}$ the diffusion flux relative to the bulk velocity
$\bu$ of $f$, $(\rho E)_Z=\avg{|\bxi|^2Z}/2$ the total energy and
$p_Z=\avg{|\bc|^2Z}/3$ the pressure of the component, so that its
conservative state is
$\bU_Z=\avg{\bpsi Z}=\big(\rho_Z,(\rho_Z\bu+\bj_Z)^T,(\rho E)_Z\big)^T$.
Combined with the diffusion flux, the stress \eqref{eq:transport_moments} and
the heat flux \eqref{eq:heat_flux_moment} of the component define the
generalized stress
\begin{equation}\label{eq:generalized_stress}
 \widetilde{\boldsymbol\sigma}_Z=\boldsymbol\sigma_Z+\bu\bj_Z^T+\bj_Z\bu^T
\end{equation}
and the generalized heat flux
\begin{equation}\label{eq:generalized_heat_flux}
 \widetilde\bq_Z=\bq_Z-(\bu\cdot\bj_Z)\,\bu-\frac12|\bu|^2\bj_Z,
\end{equation}
which reduce to $\boldsymbol\sigma_Z$ and $\bq_Z$ when the component carries
no diffusion flux. The relaxation and exchange terms of
\eqref{eq:wp_kinetic_system} transfer conservative moments from the particle
to the wave at the rate
\begin{equation}\label{eq:exchange_source}
 \boldsymbol{\mathcal S}=\big(s_\rho,\boldsymbol s_m^T,s_E\big)^T
 =\frac{\bU_P}{\tau}-\avg{\bpsi\,\mathcal S_\ell},
\end{equation}
in which relaxation converts particle into wave at the rate $1/\tau$ and the
horizon exchange returns the part that leaves the horizon. In terms of these
quantities, the moments of \eqref{eq:wp_kinetic_system} are the extended
Navier--Stokes equations
\begin{equation}\label{eq:wp_moment_system}
 \left\{
 \begin{aligned}
 &\partial_t\rho_Z+\nabla_{\bx}\cdot\big(\rho_Z\bu+\bj_Z\big)=\pm\,s_\rho,\\[4pt]
 &\partial_t\big(\rho_Z\bu+\bj_Z\big)+\nabla_{\bx}\cdot\big(\rho_Z\bu\bu^T
   +p_Z\mathbf I+\widetilde{\boldsymbol\sigma}_Z\big)=\pm\,\boldsymbol s_m,\\[4pt]
 &\partial_t(\rho E)_Z+\nabla_{\bx}\cdot\big[\big((\rho E)_Z+p_Z\big)\bu
   +\widetilde{\boldsymbol\sigma}_Z\bu+\widetilde\bq_Z\big]=\pm\,s_E,
 \end{aligned}
 \right.
\end{equation}
where the upper sign gives the macroscopic wave equations, $Z=W$, and the
lower sign the macroscopic particle equations, $Z=P$. The collision operator
does not appear, being conservative. In contrast to the Navier--Stokes system
\eqref{eq:ns_system}, each component has its own density, pressure and
diffusion flux, and the two components exchange mass, momentum and energy
through $\boldsymbol{\mathcal S}$. Since $\rho_W+\rho_P=\rho$,
$\bj_W+\bj_P=\boldsymbol0$, $p_W+p_P=p$ and $(\rho E)_W+(\rho E)_P=\rho E$,
the diffusion corrections in \eqref{eq:generalized_stress} and
\eqref{eq:generalized_heat_flux} and the exchange cancel in the sum, which is
the source-free total conservation law
\begin{equation}\label{eq:total_conservation}
 \left\{
 \begin{aligned}
 &\partial_t\rho+\nabla_{\bx}\cdot(\rho\bu)=0,\\[4pt]
 &\partial_t(\rho\bu)+\nabla_{\bx}\cdot\big(\rho\bu\bu^T+p\mathbf I
   +\boldsymbol\sigma_W+\boldsymbol\sigma_P\big)=\boldsymbol0,\\[4pt]
 &\partial_t(\rho E)+\nabla_{\bx}\cdot\big[(\rho E+p)\bu
   +(\boldsymbol\sigma_W+\boldsymbol\sigma_P)\bu+\bq_W+\bq_P\big]=0.
 \end{aligned}
 \right.
\end{equation}
Written compactly as
$\partial_t\bU+\nabla_{\bx}\cdot(\bF_W+\bF_P)=\boldsymbol0$, this law has
the structure of the Navier--Stokes system \eqref{eq:ns_system} without any
constitutive relation, for its stress and heat flux are exact moments of the
wave and the particle.

\subsection{Closed system and continuous spectrum}\label{sec:wave_scales}
Since the wave is determined by the conservative variables along the
characteristics, the total conservation law and the particle equation form a
closed system for $\bU$ and $P$,
\begin{equation}\label{eq:up_system}
 \left\{
 \begin{aligned}
 &\partial_t\rho+\nabla_{\bx}\cdot(\rho\bu)=0,\\[4pt]
 &\partial_t(\rho\bu)+\nabla_{\bx}\cdot\big(\rho\bu\bu^T+p\mathbf I
   +\boldsymbol\sigma_W+\boldsymbol\sigma_P\big)=\boldsymbol0,\\[4pt]
 &\partial_t(\rho E)+\nabla_{\bx}\cdot\big[(\rho E+p)\bu
   +(\boldsymbol\sigma_W+\boldsymbol\sigma_P)\bu+\bq_W+\bq_P\big]=0,\\[4pt]
 &\partial_tP+\bxi\cdot\nabla_{\bx}P=R+\frac{P_e-P}{\tau},
 \end{aligned}
 \right.
\end{equation}
where the collision remainder \eqref{eq:remainder_definition} and the endpoint
distribution are
\begin{equation}\label{eq:up_closure}
 \begin{aligned}
 &R=\frac{1}{\varepsilon}Q(f,f)-\frac{M-f}{\tau},\\[4pt]
 &P_e=\tau\mathcal S_\ell=\frac{\tau}{\tau_\ell}\,M_\ell\,e^{-\nu(\ell)}
   \left(1-\partial_t\ell-\bxi\cdot\nabla_{\bx}\ell\right),
 \end{aligned}
\end{equation}
$W$ follows from $\bU$ by \eqref{eq:wave_operator} with $M$ the
Maxwellian of $\bU$, and $f=W+P$. The last equation of \eqref{eq:up_system}
is the particle equation of \eqref{eq:wp_kinetic_system} written with the
remainder: the particle carries the collision remainder and relaxes towards
the endpoint distribution at the rate $1/\tau$. We call the system
\eqref{eq:up_system}, on which the method is built, the wave--particle
multiscale equations. For compatible initial and boundary data it is
equivalent to the Boltzmann equation for every horizon and preserves
$\bU=\avg{\bpsi(W+P)}$ (Proposition~\ref{prop:wp_equivalence}). The
horizon-to-relaxation ratio $\eta=\ell/\tau$ therefore parametrizes a
continuous spectrum of equivalent representations of the Boltzmann dynamics.
At $\eta=0$ the wave vanishes, $P_e=M$, and the particle equation is the
Boltzmann equation itself, whereas as $\eta\to\infty$ the free-transport term
disappears from the particle, which then retains only the accumulated
collision remainder.

The particle stress $\boldsymbol\sigma_P$ and heat flux $\bq_P$ in
\eqref{eq:up_system} are moments of the unknown $P$, whereas those of the wave
are functionals of $\bU$. For a locally frozen relaxation time, and with the
time derivative of $M$ taken from the Euler equations, the expansion
\eqref{eq:wave_expansion} gives them explicitly. The wave carries the fraction
$A_0$ of the conservative state,
\begin{equation}\label{eq:wave_state}
 \bU_W=A_0(\eta)\,\bU+\mathcal O(\tau^2),
\end{equation}
so that $\bj_W=\boldsymbol0$, $p_W=A_0p$,
$\widetilde{\boldsymbol\sigma}_W=\boldsymbol\sigma_W$ and
$\widetilde\bq_W=\bq_W$ at this order, and the fraction $A_1$ of the
relaxation stress and heat flux,
\begin{equation}\label{eq:wave_stress}
 \boldsymbol\sigma_W=-A_1(\eta)\,\tau p\left[\nabla_{\bx}\bu+(\nabla_{\bx}\bu)^T
       -\frac23(\nabla_{\bx}\cdot\bu)\mathbf I\right]+\mathcal O(\tau^2),
\end{equation}
\begin{equation}\label{eq:wave_heat_flux}
 \bq_W=-A_1(\eta)\,c_p\,\tau p\,\nabla_{\bx}T+\mathcal O(\tau^2).
\end{equation}
Thus $\boldsymbol\sigma_W=A_1\boldsymbol\sigma_\tau$ and
$\bq_W=A_1\bq_\tau$, where $\boldsymbol\sigma_\tau$ and $\bq_\tau$ are the
laws \eqref{eq:boltzmann_transport_laws} and \eqref{eq:fourier_law} with the
relaxation viscosity $\tau p$ and heat conductivity $c_p\tau p$, whose
Prandtl number is one.

In the continuum limit we write $\tau=\varepsilon\bar\tau$ with an order-one
coefficient $\bar\tau(\bx,t)$ and let the horizon be bounded below
independently of $\varepsilon$, so that $\eta\to\infty$ and the equilibrium
integral beyond the horizon becomes negligible. The wave then takes the
Chapman--Enskog form
\begin{equation}\label{eq:wave_ce}
 W=M+\varepsilon\,W^{(1)}+\mathcal O(\varepsilon^2),
\end{equation}
with the first-order correction
\begin{equation}\label{eq:wave_first_order}
 W^{(1)}=-\bar\tau\left(\partial_t^{(0)}M+\bxi\cdot\nabla_{\bx}M\right),
\end{equation}
and subtraction from \eqref{eq:boltzmann_ce_expansion} shows that the
particle is of first order,
\begin{equation}\label{eq:particle_ce}
 P=\varepsilon\big(f^{(1)}-W^{(1)}\big)+\mathcal O(\varepsilon^2).
\end{equation}
The total flux is therefore the Euler flux plus a first-order correction,
\begin{equation}\label{eq:total_flux_ce}
 \bF_W+\bF_P=\bF_E+\varepsilon\,\bF^{(1)}+\mathcal O(\varepsilon^2),
\end{equation}
with the Boltzmann flux correction $\bF^{(1)}=\avg{\bpsi\bxi^Tf^{(1)}}$ and
the Euler flux tensor
\begin{equation}\label{eq:euler_flux}
 \bF_E=\avg{\bpsi\bxi^TM}=
 \begin{pmatrix}
 \rho\bu^T\\[2pt]
 \rho\bu\bu^T+p\mathbf I\\[2pt]
 (\rho E+p)\bu^T
 \end{pmatrix}.
\end{equation}
The conservation law in \eqref{eq:up_system} accordingly reduces to the Euler
equations at leading order and to the Navier--Stokes equations
\eqref{eq:ns_system}, with the Boltzmann stress and heat flux, at first order.
Since $A_1\to1$, the particle stress and heat flux are, at first order, the
differences $\boldsymbol\sigma_{\rm NS}-\boldsymbol\sigma_\tau$ and
$\bq_{\rm NS}-\bq_\tau$: the wave carries the Euler flux and the
unit-Prandtl-number relaxation part of the viscous flux, and the particle
supplies the remaining Boltzmann transport. For the relaxation time
$\tau=\mu/p$ of Section~\ref{sec:test_gas_model}, whose viscosity is that of
the collision kernel, $\boldsymbol\sigma_\tau=\boldsymbol\sigma_{\rm NS}$.
The particle stress then vanishes at first order, and the particle heat flux
\begin{equation}\label{eq:particle_heat_flux}
 \bq_P=-\big(\kappa-c_p\mu\big)\nabla_{\bx}T+\mathcal O(\varepsilon^2)
\end{equation}
corrects the unit Prandtl number of the relaxation to the Prandtl number of
the Boltzmann operator.

For a finite, locally frozen ratio $\eta$ the wave carries only a part of
these terms. The expansion \eqref{eq:wave_expansion}, with the time derivative
of $M$ taken from the Euler equations, and the expansion
\eqref{eq:boltzmann_ce_expansion} give
\begin{equation}\label{eq:wp_spectrum}
 \begin{aligned}
 &W=A_0M+\varepsilon A_1W^{(1)}+\mathcal O(\varepsilon^2),\\[4pt]
 &P=e^{-\eta}M+\varepsilon\big[(1+\eta)e^{-\eta}\,W^{(1)}+f^{(1)}-W^{(1)}\big]
   +\mathcal O(\varepsilon^2),
 \end{aligned}
\end{equation}
so that the particle carries the complementary state $\bU_P=e^{-\eta}\bU$,
and its stress and heat flux are, at first order,
\begin{equation}\label{eq:particle_spectrum}
 \begin{aligned}
 &\boldsymbol\sigma_P=(1+\eta)e^{-\eta}\,\boldsymbol\sigma_\tau
   +\boldsymbol\sigma_{\rm NS}-\boldsymbol\sigma_\tau,\\[4pt]
 &\bq_P=(1+\eta)e^{-\eta}\,\bq_\tau+\bq_{\rm NS}-\bq_\tau .
 \end{aligned}
\end{equation}
Hence $\boldsymbol\sigma_W+\boldsymbol\sigma_P=\boldsymbol\sigma_{\rm NS}$ and
$\bq_W+\bq_P=\bq_{\rm NS}$ for every locally frozen $\eta$. The conservation
law in \eqref{eq:up_system} is the Navier--Stokes system whatever the
horizon, and only the division of the fluxes between the two components
changes along the spectrum. As $\eta$ increases, the wave takes over the
conservative state with the weight $A_0$ and the relaxation stress and heat
flux with the weight $A_1$, while the particle share $(1+\eta)e^{-\eta}$ of
the relaxation transport decays exponentially. What remains in the particle
is the correction of the transport coefficients, which does not depend on
$\eta$ and, for $\tau=\mu/p$, is the heat flux
\eqref{eq:particle_heat_flux}. The discrete counterpart of this structure is
established in Theorem~\ref{thm:ap}.

The two parts of \eqref{eq:up_system} interact only through the particle
stress and heat flux in the conservation law and through the wave, the
Maxwellian and the endpoint in the particle equation. A steady solution can
therefore be computed by alternating two solves: a macroscopic solve of the
conservation law for $\bU$ with $P$ fixed, which requires no collision
integral, and a kinetic solve of the particle equation for $P$ with $W$
fixed, in which the change of the endpoint drives the particle towards the new
macroscopic state. The conventional iterative scheme relaxes the conservative
modes only through transport and, on a fixed mesh, needs a number of
iterations of order $\varepsilon^{-2}$. In the coupled WP iteration the
macroscopic solve treats the conservative modes directly, and in the linear
model of Section~\ref{sec:acceleration} the error passed to the next
iteration is controlled by the particle share $(1+\eta)e^{-\eta}$ of the
transport, which decays exponentially as the horizon grows relative to the
relaxation time. Section~\ref{sec:algorithm} constructs this iteration, and
Theorem~\ref{thm:acceleration} bounds its convergence factor.

\section{Coupled WP Iteration}\label{sec:algorithm}
The steady computation takes the total conservative state $\bU$ and the
particle distribution $P$ as unknowns and combines two iterations built on
system \eqref{eq:up_system}. The wave is never an independent unknown; it is
always evaluated from the conservative state, and the total distribution is
$f=W+P$. The W iteration updates $\bU$ with $P$ fixed, so that only the wave
and its flux follow the macroscopic update. The P iteration then updates $P$
with the wave fixed at its value for the predicted state, driven by the
particle residual and by the endpoint traction of the macroscopic
prediction. Both residuals are evaluated on the accepted total distribution,
whereas all preconditioning quantities multiply increments only. Both
components are reconstructed to second order: the wave flux is the
gas-kinetic flux of the reconstructed $\bU$, and the particle flux is a
second-order upwind discrete-ordinate flux of MUSCL type for the
reconstructed $P$. The reconstruction and the fluxes are given in
Section~\ref{sec:alg_fluxes} and the discrete residuals in
Section~\ref{sec:alg_residuals}. Sections~\ref{sec:alg_w}
and~\ref{sec:alg_p} describe the W and P iterations, and
Section~\ref{sec:alg_wp} combines them into the coupled WP iteration,
whose name joins its macroscopic (W) and kinetic (P) iterations.

\subsection{Reconstruction and fluxes}\label{sec:alg_fluxes}
Physical space is divided into finite-volume cells $\Omega_i$ of volume
$|\Omega_i|$ and centre $\bx_i$. A face $\Gamma\subset\partial\Omega_i$ has area
$|\Gamma|$, midpoint $\bx_\Gamma$, outward unit normal $\bn_\Gamma$ and, in the plane
of the two-dimensional computations, unit tangent $\bt_\Gamma$, and $i_\Gamma$
denotes the neighbour of cell $i$ across it. The cell size in mesh direction
$d$ is $\Delta_{d,i}$, and $\Delta_i=\min_d\Delta_{d,i}$. Velocity space is
represented by nodes $\bxi_j$ with quadrature weights $\varpi_j$, and the
discrete conservative state is $\bU_i=\sum_j\varpi_j\bpsi_jf_{i,j}$ with
$\bpsi_j=\bpsi(\bxi_j)$. The kinetic horizon is taken as a cell-crossing time,
\begin{equation}\label{eq:discrete_horizon}
 \ell_i=C_\ell\,\frac{\Delta_i}{|\bu_i|+a_i},
\end{equation}
where $a_i$ is the sound speed and $C_\ell\geq0$ the horizon number. With the
relaxation time $\tau_i$ of \eqref{eq:reference_viscosity}, the
horizon-to-relaxation ratio $\eta_i=\ell_i/\tau_i$ compares the mesh scale
with the local mean free path. At a face, $\ell_\Gamma$ is the
distance-weighted average of the horizons of the two adjacent cells,
$\tau_\Gamma$ is the relaxation time of the face state, and
$\eta_\Gamma=\ell_\Gamma/\tau_\Gamma$.

Both components are reconstructed to second order with van Leer-limited
slopes~\cite{vanleer1974}. For a quantity with cell values $q_i$, a component
of the conservative state $\bU_i$ or an ordinate value $P_{i,j}$ of the
particle, the slope in the mesh direction $d$ is the limited combination
\begin{equation}\label{eq:vanleer_slope}
 (\partial_dq)_i=\big[\mathrm{sgn}(s^-)+\mathrm{sgn}(s^+)\big]
 \frac{|s^-|\,|s^+|}{|s^-|+|s^+|}
\end{equation}
of the one-sided differences $s^\mp$ of $q$ towards the neighbours $i\mp$ of
cell $i$ in direction $d$. This slope vanishes at local extrema and is
second-order accurate where $q$ is smooth and monotone, and the slopes form
the cell gradients $\nabla_hq_i$ and $\nabla_hP_{i,j}$. At face $\Gamma$ a
component of the conservative state is the average of the values
extrapolated from the two adjacent cells,
\begin{equation}\label{eq:face_reconstruction}
 q_\Gamma=\frac12\big[q_i+(\bx_\Gamma-\bx_i)\cdot\nabla_hq_i\big]
    +\frac12\big[q_{i_\Gamma}+(\bx_\Gamma-\bx_{i_\Gamma})\cdot\nabla_hq_{i_\Gamma}\big],
\end{equation}
its normal derivative is the central difference of the two cell values,
\begin{equation}\label{eq:face_normal_derivative}
 \partial_nq_\Gamma=\frac{q_{i_\Gamma}-q_i}{(\bx_{i_\Gamma}-\bx_i)\cdot\bn_\Gamma},
\end{equation}
and its tangential derivative
$\partial_\parallel q_\Gamma=(q_{\Gamma+}-q_{\Gamma-})/|\bx_{\Gamma+}-\bx_{\Gamma-}|$ is the
central difference of the values on the two neighbouring faces $\Gamma\pm$
parallel to $\Gamma$. Together, the face data $q_\Gamma$, $\partial_nq_\Gamma$ and
$\partial_\parallel q_\Gamma$ define the linear distribution
$q_\Gamma+(\bx-\bx_\Gamma)\cdot(\partial_nq_\Gamma\,\bn_\Gamma+\partial_\parallel q_\Gamma\,\bt_\Gamma)$
around the face midpoint, which is second-order accurate for smooth $q$. In
the Reynolds-number cavity flows of Section~\ref{sec:results}, which are
smooth, both components are instead reconstructed with sixth-order central
stencils, one-sided next to the walls.

Derivatives of a Maxwellian are expressed by micro slopes. For a state
with Maxwellian $M$, the micro slope of a derivative $\partial_d\bU$ is the
linear combination $a_d=\boldsymbol\alpha_d\cdot\bpsi$ of the collision
invariants with $\avg{\bpsi\,a_dM}=\partial_d\bU$, so that $\partial_dM=a_dM$.
The time slope $A=\boldsymbol\alpha_t\cdot\bpsi$ follows from the
compatibility condition of the gas-kinetic scheme~\cite{xu2001},
\begin{equation}\label{eq:compatibility}
 \avg{\bpsi\,AM}=-\sum_d\avg{\bpsi\,\xi_d\,a_dM},
\end{equation}
which makes $\avg{\bpsi AM}$ the Euler time derivative of $\bU$, and the
material derivative of the Maxwellian becomes
$\partial_tM+\bxi\cdot\nabla_{\bx}M=\big(A+\sum_d\xi_da_d\big)M$. The cell
wave is evaluated from $\bU$ alone as the first-order truncation of
\eqref{eq:wave_expansion},
\begin{equation}\label{eq:discrete_wave}
 W_i=A_0(\eta_i)M_i-A_1(\eta_i)\,\tau_i
 \big(\partial_tM+\bxi\cdot\nabla_{\bx}M\big)^h_i,
\end{equation}
where $M_i$ is the discrete Maxwellian whose quadrature moments equal
$\bU_i$, and the material derivative is formed with the micro slopes of
$\nabla_h\bU_i$. Since the moments in these slopes are evaluated by the
velocity quadrature,
$\sum_j\varpi_j\bpsi_j\big(\partial_tM+\bxi\cdot\nabla_{\bx}M\big)^h_{i,j}
=\boldsymbol0$ and $\sum_j\varpi_j\bpsi_jW_{i,j}=A_0(\eta_i)\bU_i$, which is
the discrete counterpart of \eqref{eq:wave_state}.

The wave flux is the second-order gas-kinetic flux~\cite{xu2001} weighted by
the horizon. It is built from the face state $\bU_\Gamma$ and its derivatives
$\partial_n\bU_\Gamma$ and $\partial_\parallel\bU_\Gamma$ given by
\eqref{eq:face_reconstruction} and \eqref{eq:face_normal_derivative}. Let
$g_\Gamma$ be the Maxwellian of $\bU_\Gamma$, $\xi_n=\bxi\cdot\bn_\Gamma$ and
$\xi_\parallel=\bxi\cdot\bt_\Gamma$ the normal and tangential molecular
velocities, $a_n$ and $a_\parallel$ the micro slopes of the two derivatives,
$A$ the time slope from \eqref{eq:compatibility}, and
$A_{0,\Gamma}=A_0(\eta_\Gamma)$, $A_{1,\Gamma}=A_1(\eta_\Gamma)$. The velocity-resolved
wave flux of face $\Gamma$ is
\begin{equation}\label{eq:gks_wave_flux}
 \phi_{W,\Gamma}(\bxi)=\xi_n\,g_\Gamma\Big[A_{0,\Gamma}
   -A_{1,\Gamma}\,\tau_\Gamma\big(\xi_na_n+\xi_\parallel a_\parallel\big)
   +\Big(\frac{A_{0,\Gamma}\,\Delta t}{2}-A_{1,\Gamma}\,\tau_\Gamma\Big)A\Big].
\end{equation}
It is the average over the interval $\Delta t$ of the first-order
Chapman--Enskog solution
$g_\Gamma\big[1-\tau_\Gamma(\xi_na_n+\xi_\parallel a_\parallel+A)+tA\big]$ of
the relaxation model, with the Maxwellian part weighted by $A_{0,\Gamma}$ and the
relaxation part by $A_{1,\Gamma}$, and for $A_{0,\Gamma}=A_{1,\Gamma}=1$ it reduces to
the gas-kinetic Navier--Stokes flux. The averaging interval is set by the
explicit stability limit of the velocity grid,
$\Delta t=C_{\rm exp}\big/\max_i\sum_d\xi_{\max}/\Delta_{d,i}$, with
$\xi_{\max}$ the largest molecular speed of the grid and
$C_{\rm exp}\in(0,1)$. The moments of \eqref{eq:gks_wave_flux} are evaluated
analytically, and the Chapman--Enskog structure gives the macroscopic wave
flux in the explicit form
\begin{equation}\label{eq:gks_wave_moments}
 \bF^h_{W,\Gamma}=A_{0,\Gamma}\Big[\bF_E(\bU_\Gamma)
   +\frac{\Delta t}{2}\,\partial_t\bF_E(\bU_\Gamma)\Big]\bn_\Gamma
   +A_{1,\Gamma}\,\bF_R(\bU_\Gamma)\,\bn_\Gamma,
\end{equation}
where $\bF_E$ is the Euler flux tensor \eqref{eq:euler_flux},
$\partial_t\bF_E=\avg{\bpsi\bxi^TA\,g_\Gamma}$ is its Euler time derivative,
and
\begin{equation}\label{eq:viscous_flux}
 \bF_R=
 \begin{pmatrix}
 \boldsymbol0^T\\[2pt]
 \boldsymbol\sigma_\tau\\[2pt]
 (\boldsymbol\sigma_\tau\bu+\bq_\tau)^T
 \end{pmatrix}
\end{equation}
is the relaxation flux tensor, the viscous flux with the stress and heat flux
of \eqref{eq:wave_stress} and \eqref{eq:wave_heat_flux} at unit weight,
evaluated with $\tau_\Gamma$ and the gradients of $\bu$ and $T$ implied by the
face derivatives. The wave flux thus contains the Euler flux weighted by
$A_{0,\Gamma}$ and the Navier--Stokes viscous and heat fluxes, with viscosity
$\tau_\Gamma p_\Gamma$ and conductivity $c_p\tau_\Gamma p_\Gamma$, weighted by $A_{1,\Gamma}$.

The particle flux is the second-order upwind discrete-ordinate flux in MUSCL
form~\cite{vanleer1979}. For each node $j$ the particle is reconstructed
linearly in the upwind cell $c$ of face $\Gamma$, with $c=i$ for
$\bxi_j\cdot\bn_\Gamma>0$ and $c=i_\Gamma$ otherwise, from the limited slopes
\eqref{eq:vanleer_slope}. Its free transport along $\bxi_j$, averaged over
the same interval $\Delta t$ as the wave flux, gives
\begin{equation}\label{eq:particle_flux}
 \phi_{P,\Gamma,j}=\xi_{n,j}\Big[P_{c,j}+(\bx_\Gamma-\bx_c)\cdot\nabla_hP_{c,j}
 -\frac{\Delta t}{2}\,\xi_{n,j}\,\partial_nP_{c,j}\Big],
\end{equation}
with $\xi_{n,j}=\bxi_j\cdot\bn_\Gamma$ and
$\partial_nP_{c,j}=\bn_\Gamma\cdot\nabla_hP_{c,j}$. The term proportional to
$\Delta t$ shifts the face value upstream along $\bxi_j$ by half the
averaging interval and has the same origin as the time average in the wave
flux \eqref{eq:gks_wave_flux}, while the limited upwind reconstruction
supplies the dissipation that stabilizes the transport. The quadrature
moments of \eqref{eq:particle_flux} form the macroscopic particle flux
$\bF^h_{P,\Gamma}=\sum_j\varpi_j\bpsi_j\phi_{P,\Gamma,j}$.

At a boundary face the particle is extrapolated from the interior cell. At a
diffuse wall with velocity $\bu_{\rm w}$ and temperature $T_{\rm w}$, the
distributions entering the domain are Maxwellians of
$(\bu_{\rm w},T_{\rm w})$ whose densities are chosen separately for the wave
and the particle, so that the normal mass flux of each component through the
wall vanishes.

The particle equation requires the wave flux at the velocity nodes. The
nodal values of \eqref{eq:gks_wave_flux} are corrected by a
Maxwellian-weighted combination of the collision invariants,
\begin{equation}\label{eq:wave_flux_lifting}
 \phi^h_{W,\Gamma,j}=\phi_{W,\Gamma}(\bxi_j)+g_\Gamma(\bxi_j)\,\boldsymbol\lambda_\Gamma\cdot\bpsi_j,
\end{equation}
where the five coefficients $\boldsymbol\lambda_\Gamma$ solve the linear system
$\sum_j\varpi_j\bpsi_j\phi^h_{W,\Gamma,j}=\bF^h_{W,\Gamma}$, so that the correction
removes the quadrature error of the velocity grid. The cell divergence of the
velocity-resolved wave flux is
\begin{equation}\label{eq:cell_divergence}
 D^h_W(W)_{i,j}=\frac{1}{|\Omega_i|}\sum_{\Gamma\subset\partial\Omega_i}|\Gamma|\,
 \phi^h_{W,\Gamma,j},
\end{equation}
and $D^h_P(P)$ and the finite-volume divergence $\nabla_h\cdot$ of face
moment fluxes are formed in the same way from $\phi_{P,\Gamma,j}$ and from
$\bF^h_{W,\Gamma}$ and $\bF^h_{P,\Gamma}$. By construction, at interior and boundary
faces alike,
\begin{equation}\label{eq:flux_compatibility}
 \sum_j\varpi_j\bpsi_j\big[D^h_W(W)+D^h_P(P)\big]_{i,j}
 =\nabla_h\cdot\big(\bF^h_W+\bF^h_P\big)\big|_i .
\end{equation}

\subsection{Discrete residuals}\label{sec:alg_residuals}
The collision term is evaluated by the fast spectral method
(FSM)~\cite{mouhot2006,wu2013}, which writes the collision integral on a
periodized velocity box as a sum of convolutions in Fourier space at a cost
of $\mathcal O(N_{\rm ang}^2N_v\log N_v)$ operations for $N_v$ velocity nodes
and $N_{\rm ang}^2$ angular sections. Let $\widehat Q^h(f,f)$ denote its
output after a Maxwellian-weighted projection that restores the discrete
collision invariants, $\sum_j\varpi_j\bpsi_j\widehat Q^h_j=\boldsymbol0$. The
collision term used in the residuals is
\begin{equation}\label{eq:fsm_equilibrium}
 Q^h_{\rm FSM}(f,f)=\widehat Q^h(f,f)-\widehat Q^h(M^h,M^h),
\end{equation}
where $M^h$ is the discrete Maxwellian of $f$. This term is conservative and
has the discrete Maxwellian as an exact equilibrium,
$Q^h_{\rm FSM}(M^h,M^h)=0$; it contains the factor $\varepsilon^{-1}$ of
\eqref{eq:boltzmann}, and the same evaluation provides the discrete loss rate
$\Lambda_{i,j}$, the loss part of the collision term divided by $f_{i,j}$.

The discrete counterpart of the endpoint $P_e=\tau\mathcal S_\ell$ is
\begin{equation}\label{eq:discrete_endpoint}
 P_e^{h}=M-W-\tau D^h_W(W),
\end{equation}
which satisfies $(M-W-P_e^{h})/\tau=D^h_W(W)$ identically, and the discrete
collision remainder, the counterpart of \eqref{eq:remainder_definition}, is
\begin{equation}\label{eq:discrete_remainder}
 R^h=Q^h_{\rm FSM}(f,f)-\frac{M-f}{\tau}.
\end{equation}
The physical residuals of an accepted state $f=W+P$ are the macroscopic
residual
\begin{equation}\label{eq:macro_residual}
 R_U=-\nabla_h\cdot\big(\bF^h_W+\bF^h_P\big)
\end{equation}
and the particle residual
\begin{equation}\label{eq:discrete_residuals}
 R_P=R^h+\frac{P_e^h-P}{\tau}-D^h_P(P).
\end{equation}
The particle residual thus collects the collision remainder and the
relaxation of the particle towards the endpoint at the rate $1/\tau$, less the
particle transport. It is the steady discrete form of the particle equation
in \eqref{eq:up_system}, with $R^h$ carrying the whole non-relaxation part of
the Boltzmann collision term. Substituting \eqref{eq:discrete_endpoint} gives
the equivalent form
\begin{equation}\label{eq:discrete_residual_direct}
 R_P=Q^h_{\rm FSM}(f,f)-D^h_W(W)-D^h_P(P),
\end{equation}
in which the residual is evaluated. Since the endpoint and the residual use
the same wave divergence, the endpoint introduces no wave-truncation defect
into the stationary equation. The relaxation towards the endpoint is part of
the discrete equation and is balanced, not removed, at a stationary state;
the endpoint traction of Section~\ref{sec:alg_p}, by contrast, acts on the
change of the endpoint with the preconditioning rate $1/\widehat\tau$ and
vanishes at convergence. The conservation of $Q^h_{\rm FSM}$ and
\eqref{eq:flux_compatibility} imply that the quadrature moments of $R_P$
equal $R_U$, so that the macroscopic residual is the conservative part of the
particle residual. A superscript $k$ marks quantities evaluated on the
accepted state of outer iteration $k$.

\subsection{Macroscopic W iteration}\label{sec:alg_w}
The W iteration solves the total conservation law for the conservative state
with the particle distribution fixed at $P^k$, so that the particle flux
keeps its accepted value $\bF^{h,k}_P$ and only the wave follows the update
of $\bU$. Its $n$th update, $n=1,\dots,N_U$, starts from $\bU^{(n-1)}$, with
$\bU^{(0)}=\bU^k$, and from the residual
\begin{equation}\label{eq:w_residual}
 R_U^{(n-1)}=-\nabla_h\cdot\big(\bF^h_W(\bU^{(n-1)})+\bF^{h,k}_P\big),
\end{equation}
in which only the wave flux \eqref{eq:gks_wave_moments} is recomputed. As in
the implicit unified gas-kinetic scheme~\cite{zhu2016implicit}, the
conservation law is advanced in pseudo-time by the backward Euler method in
delta form,
\begin{equation}\label{eq:macro_delta_form}
 \frac{|\Omega_i|}{\Delta t_{\rm num}}\Delta\bU_i
 +\sum_{\Gamma\subset\partial\Omega_i}|\Gamma|\,\Delta\bF_\Gamma=|\Omega_i|\,R^{(n-1)}_{U,i},
\end{equation}
where $\Delta\bF_\Gamma$ is the increment of the normal flux through face
$\Gamma$ and the pseudo-time step is set by the implicit CFL number
$C_{\rm imp}$,
\begin{equation}\label{eq:dt_num}
 \Delta t_{\rm num}=\frac{C_{\rm imp}}
 {\max_i\sum_d\big(|u_{d,i}|+a_i\big)/\Delta_{d,i}},
\end{equation}
with $u_{d,i}$ the component of $\bu_i$ in mesh direction $d$. The
right-hand side is the physical residual and determines the converged state,
whereas the flux increment affects only the path of the iteration and is
therefore approximated by the Euler flux splitting
\begin{equation}\label{eq:macro_flux_splitting}
 \Delta\bF_\Gamma=\frac12\Big[\big(\Delta\bF_{E,i}+\Delta\bF_{E,i_\Gamma}\big)\bn_\Gamma
 +r_\Gamma\big(\Delta\bU_i-\Delta\bU_{i_\Gamma}\big)\Big],
\end{equation}
where $\Delta\bF_{E,m}=\bF_E(\bU_m+\Delta\bU_m)-\bF_E(\bU_m)$ is the
increment of the Euler flux tensor \eqref{eq:euler_flux} in cell
$m\in\{i,i_\Gamma\}$ and $r_\Gamma$ is a face spectral radius. Because
$\sum_\Gamma|\Gamma|\bn_\Gamma=\boldsymbol0$ on a closed cell, the contribution of
$\Delta\bF_{E,i}$ vanishes, and \eqref{eq:macro_delta_form} becomes
\begin{equation}\label{eq:macro_increment}
 d_i\,\Delta\bU_i+\sum_{\Gamma\subset\partial\Omega_i}G_{i,\Gamma}\big(\Delta\bU_{i_\Gamma}\big)
 =|\Omega_i|\,R^{(n-1)}_{U,i}.
\end{equation}
Its diagonal coefficient collects the pseudo-time term and the dissipation
of all faces,
\begin{equation}\label{eq:macro_diagonal}
 d_i=\frac{|\Omega_i|}{\Delta t_{\rm num}}
   +\frac12\sum_{\Gamma\subset\partial\Omega_i}r_\Gamma|\Gamma|,
\end{equation}
and a neighbour increment $\boldsymbol\delta$ enters through
\begin{equation}\label{eq:macro_offdiag}
 G_{i,\Gamma}(\boldsymbol\delta)=\frac{|\Gamma|}{2}
 \Big[\big(\bF_E(\bU_{i_\Gamma}+\boldsymbol\delta)-\bF_E(\bU_{i_\Gamma})\big)\bn_\Gamma
 -r_\Gamma\boldsymbol\delta\Big].
\end{equation}
The spectral radius combines the acoustic speed with a viscous contribution,
\begin{equation}\label{eq:spectral_radius}
 r_\Gamma=\max_{m\in\{i,i_\Gamma\}}\big(|\bu_m\cdot\bn_\Gamma|+a_m\big)
    +\frac{2\,\chi_\Gamma\,\mu_\Gamma}{\rho_\Gamma \Delta_\Gamma},
\end{equation}
where $\mu_\Gamma/\rho_\Gamma$ is the kinematic viscosity averaged over the two
cells, $\Delta_\Gamma$ is the cell size normal to face $\Gamma$, and
$\chi_\Gamma\in[0,1]$ weights the viscous part, for instance by the wave share
$A_{1,\Gamma}$. The radius keeps the system diagonally dominant. A larger
$\Delta t_{\rm num}$ strengthens the implicit coupling, whereas a larger
$r_\Gamma$ improves diagonal dominance at the expense of the convergence
rate~\cite{zhu2016implicit}.

Collected over all cells, \eqref{eq:macro_increment} is a sparse block system
whose diagonal consists of the coefficients $d_i$ and whose off-diagonal
part, for a chosen cell ordering, splits into the neighbour terms
$G_{i,\Gamma}$ of the cells that precede and of those that follow cell $i$.
The lower--upper symmetric Gauss--Seidel (LU--SGS)
method~\cite{yoon1988,zhu2016implicit} replaces this matrix by the product
of its lower and upper triangular parts, each including the diagonal, with
the inverse of the diagonal between them, and inverts the product by a
forward and a backward sweep. The point-relaxation (PR) scheme repeats the
two sweeps with the most recent neighbour values. Let $\partial\Omega_i^-$
and $\partial\Omega_i^+$ be the faces shared with the preceding and the
following neighbours of cell $i$. Starting from
$\Delta\bU^{[0]}=\boldsymbol0$, the $q$th of $m_U$ pairs consists of a
forward sweep in increasing cell order,
\begin{equation}\label{eq:macro_pr_sweeps}
 d_i\Delta\bU_i^{[q-1/2]}=|\Omega_i|R^{(n-1)}_{U,i}
   -\sum_{\Gamma\subset\partial\Omega_i^-}G_{i,\Gamma}\big(\Delta\bU^{[q-1/2]}_{i_\Gamma}\big)
   -\sum_{\Gamma\subset\partial\Omega_i^+}G_{i,\Gamma}\big(\Delta\bU^{[q-1]}_{i_\Gamma}\big),
\end{equation}
followed by a backward sweep in decreasing cell order,
\begin{equation}\label{eq:macro_pr_backward}
 d_i\Delta\bU_i^{[q]}=|\Omega_i|R^{(n-1)}_{U,i}
   -\sum_{\Gamma\subset\partial\Omega_i^-}G_{i,\Gamma}\big(\Delta\bU^{[q-1/2]}_{i_\Gamma}\big)
   -\sum_{\Gamma\subset\partial\Omega_i^+}G_{i,\Gamma}\big(\Delta\bU^{[q]}_{i_\Gamma}\big).
\end{equation}
For $m_U=1$ the pair is one LU--SGS sweep pair, and further pairs solve the
same linear system more accurately without accepting the state or rebuilding
the wave. At boundary faces the neighbour increment is set to zero in both
sweeps, so that boundary data enter the W iteration through the residual
only. The state is then updated as
\begin{equation}\label{eq:w_update}
 \bU^{(n)}=\bU^{(n-1)}+\theta_U\,\Delta\bU^{[m_U]},
\end{equation}
with the largest $\theta_U\in\{1,\tfrac12,\tfrac14,\dots\}$ for which density
and temperature remain positive and the total distribution
$W(\bU^{(n)})+P^k$ is non-negative, where $W(\bU)$ denotes the wave evaluated
from the state $\bU$. The wave, its flux and the residual
\eqref{eq:w_residual} are then re-evaluated with $P^k$ unchanged, and no
collision term is evaluated. After $N_U$ updates,
$\widetilde\bU^{k+1}=\bU^{(N_U)}$ defines the predicted wave
$\widetilde W^{k+1}=W(\widetilde\bU^{k+1})$ and, through
\eqref{eq:discrete_endpoint}, the predicted endpoint
$\widetilde P_e^{h,k+1}$. These steps form the left panel of
Fig.~\ref{fig:algorithm}.

\subsection{Particle P iteration}\label{sec:alg_p}
The P iteration solves the particle equation for the particle distribution
with the wave fixed at the predicted value $\widetilde W^{k+1}$. Its $n$th
update, $n=1,\dots,N_P$, starts from $P^{k,n-1}$, with $P^{k,0}=P^k$, and
from the total distribution $f^{k,n-1}=\widetilde W^{k+1}+P^{k,n-1}$, and it
solves the linear system
\begin{equation}\label{eq:particle_increment}
 \left[\left(\frac{1}{\Delta t_{\rm num}}+\frac{1}{\widehat\tau}\right)\mathbf I
 +D^h_{\rm up}\right]\Delta P=b^{(n)},
\end{equation}
whose operator is denoted $\mathbf A_P$; here $D^h_{\rm up}$ is the
first-order upwind transport operator acting on increments, and $\mathbf I$
is the identity on the discrete distribution space. The first update uses
the residual of the accepted state together with the endpoint traction
supplied by the W iteration,
\begin{equation}\label{eq:particle_rhs}
 b^{(1)}=R_P^k+\Theta^k,
\end{equation}
where the traction is the change of the endpoint scaled by the
preconditioning rate,
\begin{equation}\label{eq:endpoint_traction}
 \Theta^k=\frac{\widetilde P_e^{h,k+1}-P_e^{h,k}}{\widehat\tau^k}.
\end{equation}
The later updates, $n\geq2$, use the direct residual
\eqref{eq:discrete_residual_direct} of $f^{k,n-1}$ with the wave fixed at
$\widetilde W^{k+1}$,
\begin{equation}\label{eq:particle_rhs_later}
 b^{(n)}=Q^h_{\rm FSM}\big(f^{k,n-1},f^{k,n-1}\big)
  -D^h_W\big(\widetilde W^{k+1}\big)-D^h_P\big(P^{k,n-1}\big),
\end{equation}
which requires no endpoint, so that the endpoint traction enters the first
update only. The numerical preconditioning time is
\begin{equation}\label{eq:preconditioning_time}
 \frac{1}{\widehat\tau_i}=
 \max\left(\frac{1}{\tau_i},\,C_{\rm col}\max_j\Lambda_{i,j}\right),
\end{equation}
with a safety factor $C_{\rm col}\geq1$, evaluated on the state that provides
the residual. The same coefficient multiplies $\Delta P$ and the endpoint
change in $\Theta^k$. Near equilibrium this pairing makes the particle follow
the Maxwellian share carried by the endpoint, while the value of
$\widehat\tau$ affects only the path of the iteration.

For velocity node $j$, the row of \eqref{eq:particle_increment} in cell $i$
reads
\begin{equation}\label{eq:particle_row}
 d_{i,j}\Delta P_{i,j}+\sum_{\Gamma\subset\partial\Omega_i}c_{i,\Gamma,j}\Delta P_{i_\Gamma,j}
 =b_{i,j},
\end{equation}
with $x^\pm=(x\pm|x|)/2$, the neighbour coefficient
$c_{i,\Gamma,j}=(\bxi_j\cdot\bn_\Gamma)^-|\Gamma|/|\Omega_i|$ and the diagonal
coefficient
\begin{equation}\label{eq:particle_diagonal}
 d_{i,j}=\frac{1}{\Delta t_{\rm num}}+\frac{1}{\widehat\tau_i}
 +\sum_{\Gamma\subset\partial\Omega_i}\frac{(\bxi_j\cdot\bn_\Gamma)^+|\Gamma|}{|\Omega_i|},
\end{equation}
so that only inflow faces, where $c_{i,\Gamma,j}<0$, couple the cell to its
neighbours. As in the W iteration, the matrix is approximately factorized
into its lower and upper parts~\cite{zhu2016implicit}. For a right-hand side
$y$, the kinetic LU--SGS sweep pair $\delta P=\mathbf M_P^{-1}y$ consists, for
each velocity node, of a forward sweep
\begin{equation}\label{eq:particle_forward}
 d_{i,j}\,\delta P^{*}_{i,j}=y_{i,j}
   -\sum_{\Gamma\subset\partial\Omega_i^-}c_{i,\Gamma,j}\,\delta P^{*}_{i_\Gamma,j}
\end{equation}
and a backward sweep
\begin{equation}\label{eq:particle_backward}
 d_{i,j}\,\delta P_{i,j}=d_{i,j}\,\delta P^{*}_{i,j}
   -\sum_{\Gamma\subset\partial\Omega_i^+}c_{i,\Gamma,j}\,\delta P_{i_\Gamma,j}.
\end{equation}
The linear system is solved by $m_P$ preconditioned iterations, each
applying one sweep pair,
\begin{equation}\label{eq:particle_lusgs}
 \Delta P^{[q]}=\Delta P^{[q-1]}
 +\mathbf M_P^{-1}\big(b^{(n)}-\mathbf A_P\Delta P^{[q-1]}\big),
 \qquad q=1,\dots,m_P,
\end{equation}
starting from $\Delta P^{[0]}=0$ and stopping earlier when the linear
residual falls below a prescribed tolerance. The velocity nodes are
independent and are processed in parallel.

As in the W iteration, the neighbour increment is set to zero at boundary
faces, so that boundary data enter the P iteration through the residual
only. The particle is then updated as $P^{k,n}=P^{k,n-1}+\theta\,\Delta
P^{[m_P]}$, where the cellwise factor $0<\theta\leq1$ is the largest value,
reduced by a safety margin, for which the resolved ordinates of
$f^{k,n}=\widetilde W^{k+1}+P^{k,n}$ remain non-negative; positivity is
thus imposed on $f$ rather than on the signed particle. The collision term,
$R_P$ and $\widehat\tau$ are evaluated on $f^{k,n}$ for the next update.
After $N_P$ updates the distribution $f^{k+1}=\widetilde W^{k+1}+P^{k,N_P}$
defines the new conservative state as its moments,
\begin{equation}\label{eq:acceptance}
 \bU^{k+1}=\sum_j\varpi_j\bpsi_jf^{k+1}_j,
\end{equation}
and the new particle as its complement to the new wave,
$P^{k+1}=f^{k+1}-W(\bU^{k+1})$. The acceptance leaves the total distribution
unchanged and restores $\bU=\sum_j\varpi_j\bpsi_j(W+P)_j$. Because the
accepted conservative state is the moment of the corrected distribution
rather than the prediction, the particle update feeds back on the
macroscopic state and closes the two-way coupling. These steps form the
right panel of Fig.~\ref{fig:algorithm}.

\subsection{WP algorithm}\label{sec:alg_wp}
The WP algorithm alternates the two iterations within each outer iteration.
The W iteration receives the fixed particle from the P iteration, the P
iteration receives the predicted wave and the endpoint traction from the W
iteration, and the accepted pair $(\bU^{k+1},P^{k+1})$ determines both
residuals of the next outer iteration. The steps are summarized in
Algorithm~\ref{alg:wp} and Fig.~\ref{fig:algorithm}.

\begin{center}
\begin{minipage}{0.95\linewidth}\small
\refstepcounter{algcounter}\label{alg:wp}
\textbf{Algorithm~\thealgcounter} (WP iteration).
\emph{Input:} initial macroscopic field, controls $N_U$, $m_U$, $N_P$, $m_P$
and tolerance. Set $\bU^0$ from the initial field and $P^0=M^0-W(\bU^0)$,
where $M^0$ is the Maxwellian of $\bU^0$, and evaluate $Q^h_{\rm FSM}$ on
$f^0=M^0$. For $k=0,1,2,\dots$:
\begin{enumerate}\setlength{\itemsep}{0pt}\setlength{\parskip}{0pt}
\item On the accepted pair $(\bU^k,P^k)$, with $f^k=W(\bU^k)+P^k$, form
  $R_U^k$ and $R_P^k$ by \eqref{eq:macro_residual} and \eqref{eq:discrete_residuals}, $P_e^{h,k}$ by
  \eqref{eq:discrete_endpoint} and $\widehat\tau^k$ by
  \eqref{eq:preconditioning_time}; stop if both normalized residuals are
  below the tolerance.
\item W iteration with $P=P^k$ fixed: perform $N_U$ updates of $\bU$, each
  consisting of $m_U$ PR pairs \eqref{eq:macro_pr_sweeps}--\eqref{eq:macro_pr_backward} and the update
  \eqref{eq:w_update}; this gives $\widetilde\bU^{k+1}$,
  $\widetilde W^{k+1}$ and $\widetilde P_e^{h,k+1}$.
\item Form the endpoint traction $\Theta^k$ by \eqref{eq:endpoint_traction}.
\item P iteration with $W=\widetilde W^{k+1}$ fixed: perform $N_P$ updates of
  $P$, each consisting of $m_P$ iterations \eqref{eq:particle_lusgs}, the
  positivity-limited update and one FSM evaluation, the first with
  $b^{(1)}=R_P^k+\Theta^k$.
\item Set $\bU^{k+1}$ and $P^{k+1}$ by \eqref{eq:acceptance}.
\end{enumerate}
\end{minipage}
\end{center}

\begin{figure}[!htbp]
\centering
\resizebox{\linewidth}{!}{%
\begin{tikzpicture}[
  font=\footnotesize, >={Stealth[length=2.2mm,width=1.7mm]},
  node distance=5mm,
  box/.style={rounded corners=3pt, align=center, inner sep=4pt, minimum height=9mm},
  outer/.style={box, draw=black!55, fill=black!4, text width=43mm},
  wmain/.style={box, draw=blue!60!black, fill=blue!13, very thick, text width=43mm},
  pmain/.style={box, draw=orange!80!black, fill=orange!17, very thick, text width=43mm},
  decide/.style={chamfered rectangle, chamfered rectangle xsep=2.5mm, draw=black!55,
                 fill=black!4, align=center, text width=36mm, inner sep=3pt, minimum height=9mm},
  wsub/.style={box, draw=blue!60!black, fill=white, text width=43mm},
  psub/.style={box, draw=orange!80!black, fill=white, text width=43mm},
  flow/.style={->, thick, black!65},
  wflow/.style={->, thick, blue!60!black},
  pflow/.style={->, thick, orange!80!black}]
\node[outer] (o1) {Accepted $\bU^k$ and $P^k$\\ $f^k=W(\bU^k)+P^k$};
\node[outer, below=of o1] (o2) {Residuals $R_U^k$, $R_P^k$ \eqref{eq:macro_residual}--\eqref{eq:discrete_residuals}\\ endpoint $P_e^{h,k}$ and $\widehat\tau^k$};
\node[wmain, below=of o2] (o3) {\textbf{W iteration}\\ $N_U$ updates $\to\widetilde{\bU}^{k+1}$};
\node[outer, below=of o3] (o4) {Endpoint traction \eqref{eq:endpoint_traction}\\ $\Theta^k=(\widetilde P_e^{h,k+1}-P_e^{h,k})/\widehat\tau^k$};
\node[pmain, below=of o4] (o5) {\textbf{P iteration}\\ $N_P$ updates $\to\bU^{k+1}$, $P^{k+1}$};
\node[decide, below=of o5] (o6) {$R_U$ and $R_P$\\ below tolerance?};
\node[box, draw=green!45!black, fill=green!10, text width=43mm, below=6mm of o6] (o7) {Fixed point $f^*$};
\draw[flow] (o1) -- (o2);
\draw[flow] (o2) -- (o3);
\draw[flow] (o3) -- (o4);
\draw[flow] (o4) -- (o5);
\draw[flow] (o5) -- (o6);
\draw[flow] (o6) -- node[right]{yes} (o7);
\coordinate (L) at (-9.9,0);
\coordinate (T) at (0,1.35);
\draw[flow] (o6.west) -- node[above, pos=0.1]{no} (o6.west -| L) -- (L |- T)
  -- node[above]{next outer iteration, $k\leftarrow k+1$} (T) -- (o1.north);
\node[wsub] (w2) at ($(o3)+(-6.0,0.8)$) {Assemble $d_i$, $G_{i,\Gamma}$ and $r_\Gamma$\\ of \eqref{eq:macro_increment}};
\node[wsub, above=3.5mm of w2] (w1) {Fix $P=P^k$ and its flux $\bF^{h,k}_P$};
\node[wsub, below=3.5mm of w2] (w3) {Point relaxation \eqref{eq:macro_pr_sweeps}:\\ forward $\rightarrow$ backward sweep, $m_U$ pairs};
\node[wsub, below=3.5mm of w3] (w4) {Update \eqref{eq:w_update} $\bU\leftarrow\bU+\theta_U\Delta\bU$\\ evaluate $W(\bU)$ and $R_U$};
\draw[wflow] (w1) -- (w2);
\draw[wflow] (w2) -- (w3);
\draw[wflow] (w3) -- (w4);
\draw[wflow] (w4.west) -- ++(-4mm,0) coordinate (wl) |- node[pos=0.25, rotate=90, anchor=south, inner sep=2pt, text=blue!60!black]{$\times N_U$} (w2.west);
\coordinate (wll) at ($(wl)+(-4mm,0)$);
\node[psub] (p2) at ($(o5)+(6.0,0.8)$) {LU--SGS \eqref{eq:particle_lusgs} over all $\bxi_j$:\\ $m_P$ forward--backward sweeps};
\node[psub, above=3.5mm of p2] (p1) {Fix $W=\widetilde W^{k+1}$\\ right-hand side $R_P+\Theta$};
\node[psub, below=3.5mm of p2] (p3) {Update $P\leftarrow P+\theta\Delta P$\\ with $f=\widetilde W^{k+1}+P\geq0$};
\node[psub, below=3.5mm of p3] (p4) {FSM on $f$:\\ $Q^h_{\rm FSM}$, $R_P$ and $\widehat\tau$};
\draw[pflow] (p1) -- (p2);
\draw[pflow] (p2) -- (p3);
\draw[pflow] (p3) -- (p4);
\draw[pflow] (p4.east) -- ++(4mm,0) coordinate (pr) |- node[pos=0.25, rotate=90, anchor=north, inner sep=2pt, text=orange!80!black]{$\times N_P$} (p1.east);
\coordinate (prr) at ($(pr)+(4mm,0)$);
\begin{scope}[on background layer]
\node[draw=blue!60!black, dashed, rounded corners=6pt, fill=blue!4, inner sep=6pt,
      fit=(w1)(w4)(wll), label={[text=blue!60!black, font=\small\bfseries]above:W iteration}] (wf) {};
\fill[blue!7] (o3.north west) -- (wf.north east) -- (wf.south east) -- (o3.south west) -- cycle;
\node[draw=orange!80!black, dashed, rounded corners=6pt, fill=orange!5, inner sep=6pt,
      fit=(p1)(p4)(prr), label={[text=orange!80!black, font=\small\bfseries]above:P iteration}] (pf) {};
\fill[orange!9] (o5.north east) -- (pf.north west) -- (pf.south west) -- (o5.south east) -- cycle;
\end{scope}
\draw[blue!60!black, dashed] (o3.north west) -- (wf.north east);
\draw[blue!60!black, dashed] (o3.south west) -- (wf.south east);
\draw[orange!80!black, dashed] (o5.north east) -- (pf.north west);
\draw[orange!80!black, dashed] (o5.south east) -- (pf.south west);
\end{tikzpicture}}
\caption{WP algorithm. The outer iteration (centre) performs $N_U$ W updates
of $\bU$ with $P$ fixed (left) and $N_P$ P updates of $P$ with the wave fixed
(right), after which the accepted distribution defines $\bU^{k+1}$ and
$P^{k+1}$ by \eqref{eq:acceptance}. Each W update contains $m_U$
forward--backward point-relaxation pairs without collision evaluation, and
each P update $m_P$ forward--backward LU--SGS sweeps over every velocity node
followed by one FSM evaluation.}
\label{fig:algorithm}
\end{figure}
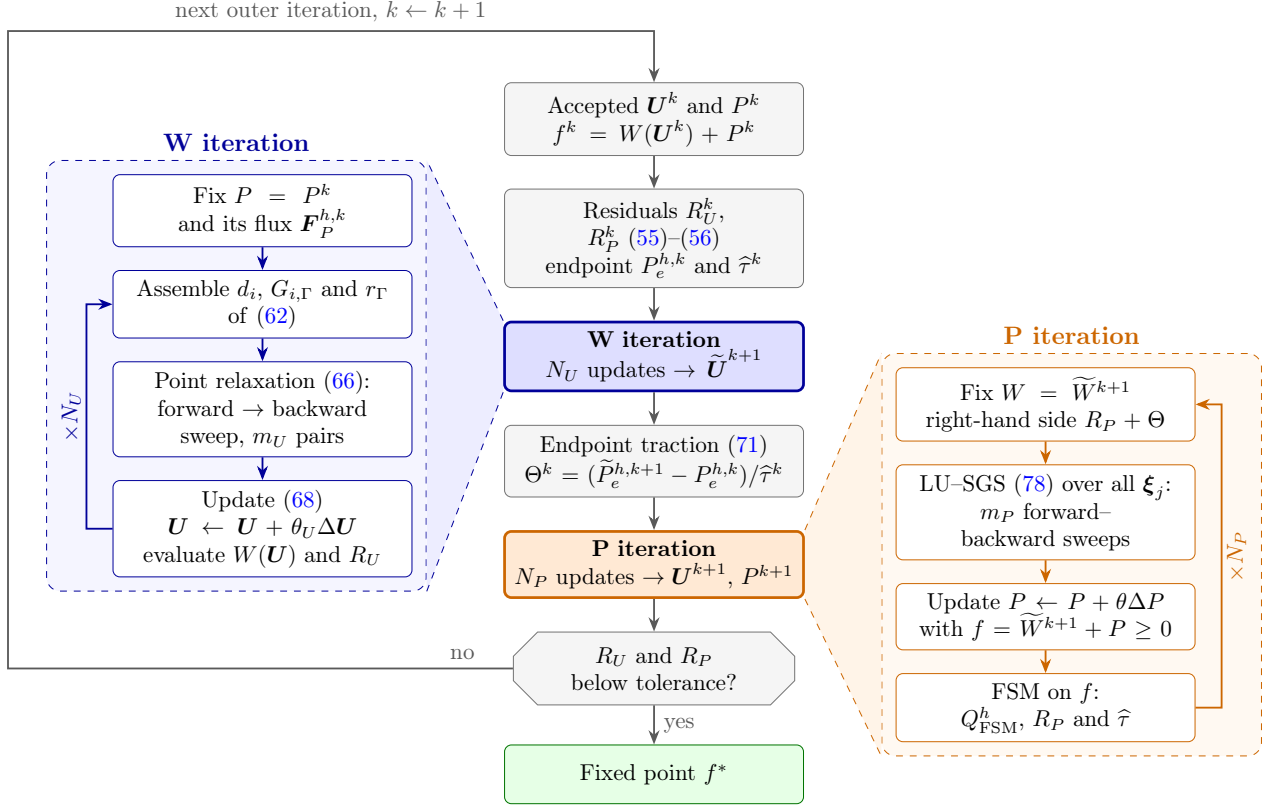

One outer iteration therefore contains $N_U$ macroscopic updates with $m_U$
PR pairs each and $N_P$ particle updates with $m_P$ LU--SGS iterations each.
The collision operator is evaluated once for every accepted particle update,
that is, $N_P$ times per outer iteration, whereas the macroscopic updates
only rebuild the wave and its flux; additional macroscopic work therefore
adds no collision evaluations. In Section~\ref{sec:results} the
configuration is denoted W$N_U$/PR$m_U$/P$N_P$, as in the companion
study~\cite{liu_wpd_shakhov}, and the number $m_P$ of kinetic iterations is
stated separately.

\section{Numerical Analysis}\label{sec:analysis}
This section proves the properties stated in Section~\ref{sec:wave_scales}.
Section~\ref{sec:ap} establishes the equivalence of the coupled system with
the Boltzmann equation and its continuum limit, and shows that the fixed
point of the coupled WP iteration is the discrete Boltzmann solution with the
correct Euler and Navier--Stokes limits. Section~\ref{sec:acceleration} then
bounds the convergence factor of the iteration in a linear model and explains
why it improves as the flow approaches the continuum.

\subsection{Asymptotic preservation}\label{sec:ap}
\begin{proposition}[Equivalence and continuum limit]\label{prop:wp_equivalence}
Let $\tau>0$ be velocity independent, let the characteristic integrals and
their derivatives exist, and let the required velocity moments be finite.
\begin{enumerate}
\item[(i)] With $\bU=\avg{\bpsi f}$ and compatible initial and boundary data,
\eqref{eq:wave_operator} maps a Boltzmann solution to a solution of
\eqref{eq:wp_kinetic_system} and \eqref{eq:up_system}. Conversely, a solution
of \eqref{eq:up_system} with compatible data gives the Boltzmann solution
$f=W+P$, and $\bU-\avg{\bpsi(W+P)}$ is constant in time.
\item[(ii)] Assume in addition the Chapman--Enskog expansion
\eqref{eq:boltzmann_ce_expansion}, smooth macroscopic fields,
$\tau=\varepsilon\bar\tau$ with $\bar\tau$ bounded above and away from zero
and with bounded derivatives along the retained characteristics, and a
horizon for which the equilibrium integral beyond it is
$\mathcal O(\varepsilon^2)$. Then \eqref{eq:wave_ce}--\eqref{eq:total_flux_ce}
hold, together with
$\boldsymbol\sigma_W+\boldsymbol\sigma_P=\boldsymbol\sigma_{\rm NS}+\mathcal O(\varepsilon^2)$
and $\bq_W+\bq_P=\bq_{\rm NS}+\mathcal O(\varepsilon^2)$.
\end{enumerate}
\end{proposition}
\begin{proof}
(i) For the characteristic integrand $K(s)=M_s e^{-\nu(s)}/\tau_s$ and the
relaxation time $\tau_0=\tau(\bx,t)$ at the observation point,
\begin{equation}\label{eq:kernel_identity}
 (\partial_t+\bxi\cdot\nabla_{\bx}+\partial_s)K(s)=-\frac{K(s)}{\tau_0},
\end{equation}
because $(\partial_t+\bxi\cdot\nabla_{\bx})a_s=-\partial_sa_s$ for any field
evaluated at the characteristic point, while
$(\partial_t+\bxi\cdot\nabla_{\bx})\nu(s)=1/\tau_0-1/\tau_s$ and
$\partial_s\nu(s)=1/\tau_s$. In the Leibniz differentiation of the wave
integral the integrand contributes $K(0)-K(\ell)-W/\tau_0$ with
$K(0)=M/\tau_0$, and the upper limit contributes
$K(\ell)(\partial_t\ell+\bxi\cdot\nabla_{\bx}\ell)$. Together they give the
wave equation of \eqref{eq:wp_kinetic_system}, and subtracting it from
\eqref{eq:boltzmann} gives the particle equation, which the identity
$\varepsilon^{-1}Q(f,f)-(M-W)/\tau+\mathcal S_\ell=R+(\tau\mathcal S_\ell-P)/\tau$
turns into the particle equation of \eqref{eq:up_system}; the moments then
give \eqref{eq:wp_moment_system} and \eqref{eq:total_conservation}.
Conversely, let $(\bU,P)$ solve \eqref{eq:up_system} with $W$ defined from
$\bU$. The wave equation then holds identically, and the moments of the wave
and particle equations add up to
$\partial_t\avg{\bpsi(W+P)}+\nabla_{\bx}\cdot(\bF_W+\bF_P)=\boldsymbol0$;
combined with the conservation law this gives
$\partial_t\big(\bU-\avg{\bpsi(W+P)}\big)=\boldsymbol0$. Compatible data
therefore keep $M$ the Maxwellian of $f=W+P$, and adding the two kinetic
equations gives the Boltzmann equation for $f$.

(ii) For a locally frozen relaxation time, \eqref{eq:wave_expansion} follows
from the second-order Taylor expansion of $M_s$, whose remainder, integrated
against $e^{-s/\tau}/\tau$, has a weighted second moment of at most
$2\tau^2$. Integrating the wave integral by parts with
$\mathrm d(e^{-\nu(s)})/\mathrm ds=-e^{-\nu(s)}/\tau_s$, which avoids
freezing $\bar\tau$, gives the expansion of $W$ in \eqref{eq:wave_ce}; the
variation of $\bar\tau$ enters at second order, and the equilibrium integral
beyond the horizon is $\mathcal O(\varepsilon^2)$ by assumption. Subtracting
this expansion from \eqref{eq:boltzmann_ce_expansion} gives the expansion of
$P$. The flux moments and the moment maps \eqref{eq:transport_moments} and
\eqref{eq:heat_flux_moment} are linear when evaluated with the same bulk
velocity, so the wave corrections cancel the corresponding terms of the
particle and leave the flux, stress and heat flux determined by $f^{(1)}$.
\end{proof}

The discrete results concern the fixed points of the coupled WP iteration. A
state is called a fixed point when every increment produced by one outer
iteration vanishes; states left unchanged only because an update factor
$\theta_U$ or $\theta$ is zero are excluded.

\begin{theorem}[Discrete fixed point and continuum limit]\label{thm:ap}
Let $Q^h_{\rm FSM}$ conserve the discrete collision invariants, let the
fluxes, including those at boundary faces, satisfy
\eqref{eq:flux_compatibility}, and let the update factors satisfy
$0<\theta_U,\theta\leq1$ and $N_U,N_P\geq1$. Assume that, for the chosen
finite numbers $m_U$ and $m_P$ of inner iterations, the W and P increment
maps have trivial nullspaces, so that a computed increment vanishes if and
only if its right-hand side vanishes.
\begin{enumerate}
\item[(i)] A state $f^*=W^*+P^*$ is a fixed point of the iteration if and
only if
\begin{equation}\label{eq:discrete_boltzmann}
 D^h_W(W^*)+D^h_P(P^*)=Q^h_{\rm FSM}(f^*,f^*),
\end{equation}
and then $R_U^*=\boldsymbol0$. The fixed-point equation contains none of
$\Delta t_{\rm num}$, $\widehat\tau$, $r_\Gamma$, $N_U$, $m_U$, $N_P$ and
$m_P$; these quantities affect the path of the iteration but not its fixed
points.
\item[(ii)] Consider a family of such fixed points whose moments remain
bounded and whose distributions stay in a neighbourhood of their discrete
Maxwellians in which the equilibrium is isolated. Assume further that
$Q^h_{\rm FSM}$ has the form \eqref{eq:fsm_equilibrium}, so that every
discrete Maxwellian is an exact equilibrium, that the linearized discrete
collision operator is invertible on the complement of the collision
invariants with bounds independent of $\varepsilon$, that
$\ell_i\geq\ell_{\min}>0$ independently of $\varepsilon$, that $\bF^h_W$ is
the gas-kinetic Navier--Stokes flux of the wave, and that the discrete
solution is smooth. Then, formally as $\varepsilon\to0$ on a fixed mesh of
cell size $\Delta_i$,
\begin{equation}\label{eq:discrete_ce}
 \begin{aligned}
 &f^*=M^*+\varepsilon f^{(1),h}+\mathcal O(\varepsilon^2),\\[4pt]
 &P^*=\varepsilon\big(f^{(1),h}-W^{(1),h}\big)+\mathcal O(\varepsilon^2),
 \end{aligned}
\end{equation}
where $f^{(1),h}$ and $W^{(1),h}$ are the discrete first-order corrections of
the distribution and the wave, and $R_U^*=\boldsymbol0$ reduces to the
finite-volume Euler system at leading order and to the gas-kinetic
Navier--Stokes scheme with the Boltzmann stress and heat flux at first
order, up to $\mathcal O(\varepsilon\Delta_i)$.
\item[(iii)] Let the accepted and predicted states be smooth
near-equilibrium states that differ by an $\mathcal O(1)$ macroscopic
increment, let $\eta$ be locally frozen, $\varepsilon\ll\Delta t_{\rm num}$
and $\theta=1$. Then, for every $\eta\geq0$,
\begin{equation}\label{eq:hard_pull}
 \begin{aligned}
 &\Delta W=\Delta\big[(1-e^{-\eta})M\big]+\mathcal O(\varepsilon),\\[4pt]
 &\Delta P=\Delta\big(e^{-\eta}M\big)+\mathcal O(\varepsilon),
 \end{aligned}
\end{equation}
so that $\Delta f=\Delta M+\mathcal O(\varepsilon)$ and
$\bU^{k+1}=\widetilde\bU^{k+1}+\mathcal O(\varepsilon)$.
\end{enumerate}
\end{theorem}

\begin{proof}
(i) At a fixed point the macroscopic increments vanish, so the predicted
endpoint coincides with the accepted one, $\Theta^k=0$, and the right-hand
side of \eqref{eq:particle_increment} reduces to $R_P^*$. Since the particle
increment vanishes, the assumption gives $R_P^*=0$, which is
\eqref{eq:discrete_boltzmann} by \eqref{eq:discrete_residual_direct}, and its
quadrature moments give $R_U^*=\boldsymbol0$. Conversely,
\eqref{eq:discrete_boltzmann} gives $R_P^*=0$ and $R_U^*=\boldsymbol0$; the
macroscopic increments then vanish, the endpoint traction vanishes, the
particle increment is zero, and the iteration returns $f^*$. The
preconditioning quantities multiply increments only and do not enter
\eqref{eq:discrete_boltzmann}.

(ii) Write $Q^h_{\rm FSM}=\varepsilon^{-1}Q^h$ and $\tau=\varepsilon\bar\tau$.
Since $\eta\geq\ell_{\min}/(\varepsilon\bar\tau)$, the factor $e^{-\eta}$ is
smaller than any power of $\varepsilon$, so $A_0$ and $A_1$ equal one up to
such terms and the reconstructed wave has the form
$W^*=M^*+\varepsilon W^{(1),h}+\mathcal O(\varepsilon^2)$. Multiplying
\eqref{eq:discrete_boltzmann} by $\varepsilon$ gives
$Q^h(f^*,f^*)=\mathcal O(\varepsilon)$. Because the discrete equilibrium is
exact and isolated and the linearization has a uniformly bounded inverse on
the complement of the invariants, it follows locally that
$f^*=M^*+\mathcal O(\varepsilon)$. With $L^h_M$ the linearization of $Q^h$
at $M^*$, the next order is
$L^h_Mf^{(1),h}=D^h_W(M^*)+\mathcal O(\varepsilon)$. Since
$P^*=\mathcal O(\varepsilon)$, the divergence $D^h_P(P^*)$ enters this
kinetic balance one order beyond the equation that determines $f^{(1),h}$,
whereas the first-order moment flux of $P^*$ remains part of the
Navier--Stokes correction below. Solvability requires the quadrature moments
of $D^h_W(M^*)$ to vanish at leading order, which by
\eqref{eq:flux_compatibility} is the finite-volume Euler system, and the
solution on the complement of the invariants defines $f^{(1),h}$ and
\eqref{eq:discrete_ce}. The total flux is the Euler flux plus $\varepsilon$
times the moments of $W^{(1),h}$ and of the particle distribution
$f^{(1),h}-W^{(1),h}$. The upwind reconstruction and the transport
correction in \eqref{eq:particle_flux} change the face moments of a smooth
$\mathcal O(\varepsilon)$ distribution by $\mathcal O(\varepsilon\Delta_i)$,
since $\Delta t=\mathcal O(\Delta_i)$. The first-order flux is therefore that
of $f^{(1),h}$, with the wave supplying the relaxation stress for $\mu=\tau p$
and the particle adding the difference from the Boltzmann response,
including the heat-flux correction of the unit-Prandtl-number relaxation.

(iii) By \eqref{eq:discrete_wave}, in which $A_1\tau=\mathcal O(\varepsilon)$
uniformly in $\eta$, $W=(1-e^{-\eta})M+\mathcal O(\varepsilon)$ for both
states. Since $\tau D^h_W(W)=\mathcal O(\varepsilon)$ on a smooth state,
\eqref{eq:discrete_endpoint} gives $P_e^h=e^{-\eta}M+\mathcal O(\varepsilon)$.
Multiplying \eqref{eq:particle_increment} by $\widehat\tau\leq\tau$ gives
$[\mathbf I+\widehat\tau(\Delta t_{\rm num}^{-1}\mathbf I+D^h_{\rm up})]\Delta P
=\widehat\tau R_P^k+\widetilde P_e^{h,k+1}-P_e^{h,k}$. On a smooth
near-equilibrium state $R_P^k=\mathcal O(1)$, because the factor
$\varepsilon^{-1}$ multiplies a collision term of order $\varepsilon$, and
the bracket equals $\mathbf I+\mathcal O(\varepsilon)$. Hence
$\Delta P=\Delta(e^{-\eta}M)+\mathcal O(\varepsilon)$. The same balance holds
for one LU--SGS sweep, whose diagonal is dominated by $1/\widehat\tau$.
Adding the wave increment gives $\Delta f=\Delta M$; the two shares need not
be separately velocity independent.
\end{proof}

The assumptions in (ii) identify the two discrete ingredients of the limit.
The first is an exact discrete equilibrium, which keeps the defect
$\varepsilon^{-1}Q^h(M,M)$ from growing as $\varepsilon\to0$; the second is
the gas-kinetic wave flux, whose dissipation in the limit is the
Navier--Stokes viscosity rather than a numerical viscosity of order
$\Delta_i$. Part (iii) explains the paired coefficient in
\eqref{eq:particle_increment}. The endpoint traction moves the particle by the
Maxwellian share $e^{-\eta}$ of the prediction, so that the accepted moments
recover the full macroscopic prediction rather than only its wave share
$1-e^{-\eta}$.

\subsection{Convergence acceleration}\label{sec:acceleration}
The rate of the coupled WP iteration is analyzed in the simplest model that
retains transport, relaxation and the horizon split. We linearize about a
uniform equilibrium at rest, keep a single molecular speed normalized to one
with isotropic directions, and let the collision be relaxation of the density
at constant $\tau$, so that the collision remainder vanishes. For a periodic
mode of wavenumber $\varkappa$, let $\zeta\in[-1,1]$ be the cosine between
$\bxi$ and the wave vector, with angular average
$\frac12\int_{-1}^{1}\cdot\,\dd\zeta$, and put $z=\varkappa\tau$ and
$\eta=\ell/\tau$. The particle increment is solved exactly, with
$\Delta t_{\rm num}\to\infty$, $m_P\to\infty$ and $\widehat\tau=\tau$, the W
iteration is converged with the particle fixed, and $N_P=1$ and $\theta=1$.
For comparison, the conventional iterative scheme (CIS) obtains $f^{k+1}$
from $\bxi\cdot\nabla_{\bx}f^{k+1}+f^{k+1}/\tau=\rho^k/\tau$, where $\rho^k$
is the density of $f^k$.

\begin{theorem}[Convergence acceleration]\label{thm:acceleration}
In this model, one iteration multiplies the error of the density mode by
\begin{equation}\label{eq:iteration_factors}
 \begin{aligned}
 &g_{\rm CIS}(z)=\frac{\arctan z}{z}\quad\text{for CIS},\\[4pt]
 &g_{\rm WP}(z,\eta)=\left|1-g_{\rm CIS}(z)+\frac{\varphi_P(z,\eta)}{\varphi_W(z,\eta)}\right|
 \quad\text{for WP},
 \end{aligned}
\end{equation}
where the particle and wave flux responses are
\begin{equation}\label{eq:flux_responses}
 \begin{aligned}
 &\varphi_P=\frac{e^{-\eta}}{2}\int_{-1}^{1}
   \frac{\zeta\,e^{-iz\eta\zeta}}{1+iz\zeta}\,\dd\zeta,\\[4pt]
 &\varphi_W=\frac12\int_{-1}^{1}
   \frac{\zeta\,\big(1-e^{-\eta}e^{-iz\eta\zeta}\big)}{1+iz\zeta}\,\dd\zeta .
 \end{aligned}
\end{equation}
With $\beta(\eta)=(1+\eta)e^{-\eta}$, for $0<z\leq1$ and
$\beta(\eta)+3z^2/5<1$,
\begin{equation}\label{eq:factor_bounds}
 \begin{aligned}
 &g_{\rm CIS}(z)\geq1-\frac{z^2}{3},\\[4pt]
 &g_{\rm WP}(z,\eta)\leq\frac{z^2}{3}
   +\frac{\beta(\eta)}{1-\beta(\eta)-3z^2/5}.
 \end{aligned}
\end{equation}
\end{theorem}

\begin{proof}
Scale time by $\tau$. For density amplitude $\rho$, the wave operator
\eqref{eq:wave_operator} gives $W=w\rho$ with
$w=(1-e^{-\eta}e^{-iz\eta\zeta})/(1+iz\zeta)$, and the endpoint is
$P_e=e^{-\eta}e^{-iz\eta\zeta}\rho$, so that $(1+iz\zeta)w+P_e/\rho=1$.
CIS gives $f^{k+1}=\rho^k/(1+iz\zeta)$, whose average is
$g_{\rm CIS}\rho^k$. For WP, $R_P^k=\rho^k-(1+iz\zeta)f^k$, and the particle
increment satisfies
$(1+iz\zeta)\Delta P=R_P^k+(\widetilde\rho^{k+1}-\rho^k)e^{-\eta}e^{-iz\eta\zeta}$,
with $\widetilde\rho^{k+1}$ the predicted density. The identity above gives
$f^{k+1}=w\widetilde\rho^{k+1}+P^k+\Delta P=\widetilde\rho^{k+1}/(1+iz\zeta)$,
and hence $\rho^{k+1}=g_{\rm CIS}\widetilde\rho^{k+1}$. The converged
macroscopic iteration balances the wave flux of the new prediction against
the fixed particle flux,
$\varphi_W\widetilde\rho^{k+2}+\frac12\int\zeta P^{k+1}\dd\zeta=0$. Since
$P^{k+1}=f^{k+1}-w\rho^{k+1}$ and
$\frac12\int\zeta(1+iz\zeta)^{-1}\dd\zeta=\varphi_W+\varphi_P$, the particle
flux is $(\varphi_W+\varphi_P-g_{\rm CIS}\varphi_W)\widetilde\rho^{k+1}$, and
$\widetilde\rho^{k+2}=-(1-g_{\rm CIS}+\varphi_P/\varphi_W)\widetilde\rho^{k+1}$,
which is \eqref{eq:iteration_factors}.

The alternating series of $\arctan z$ gives $g_{\rm CIS}\geq1-z^2/3$. In
$\varphi_P$ the odd parts of the integrand cancel, leaving
$\varphi_P=-\frac{i}{2}e^{-\eta}\int_{-1}^{1}
[\zeta\sin(z\eta\zeta)+z\zeta^2\cos(z\eta\zeta)](1+z^2\zeta^2)^{-1}\dd\zeta$,
and $|\sin x|\leq|x|$ yields $|\varphi_P|\leq\beta z/3$. Moreover
$\varphi_W+\varphi_P=-i(z-\arctan z)/z^2$, whose modulus is at least
$z/3-z^3/5$ for $0<z\leq1$. Hence $|\varphi_W|\geq(z/3)(1-\beta-3z^2/5)$,
and the triangle inequality gives \eqref{eq:factor_bounds}.
\end{proof}

The bound explains how the convergence depends on the Knudsen number. On a
fixed mesh of cell size $\Delta_i$, a resolved mode has
$\varkappa\leq\pi/\Delta_i$ and therefore $z=\mathcal O(\varepsilon)$, while
$\eta=\ell/\tau$ grows like $\varepsilon^{-1}$. CIS multiplies the error by
only $1-\mathcal O(\varepsilon^2)$ per iteration and needs
$\mathcal O(\varepsilon^{-2})$ iterations, which is the classical slow
convergence of source iteration in diffusive regimes~\cite{adams2002}. The
WP factor, in contrast, is bounded by
$\mathcal O(\varepsilon^2)+\mathcal O\big((1+\eta)e^{-\eta}\big)$. The
coefficient $\beta=1-A_1$ is the part of the first-order transport not
carried by the wave; as the horizon-to-relaxation ratio grows, the particle
share of the transport decays exponentially, and the coupled WP iteration
inherits this decay. For $\eta=4$ and $z=10^{-2}$, for example, the bound
\eqref{eq:factor_bounds} gives $g_{\rm WP}\leq0.10$, whereas
$g_{\rm CIS}=1-3.33\times10^{-5}$. The acceleration thus reduces the number
of outer iterations without increasing the number of FSM evaluations in
each particle update. In the long-wave limit $z\to0$ the bound
\eqref{eq:factor_bounds} tends to $\beta/(1-\beta)$, which is smaller than
one when $\beta(\eta)<1/2$, that is, when the wave carries more than half of
the first-order transport, $\eta>\eta_c\approx1.68$; for resolved modes with
finite $z$ the complete bound applies.

The collision remainder, which is absent from the linear model, is treated
through the preconditioning time \eqref{eq:preconditioning_time}. With the
macroscopic state fixed and transport neglected, the particle update
multiplies the remainder error by $\mathbf I+d^{-1}J^h$, where
$d=\Delta t_{\rm num}^{-1}+\widehat\tau^{-1}$ and $J^h$ is the Jacobian of
$Q^h_{\rm FSM}$. Perturbations along the discrete Maxwellians lie in the null
space of $J^h$, because \eqref{eq:fsm_equilibrium} vanishes on every discrete
Maxwellian; on these conservative modes the factor is one, and they are
advanced by the W iteration. Assume that $-J^h$ is self-adjoint and
non-negative in a weighted inner product, as the linearized Boltzmann
operator is in the $M^{-1}$-weighted inner product, and that its largest
eigenvalue satisfies $\lambda_{\max}\leq C_{\rm col}\max_j\Lambda_{i,j}$.
Then, on the complement of the null space, the eigenvalues $1-\lambda/d$,
with $\lambda\in(0,\lambda_{\max}]$ an eigenvalue of $-J^h$, lie in $[0,1)$
for every $\Delta t_{\rm num}$, because
$d\geq\widehat\tau^{-1}\geq\lambda_{\max}$ by
\eqref{eq:preconditioning_time}.

\section{Numerical Tests}\label{sec:results}
The tests proceed from a normal shock through Couette and lid-driven cavity
flows to hypersonic flows around a cylinder and an Apollo section. For each
case we give the physical setting and the controls of the coupled WP iteration,
compare the WP solution with a kinetic or continuum reference, and report the
ratios of iteration counts and wall times with respect to the conventional
iterative scheme.

\subsection{Gas model and comparison protocol}\label{sec:test_gas_model}
All tests use the reference speed $c_0=\sqrt{2R_gT_0}$, where $T_0$ is the
dimensional reference temperature, so that the nondimensional reference
temperature is $T_{\rm ref}=1/2$. The gas is argon, modelled by variable hard
spheres (VHS) with viscosity exponent $\omega=0.81$, and the Knudsen number
$\varepsilon$ is the reference mean free path divided by $L_0$. The viscosity
follows the power law
\begin{equation}\label{eq:viscosity_law}
 \mu(T)=\mu_{\rm ref}\left(\frac{T}{T_{\rm ref}}\right)^{\omega},
\end{equation}
whose reference value $\mu_{\rm ref}=\mu(T_{\rm ref})$ is fixed by the
viscosity-based normalization,
\begin{equation}\label{eq:reference_mu}
 \mu_{\rm ref}=\frac{15\sqrt{\pi}}{2(5-2\omega)(7-2\omega)}\,\varepsilon,
\end{equation}
so that $\mu_{\rm ref}=0.73$ at $\varepsilon=1$, and the wave relaxation time
is
\begin{equation}\label{eq:reference_viscosity}
 \tau=\frac{\mu(T)}{p}.
\end{equation}
This relaxation time enters only the wave. The collision operator itself uses
the single-term power-law kernel
\begin{equation}\label{eq:power_law_kernel}
 B(|\bg|,\theta)=
 \frac{|\bg|^\alpha\sin^{\alpha-1}\!\left(\theta/2\right)}
  {4\pi\,2^{\alpha/2}\Gamma\!\left((\alpha+3)/2\right)}
\end{equation}
with $\alpha=2(1-\omega)=0.38$, where $\Gamma$ is the gamma function. The
exponent $\alpha$ reproduces the temperature dependence of the viscosity law
\eqref{eq:viscosity_law}, and the FSM applies this normalization together
with the factor $\varepsilon^{-1}$ of \eqref{eq:boltzmann}.

In the tests, WP denotes Algorithm~\ref{alg:wp} and CIS the conventional
iterative scheme for the full distribution, which transports $f$ with the
flux \eqref{eq:particle_flux} and performs one kinetic LU--SGS sweep pair and
one FSM evaluation per iteration. NS denotes the gas-kinetic Navier--Stokes
scheme that serves as the continuum reference, namely the flux
\eqref{eq:gks_wave_moments} with $A_{0,\Gamma}=A_{1,\Gamma}=1$ combined with
the W iteration of Section~\ref{sec:alg_w} and the reconstruction of WP. The
configuration W$N_U$/PR$m_U$/P$N_P$ of Section~\ref{sec:alg_wp} is given with
each case. All tests use $N_P=1$ and $m_P=1$, so that one outer iteration of
WP, like one CIS iteration, performs a single kinetic LU--SGS sweep pair and a
single FSM evaluation. All calculations use the second-order van Leer-limited
reconstruction of Section~\ref{sec:alg_fluxes} for both the wave and the
particle, except those of the Reynolds-number cavities, which use the
sixth-order central reconstruction. The horizon number of
\eqref{eq:discrete_horizon} is $C_\ell=1$ throughout. The pseudo-time step
$\Delta t_{\rm num}$ of \eqref{eq:dt_num} is set by the implicit CFL number
$C_{\rm imp}=1000$, or 150 for the Reynolds-number cavities, and the
averaging interval $\Delta t$ of the fluxes by the explicit CFL number
$C_{\rm exp}=0.2$, or $0.25$ for the normal shock and $0.02$ for Couette
flow. All wall times are solver wall times measured with 64 OpenMP threads on
one node, for WP, CIS and NS alike and including the calculations taken from
the companion study~\cite{liu_wpd_shakhov}. Every calculation is carried below
the normalized macroscopic residual $R_U=10^{-7}$, the threshold at which
iteration counts and wall times are compared; the only exceptions are the CIS
solution of the normal shock, discussed in Section~\ref{sec:shock}, and the
Boltzmann CIS calculations in the continuum regime, discussed below. In the
continuum regime the distribution stays close to equilibrium and the particle
carries only a small correction, so that coarser velocity grids suffice, with
$16\times16\times8$ nodes for Couette flow at $\varepsilon=10^{-4}$ and
$32\times32\times12$ nodes for the Reynolds-number cavities and for the
cylinder and the Apollo section at $\varepsilon=10^{-4}$.

Accuracy is measured by the relative difference
$E_{L_2}(q)=\|q_{\rm WP}-q_{\rm ref}\|_2/\|q_{\rm ref}\|_2$, where $q$ is a
macroscopic field, a profile or a sampled distribution and the norm is taken
over the corresponding cells, samples or velocity nodes. The values quoted
below are those printed in the figures, given to three significant figures.
Distributions are shown as reduced distributions of $f$, integrated over the
velocity components that are not displayed, and for WP the distribution is
reconstructed as $W+P$. Contour plots show WP by filled contours and the
reference by lines at common levels, and profile plots show WP by lines and
the reference by symbols. Efficiency is quantified, as in the companion
study~\cite{liu_wpd_shakhov}, by the ratio $S_{\rm iter}=N_{\rm CIS}/N_{\rm WP}$
of the numbers of outer iterations and the ratio
$S_t=t_{\rm CIS}/t_{\rm WP}$ of the solver wall times, both evaluated at the
residual threshold and printed in the residual figures, where WP is drawn
with solid, CIS with dashed and NS with dotted lines.

In the continuum regime, Boltzmann CIS on the three-dimensional velocity grid
did not reach the threshold within the available computing time, except for
the cavity at $\mathrm{Re}=100$. For Couette flow at $\varepsilon=10^{-4}$ on
the finer $32^3$ velocity grid, for instance, its residual was still
$R_U=5.2\times10^{-4}$ after $10^6$ iterations and $3.4\times10^5$~s. For the
four remaining continuum cases, which are Couette flow at
$\varepsilon=10^{-4}$, the cavity at $\mathrm{Re}=1000$, and the cylinder and
the Apollo section at $\varepsilon=10^{-4}$, the CIS reference is therefore
the Shakhov model with a reduced two-dimensional velocity distribution on the
same physical mesh, computed with the Shakhov solver of the companion
study~\cite{liu_wpd_shakhov}. This reference supplies the curves labelled CIS
in the iteration and wall-time panels of these cases, as well as the CIS
profiles of Couette flow at $\varepsilon=10^{-4}$ and of the cavity at
$\mathrm{Re}=1000$ and the CIS fields of the cylinder and the Apollo section
at $\varepsilon=10^{-4}$. The companion study reports the Couette, cavity and
cylinder cases on $32\times32$, $16\times16$ and $64\times64$ reduced velocity
nodes, and the same solver is applied here to the Apollo section on
$64\times64$ reduced velocity nodes, as for the cylinder. The substitution is
justified because, near equilibrium, the convergence of CIS is limited by the
relaxation of the conservative variables through transport, which is
governed by the viscosity and the heat conduction rather than by the details
of the collision operator. The Shakhov iteration counts therefore define
$S_{\rm iter}$ in place of those of Boltzmann CIS, and the wall times of the
same calculations define $S_t$. The macroscopic sweeps of the W iteration do
not depend on the velocity grid, whereas the cost of a CIS iteration grows
with the number of velocity nodes, so that the cost of one WP outer iteration
measured in CIS iterations is smallest on the finest velocity grids, as
observed in the companion study~\cite{liu_wpd_shakhov}. A Boltzmann CIS
iteration on the three-dimensional grids of the present tests, which in
addition evaluates the FSM collision operator at every node, costs far more
than a Shakhov iteration on the reduced grid. The Shakhov iteration is thus
the faster of the two references, and the values of $S_t$ obtained with it
understate the wall-time gain over Boltzmann CIS, which makes the comparison a
conservative one. At $\mathrm{Re}=100$ both ratios refer to Boltzmann CIS, as
in the rarefied cases.

In the continuum regime WP is also timed against the NS solver, and this
comparison is a strict one. The NS solver uses the same mesh, reconstruction,
CFL numbers and implicit macroscopic iteration as the W iteration of WP but
performs no particle transport and no collision evaluation, so that its wall
time is the reference cost of the macroscopic solution on the same mesh and
the ratio $t_{\rm WP}/t_{\rm NS}$ measures the cost of the kinetic solution
relative to it. The NS solutions and residual histories are computed with
this macroscopic solver, and for the Couette, cavity and cylinder cases they
are those reported in the companion study~\cite{liu_wpd_shakhov}.

In relation to the analysis, the agreement between WP and CIS in the rarefied
and transition cases tests the common fixed point of
Theorem~\ref{thm:ap}(i), the Reynolds-number cavity flows compare WP with the
Navier--Stokes limit of Theorem~\ref{thm:ap}(ii), and the external flows at
$\varepsilon=10^{-4}$ provide further comparisons in the continuum regime.
The dependence of the speedups on the Knudsen number is compared with
Theorem~\ref{thm:acceleration}.

\subsection{Normal shock}\label{sec:shock}
The first test is a normal shock with upstream Mach number $10$, Knudsen
number $\varepsilon=1$ and upstream temperature $0.5$, whose upstream and
downstream equilibrium states are related by the Rankine--Hugoniot
conditions. The physical domain $x\in[-80,40]$ is divided into 176 uniform
cells, the velocity grid has $100\times32\times32$ nodes on
$[-26.40,26.40]^3$, and WP runs W1/PR4/P1. Because the location of a normal
shock is a neutral mode of the steady problem, the shock computed by CIS
drifts slowly along the domain and its residual stagnates near $10^{-4}$,
whereas WP reaches $R_U<10^{-8}$. All comparisons are therefore made relative
to the density midpoint of each solution, which removes this drift.

Figure~\ref{fig:shock-moments} compares the macroscopic profiles, with
positions measured from the density midpoint $x_{\rho\mathrm{-mid}}$ of each
solution. Density and temperature are normalized to rise from zero upstream
to one downstream, velocity to fall from one to zero, and heat flux and normal
stress are scaled by their maximum magnitudes. The relative differences are
$4.33\times10^{-4}$, $5.73\times10^{-4}$ and $6.96\times10^{-4}$ in density,
velocity and temperature and $1.45\times10^{-3}$ and $1.87\times10^{-3}$ in
heat flux and normal stress, so that the agreement extends from the
conservative moments to the non-equilibrium stress and energy transport
inside the shock layer. Figure~\ref{fig:shock-dist} compares the reduced
distribution
$f_u(x,\xi_x)=\int_{\mathbb R^2}f(x,\xi_x,\xi_y,\xi_z)
\,\mathrm d\xi_y\,\mathrm d\xi_z$, reconstructed as $W+P$ for WP, at four
positions from downstream to upstream, with relative differences
$6.88\times10^{-4}$, $4.20\times10^{-4}$, $7.21\times10^{-4}$ and
$1.31\times10^{-3}$. Inside the layer a broad downstream population coexists
with a narrow upstream peak. This two-peak structure cannot be represented by
a single local Maxwellian, and resolving it is the main kinetic test of this
case.

\begin{figure}[tbp]
\centering\singlespacing\setlength{\figW}{\linewidth}
\begin{subfigure}[t]{0.4337\figW}\centering
\raisebox{0.0002\figW}{\includegraphics[width=0.4337\figW]{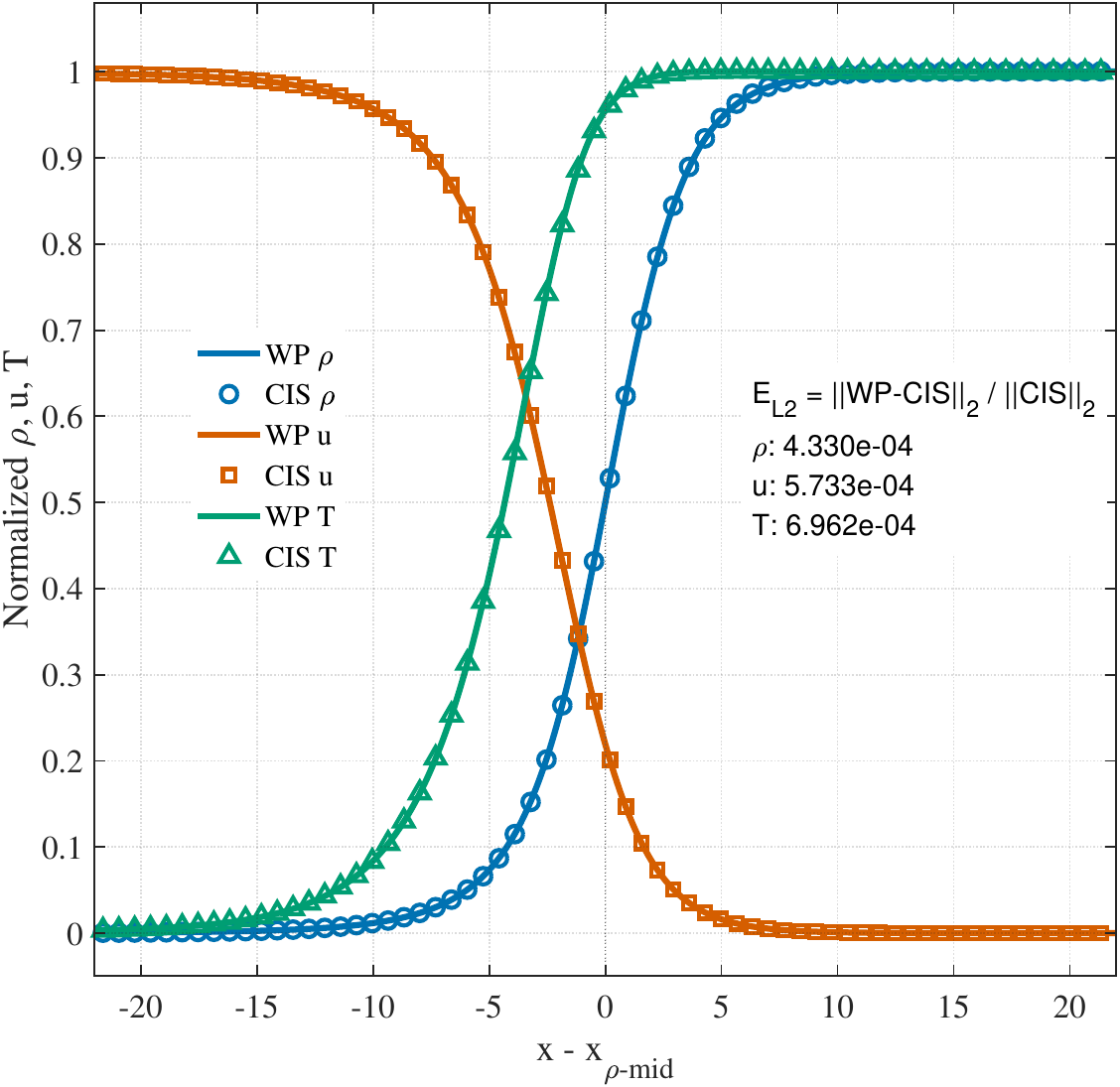}}
\caption{Density, velocity and temperature}
\end{subfigure}\hspace{0.0200\figW}%
\begin{subfigure}[t]{0.5413\figW}\centering
\raisebox{0.0000\figW}{\includegraphics[width=0.5413\figW]{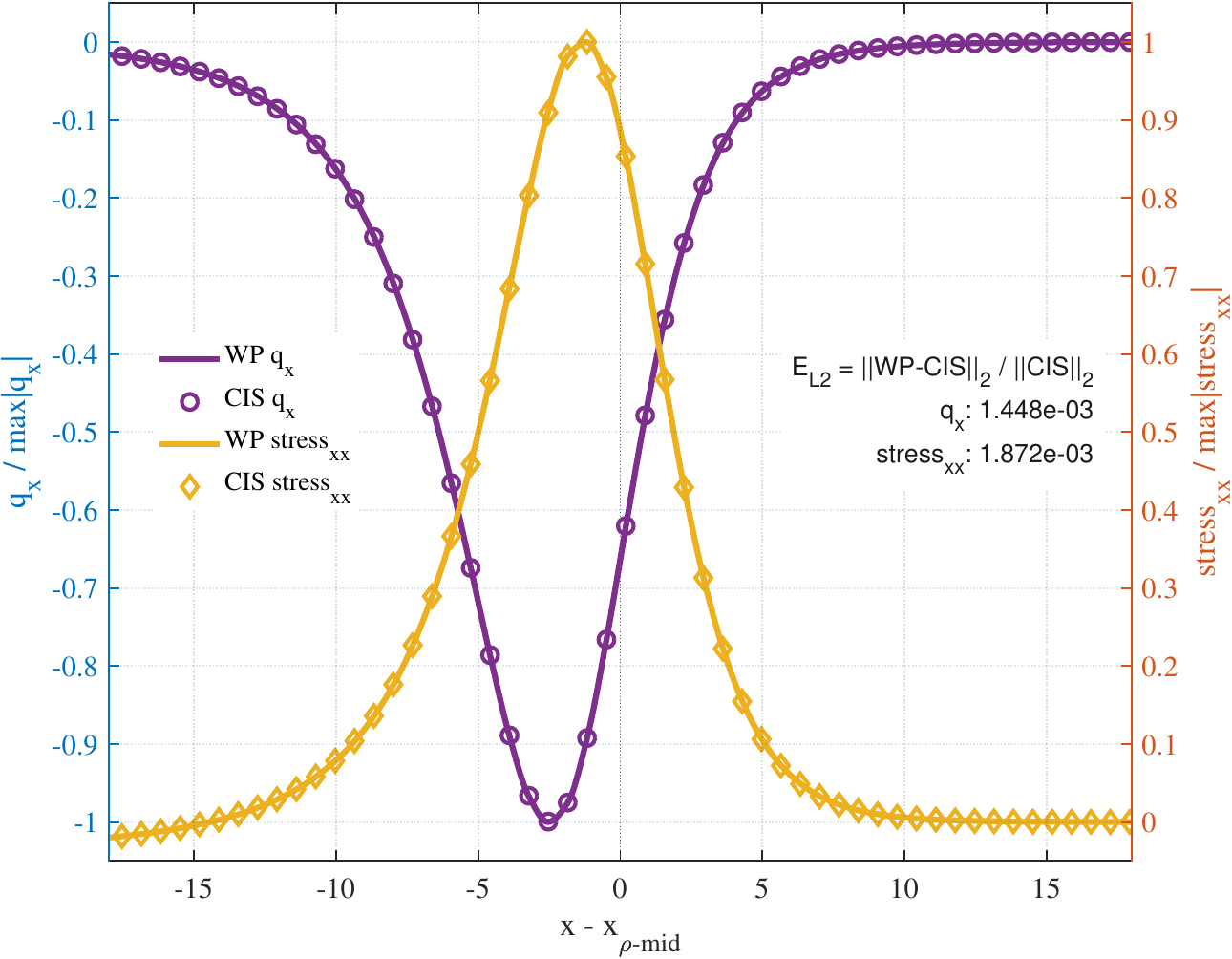}}
\caption{Heat flux and normal stress}
\end{subfigure}
\caption{Normal shock at Mach number 10 and $\varepsilon=1$: normalized macroscopic profiles against the distance from the density midpoint. WP is shown by lines and CIS by symbols, and $E_{L_2}$ is printed in each panel.}\label{fig:shock-moments}
\end{figure}
\begin{figure}[tbp]
\centering\singlespacing\setlength{\figW}{\linewidth}
\begin{subfigure}[t]{0.4875\figW}\centering
\raisebox{0.0000\figW}{\includegraphics[width=0.4875\figW]{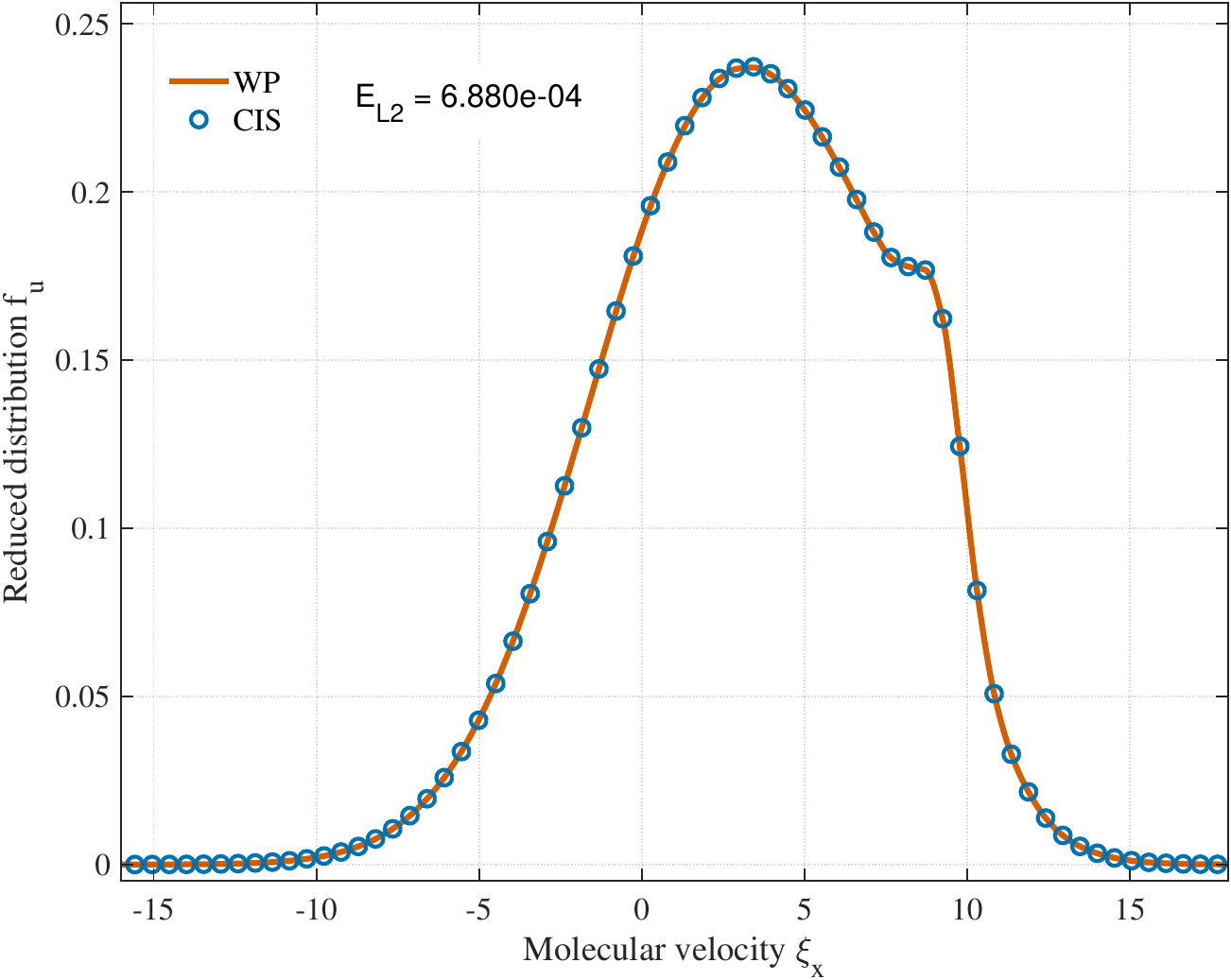}}
\caption{$x-x_{\rho\mathrm{-mid}}=0.87$}
\end{subfigure}\hspace{0.0200\figW}%
\begin{subfigure}[t]{0.4875\figW}\centering
\raisebox{0.0000\figW}{\includegraphics[width=0.4875\figW]{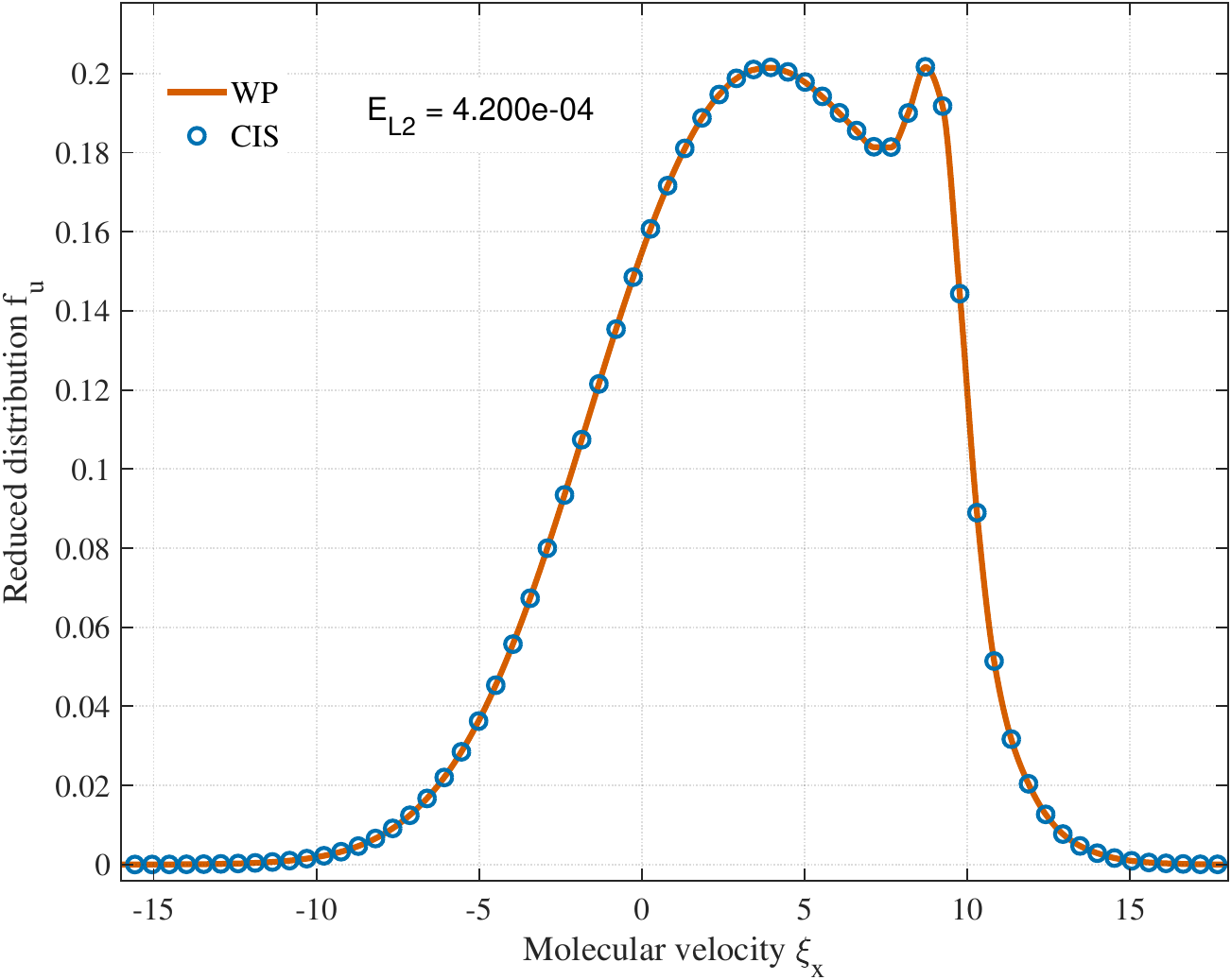}}
\caption{$x-x_{\rho\mathrm{-mid}}=0.18$}
\end{subfigure}
\par\smallskip
\begin{subfigure}[t]{0.4875\figW}\centering
\raisebox{0.0000\figW}{\includegraphics[width=0.4875\figW]{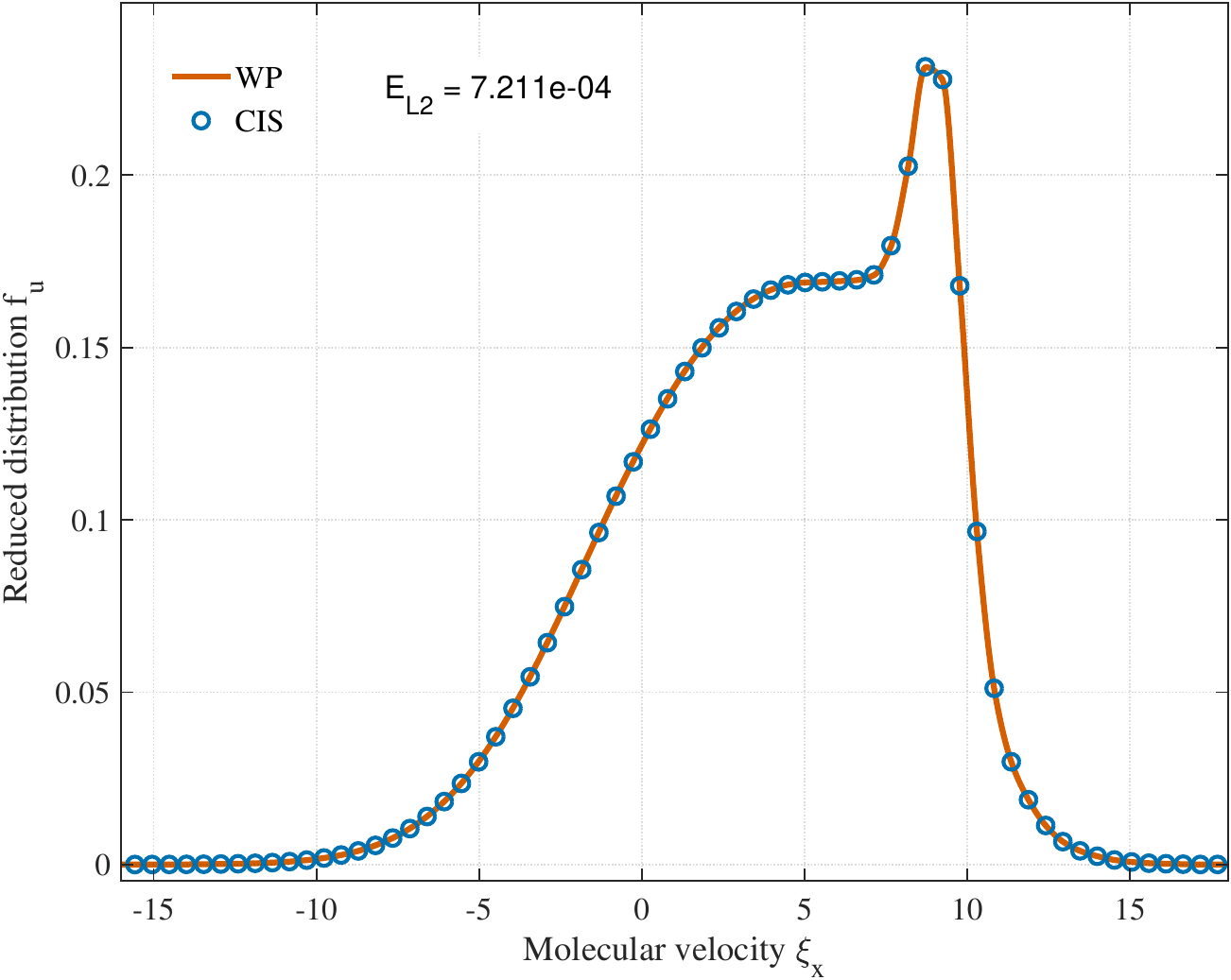}}
\caption{$x-x_{\rho\mathrm{-mid}}=-0.50$}
\end{subfigure}\hspace{0.0200\figW}%
\begin{subfigure}[t]{0.4875\figW}\centering
\raisebox{0.0000\figW}{\includegraphics[width=0.4875\figW]{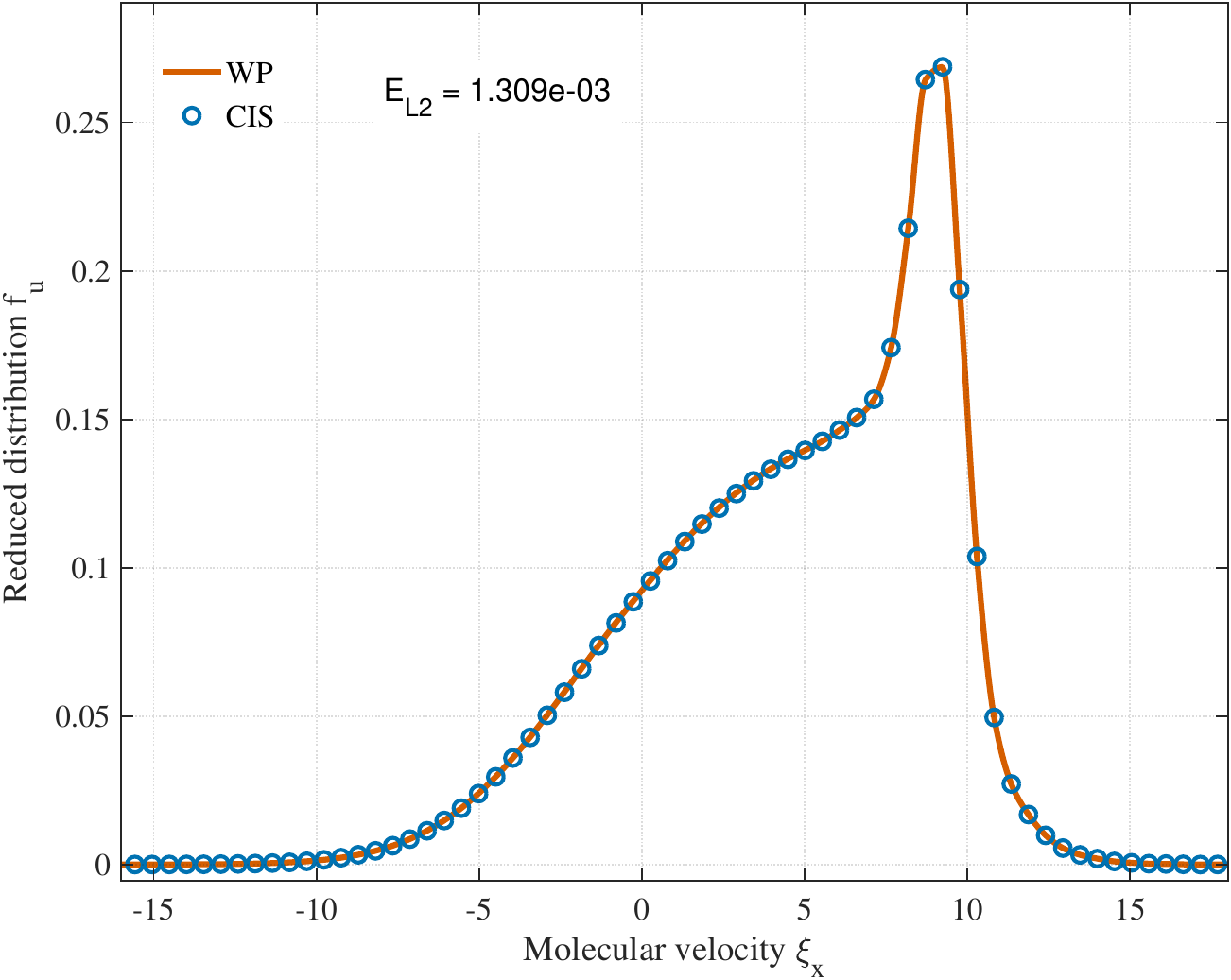}}
\caption{$x-x_{\rho\mathrm{-mid}}=-1.18$}
\end{subfigure}
\caption{Normal shock: reduced distribution $f_u(x,\xi_x)$ at four positions from downstream to upstream of the density midpoint. WP (reconstructed as $W+P$) is shown by lines and CIS by symbols, and $E_{L_2}$ is printed in each panel.}\label{fig:shock-dist}
\end{figure}

\subsection{Couette flow}
Planar Couette flow is computed between two diffuse walls at temperature
$0.5$ a unit distance apart, with periodic conditions in the streamwise
direction. The lower wall is at rest and the upper wall moves at speed
$0.15$, and the gas is initially at rest with $\rho=1$ and $T=0.5$. The
physical mesh has $4\times100$ cells. The velocity grid has $32^3$ nodes on
$[-5,5]^3$ at $\varepsilon=1$, $0.6$, $0.2$ and $10^{-2}$ and
$16\times16\times8$ nodes at $\varepsilon=10^{-4}$, and the coupled WP iteration
runs W1/PR6/P1 at the four larger Knudsen numbers and W256/PR6/P1 at
$\varepsilon=10^{-4}$.

Figure~\ref{fig:couette-profiles} compares the velocity and temperature
profiles at $\varepsilon=1$, $0.6$, $0.2$, $10^{-2}$ and $10^{-4}$. At the
larger Knudsen numbers the velocity slips at the walls and is curved across
the channel, whereas at $\varepsilon=10^{-4}$ it approaches the linear
continuum profile. The relative velocity differences between WP and CIS are
$3.06\times10^{-7}$, $1.43\times10^{-6}$, $1.06\times10^{-6}$,
$4.55\times10^{-8}$ and $8.73\times10^{-4}$ for the five Knudsen numbers, and
the temperature differences are $5.80\times10^{-10}$, $6.10\times10^{-9}$,
$6.52\times10^{-9}$, $8.85\times10^{-11}$ and $1.56\times10^{-5}$. At
$\varepsilon=10^{-4}$ the reference is CIS of the Shakhov
model~\cite{liu_wpd_shakhov}, because Boltzmann CIS could not be converged at
this Knudsen number (Section~\ref{sec:test_gas_model}), and the differences
there include those between the Boltzmann and Shakhov solutions on their
respective velocity grids.

The residual histories in Fig.~\ref{fig:couette-residuals} give
$S_{\rm iter}=1.56$ and $S_t=1.26$ at $\varepsilon=1$, and $50.83$ and
$35.27$ at $\varepsilon=10^{-2}$. At $\varepsilon=10^{-4}$ WP reaches
$R_U=10^{-7}$ in 323 outer iterations, whereas CIS of the Shakhov model on
the same mesh needs $2.385\times10^{6}$ iterations~\cite{liu_wpd_shakhov}, so
that $S_{\rm iter}=7383.90$. The speedup thus grows rapidly as the Knudsen
number decreases, in agreement with Theorem~\ref{thm:acceleration}. CIS
relaxes the macroscopic shear mode only through transport, whereas the W
iteration solves for this mode directly and the endpoint traction transmits
the macroscopic prediction to the particle update. In wall time, WP reaches
$R_U=10^{-7}$ at $\varepsilon=10^{-4}$ in $320.5$~s and the NS solver in
$163.0$~s. The converged Boltzmann solution thus costs $1.97$ times the NS
wall time, and its wall-time ratio to CIS of the Shakhov model is
$S_t=18.14$.

\begin{figure}[tbp]
\centering\singlespacing\setlength{\figW}{\linewidth}
\begin{subfigure}[t]{0.4905\figW}\centering
\raisebox{0.0000\figW}{\includegraphics[width=0.4905\figW]{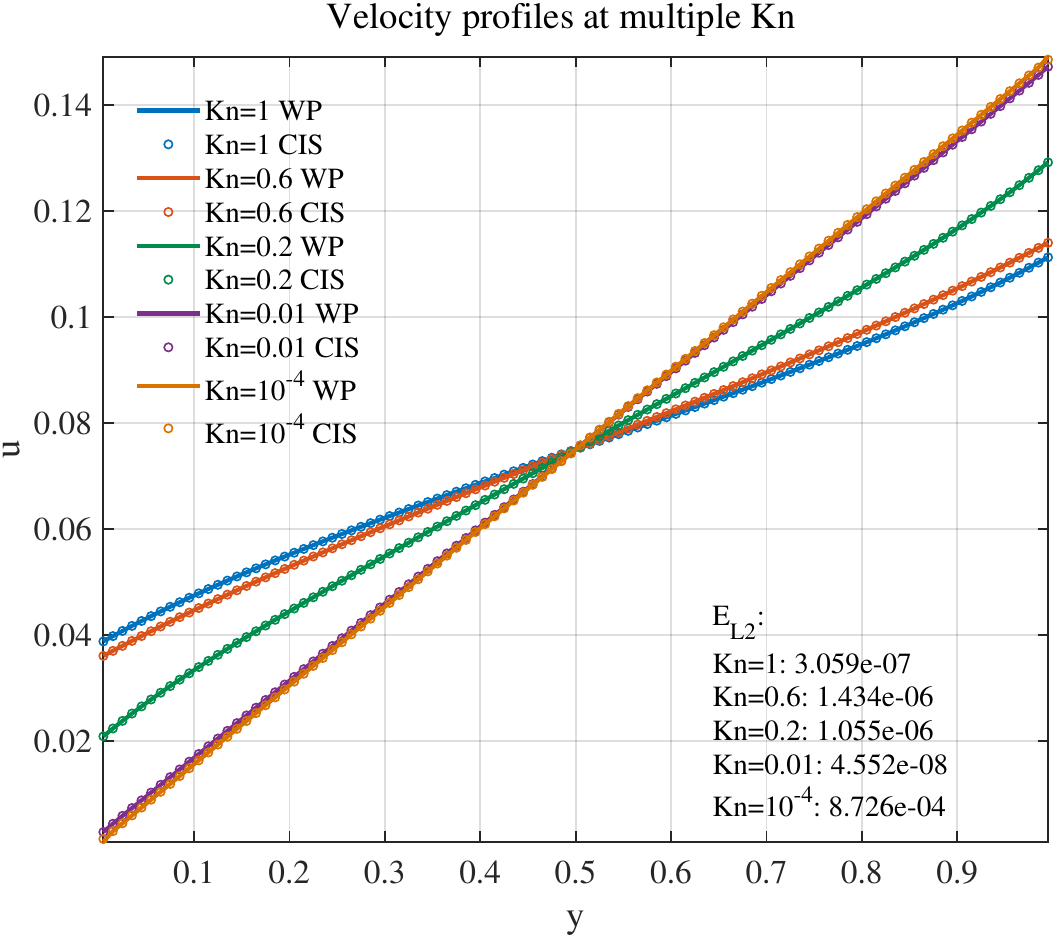}}
\caption{Streamwise velocity}
\end{subfigure}\hspace{0.0200\figW}%
\begin{subfigure}[t]{0.4845\figW}\centering
\raisebox{0.0038\figW}{\includegraphics[width=0.4845\figW]{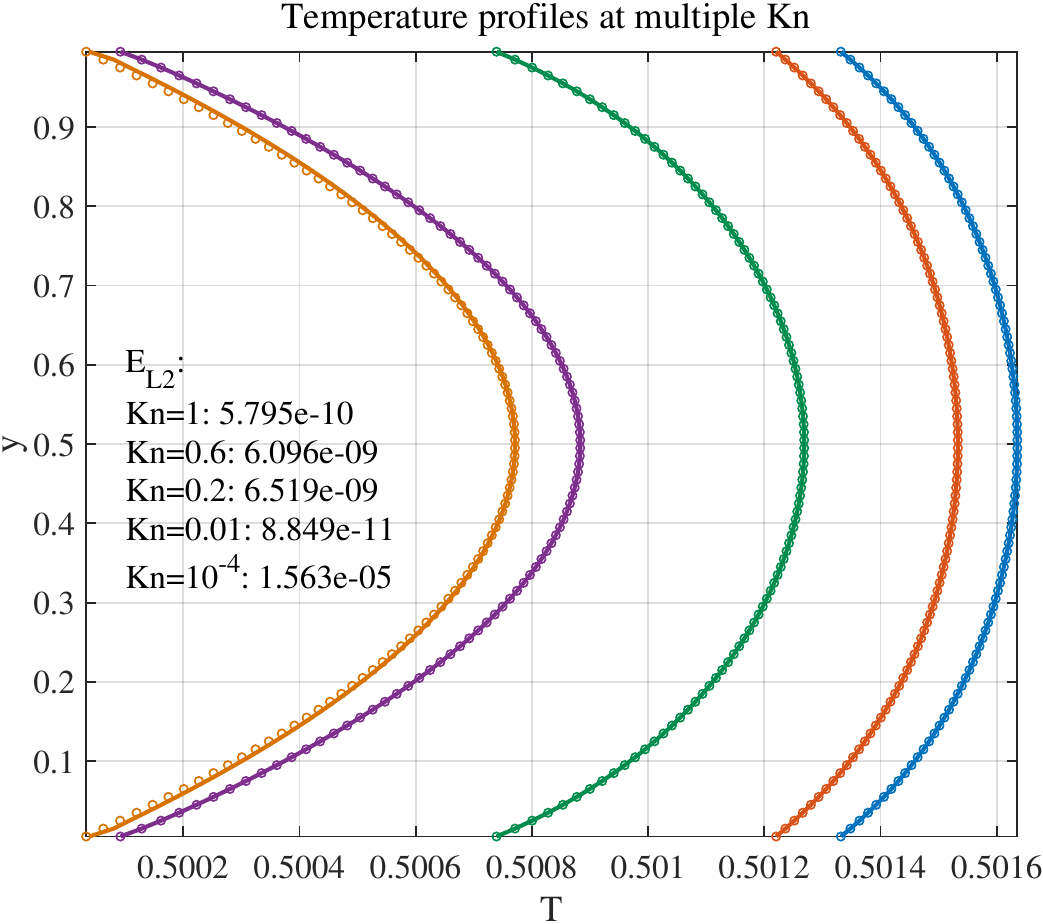}}
\caption{Temperature}
\end{subfigure}
\caption{Couette flow at $\varepsilon=1$, $0.6$, $0.2$, $10^{-2}$ and $10^{-4}$: velocity and temperature across the channel. WP is shown by lines and CIS by symbols, and $E_{L_2}$ is listed for each Knudsen number; the colours and symbols of panel (b) are those of the legend of panel (a). At $\varepsilon=10^{-4}$ the symbols are CIS of the Shakhov model~\cite{liu_wpd_shakhov}.}\label{fig:couette-profiles}
\end{figure}
\begin{figure}[tbp]
\centering\singlespacing\setlength{\figW}{\linewidth}
\begin{subfigure}[t]{0.9463\figW}\centering
\makebox[\linewidth]{\raisebox{0.0019\figW}{\includegraphics[width=0.4726\figW]{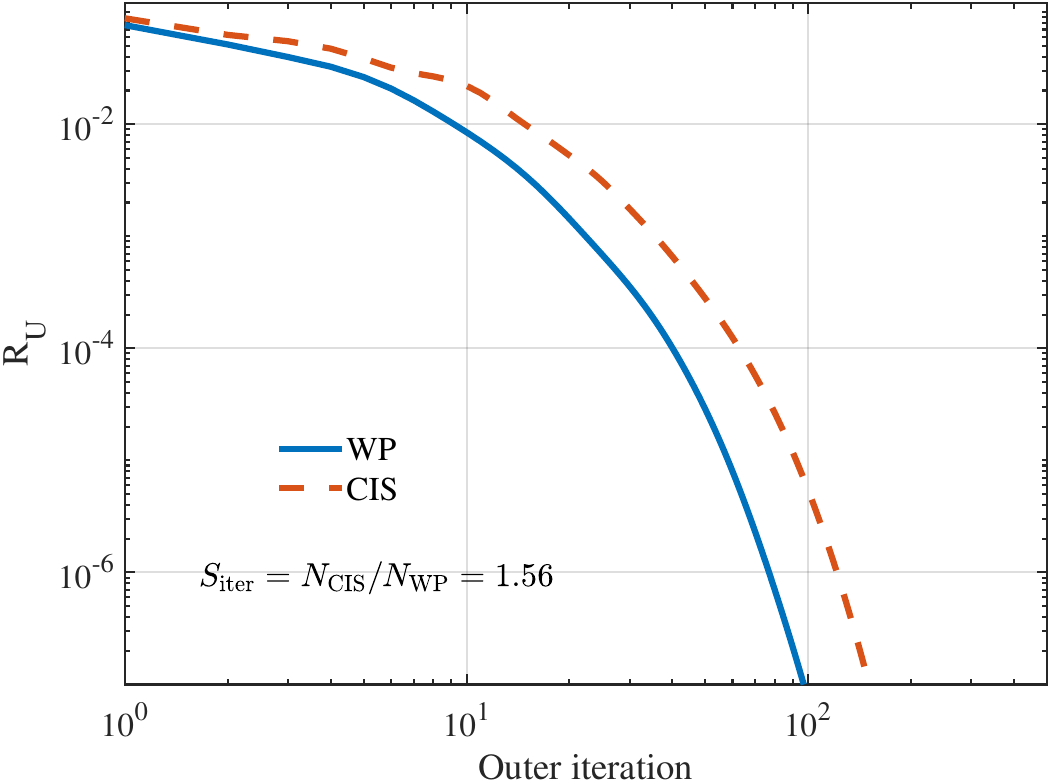}}\hspace{0.0200\figW}\raisebox{0.0000\figW}{\includegraphics[width=0.4537\figW]{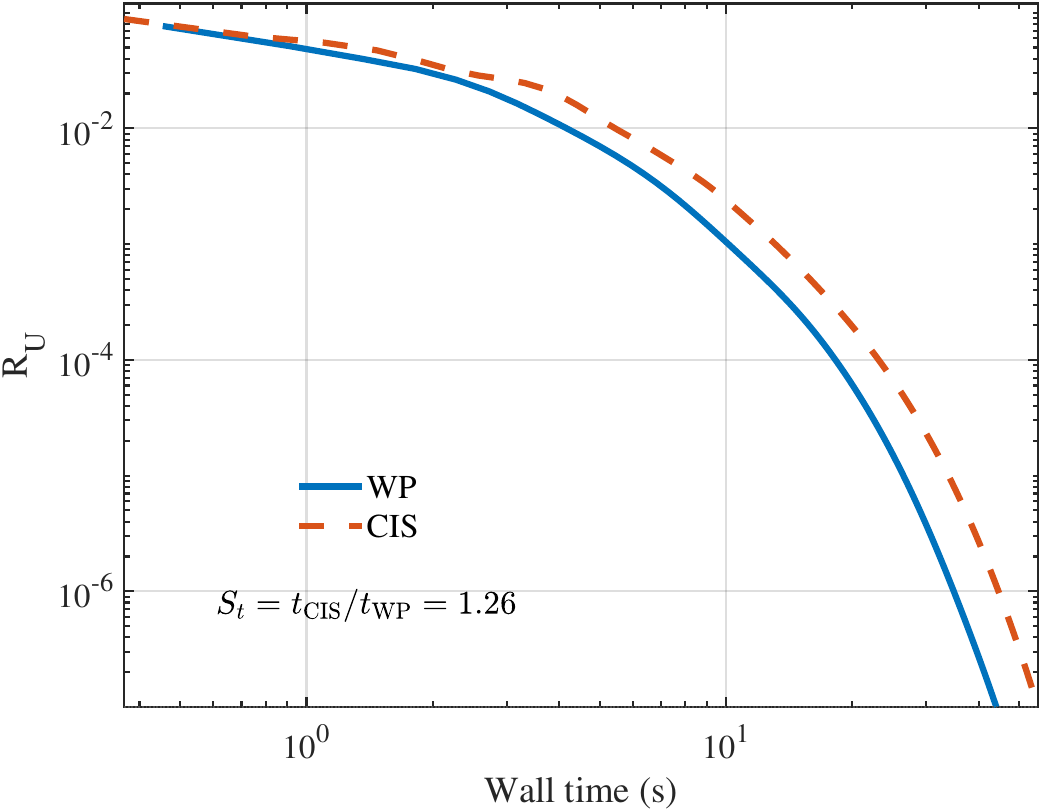}}}
\caption{$\varepsilon=1$}\label{fig:couette-res-kn1}
\end{subfigure}
\par\smallskip
\begin{subfigure}[t]{0.9463\figW}\centering
\makebox[\linewidth]{\raisebox{0.0063\figW}{\includegraphics[width=0.4579\figW]{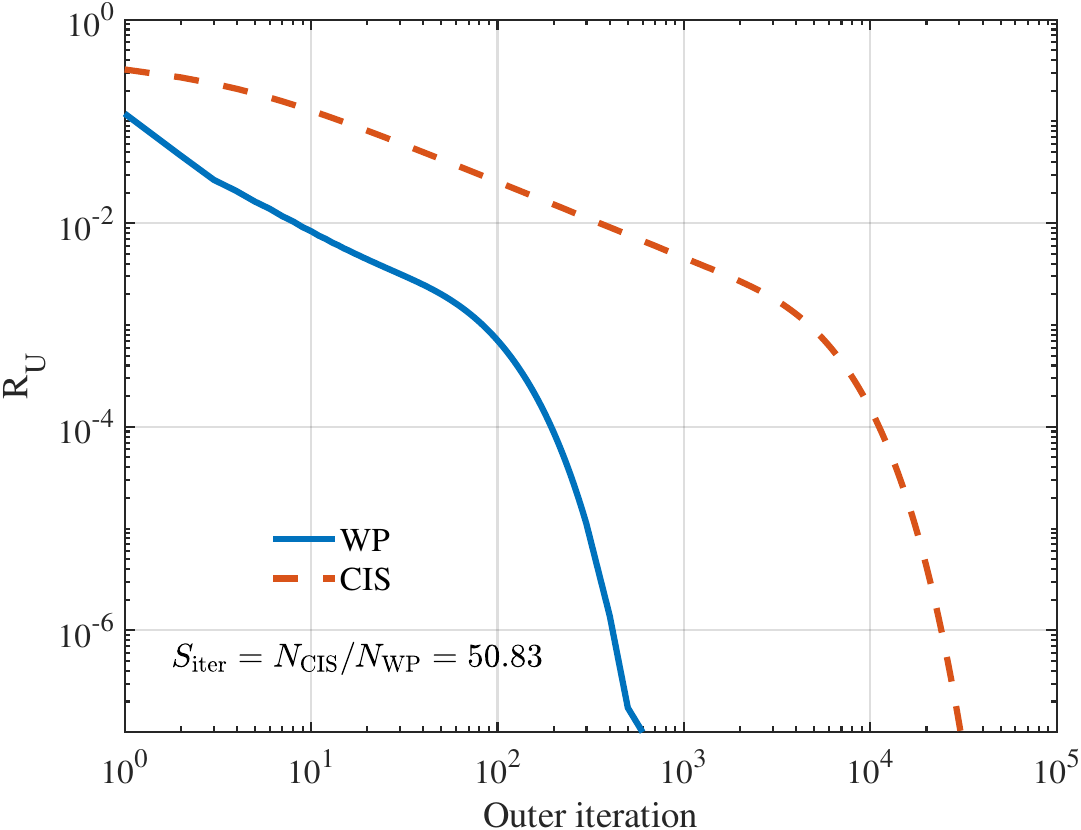}}\hspace{0.0200\figW}\raisebox{0.0000\figW}{\includegraphics[width=0.4684\figW]{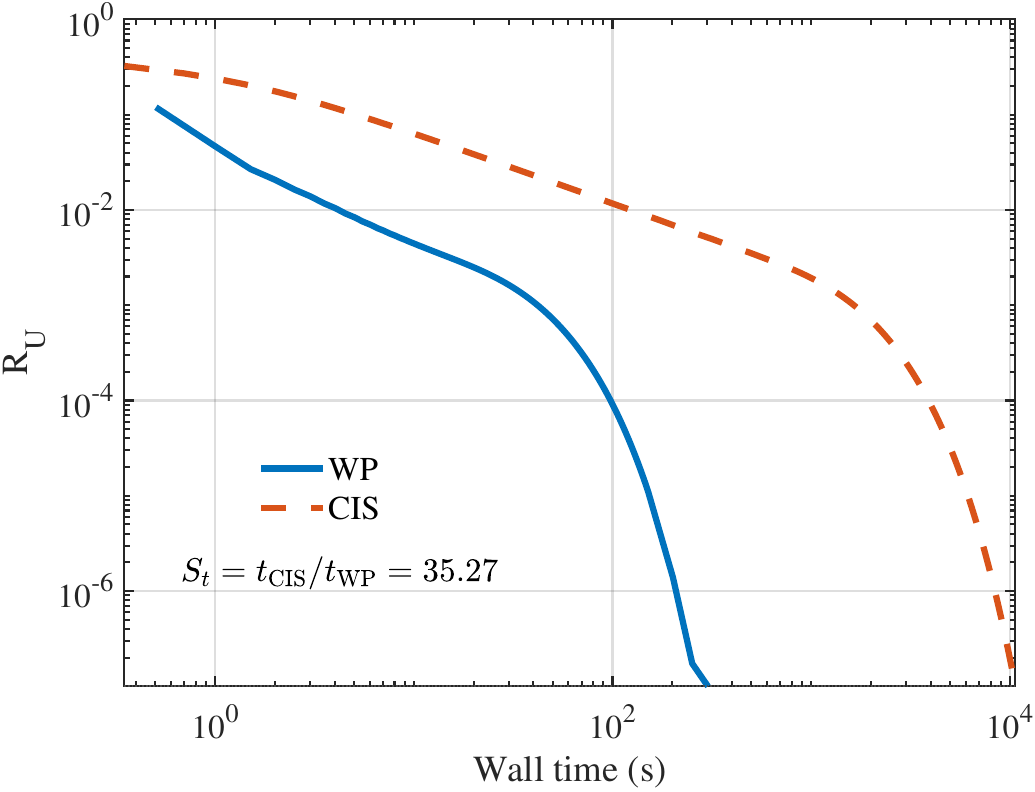}}}
\caption{$\varepsilon=10^{-2}$}\label{fig:couette-res-kn1em2}
\end{subfigure}
\par\smallskip
\begin{subfigure}[t]{0.9463\figW}\centering
\makebox[\linewidth]{\raisebox{0.0029\figW}{\includegraphics[width=0.4632\figW]{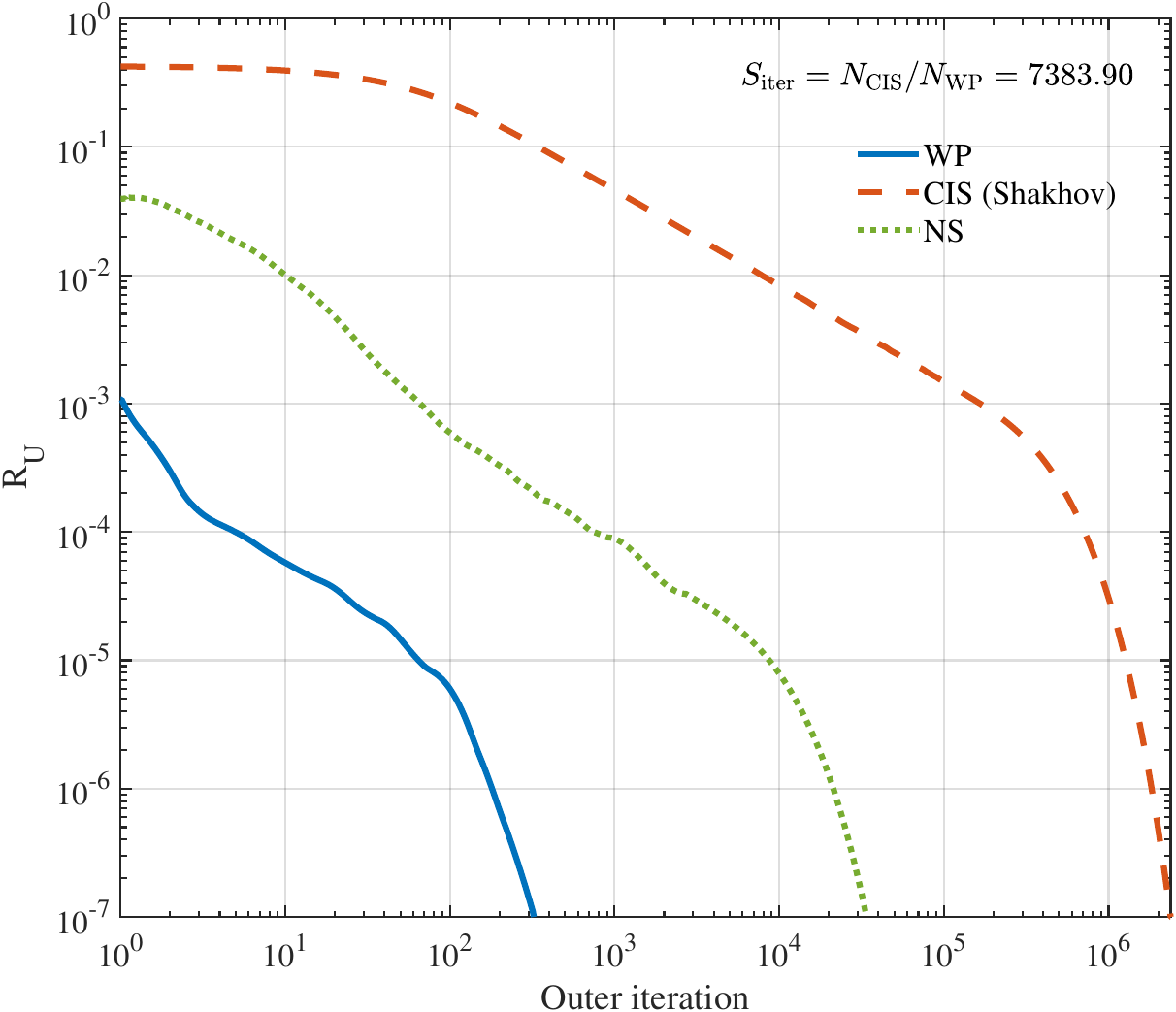}}\hspace{0.0200\figW}\raisebox{0.0000\figW}{\includegraphics[width=0.4632\figW]{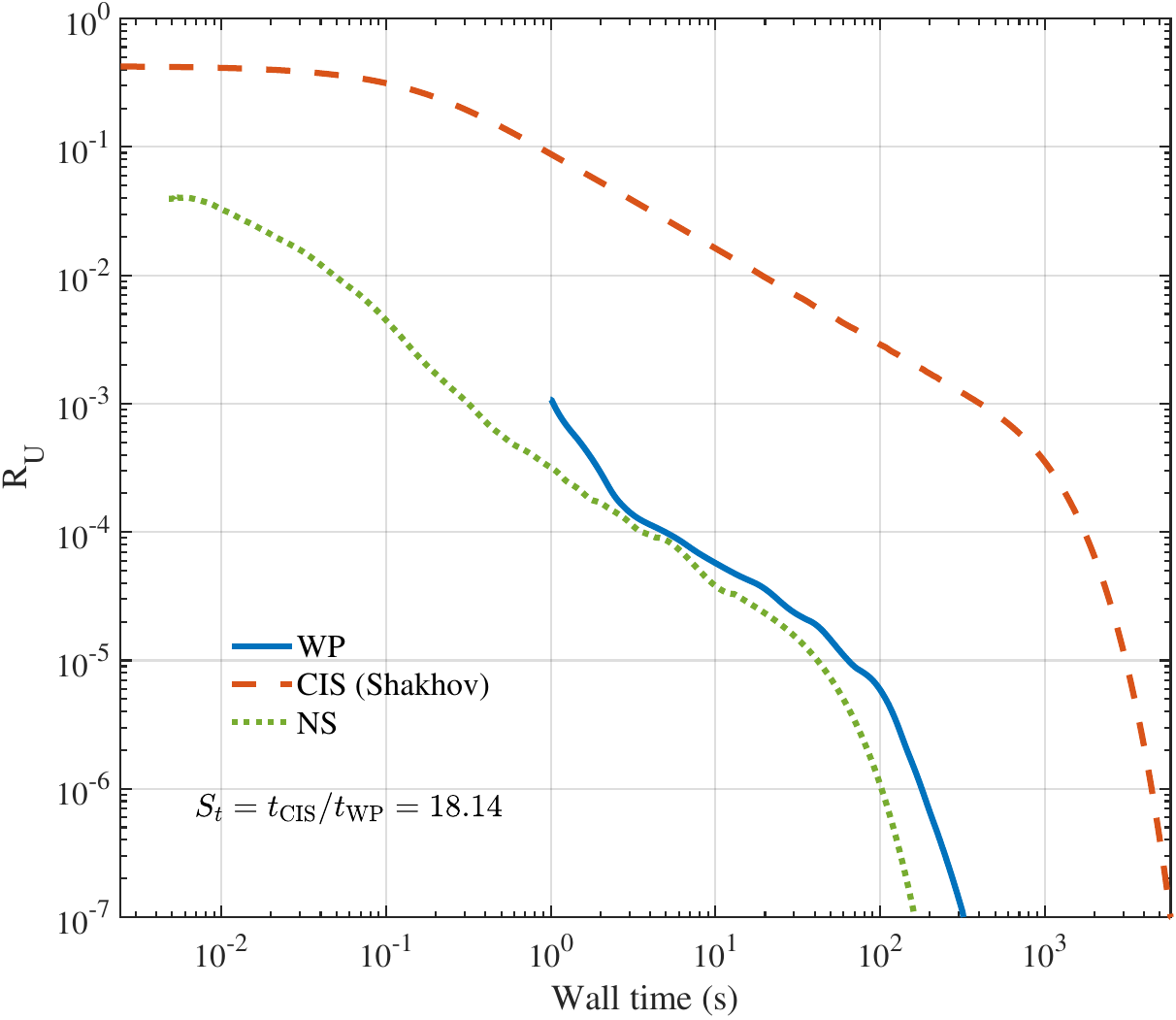}}}
\caption{$\varepsilon=10^{-4}$}\label{fig:couette-res-kn1em4}
\end{subfigure}
\caption{Couette flow: histories of the normalized macroscopic residual against outer iteration (left) and wall time (right). At $\varepsilon=1$ and $10^{-2}$ WP is compared with CIS, and the CIS-to-WP speedups are printed in the panels. At $\varepsilon=10^{-4}$ both panels show CIS of the Shakhov model~\cite{liu_wpd_shakhov} and the NS solver, whose abscissa in the iteration panel counts implicit steps of the NS solver.}\label{fig:couette-residuals}
\end{figure}

\subsection{Lid-driven cavity}
The lid-driven cavity occupies the unit square. Its upper wall moves at
$U_{\rm lid}=0.15$ while the other walls are at rest, and all walls are
diffuse with temperature $0.5$. At $\varepsilon=1$ and $0.075$ the mesh has
$64^2$ cells, the velocity grid has $64\times64\times32$ nodes on $[-5,5]^3$,
and WP runs W1/PR6/P1. The two Reynolds-number cases,
$\mathrm{Re}=\rho_{\rm ref}U_{\rm lid}L/\mu_{\rm ref}=100$ and $1000$ with the
reference density $\rho_{\rm ref}$ and the cavity length $L=1$, use $64^2$
and $100^2$ cells and $32\times32\times12$ velocity nodes, and they replace
the second-order reconstructions of Section~\ref{sec:alg_fluxes} by the
sixth-order central reconstruction of both components. In these cases WP runs
W24/PR6/P1 with implicit CFL number 150, and the NS reference is computed with
the same code, mesh, reconstruction and CFL numbers. Velocities are
normalized by $U_{\rm lid}$, and the vorticity is
$\partial(v/U_{\rm lid})/\partial x-\partial(u/U_{\rm lid})/\partial y$,
where $u$ and $v$ are the horizontal and vertical velocity components.

In the rarefied cases the WP fields in Figs.~\ref{fig:cavity-kn1-fields}
and~\ref{fig:cavity-kn0075-fields} differ from those of CIS by
$8.79\times10^{-10}$ and $9.28\times10^{-10}$ in density,
$3.99\times10^{-10}$ and $5.20\times10^{-10}$ in temperature,
$7.93\times10^{-8}$ and $4.99\times10^{-8}$ in speed and $3.96\times10^{-8}$
and $1.61\times10^{-8}$ in vorticity at the two Knudsen numbers, and the
centreline velocities of Fig.~\ref{fig:cavity-rarefied-profiles} differ by
$9.41\times10^{-8}$ and $1.47\times10^{-7}$ at $\varepsilon=1$ and by
$5.42\times10^{-8}$ and $1.14\times10^{-7}$ at $\varepsilon=0.075$. These
differences lie at the level of the iteration tolerance, as expected for two
iterations that share the same stationary equation.

\begin{figure}[tbp]
\centering\singlespacing\setlength{\figW}{\linewidth}
\begin{subfigure}[t]{0.4871\figW}\centering
\raisebox{0.0000\figW}{\includegraphics[width=0.4871\figW]{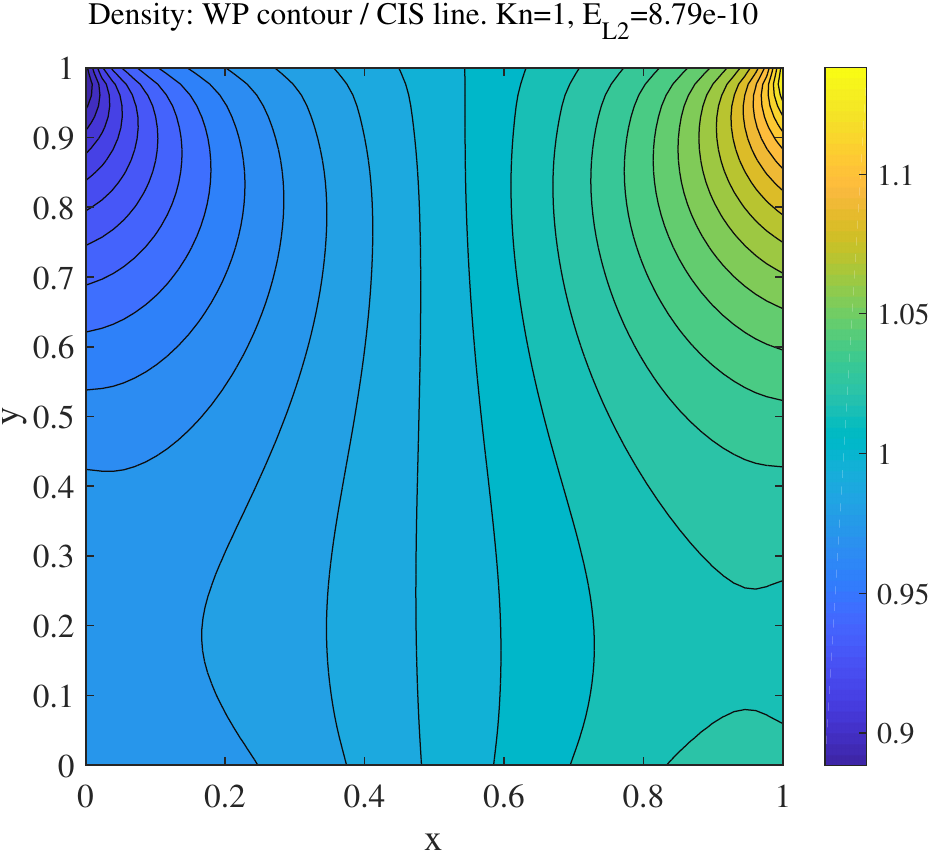}}
\caption{Density}
\end{subfigure}\hspace{0.0200\figW}%
\begin{subfigure}[t]{0.4879\figW}\centering
\raisebox{0.0029\figW}{\includegraphics[width=0.4879\figW]{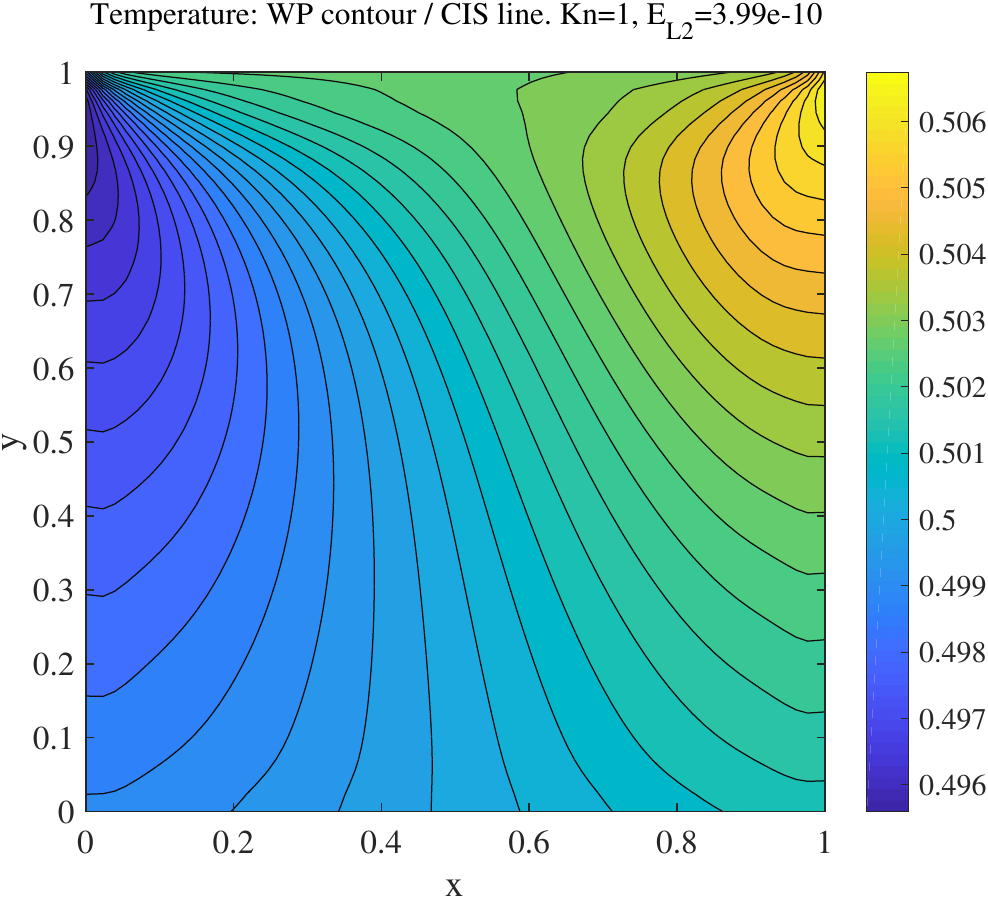}}
\caption{Temperature}
\end{subfigure}
\par\smallskip
\begin{subfigure}[t]{0.4825\figW}\centering
\raisebox{0.0094\figW}{\includegraphics[width=0.4825\figW]{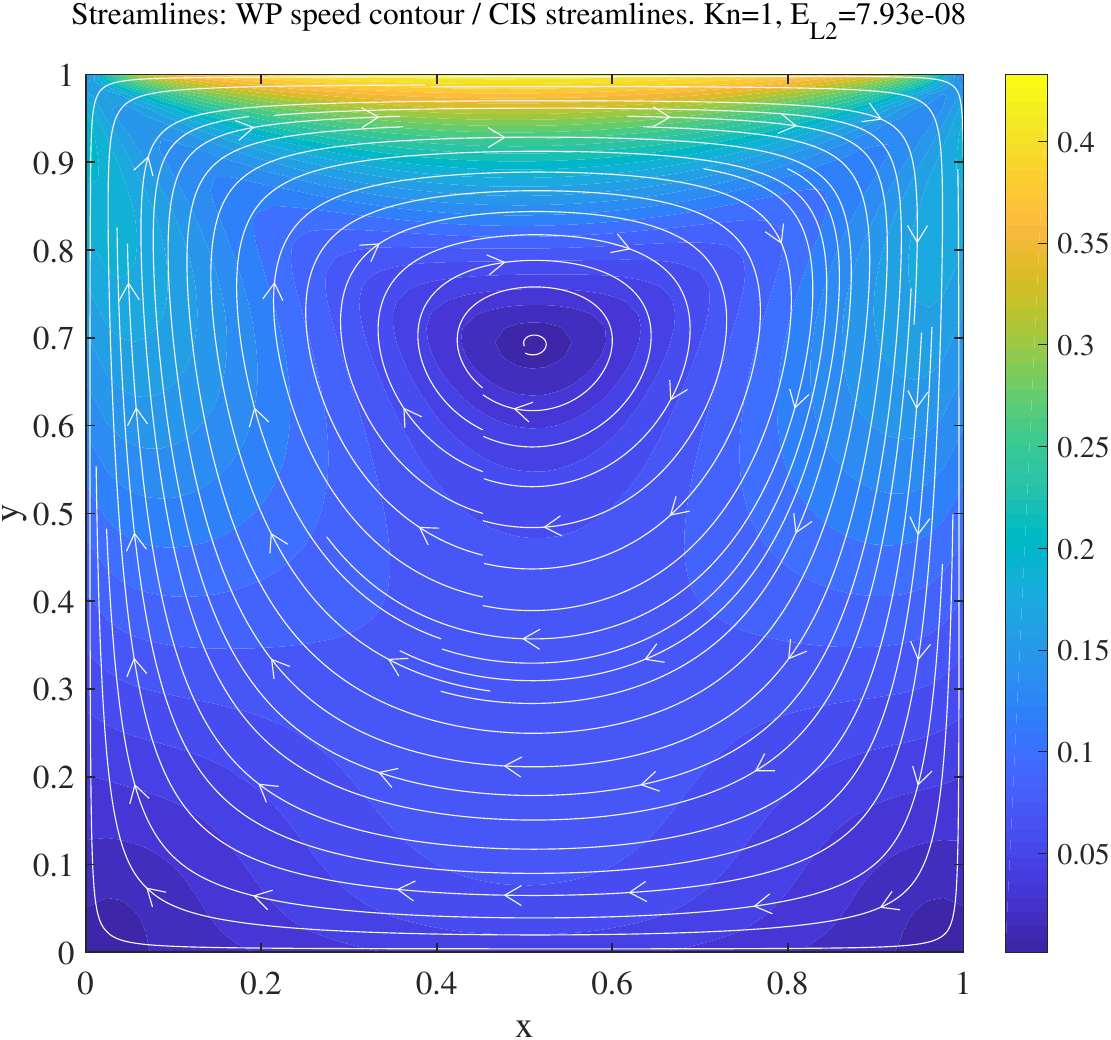}}
\caption{Speed and streamlines}
\end{subfigure}\hspace{0.0200\figW}%
\begin{subfigure}[t]{0.4925\figW}\centering
\raisebox{0.0000\figW}{\includegraphics[width=0.4925\figW]{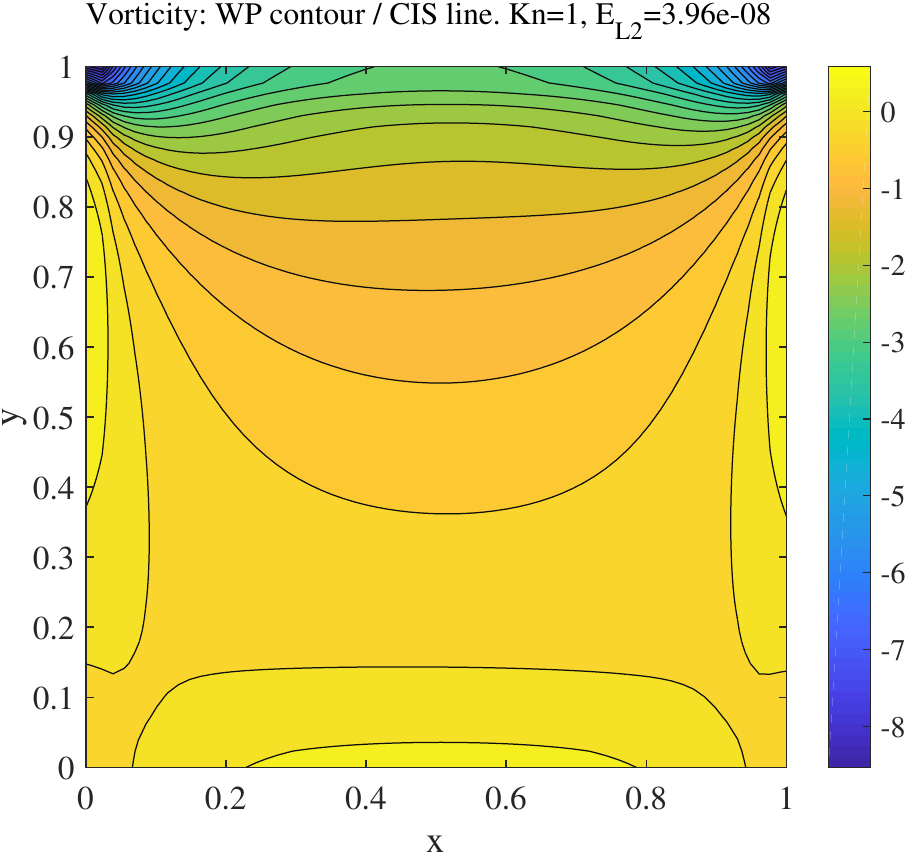}}
\caption{Vorticity}
\end{subfigure}
\caption{Lid-driven cavity at $\varepsilon=1$: density, temperature, speed with streamlines, and vorticity of WP compared with CIS. WP is shown by filled contours and the reference by lines at the same levels, and $E_{L_2}$ is printed above each panel.}\label{fig:cavity-kn1-fields}
\end{figure}
\begin{figure}[tbp]
\centering\singlespacing\setlength{\figW}{\linewidth}
\begin{subfigure}[t]{0.4874\figW}\centering
\raisebox{0.0000\figW}{\includegraphics[width=0.4874\figW]{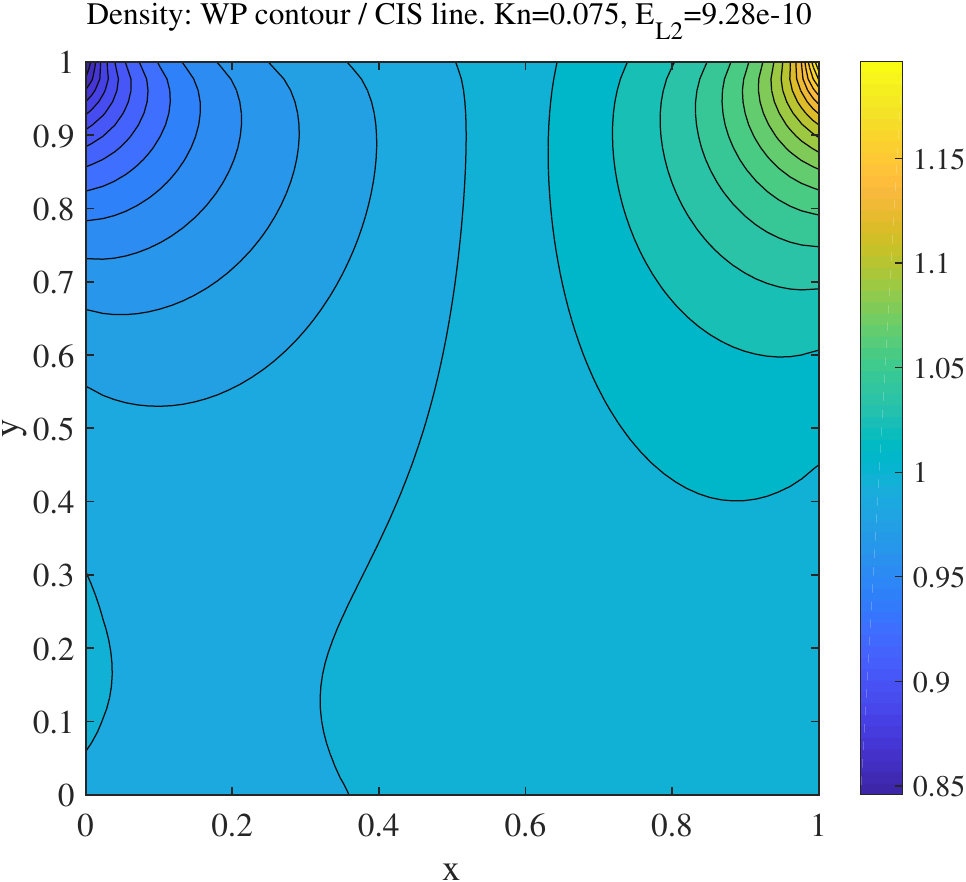}}
\caption{Density}
\end{subfigure}\hspace{0.0200\figW}%
\begin{subfigure}[t]{0.4876\figW}\centering
\raisebox{0.0028\figW}{\includegraphics[width=0.4876\figW]{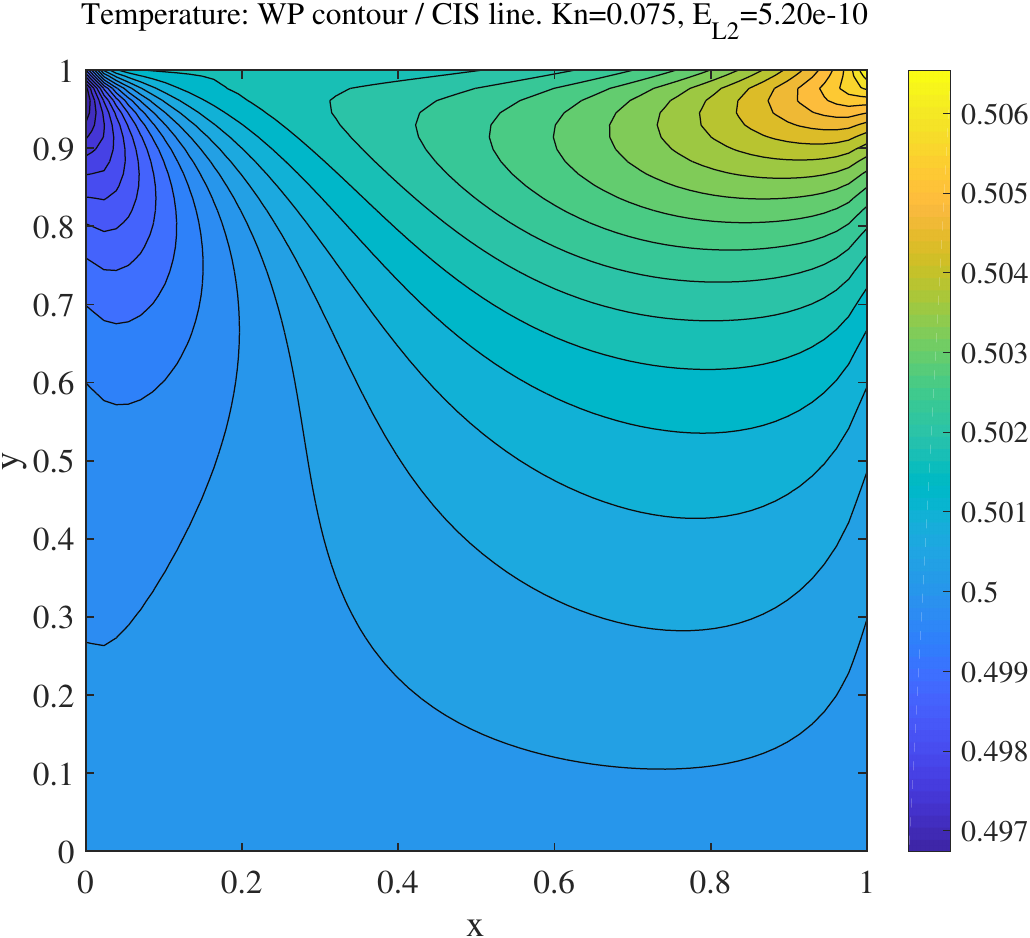}}
\caption{Temperature}
\end{subfigure}
\par\smallskip
\begin{subfigure}[t]{0.4825\figW}\centering
\raisebox{0.0034\figW}{\includegraphics[width=0.4825\figW]{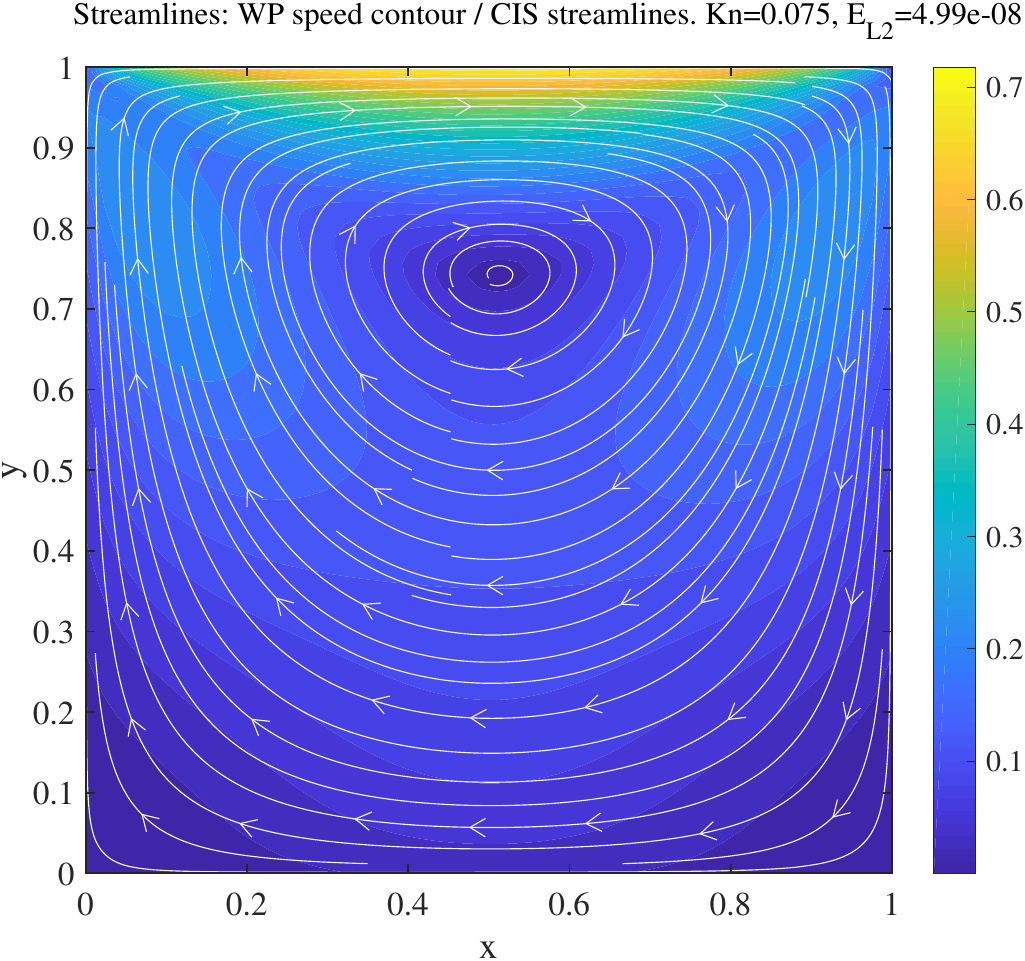}}
\caption{Speed and streamlines}
\end{subfigure}\hspace{0.0200\figW}%
\begin{subfigure}[t]{0.4925\figW}\centering
\raisebox{0.0000\figW}{\includegraphics[width=0.4925\figW]{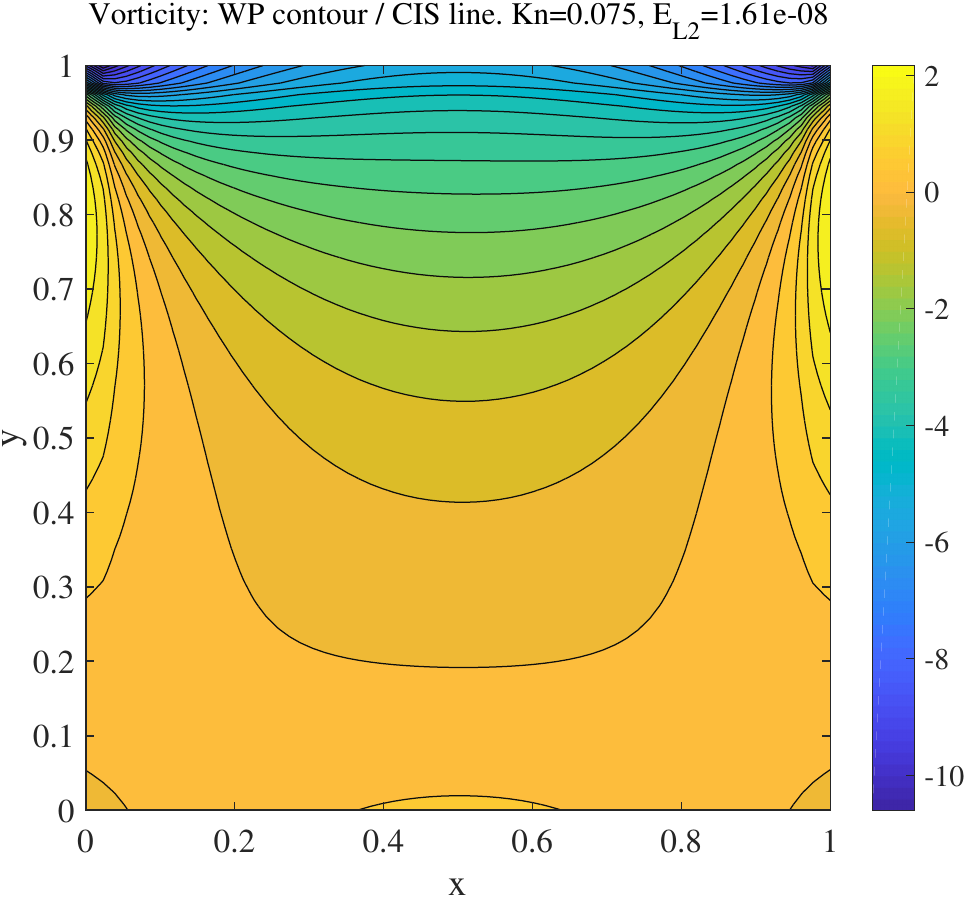}}
\caption{Vorticity}
\end{subfigure}
\caption{Lid-driven cavity at $\varepsilon=0.075$: density, temperature, speed with streamlines, and vorticity of WP compared with CIS. WP is shown by filled contours and the reference by lines at the same levels, and $E_{L_2}$ is printed above each panel.}\label{fig:cavity-kn0075-fields}
\end{figure}
\begin{figure}[tbp]
\centering\singlespacing\setlength{\figW}{\linewidth}
\begin{subfigure}[t]{0.4931\figW}\centering
\raisebox{0.0000\figW}{\includegraphics[width=0.4931\figW]{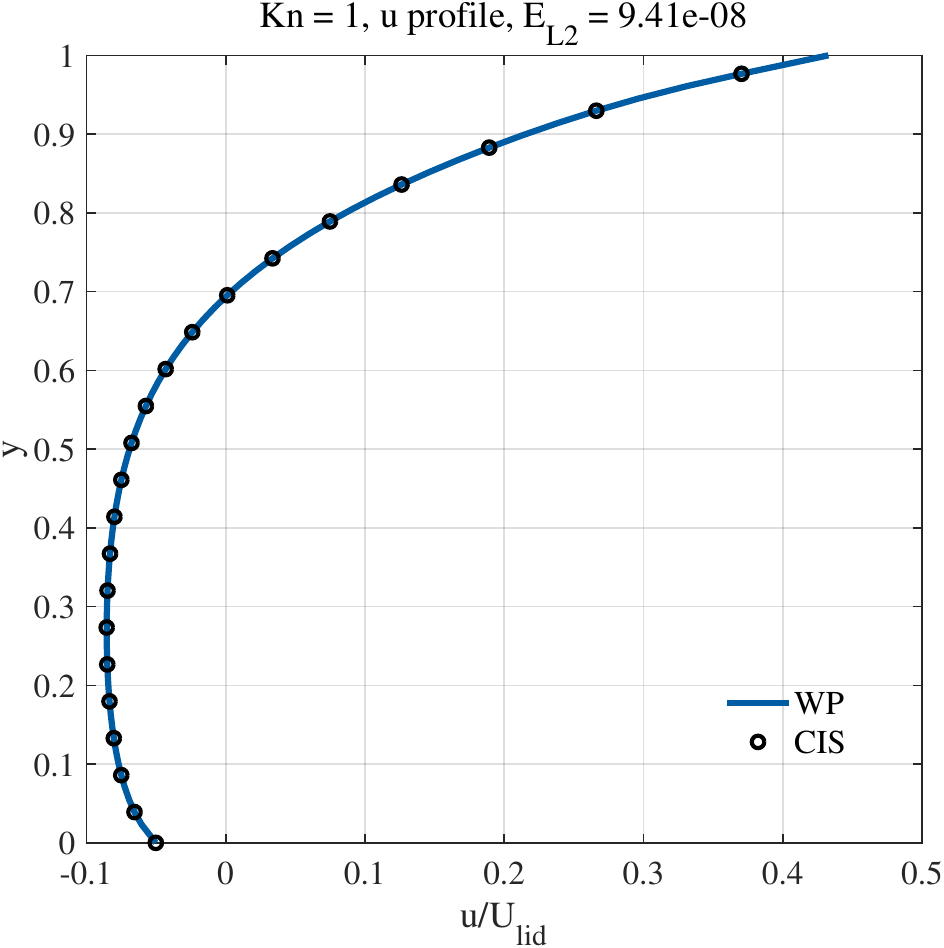}}
\caption{$u/U_{\rm lid}$, $\varepsilon=1$}
\end{subfigure}\hspace{0.0200\figW}%
\begin{subfigure}[t]{0.4819\figW}\centering
\raisebox{0.0111\figW}{\includegraphics[width=0.4819\figW]{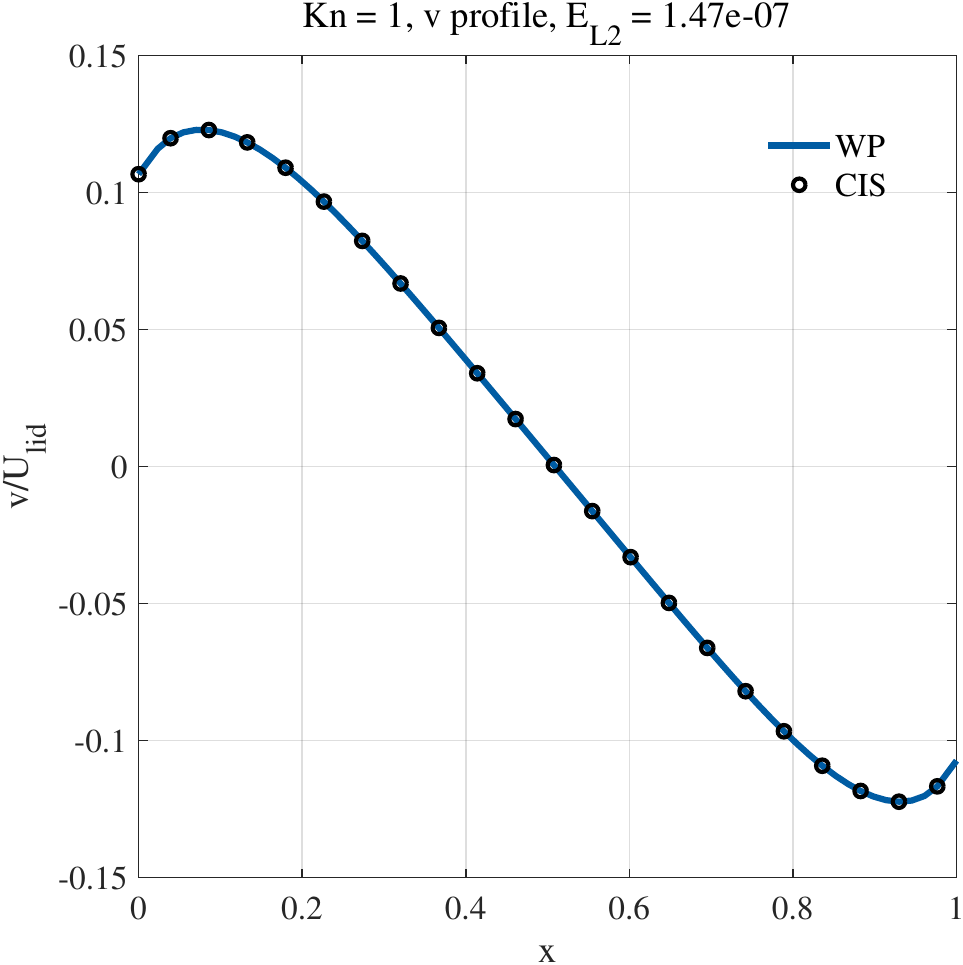}}
\caption{$v/U_{\rm lid}$, $\varepsilon=1$}
\end{subfigure}
\par\smallskip
\begin{subfigure}[t]{0.4830\figW}\centering
\raisebox{0.0000\figW}{\includegraphics[width=0.4830\figW]{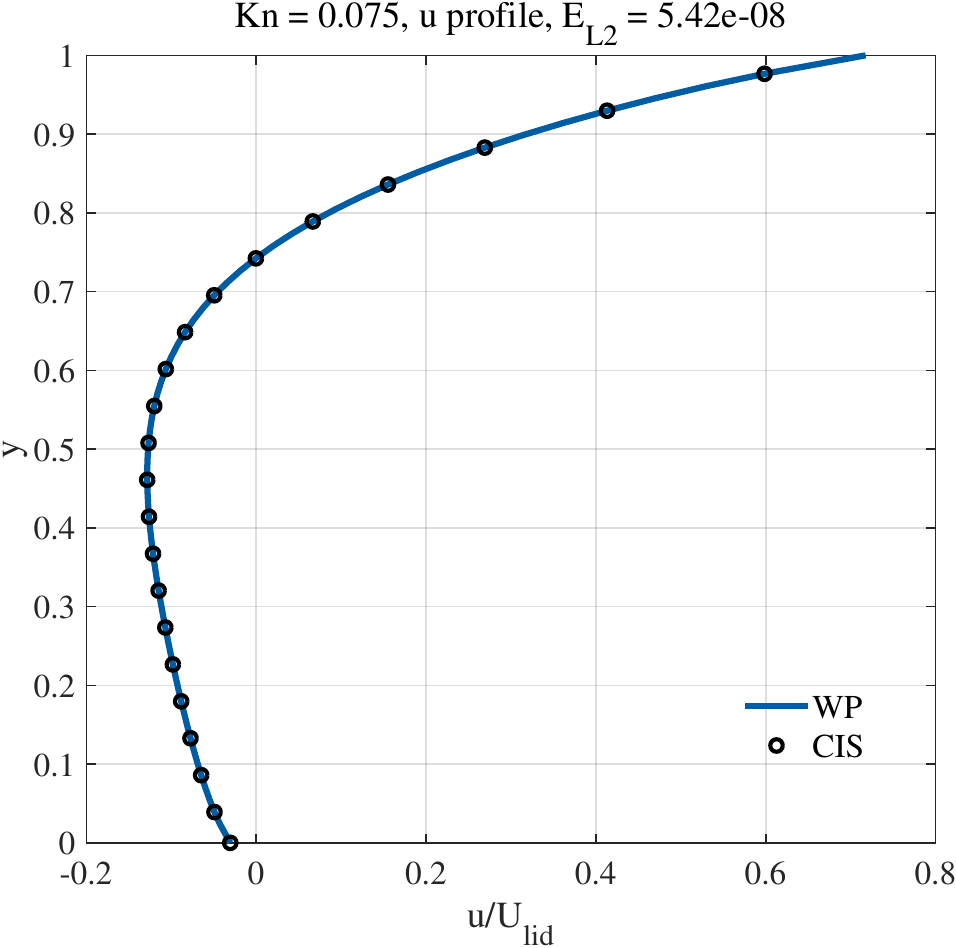}}
\caption{$u/U_{\rm lid}$, $\varepsilon=0.075$}
\end{subfigure}\hspace{0.0200\figW}%
\begin{subfigure}[t]{0.4920\figW}\centering
\raisebox{0.0107\figW}{\includegraphics[width=0.4920\figW]{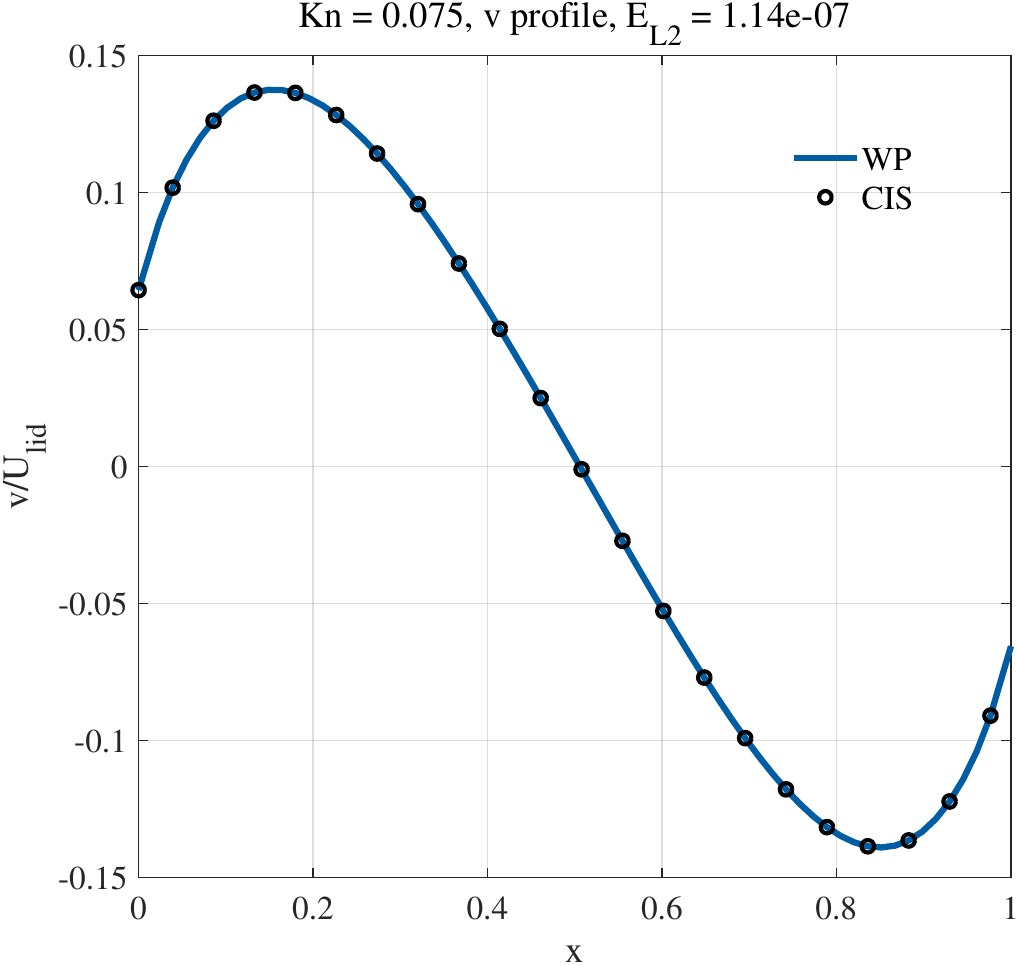}}
\caption{$v/U_{\rm lid}$, $\varepsilon=0.075$}
\end{subfigure}
\caption{Lid-driven cavity at $\varepsilon=1$ (a, b) and $\varepsilon=0.075$ (c, d): horizontal velocity on the vertical centreline and vertical velocity on the horizontal centreline. WP is shown by lines and CIS by symbols.}\label{fig:cavity-rarefied-profiles}
\end{figure}

At $\mathrm{Re}=100$ the WP solution in Fig.~\ref{fig:cavity-re100} differs
from the NS solution on the same mesh by $8.15\times10^{-3}$ in speed and
$5.14\times10^{-2}$ in vorticity, and its centreline velocity profiles differ
from the Navier--Stokes benchmark of Ghia et al.~\cite{ghia1982} by
$2.09\times10^{-3}$ and $8.46\times10^{-3}$. At $\mathrm{Re}=1000$
(Fig.~\ref{fig:cavity-re1000}) the stronger recirculation and the thinner
wall layers raise these differences to $1.58\times10^{-2}$ in speed and
$6.16\times10^{-2}$ in vorticity and to $1.44\times10^{-2}$ and
$9.31\times10^{-3}$ in the centreline profiles. Both flows are smooth and
close to equilibrium, as assumed in Theorem~\ref{thm:ap}(ii), and in both the
WP solution approaches the Navier--Stokes solution, with the largest
difference in the vorticity, which is a derivative of the velocity.

The CIS profiles in panels (c, d) of both figures depart from the benchmark,
slightly at $\mathrm{Re}=100$ and markedly at $\mathrm{Re}=1000$. The
conventional scheme transports the full distribution with a flux that does
not account for the collisions within a cell, so that its numerical
dissipation, which scales with the cell size rather than with the mean free
path, adds to the physical viscosity. The smaller the physical viscosity, the
larger the share of this numerical dissipation and the stronger the damping
of the primary vortex. The macroscopic flux of WP, by contrast, is the
gas-kinetic Navier--Stokes flux, whose dissipation in the continuum limit is
the physical viscosity (Theorem~\ref{thm:ap}(ii)), and WP therefore
reproduces the benchmark on the same mesh at both Reynolds numbers. The
comparison thus displays the asymptotic-preserving accuracy of WP in addition
to its faster convergence.

\begin{figure}[tbp]
\centering\singlespacing\setlength{\figW}{\linewidth}
\begin{subfigure}[t]{0.4783\figW}\centering
\raisebox{0.0069\figW}{\includegraphics[width=0.4783\figW]{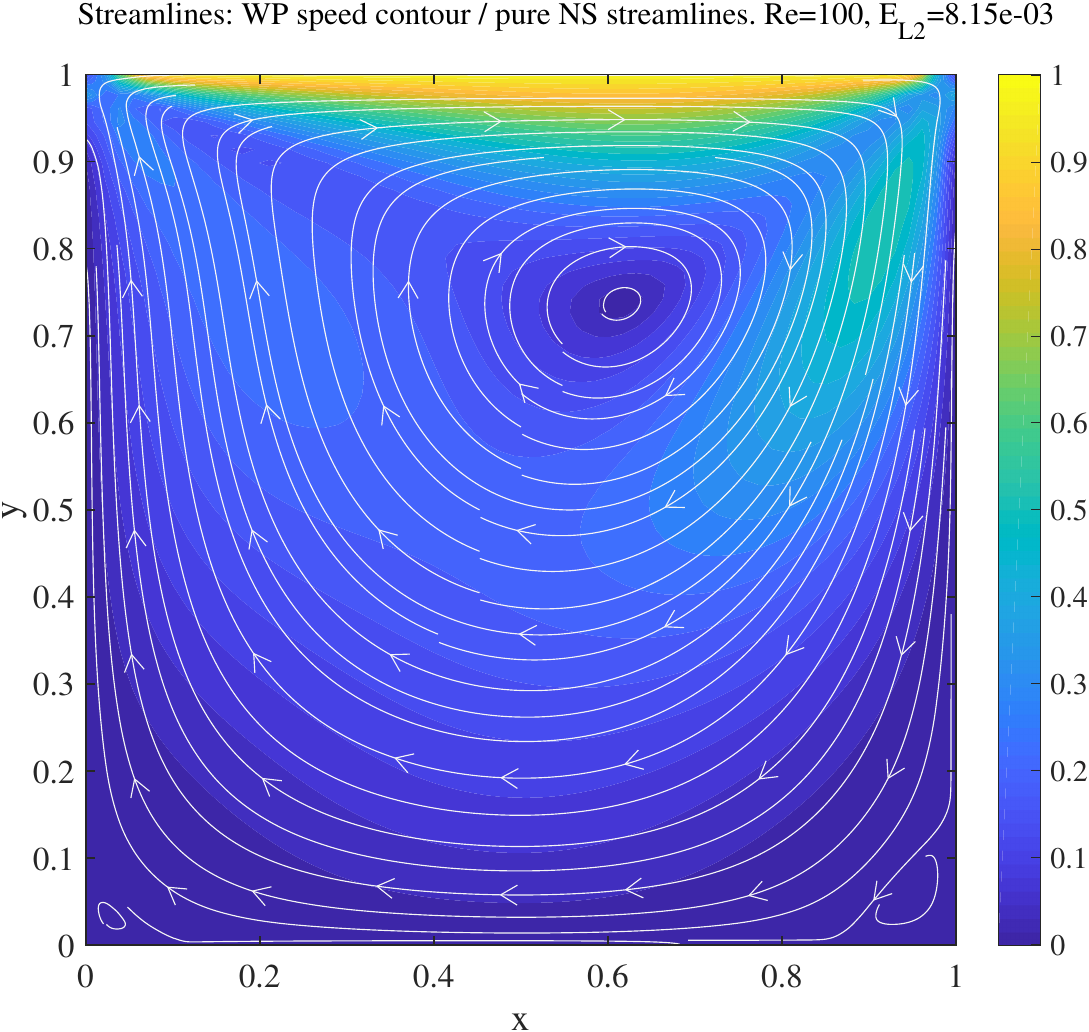}}
\caption{Speed and streamlines}
\end{subfigure}\hspace{0.0200\figW}%
\begin{subfigure}[t]{0.4967\figW}\centering
\raisebox{0.0000\figW}{\includegraphics[width=0.4967\figW]{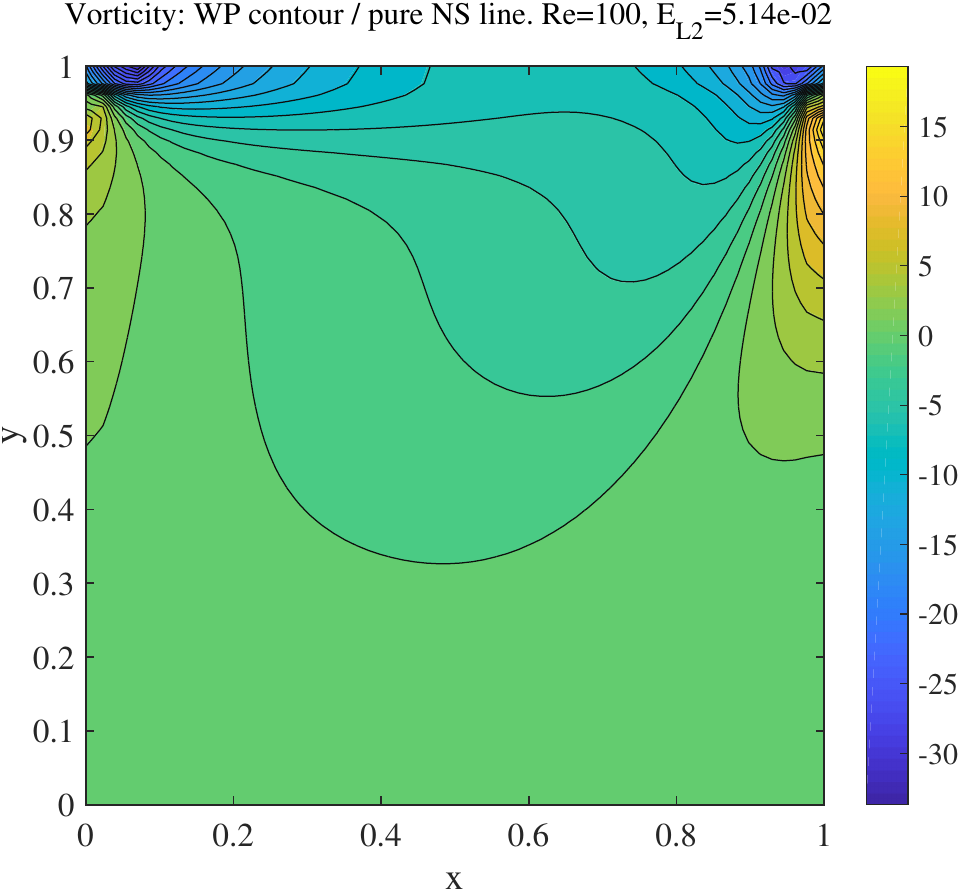}}
\caption{Vorticity}
\end{subfigure}
\par\smallskip
\begin{subfigure}[t]{0.4832\figW}\centering
\raisebox{0.0000\figW}{\includegraphics[width=0.4832\figW]{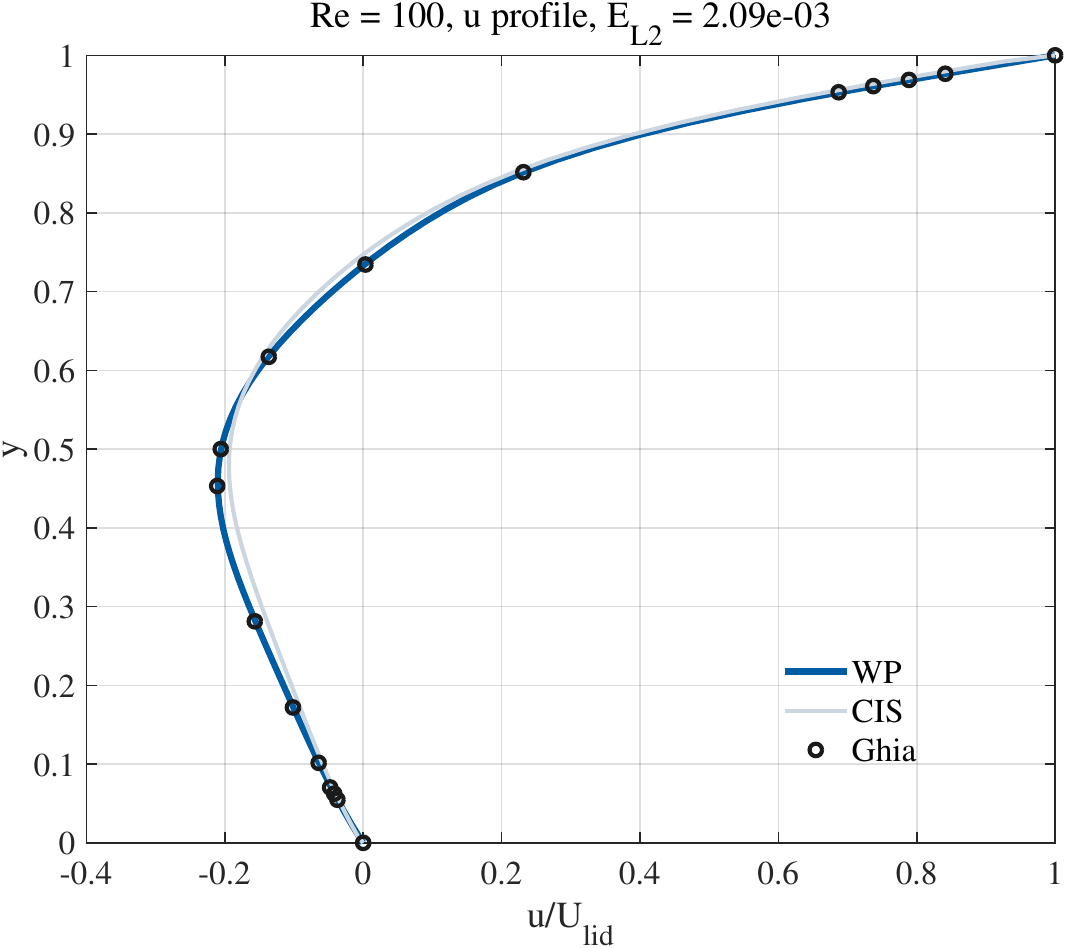}}
\caption{$u/U_{\rm lid}$ on the vertical centreline}
\end{subfigure}\hspace{0.0200\figW}%
\begin{subfigure}[t]{0.4918\figW}\centering
\raisebox{0.0099\figW}{\includegraphics[width=0.4918\figW]{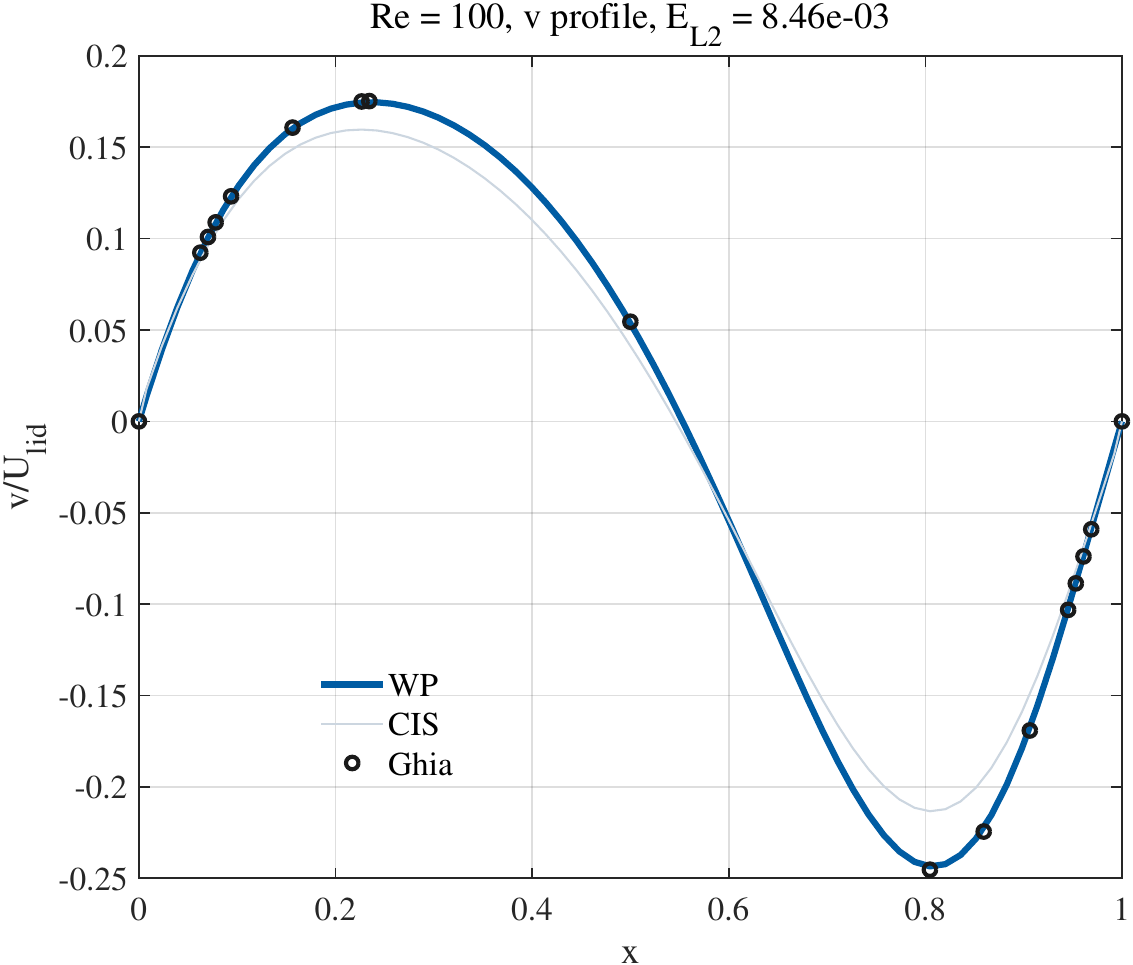}}
\caption{$v/U_{\rm lid}$ on the horizontal centreline}
\end{subfigure}
\caption{Lid-driven cavity at $\mathrm{Re}=100$: speed with streamlines (a) and vorticity (b) of WP compared with the NS solution, shown by filled contours and by lines at the same levels, and centreline velocity profiles (c, d) of WP (thick lines) compared with CIS (thin lines) and with the Navier--Stokes benchmark data of Ghia et al.~\cite{ghia1982} (symbols). $E_{L_2}$, printed above each panel, is the difference between WP and the NS solution in (a, b) and between WP and the benchmark in (c, d).}\label{fig:cavity-re100}
\end{figure}
\begin{figure}[tbp]
\centering\singlespacing\setlength{\figW}{\linewidth}
\begin{subfigure}[t]{0.4797\figW}\centering
\raisebox{0.0056\figW}{\includegraphics[width=0.4797\figW]{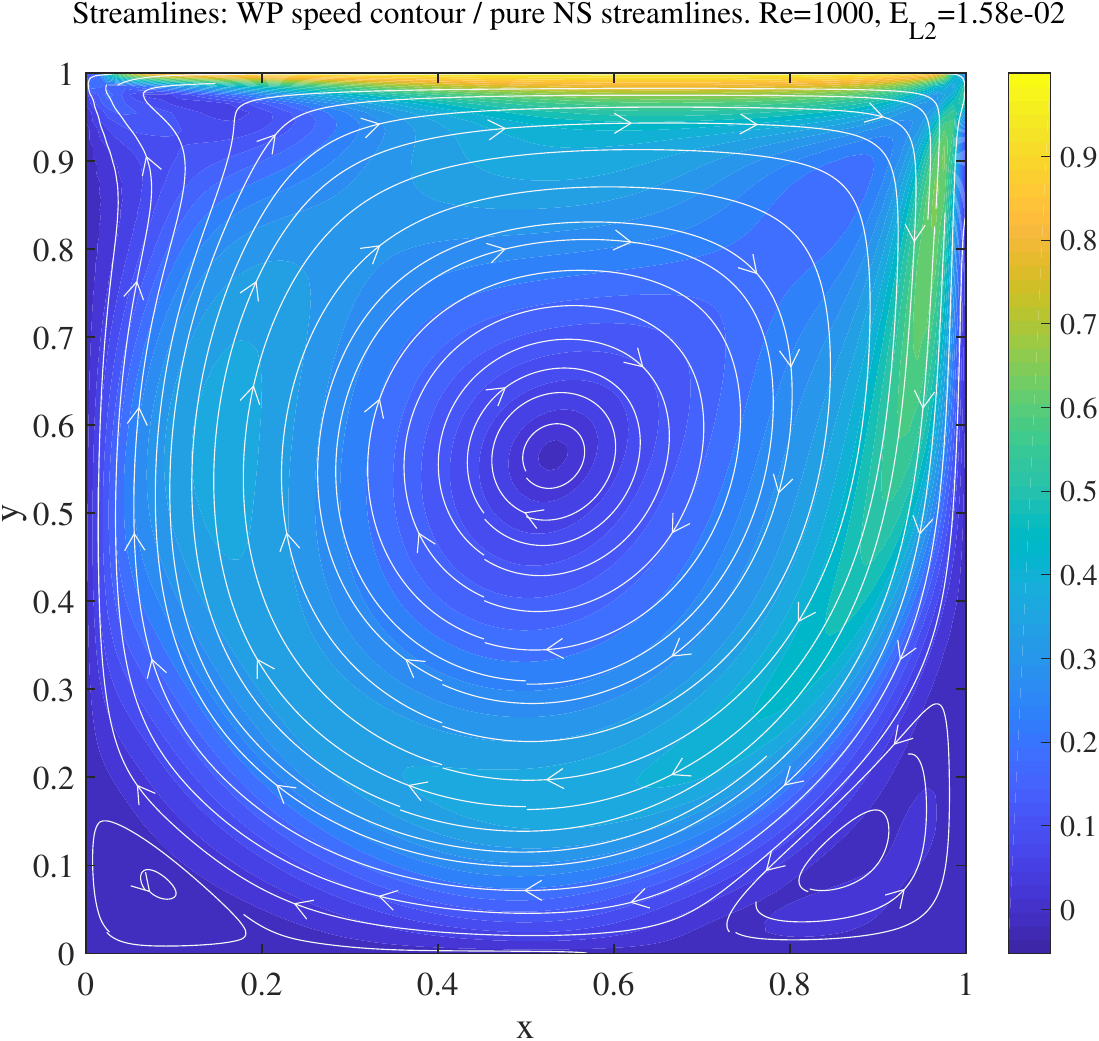}}
\caption{Speed and streamlines}
\end{subfigure}\hspace{0.0200\figW}%
\begin{subfigure}[t]{0.4953\figW}\centering
\raisebox{0.0000\figW}{\includegraphics[width=0.4953\figW]{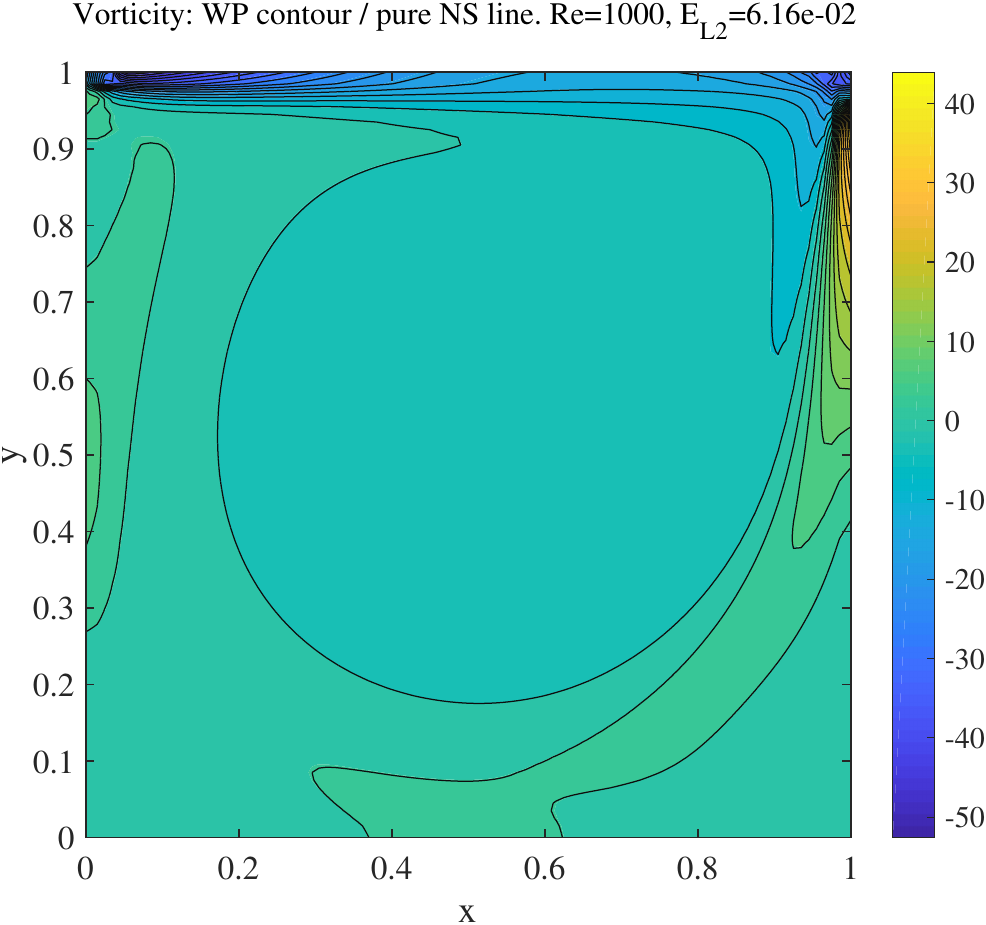}}
\caption{Vorticity}
\end{subfigure}
\par\smallskip
\begin{subfigure}[t]{0.4915\figW}\centering
\raisebox{0.0000\figW}{\includegraphics[width=0.4915\figW]{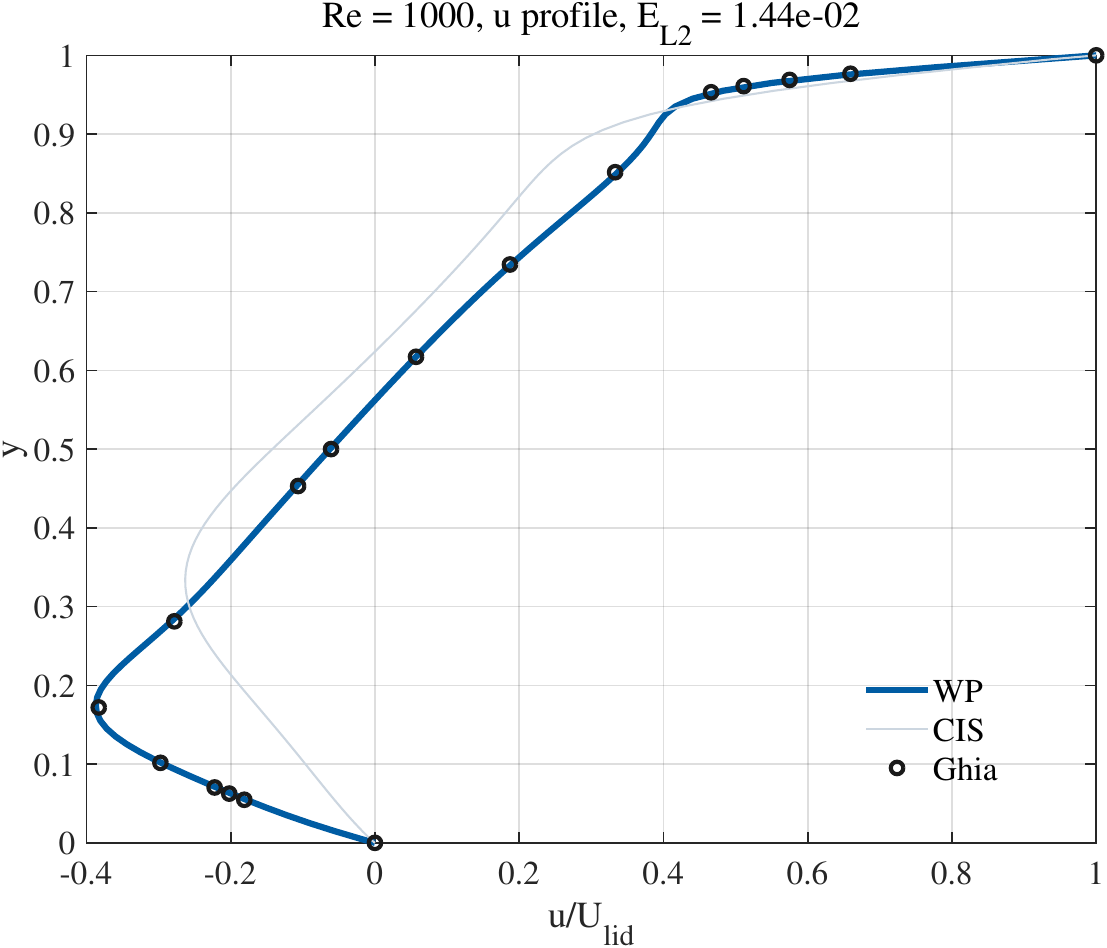}}
\caption{$u/U_{\rm lid}$ on the vertical centreline}
\end{subfigure}\hspace{0.0200\figW}%
\begin{subfigure}[t]{0.4835\figW}\centering
\raisebox{0.0094\figW}{\includegraphics[width=0.4835\figW]{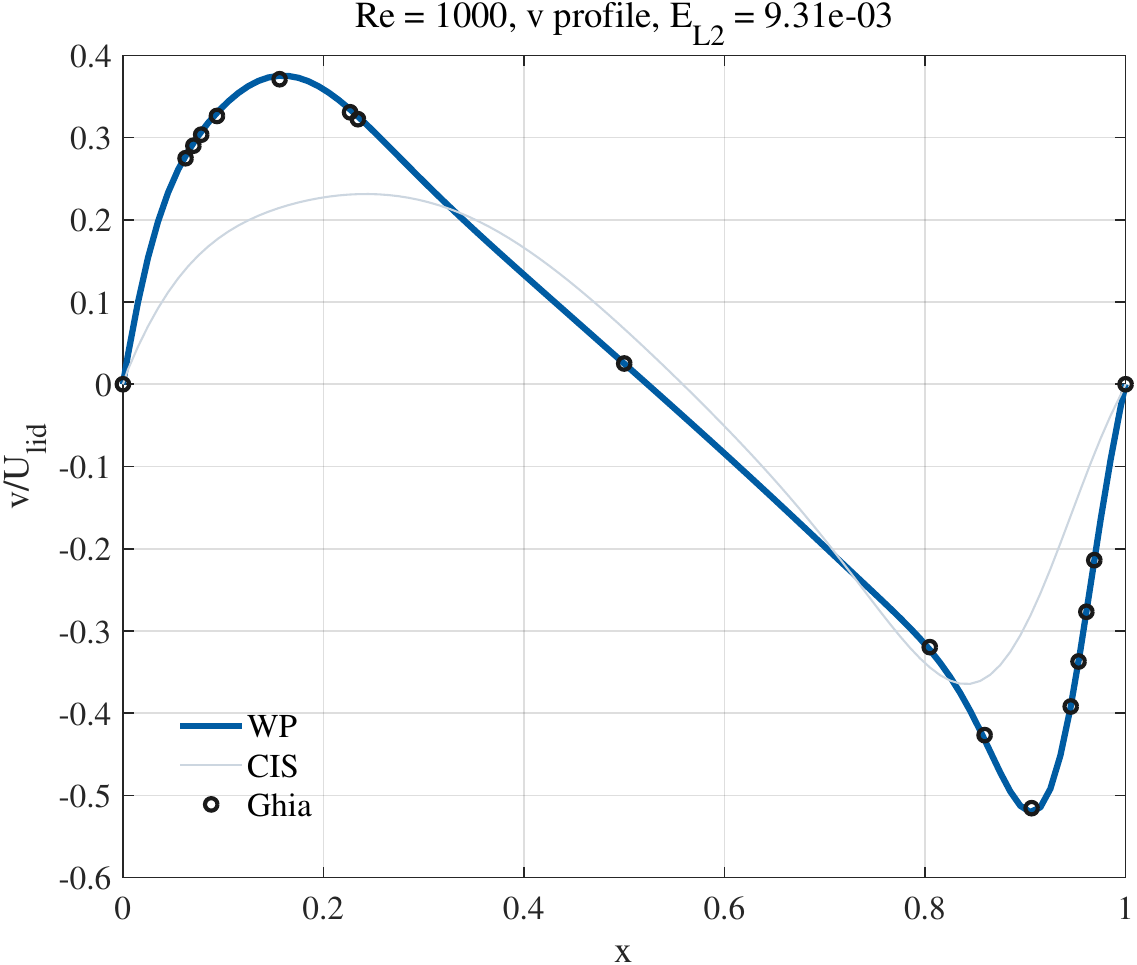}}
\caption{$v/U_{\rm lid}$ on the horizontal centreline}
\end{subfigure}
\caption{Lid-driven cavity at $\mathrm{Re}=1000$: speed with streamlines (a) and vorticity (b) of WP compared with the NS solution, shown by filled contours and by lines at the same levels, and centreline velocity profiles (c, d) of WP (thick lines) compared with CIS of the Shakhov model~\cite{liu_wpd_shakhov} (thin lines) and with the Navier--Stokes benchmark data of Ghia et al.~\cite{ghia1982} (symbols). $E_{L_2}$, printed above each panel, is the difference between WP and the NS solution in (a, b) and between WP and the benchmark in (c, d).}\label{fig:cavity-re1000}
\end{figure}

Figure~\ref{fig:cavity-residuals} shows the residual histories of the four
cavity calculations. The speedups are $S_{\rm iter}=1.78$ and $S_t=1.36$ at
$\varepsilon=1$ and increase to $9.47$ and $7.55$ at $\varepsilon=0.075$. At
$\mathrm{Re}=100$, where Boltzmann CIS converges as well, they reach
$S_{\rm iter}=758.11$ and $S_t=534.84$, so that from the rarefied to the
continuum cases the wall-time speedup grows by more than two orders of
magnitude, consistent with Theorem~\ref{thm:acceleration}. At
$\mathrm{Re}=1000$ the ratios with respect to CIS of the Shakhov
model~\cite{liu_wpd_shakhov} are $S_{\rm iter}=3797.50$ and $S_t=33.18$. For
both Reynolds numbers the wall time is also compared with that of the NS
solver, whose curves are included in Fig.~\ref{fig:cavity-residuals}. WP
reaches the residual $10^{-7}$ in $100.1$~s at $\mathrm{Re}=100$ and in
$1342.7$~s at $\mathrm{Re}=1000$, whereas the NS solver requires $48.0$~s and
$1275.0$~s. The converged full Boltzmann solution on the $32\times32\times12$
velocity grid thus costs only $2.09$ and $1.05$ times the NS wall time. At
$\mathrm{Re}=1000$ the kinetic work adds less than a tenth to the cost of the
NS computation, and the Boltzmann solution is obtained at nearly the
cost of the macroscopic solution.

\begin{figure}[tbp]
\centering\singlespacing\setlength{\figW}{\linewidth}
\begin{subfigure}[t]{0.6436\figW}\centering
\makebox[\linewidth]{\raisebox{0.0022\figW}{\includegraphics[width=0.3118\figW]{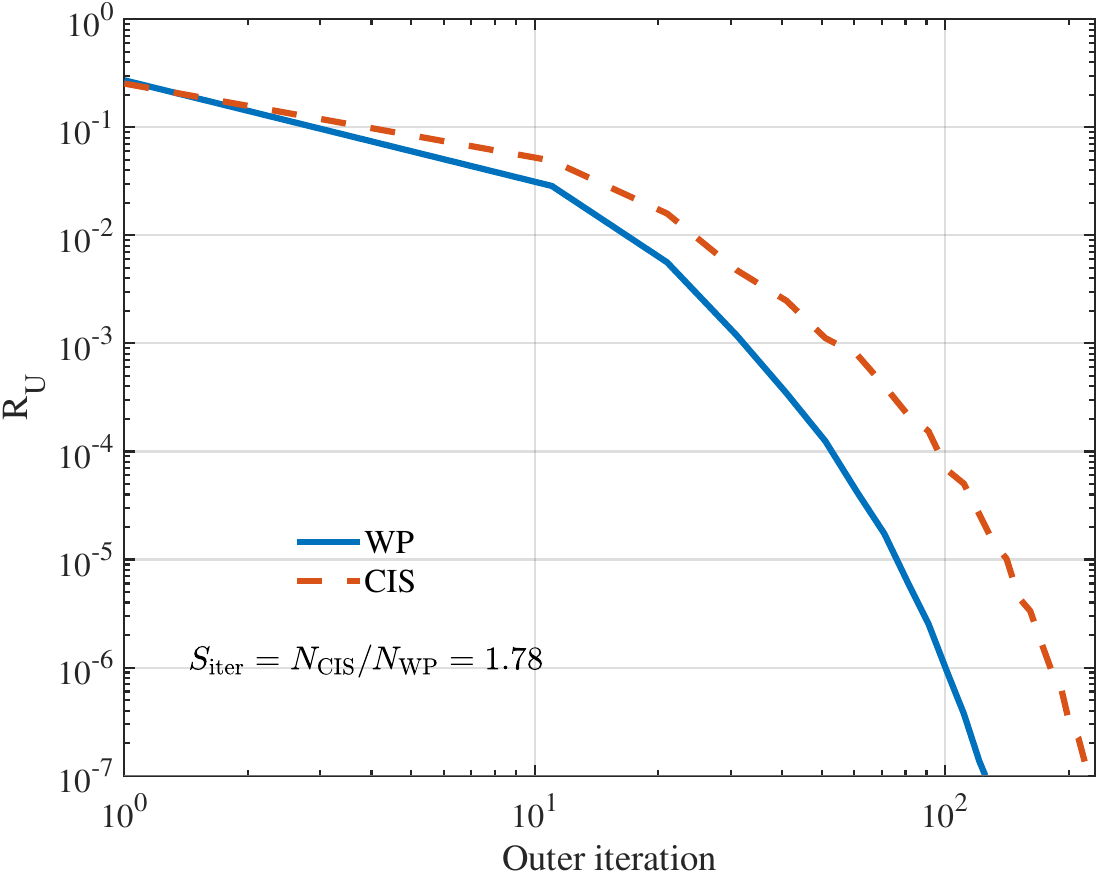}}\hspace{0.0200\figW}\raisebox{0.0000\figW}{\includegraphics[width=0.3119\figW]{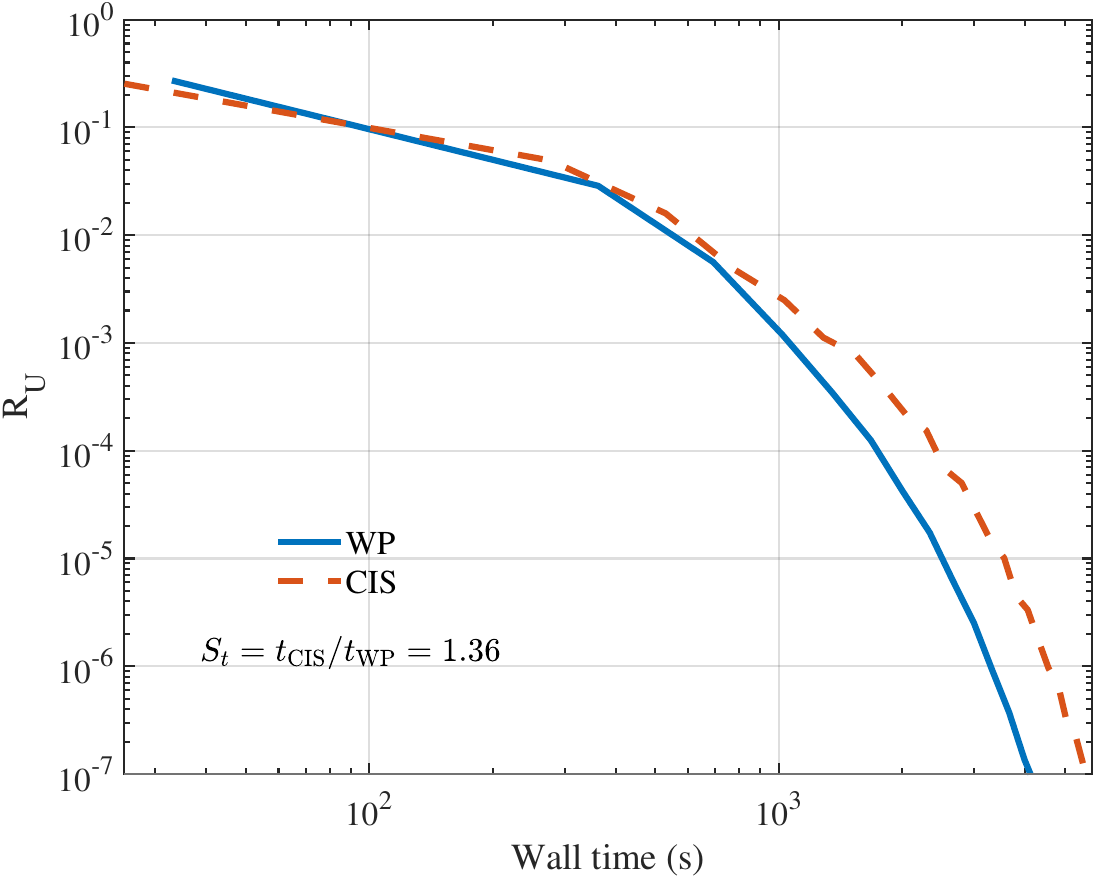}}}
\caption{$\varepsilon=1$}\label{fig:cavity-res-kn1}
\end{subfigure}
\par\smallskip
\begin{subfigure}[t]{0.6436\figW}\centering
\makebox[\linewidth]{\raisebox{0.0025\figW}{\includegraphics[width=0.3115\figW]{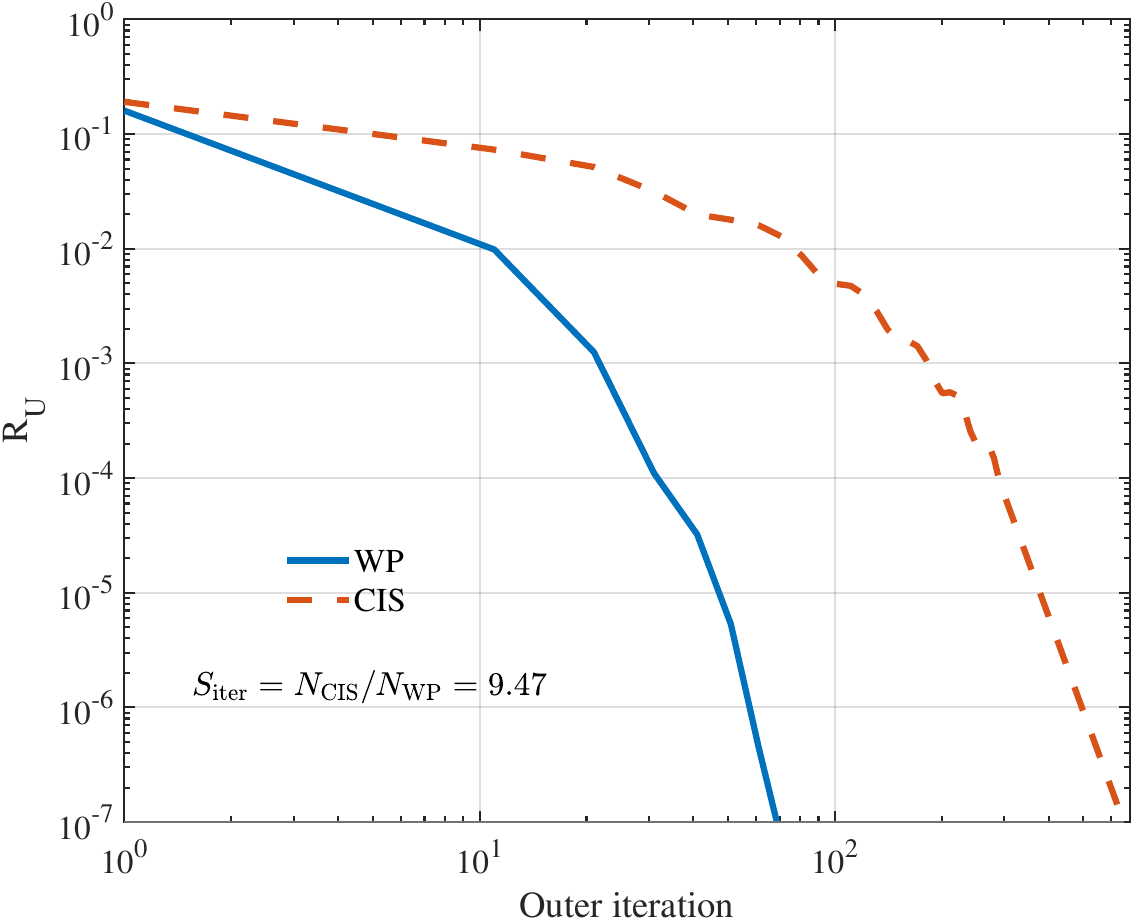}}\hspace{0.0200\figW}\raisebox{0.0000\figW}{\includegraphics[width=0.3122\figW]{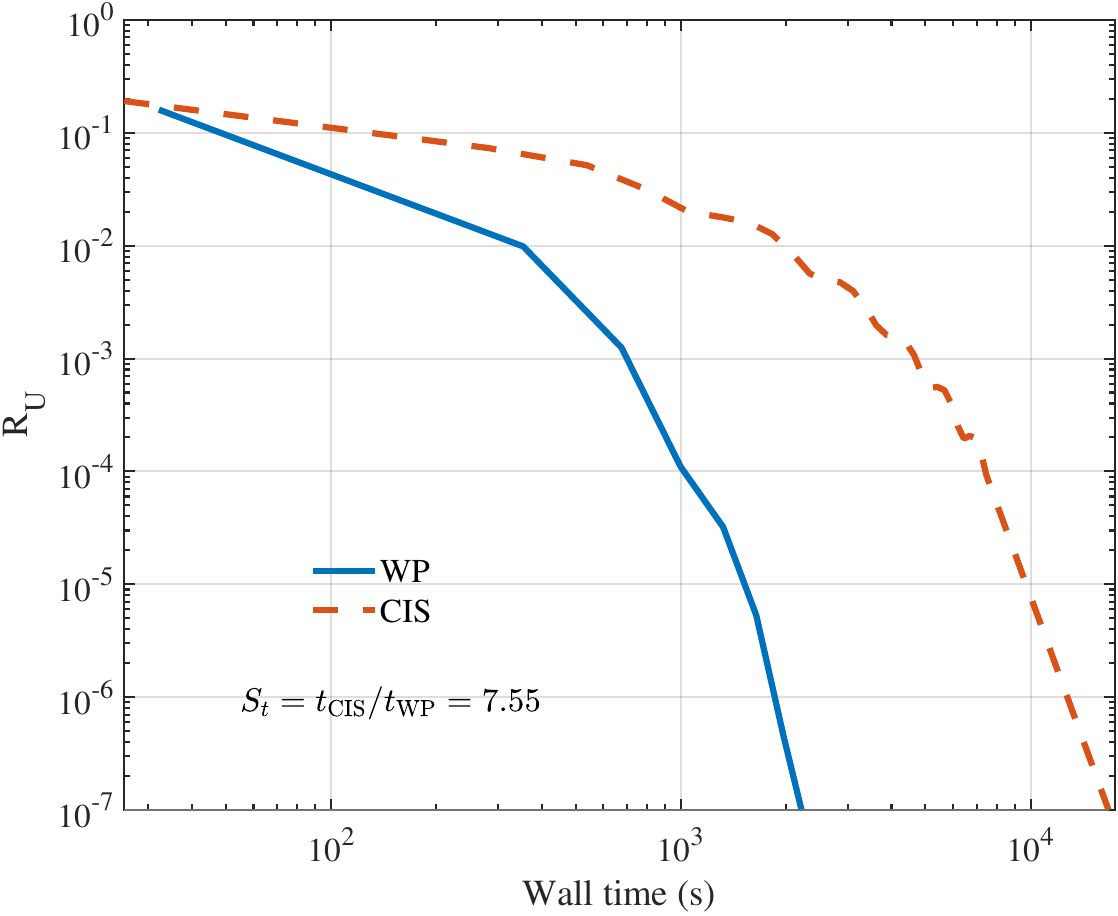}}}
\caption{$\varepsilon=0.075$}\label{fig:cavity-res-kn0075}
\end{subfigure}
\par\smallskip
\begin{subfigure}[t]{0.6436\figW}\centering
\makebox[\linewidth]{\raisebox{0.0019\figW}{\includegraphics[width=0.3117\figW]{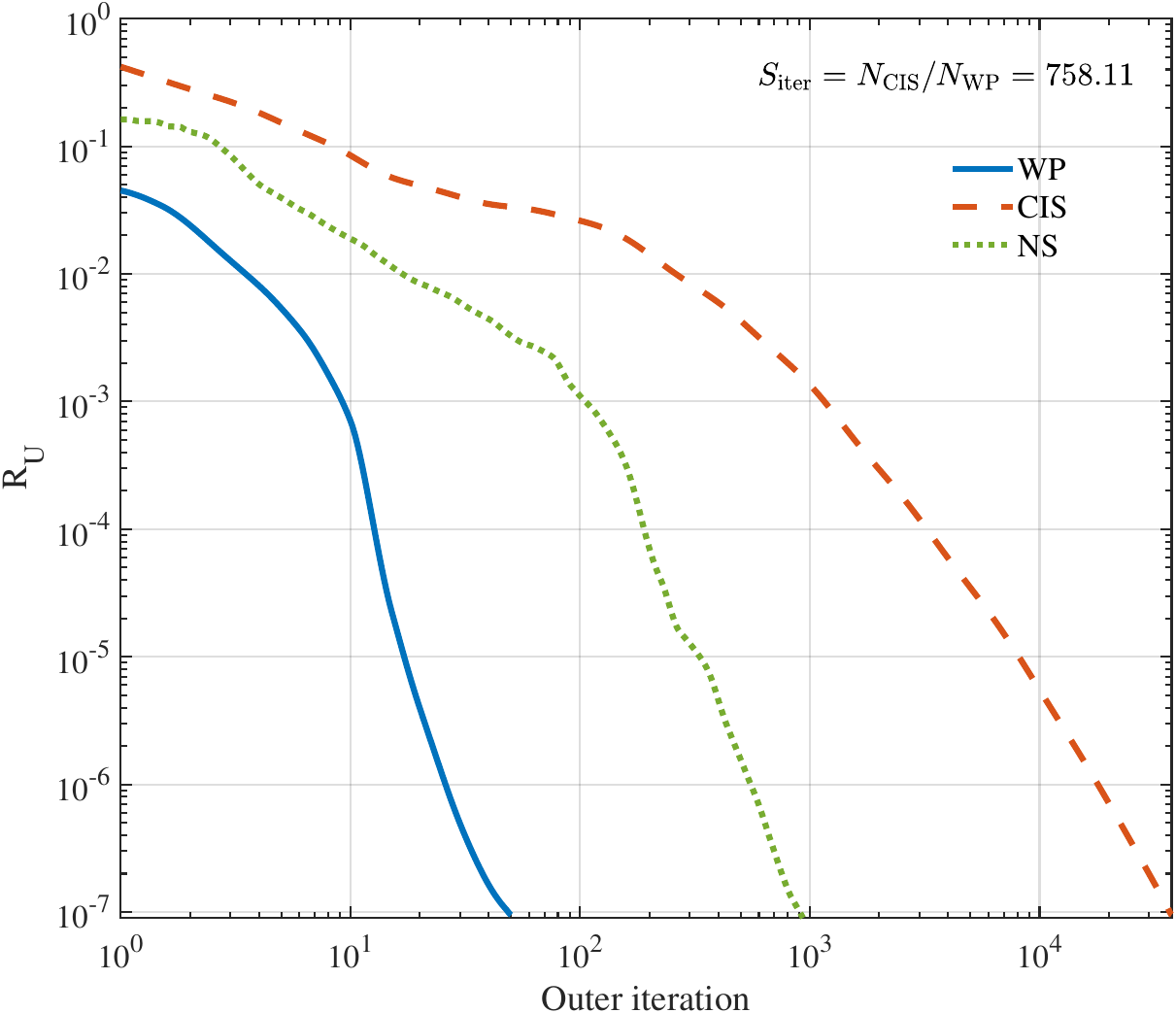}}\hspace{0.0200\figW}\raisebox{0.0000\figW}{\includegraphics[width=0.3119\figW]{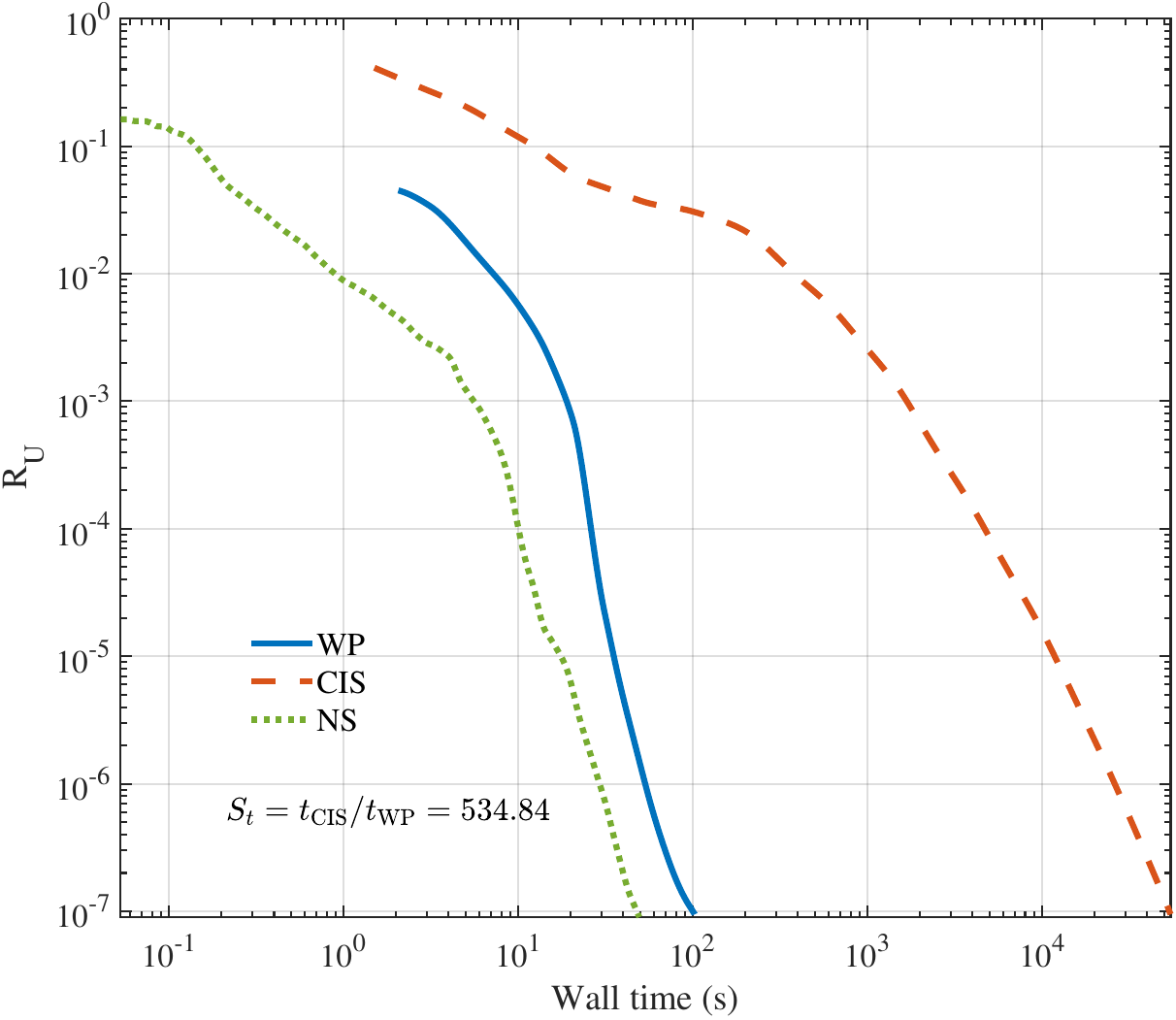}}}
\caption{$\mathrm{Re}=100$}\label{fig:cavity-res-re100}
\end{subfigure}
\par\smallskip
\begin{subfigure}[t]{0.6436\figW}\centering
\makebox[\linewidth]{\raisebox{0.0019\figW}{\includegraphics[width=0.3142\figW]{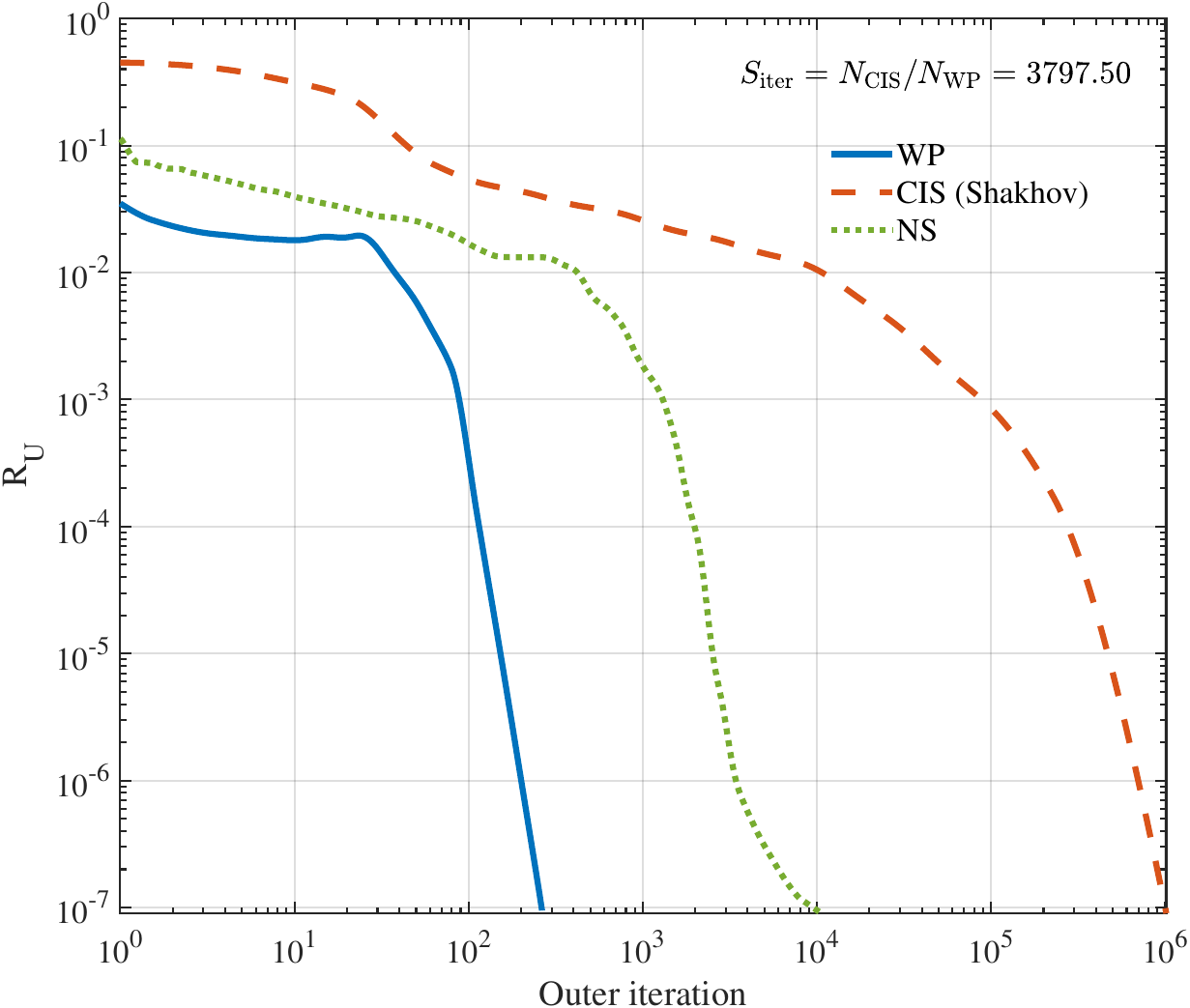}}\hspace{0.0200\figW}\raisebox{0.0000\figW}{\includegraphics[width=0.3094\figW]{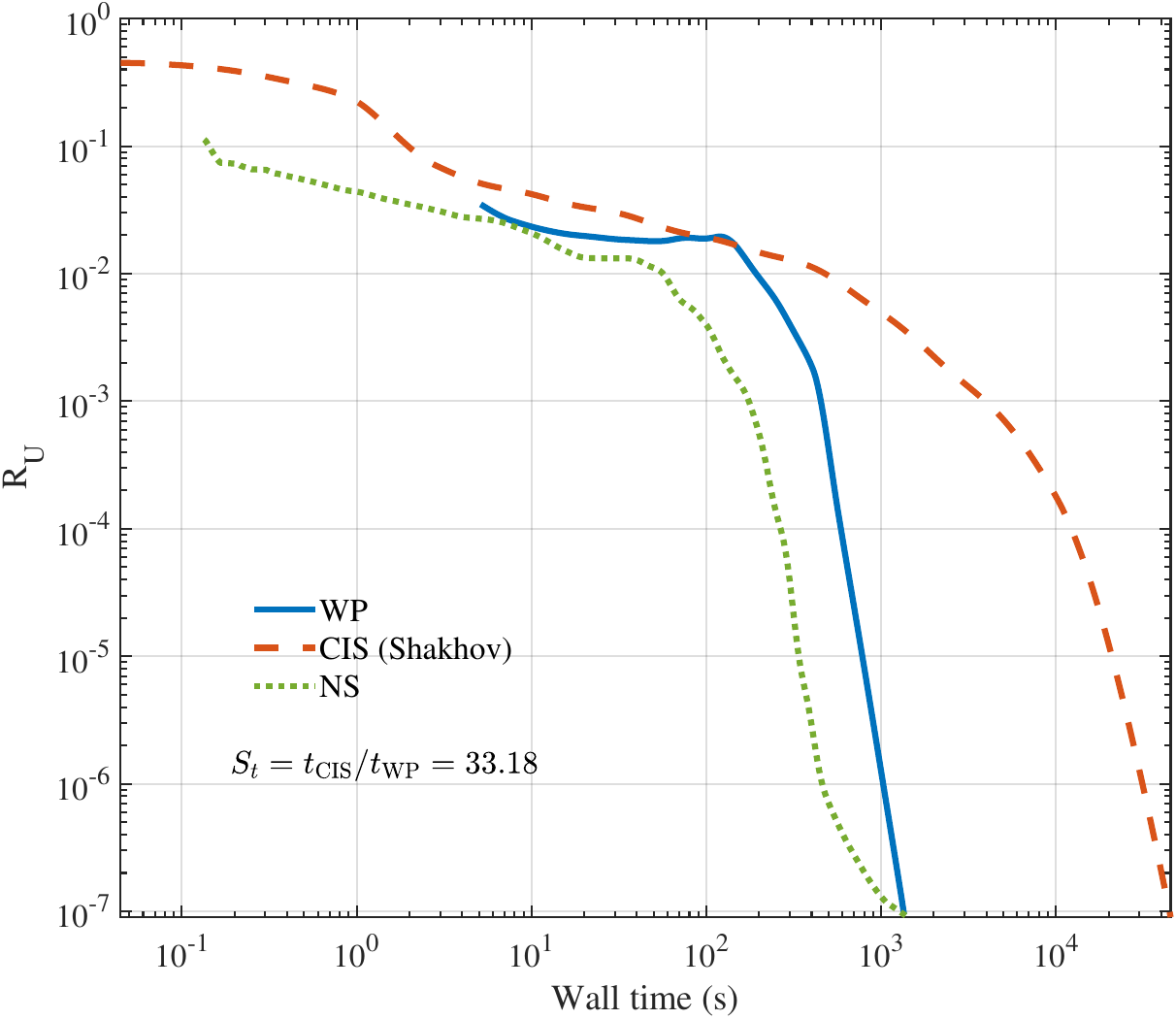}}}
\caption{$\mathrm{Re}=1000$}\label{fig:cavity-res-re1000}
\end{subfigure}
\caption{Lid-driven cavity: histories of the normalized macroscopic residual against outer iteration (left) and wall time (right). At $\varepsilon=1$, $\varepsilon=0.075$ and $\mathrm{Re}=100$ WP is compared with CIS, and the CIS-to-WP speedups are printed in the panels; at $\mathrm{Re}=100$ the NS solver is shown as well. At $\mathrm{Re}=1000$ both panels show CIS of the Shakhov model~\cite{liu_wpd_shakhov} and the NS solver. The abscissa of the NS iteration histories counts implicit steps of the NS solver.}\label{fig:cavity-residuals}
\end{figure}

\FloatBarrier
\subsection{Cylinder flow}
Mach-5 flow past a circular cylinder of unit radius tests the method on a
curved shock layer and a diffuse isothermal wall. The body-fitted O-grid has
$129\times80$ cells and outer radius $9.3953$ at $\varepsilon=1$ and
$129\times40$ cells and outer radius $6.2636$ at $\varepsilon=10^{-2}$ and
$10^{-4}$. The velocity grid has $96\times96\times24$ nodes at the two larger
Knudsen numbers and $32\times32\times12$ nodes at $\varepsilon=10^{-4}$, and
WP runs W1/PR2/P1 at $\varepsilon=1$ and $10^{-2}$ and W120/PR2/P1 at
$\varepsilon=10^{-4}$. The reference is CIS at $\varepsilon=1$ and $10^{-2}$
and CIS of the Shakhov model at $\varepsilon=10^{-4}$.

At $\varepsilon=1$ the WP and CIS fields in Fig.~\ref{fig:cyl-kn1-fields}
differ by $1.86\times10^{-5}$ in density, $1.08\times10^{-5}$ in temperature,
$2.68\times10^{-6}$ in speed and $6.45\times10^{-6}$ in pressure. The
pressure-drag coefficients of WP and CIS, both $1.25$, differ by
$2.55\times10^{-7}$, and the pressure lift is nearly zero, as symmetry
requires. The centreline and surface profiles in
Fig.~\ref{fig:cyl-kn1-profiles} compare the two solutions across the shock
layer and along the wall. The molecular distributions sampled on the
centreline, shown on linear and logarithmic scales in
Figs.~\ref{fig:cyl-kn1-dist-lin} and~\ref{fig:cyl-kn1-dist-log}, display the
non-Maxwellian structure ahead of and behind the body together with its
low-amplitude tails.

\begin{figure}[tbp]
\centering\singlespacing\setlength{\figW}{\linewidth}
\begin{subfigure}[t]{0.4856\figW}\centering
\raisebox{0.0000\figW}{\includegraphics[width=0.4856\figW]{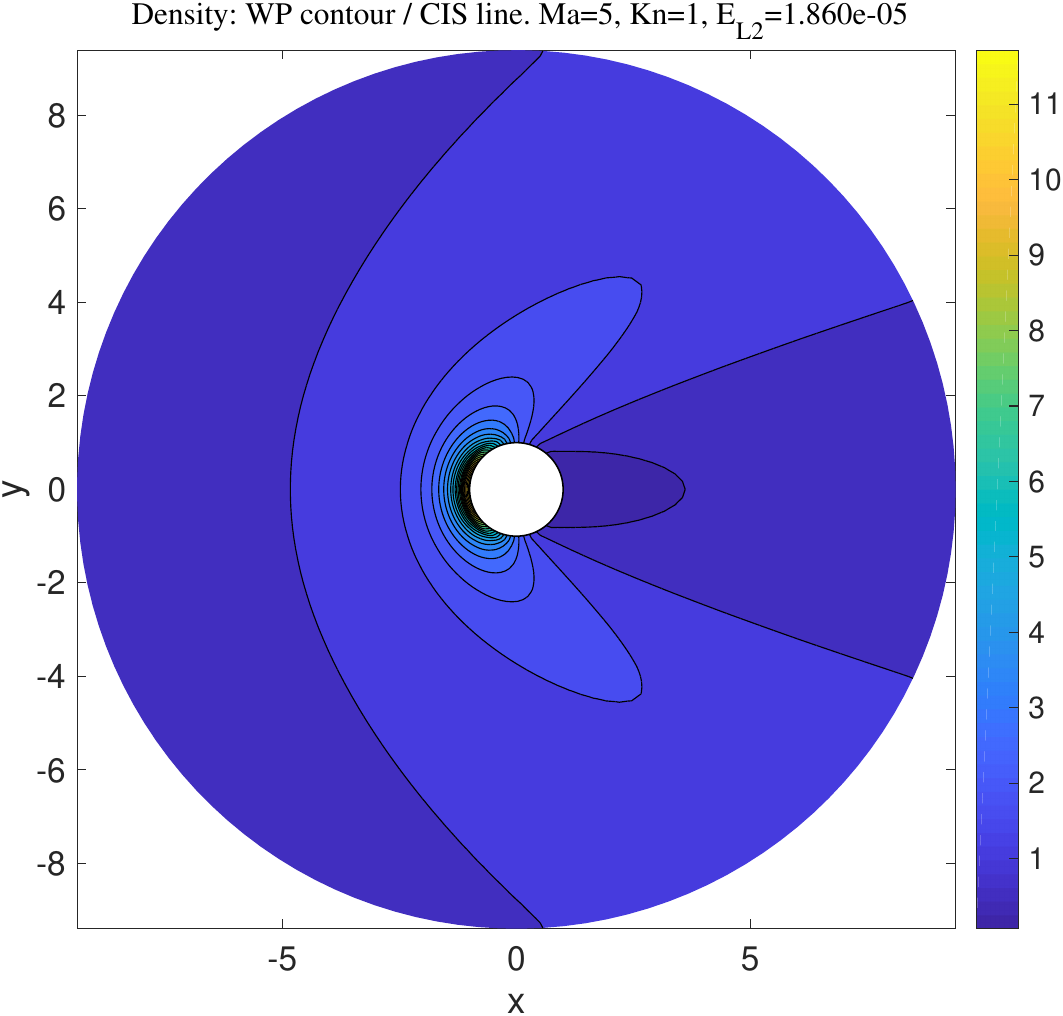}}
\caption{Density}
\end{subfigure}\hspace{0.0200\figW}%
\begin{subfigure}[t]{0.4894\figW}\centering
\raisebox{0.0000\figW}{\includegraphics[width=0.4894\figW]{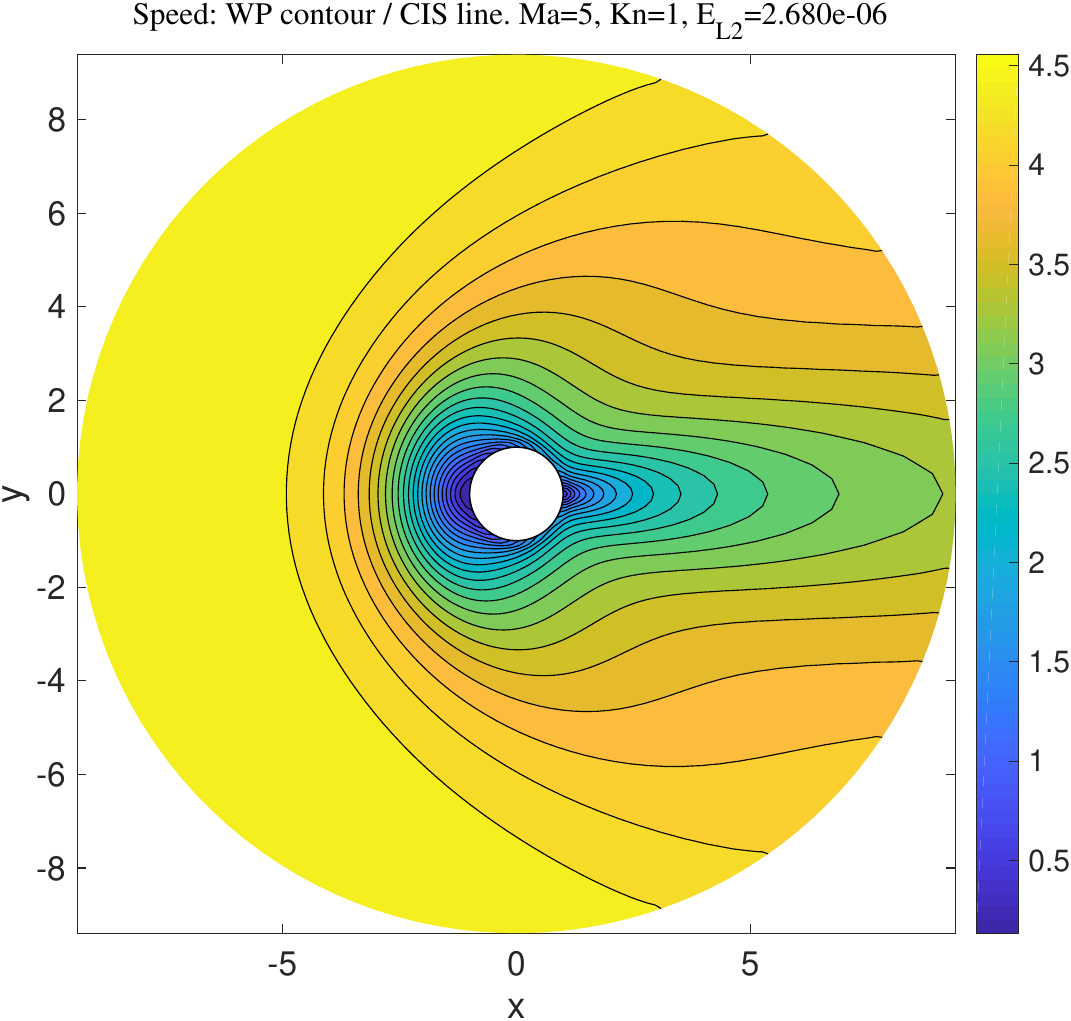}}
\caption{Speed}
\end{subfigure}
\par\smallskip
\begin{subfigure}[t]{0.4894\figW}\centering
\raisebox{0.0000\figW}{\includegraphics[width=0.4894\figW]{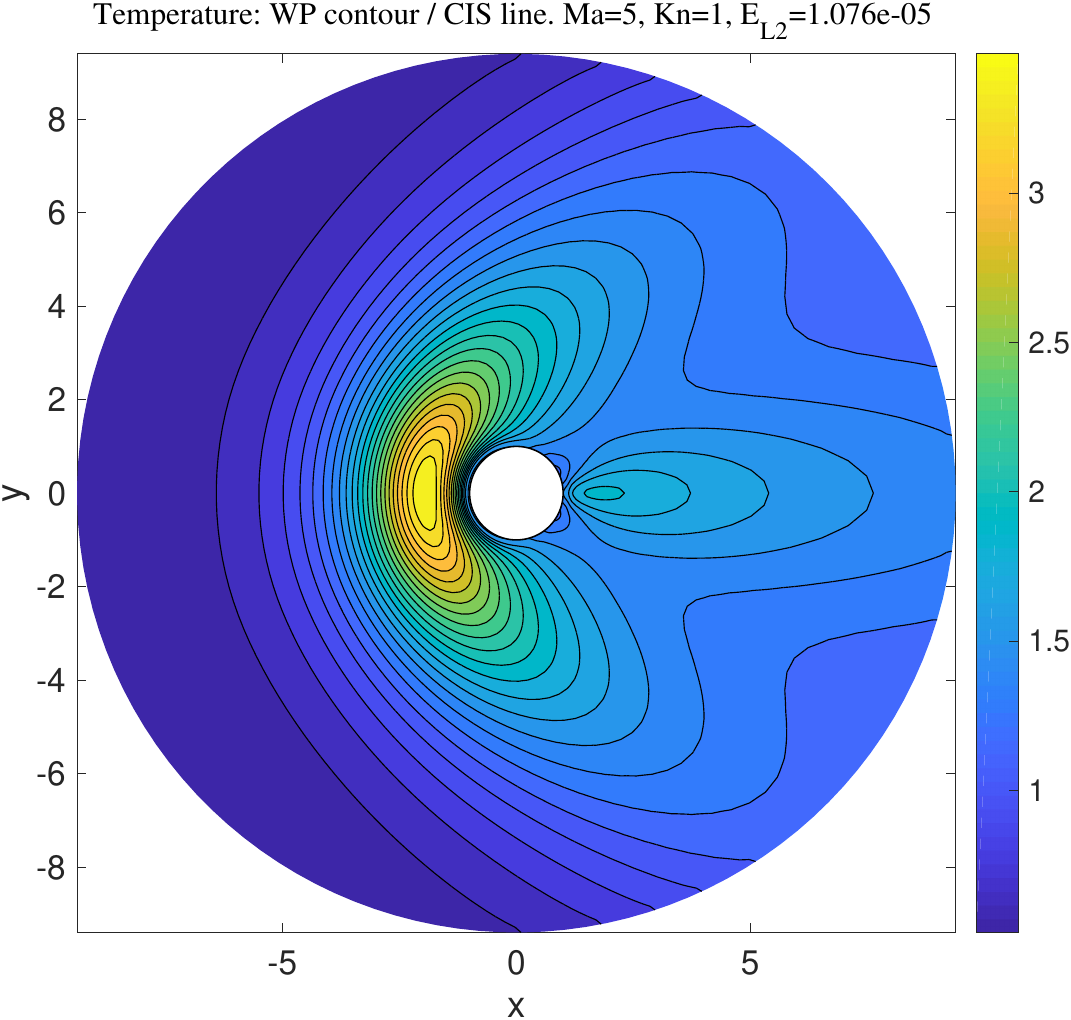}}
\caption{Temperature}
\end{subfigure}\hspace{0.0200\figW}%
\begin{subfigure}[t]{0.4856\figW}\centering
\raisebox{0.0000\figW}{\includegraphics[width=0.4856\figW]{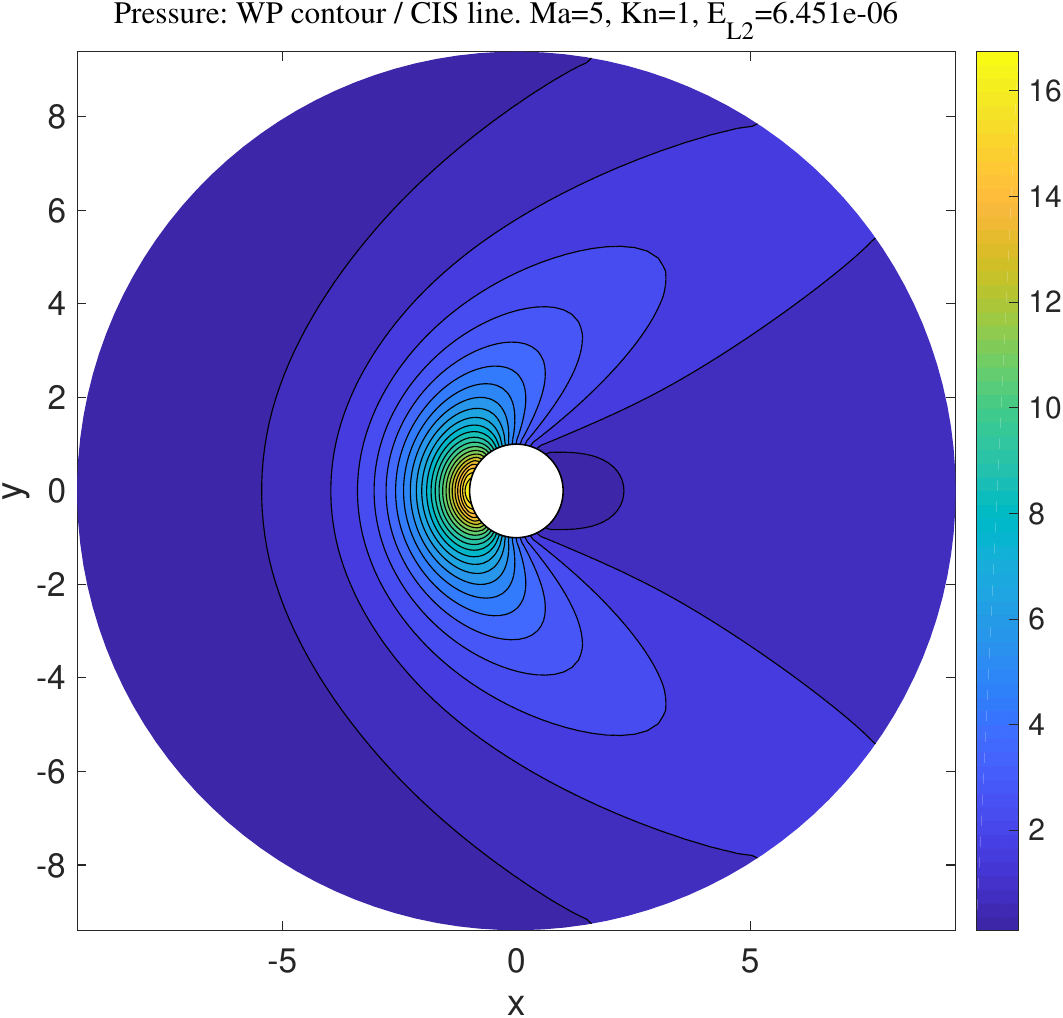}}
\caption{Pressure}
\end{subfigure}
\caption{Mach-5 flow past a cylinder, $\varepsilon=1$: density, speed, temperature and pressure fields of WP compared with CIS. WP is shown by filled contours and the reference by lines at the same levels, and $E_{L_2}$ is printed above each panel.}\label{fig:cyl-kn1-fields}
\end{figure}
\begin{figure}[tbp]
\centering\singlespacing\setlength{\figW}{\linewidth}
\begin{subfigure}[t]{0.4892\figW}\centering
\raisebox{0.0000\figW}{\includegraphics[width=0.4892\figW]{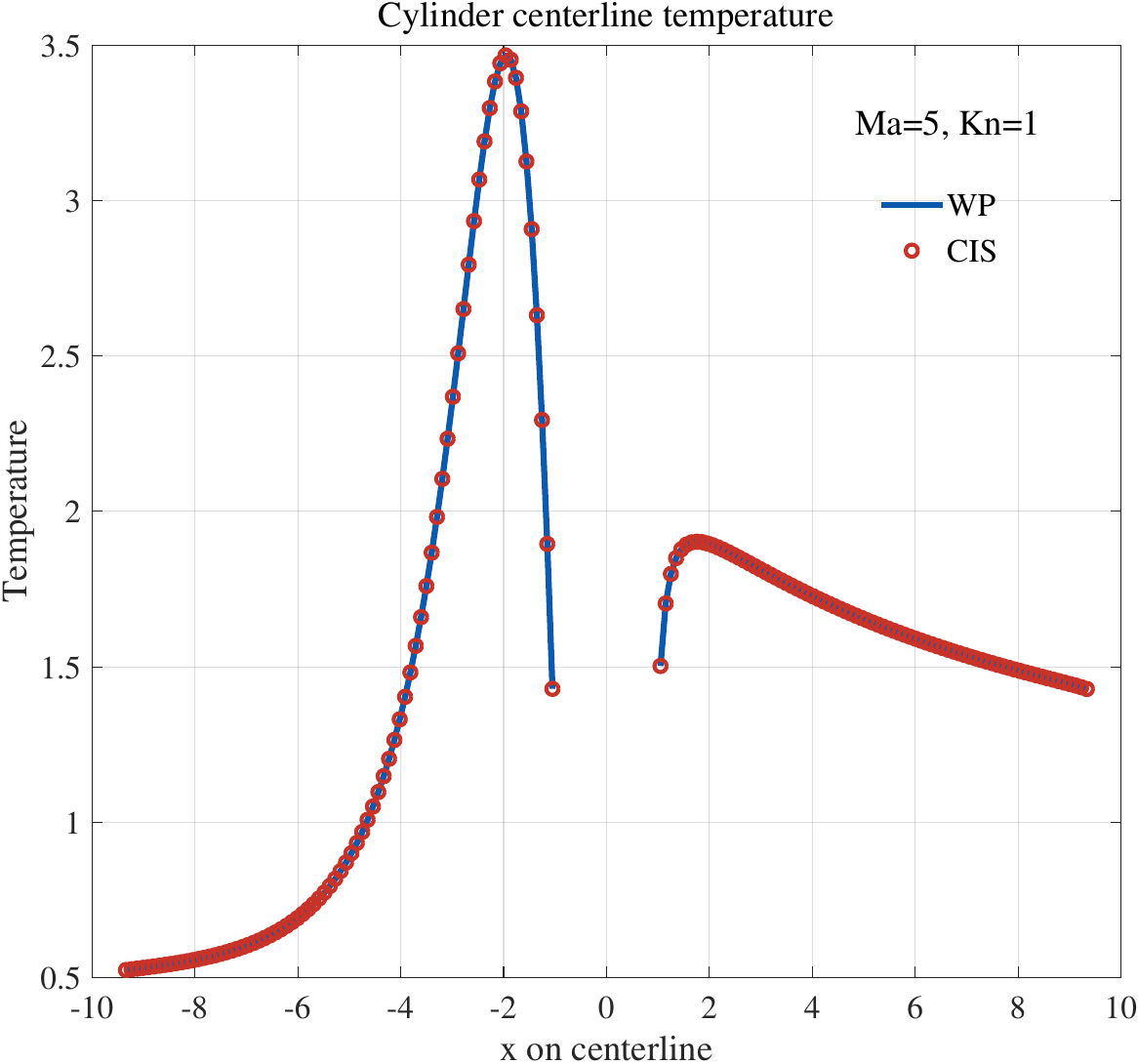}}
\caption{Centreline temperature}
\end{subfigure}\hspace{0.0200\figW}%
\begin{subfigure}[t]{0.4858\figW}\centering
\raisebox{0.0000\figW}{\includegraphics[width=0.4858\figW]{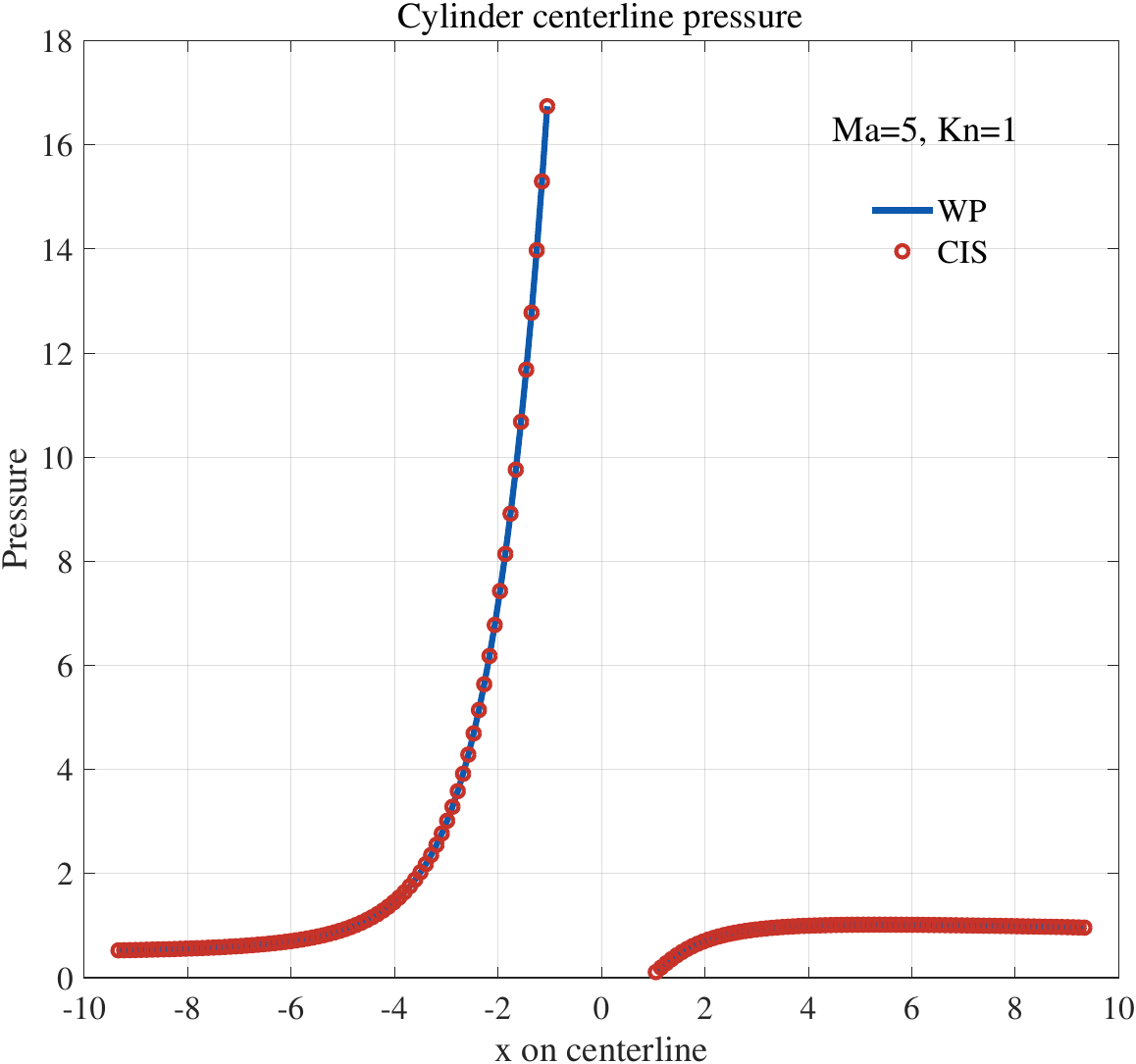}}
\caption{Centreline pressure}
\end{subfigure}
\par\smallskip
\begin{subfigure}[t]{0.4929\figW}\centering
\raisebox{0.0000\figW}{\includegraphics[width=0.4929\figW]{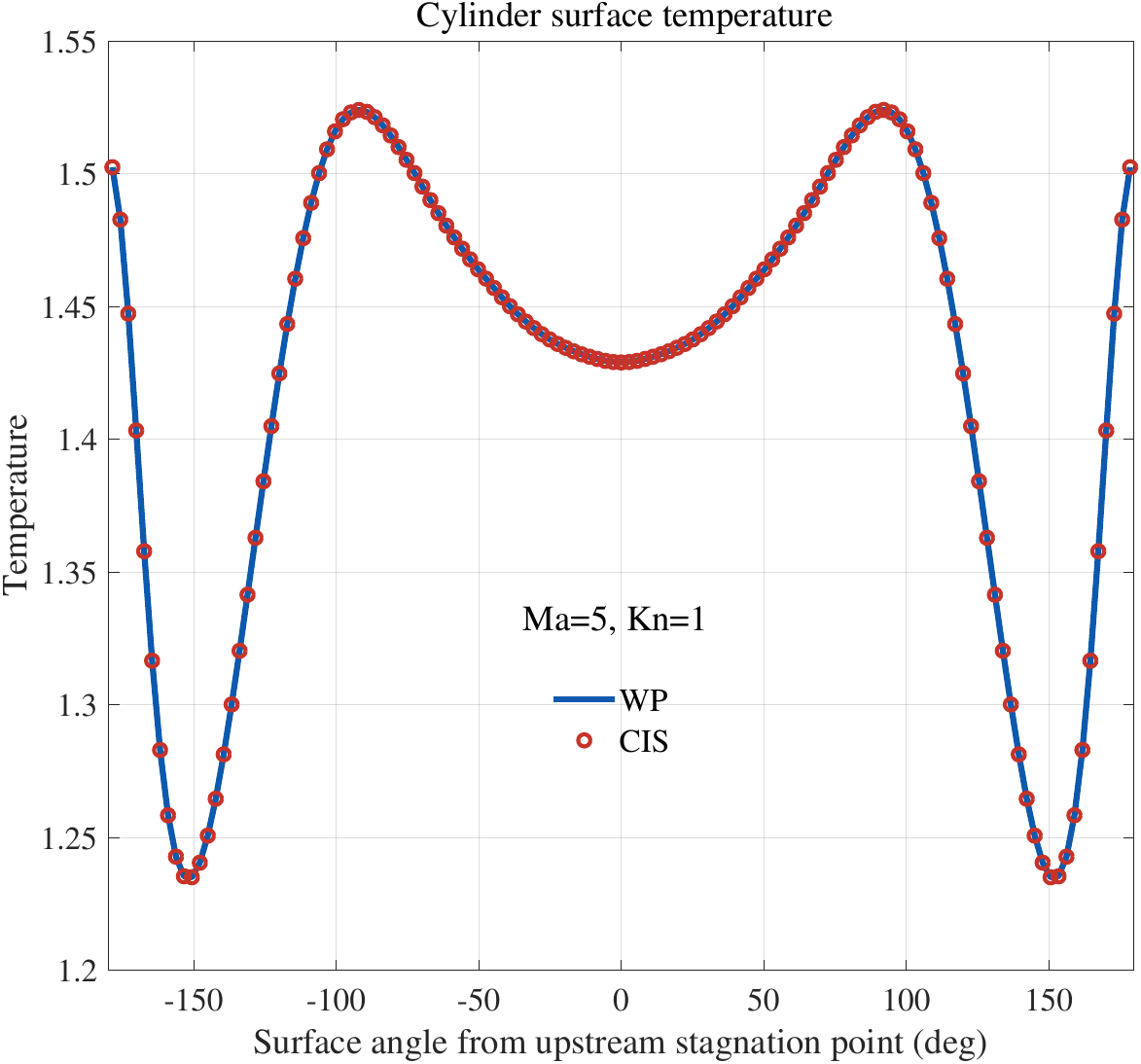}}
\caption{Surface temperature}
\end{subfigure}\hspace{0.0200\figW}%
\begin{subfigure}[t]{0.4821\figW}\centering
\raisebox{0.0000\figW}{\includegraphics[width=0.4821\figW]{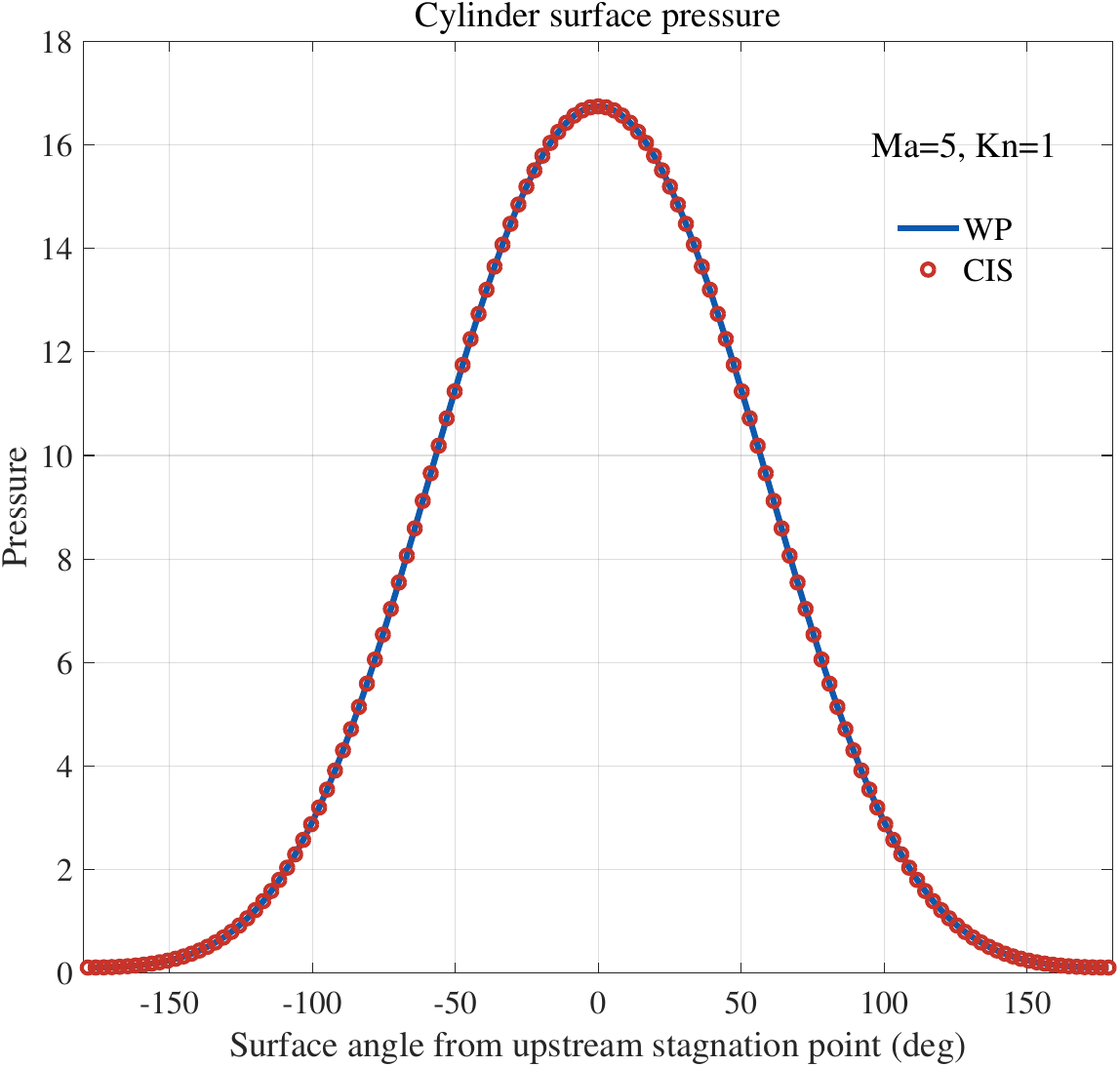}}
\caption{Surface pressure}
\end{subfigure}
\caption{Mach-5 flow past a cylinder, $\varepsilon=1$: temperature and pressure of WP and CIS along the centreline (a, b) and along the body surface (c, d). WP is shown by lines and the reference by symbols.}\label{fig:cyl-kn1-profiles}
\end{figure}
\begin{figure}[tbp]
\centering\singlespacing\setlength{\figW}{\linewidth}
\begin{subfigure}[t]{0.4913\figW}\centering
\raisebox{0.0000\figW}{\includegraphics[width=0.4913\figW]{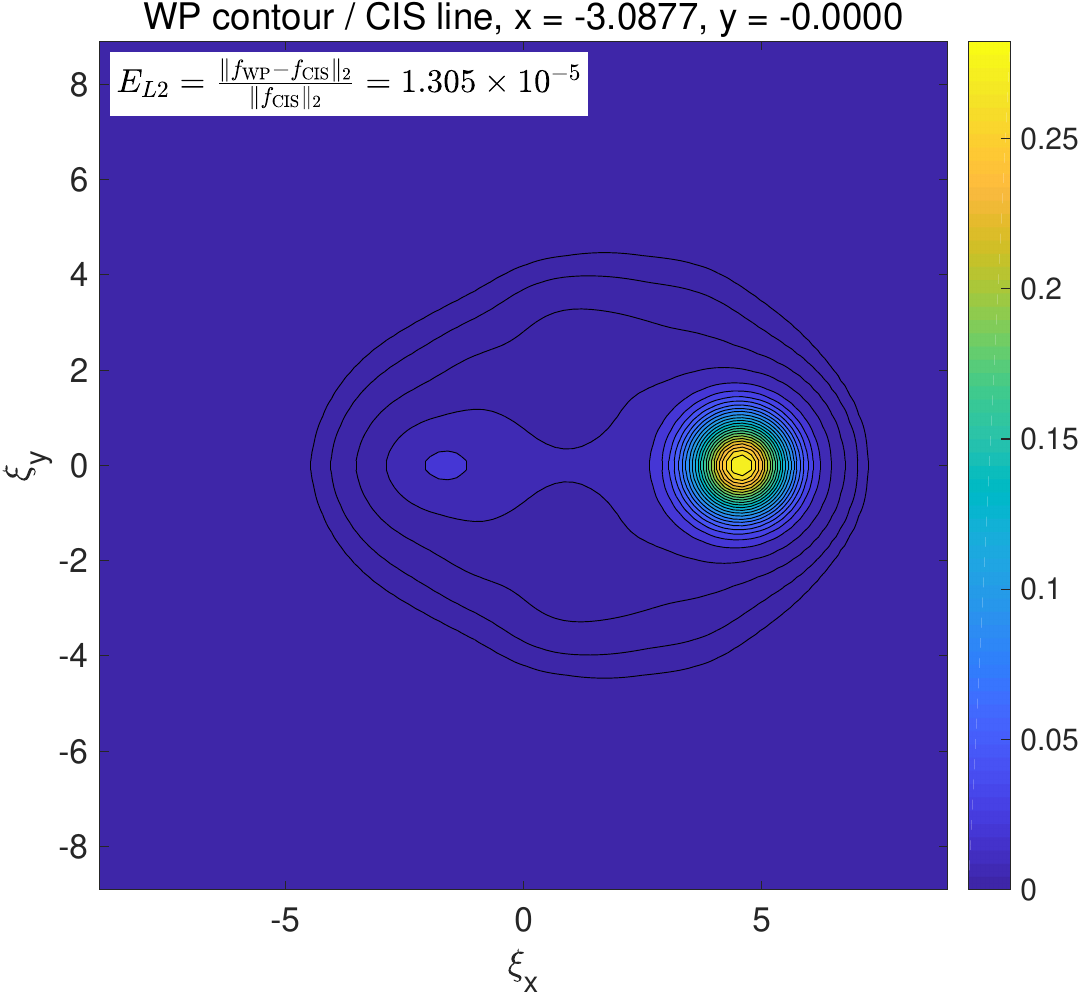}}
\caption{$(x,y)=(-3.09,\,0.00)$}
\end{subfigure}\hspace{0.0200\figW}%
\begin{subfigure}[t]{0.4837\figW}\centering
\raisebox{0.0000\figW}{\includegraphics[width=0.4837\figW]{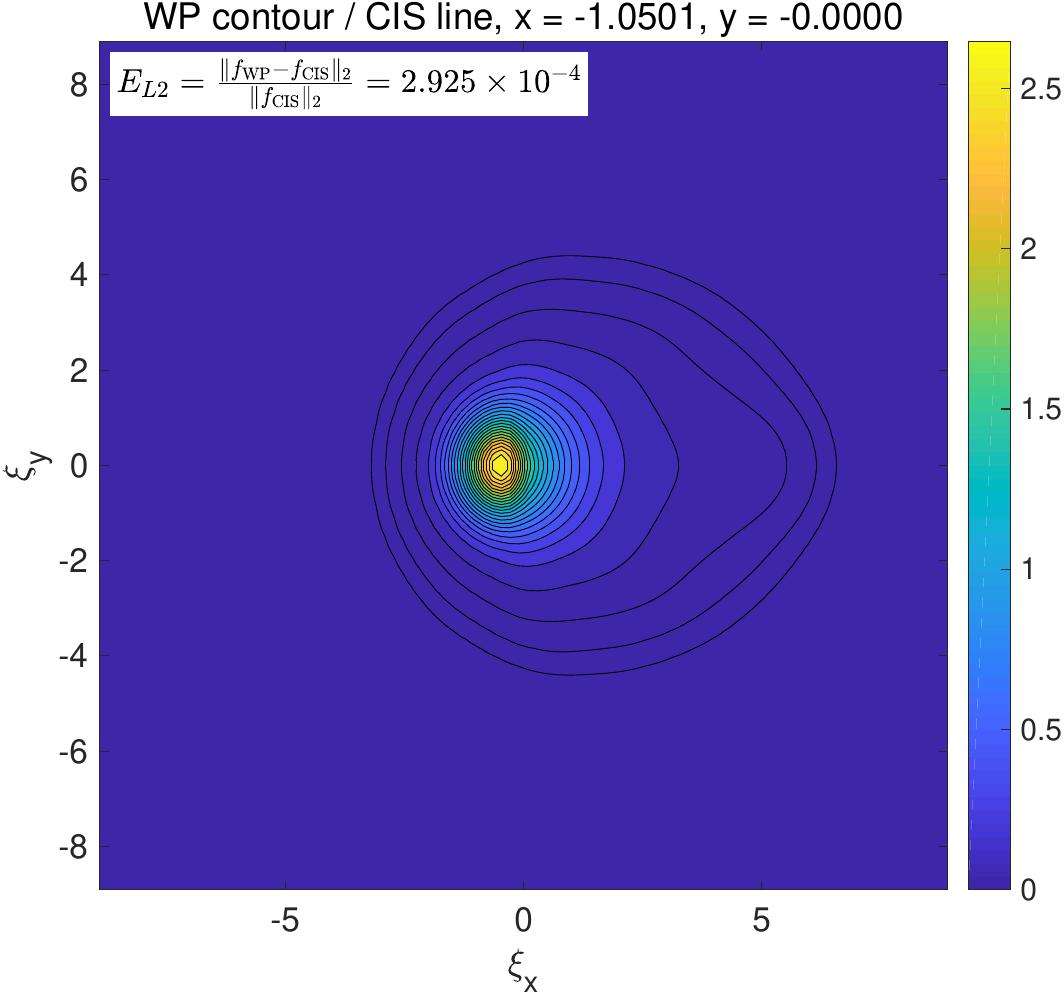}}
\caption{$(x,y)=(-1.05,\,0.00)$}
\end{subfigure}
\par\smallskip
\begin{subfigure}[t]{0.4875\figW}\centering
\raisebox{0.0000\figW}{\includegraphics[width=0.4875\figW]{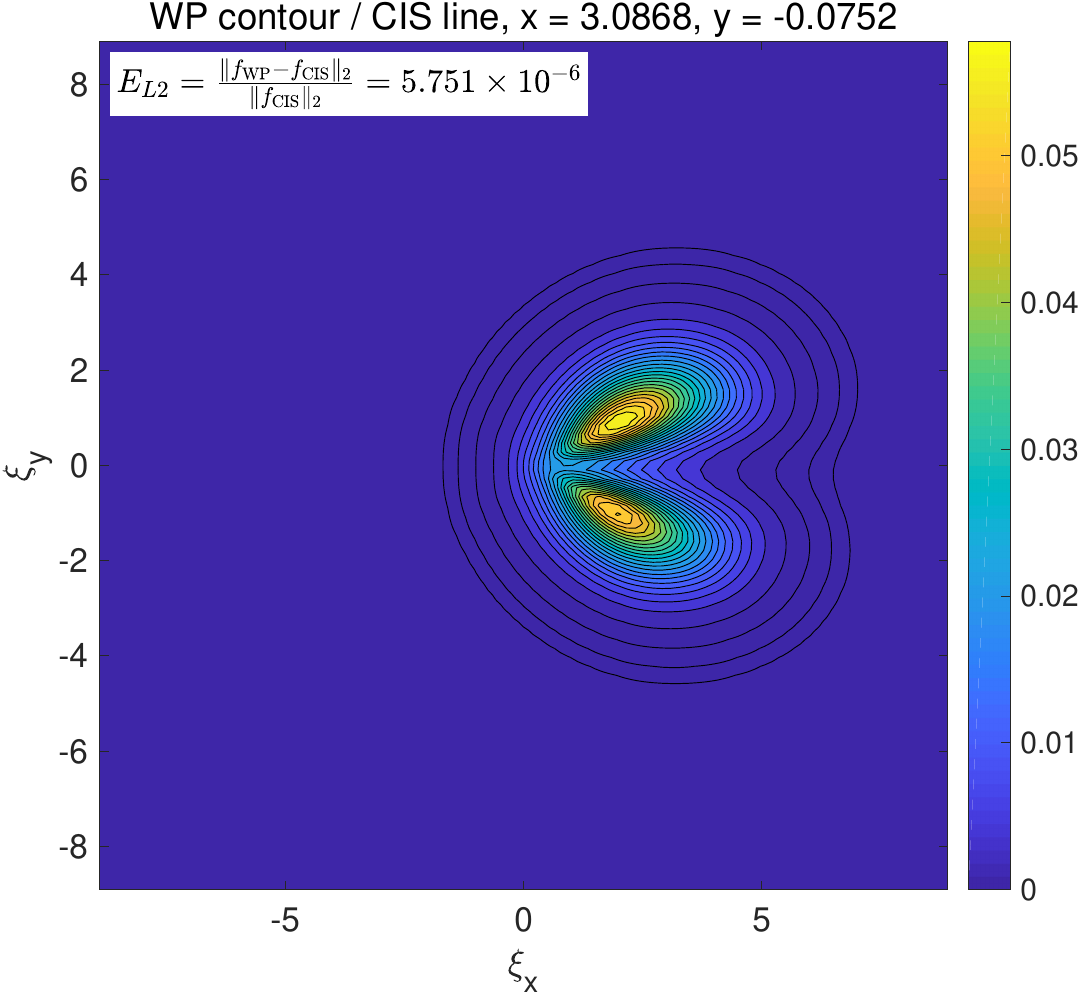}}
\caption{$(x,y)=(3.09,\,-0.08)$}
\end{subfigure}\hspace{0.0200\figW}%
\begin{subfigure}[t]{0.4875\figW}\centering
\raisebox{0.0000\figW}{\includegraphics[width=0.4875\figW]{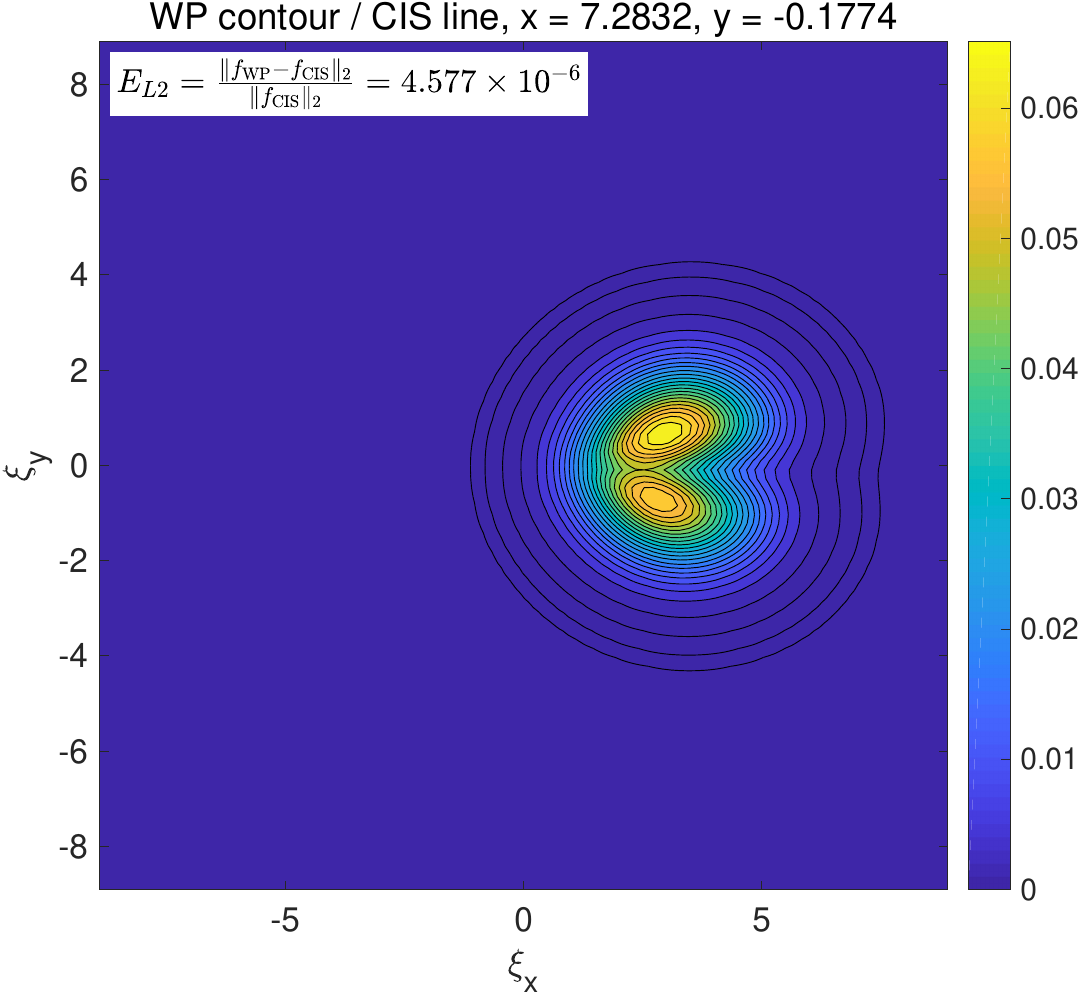}}
\caption{$(x,y)=(7.28,\,-0.18)$}
\end{subfigure}
\caption{Mach-5 flow past a cylinder, $\varepsilon=1$: reduced molecular distribution in the $(\xi_x,\xi_y)$ velocity plane at the marked sampling points, linear scale. WP is shown by filled contours and CIS by lines at the same levels, and $E_{L_2}$ is printed in each panel.}\label{fig:cyl-kn1-dist-lin}
\end{figure}
\begin{figure}[tbp]
\centering\singlespacing\setlength{\figW}{\linewidth}
\begin{subfigure}[t]{0.4875\figW}\centering
\raisebox{0.0000\figW}{\includegraphics[width=0.4875\figW]{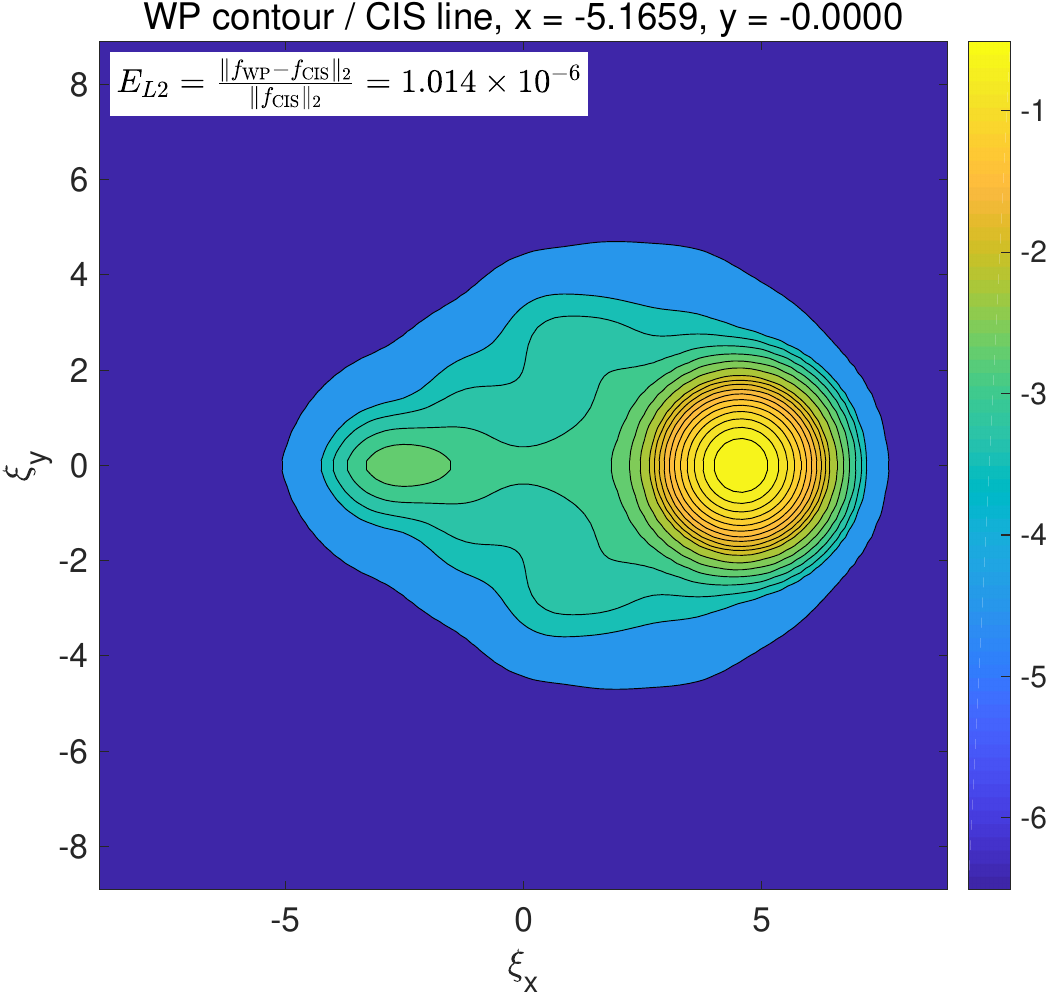}}
\caption{$(x,y)=(-5.17,\,0.00)$}
\end{subfigure}\hspace{0.0200\figW}%
\begin{subfigure}[t]{0.4875\figW}\centering
\raisebox{0.0000\figW}{\includegraphics[width=0.4875\figW]{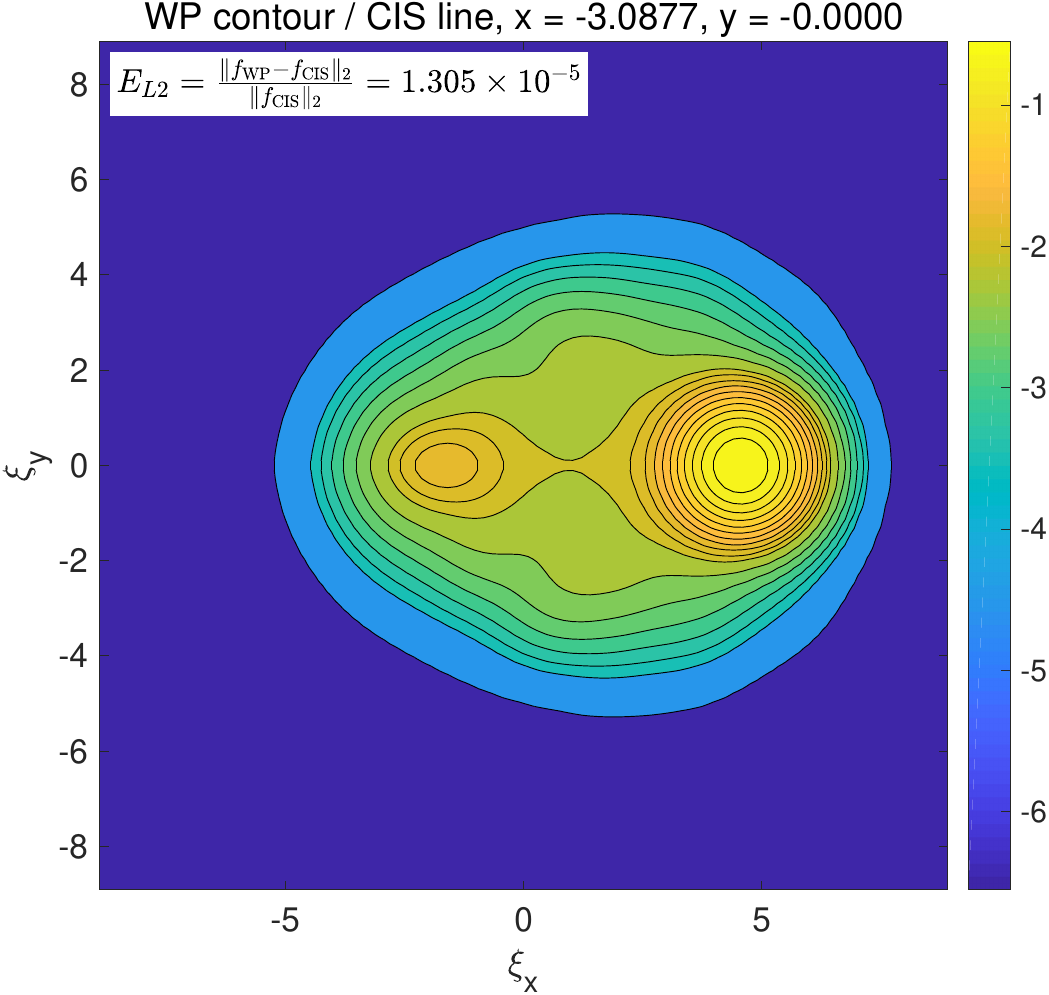}}
\caption{$(x,y)=(-3.09,\,0.00)$}
\end{subfigure}
\par\smallskip
\begin{subfigure}[t]{0.4875\figW}\centering
\raisebox{0.0000\figW}{\includegraphics[width=0.4875\figW]{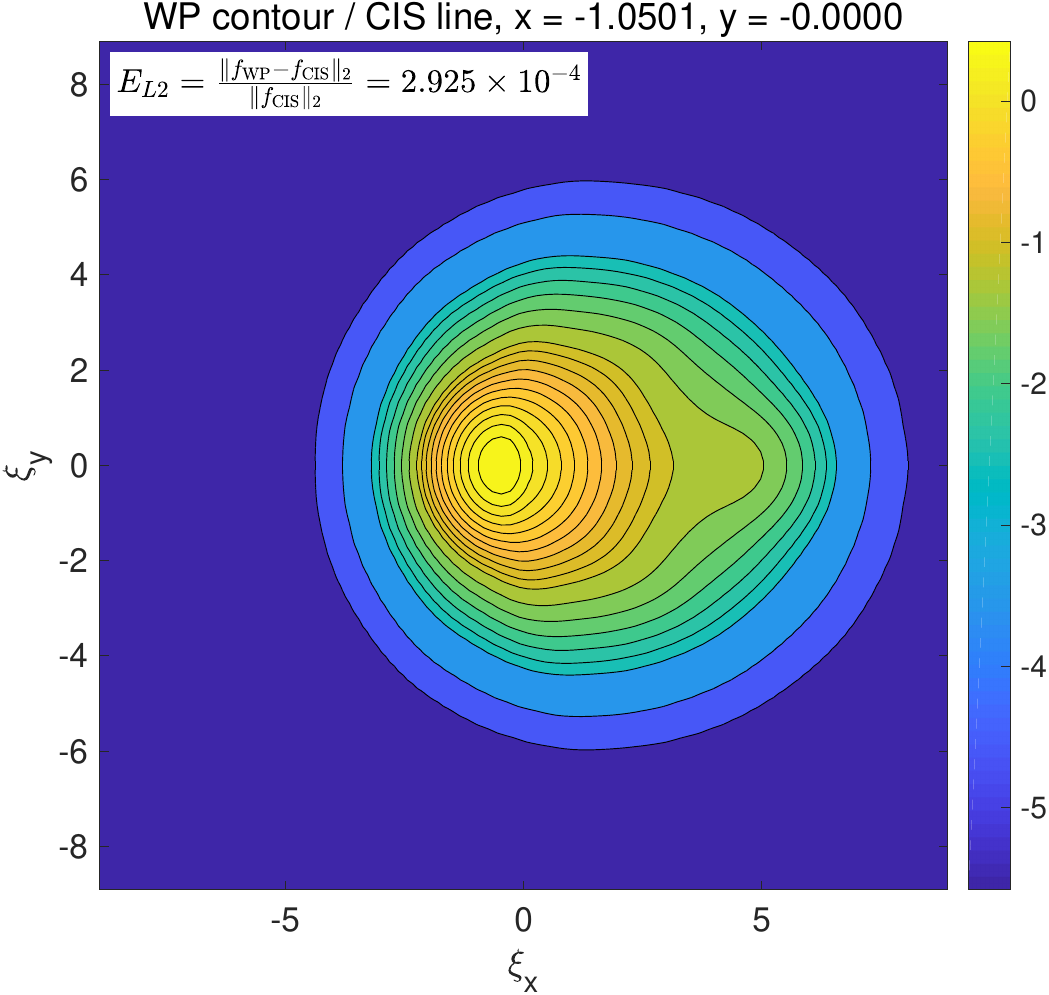}}
\caption{$(x,y)=(-1.05,\,0.00)$}
\end{subfigure}\hspace{0.0200\figW}%
\begin{subfigure}[t]{0.4875\figW}\centering
\raisebox{0.0000\figW}{\includegraphics[width=0.4875\figW]{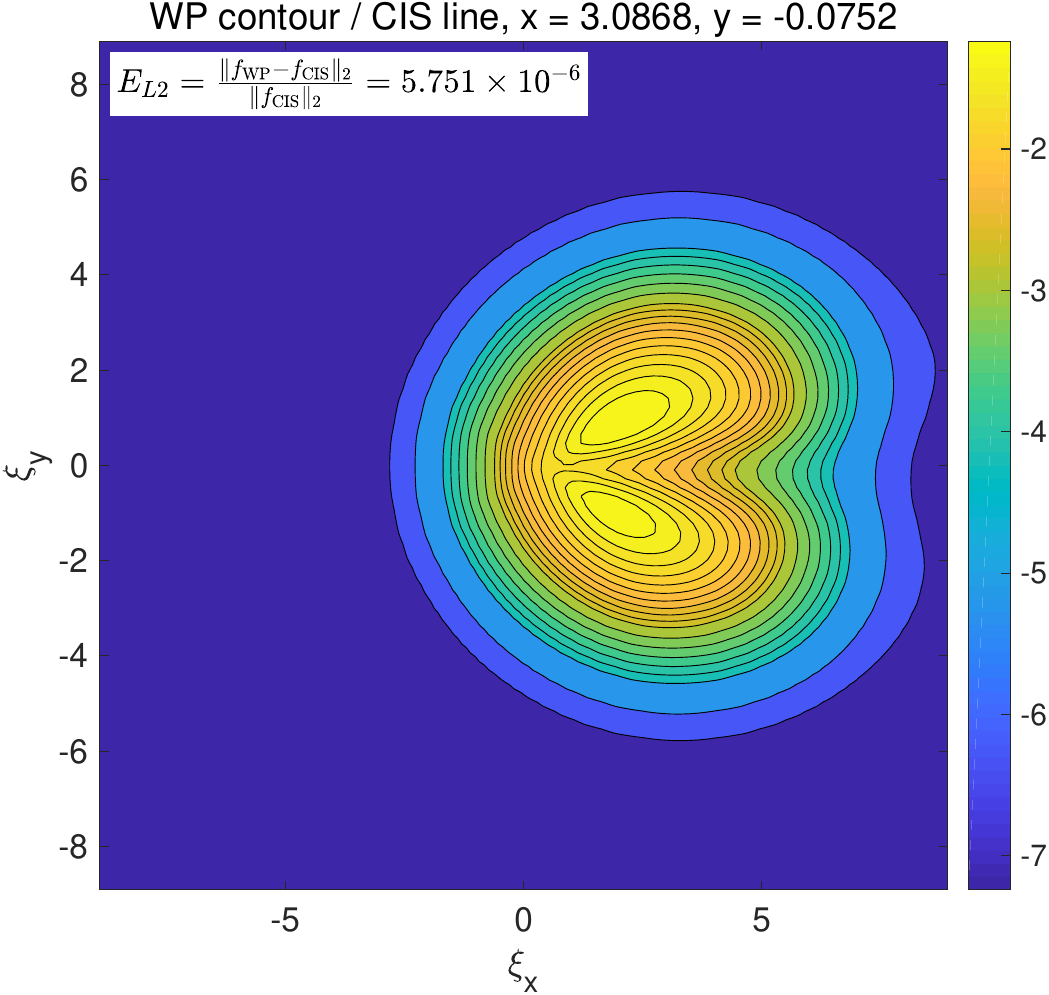}}
\caption{$(x,y)=(3.09,\,-0.08)$}
\end{subfigure}
\caption{Mach-5 flow past a cylinder, $\varepsilon=1$: reduced molecular distribution in the $(\xi_x,\xi_y)$ velocity plane at the marked sampling points, logarithmic scale. WP is shown by filled contours and CIS by lines at the same levels, and $E_{L_2}$ is printed in each panel.}\label{fig:cyl-kn1-dist-log}
\end{figure}

At $\varepsilon=10^{-2}$ the field differences in
Fig.~\ref{fig:cyl-kn1em2-fields} are $4.04\times10^{-4}$ in density,
$7.87\times10^{-4}$ in temperature, $1.82\times10^{-4}$ in speed and
$2.93\times10^{-4}$ in pressure, and the pressure-drag coefficients of WP and
CIS, both $1.29$, differ by $9.00\times10^{-5}$; the centreline and surface
profiles are compared in Fig.~\ref{fig:cyl-kn1em2-profiles}. At
$\varepsilon=10^{-4}$ the WP fields in Fig.~\ref{fig:cyl-kn1em4-fields} differ
from CIS of the Shakhov model by $4.12\times10^{-4}$,
$1.28\times10^{-3}$, $5.08\times10^{-4}$ and $7.77\times10^{-4}$ in density,
temperature, speed and pressure.

The residual histories of the three calculations are collected in
Fig.~\ref{fig:cyl-residuals}. At $\varepsilon=1$ the ratios are
$S_{\rm iter}=1.29$ and $S_t=1.08$. The particle carries the whole solution
in this regime, and the small reduction in the number of iterations pays for
the W updates, so that the two methods cost about the same. At $\varepsilon=10^{-2}$ the ratios grow to $31.74$ and
$17.52$, and at $\varepsilon=10^{-4}$ the ratios with respect to CIS of the
Shakhov model~\cite{liu_wpd_shakhov} are $S_{\rm iter}=54412.11$ and
$S_t=1434.53$. WP then reaches $R_U=10^{-7}$ in $193.0$~s, $4.17$ times the
$46.3$~s of the NS solver.

\begin{figure}[tbp]
\centering\singlespacing\setlength{\figW}{\linewidth}
\begin{subfigure}[t]{0.4856\figW}\centering
\raisebox{0.0000\figW}{\includegraphics[width=0.4856\figW]{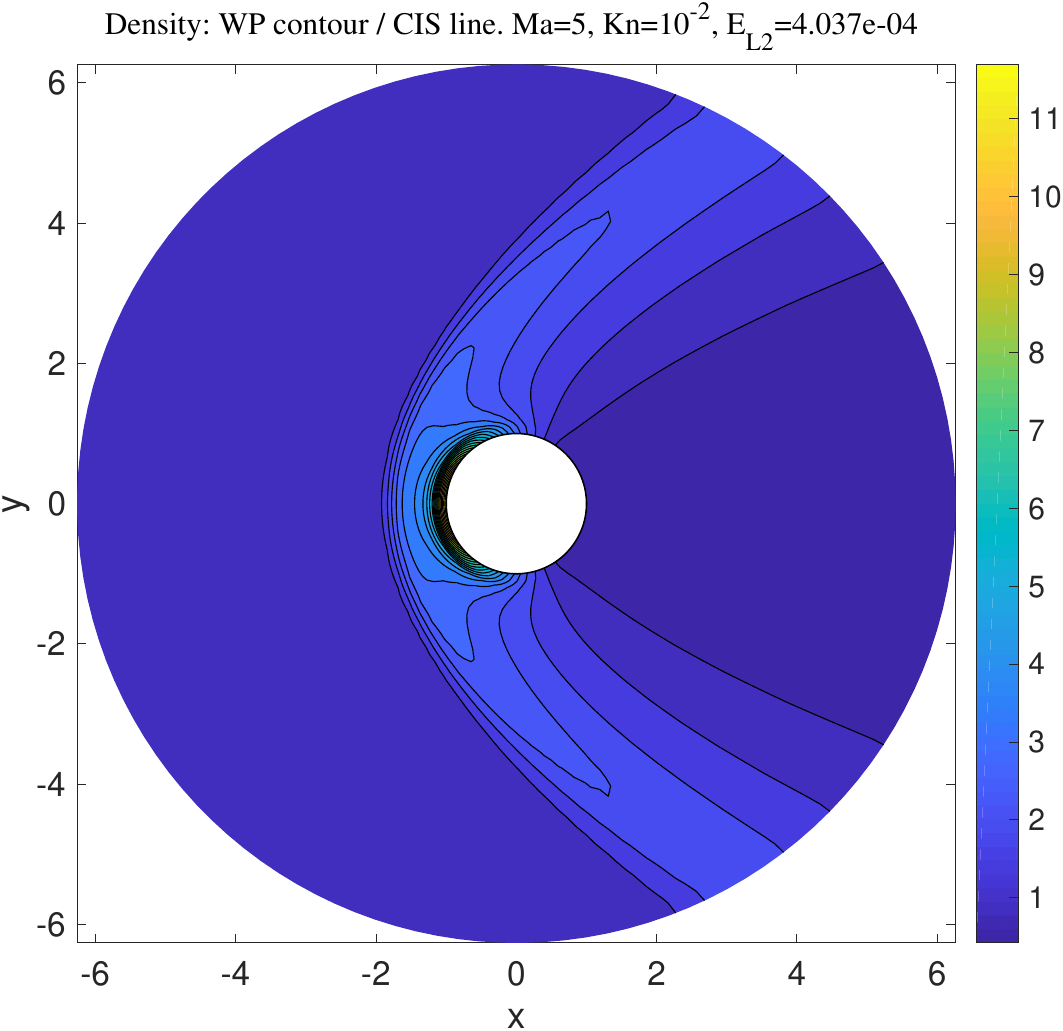}}
\caption{Density}
\end{subfigure}\hspace{0.0200\figW}%
\begin{subfigure}[t]{0.4894\figW}\centering
\raisebox{0.0000\figW}{\includegraphics[width=0.4894\figW]{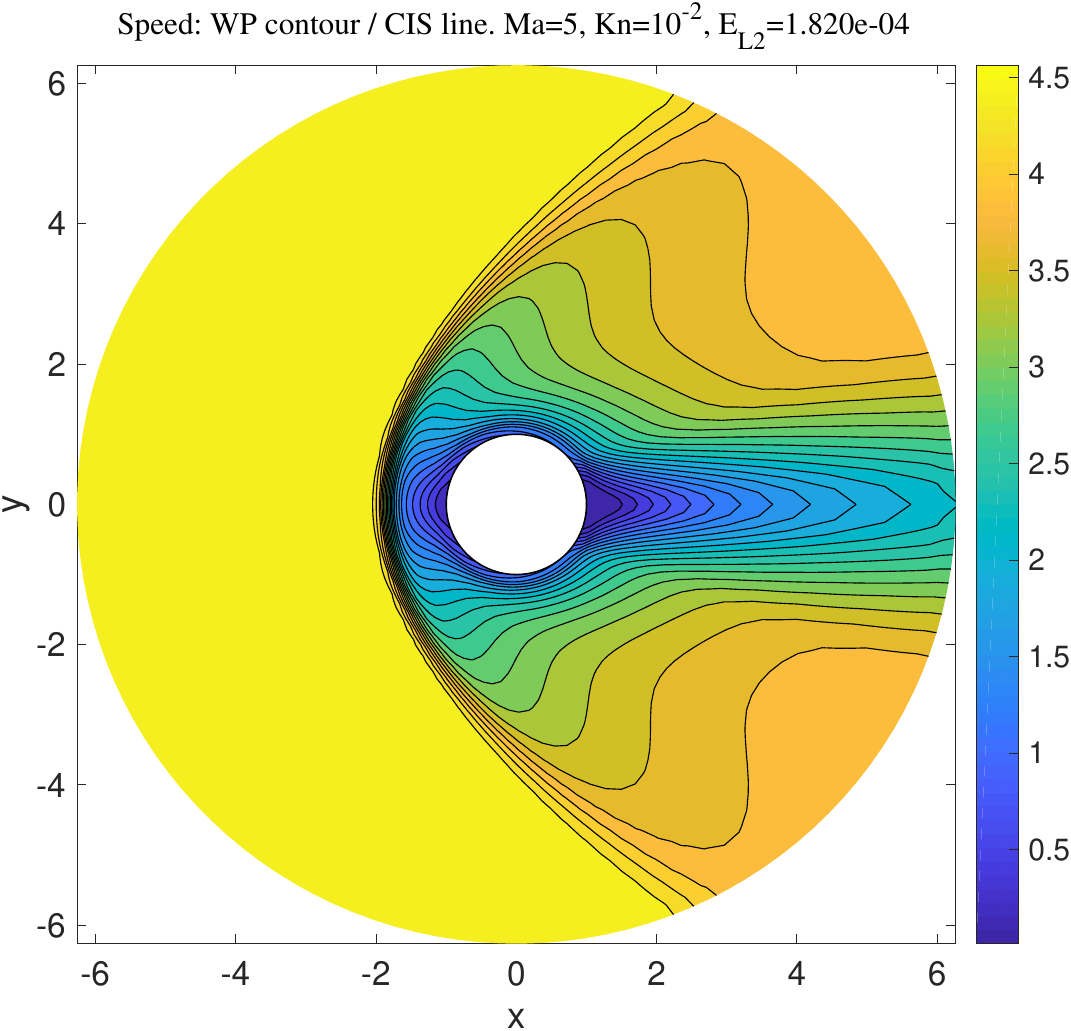}}
\caption{Speed}
\end{subfigure}
\par\smallskip
\begin{subfigure}[t]{0.4894\figW}\centering
\raisebox{0.0000\figW}{\includegraphics[width=0.4894\figW]{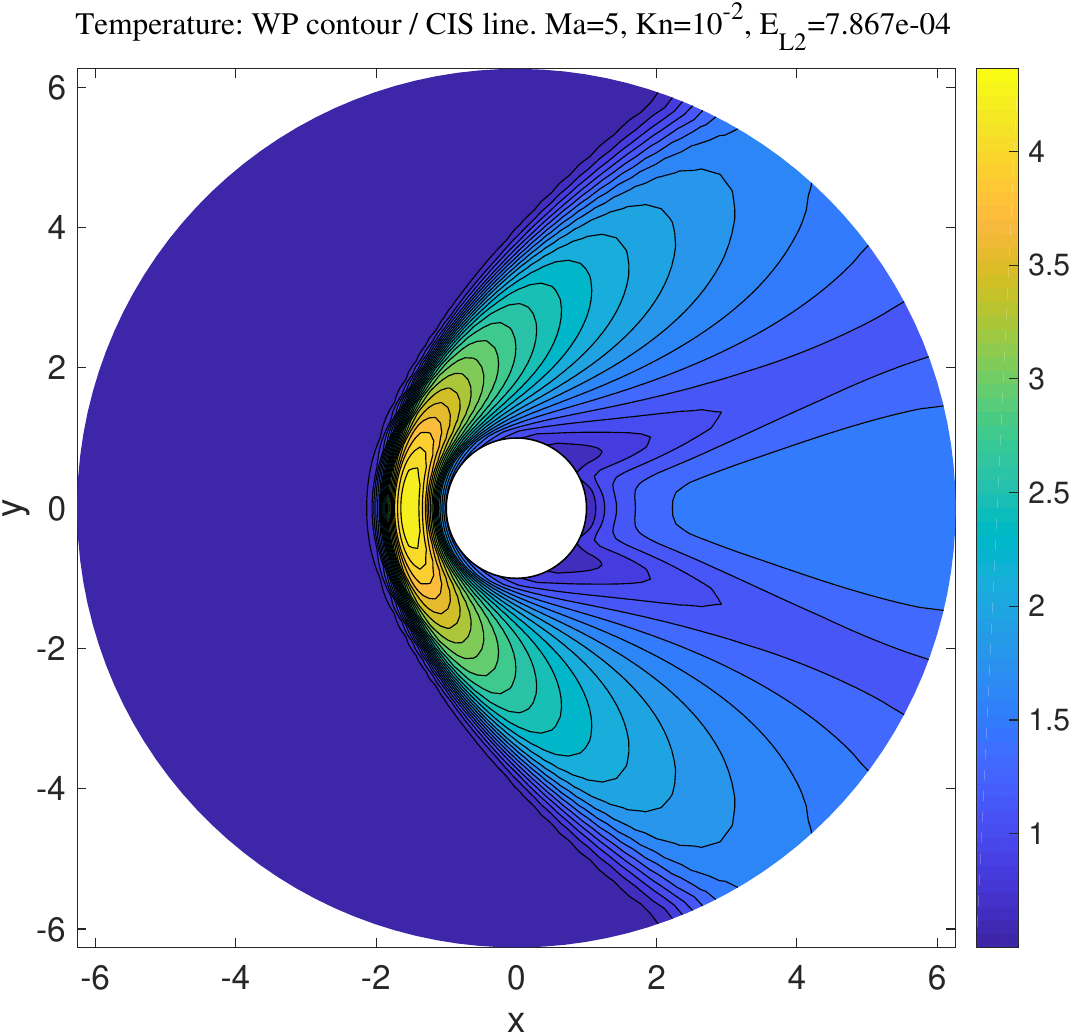}}
\caption{Temperature}
\end{subfigure}\hspace{0.0200\figW}%
\begin{subfigure}[t]{0.4856\figW}\centering
\raisebox{0.0000\figW}{\includegraphics[width=0.4856\figW]{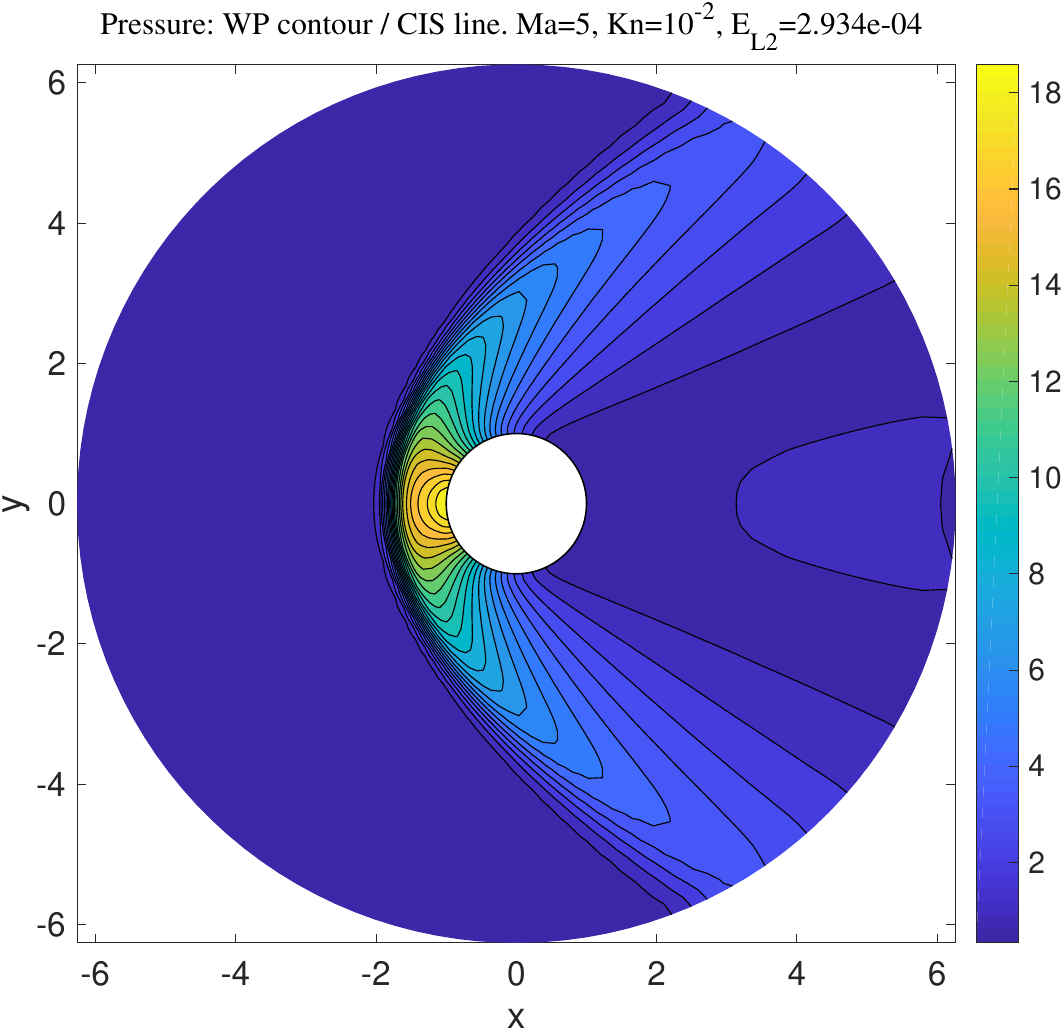}}
\caption{Pressure}
\end{subfigure}
\caption{Mach-5 flow past a cylinder, $\varepsilon=10^{-2}$: density, speed, temperature and pressure fields of WP compared with CIS. WP is shown by filled contours and the reference by lines at the same levels, and $E_{L_2}$ is printed above each panel.}\label{fig:cyl-kn1em2-fields}
\end{figure}
\begin{figure}[tbp]
\centering\singlespacing\setlength{\figW}{\linewidth}
\begin{subfigure}[t]{0.4892\figW}\centering
\raisebox{0.0000\figW}{\includegraphics[width=0.4892\figW]{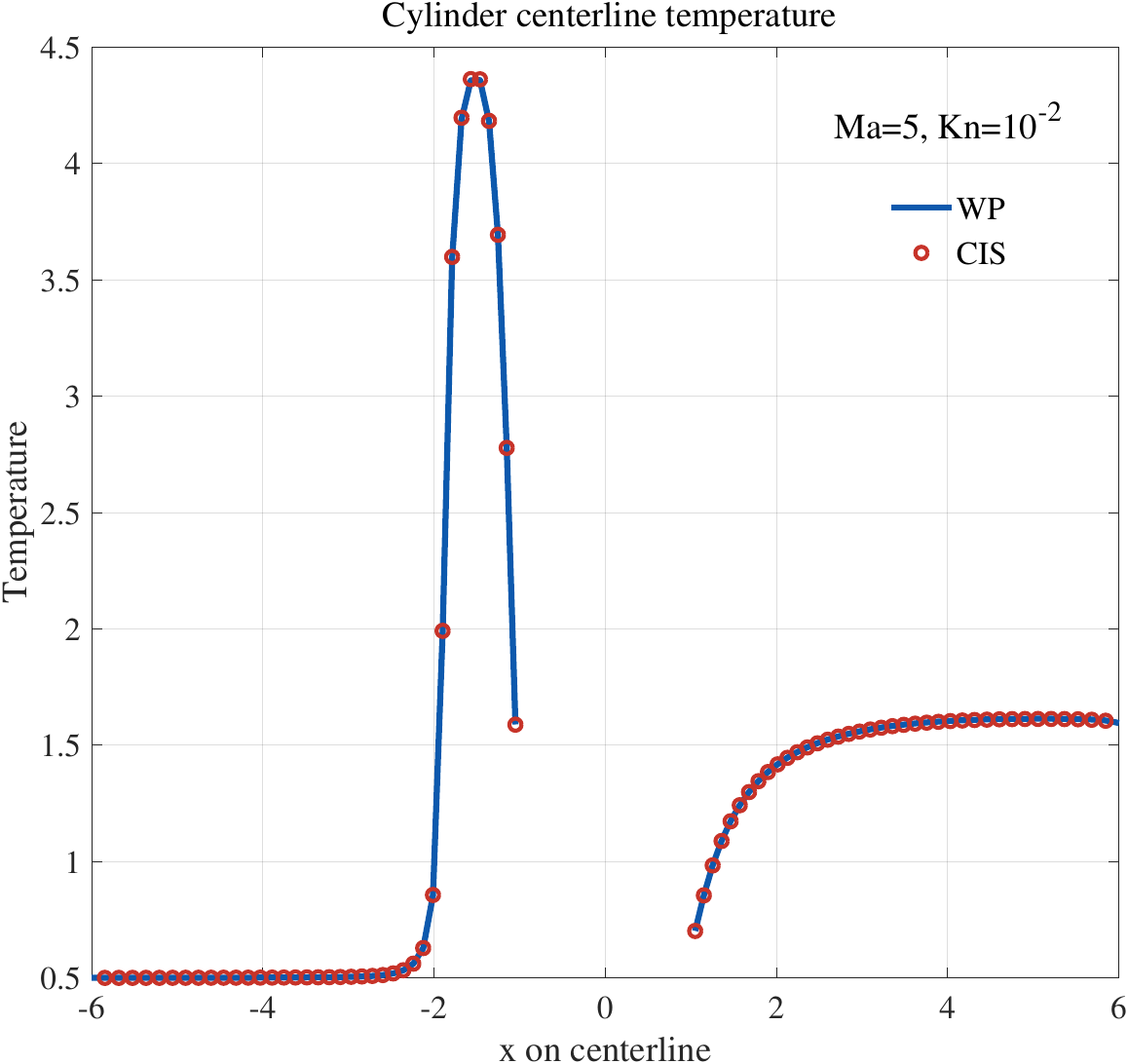}}
\caption{Centreline temperature}
\end{subfigure}\hspace{0.0200\figW}%
\begin{subfigure}[t]{0.4858\figW}\centering
\raisebox{0.0000\figW}{\includegraphics[width=0.4858\figW]{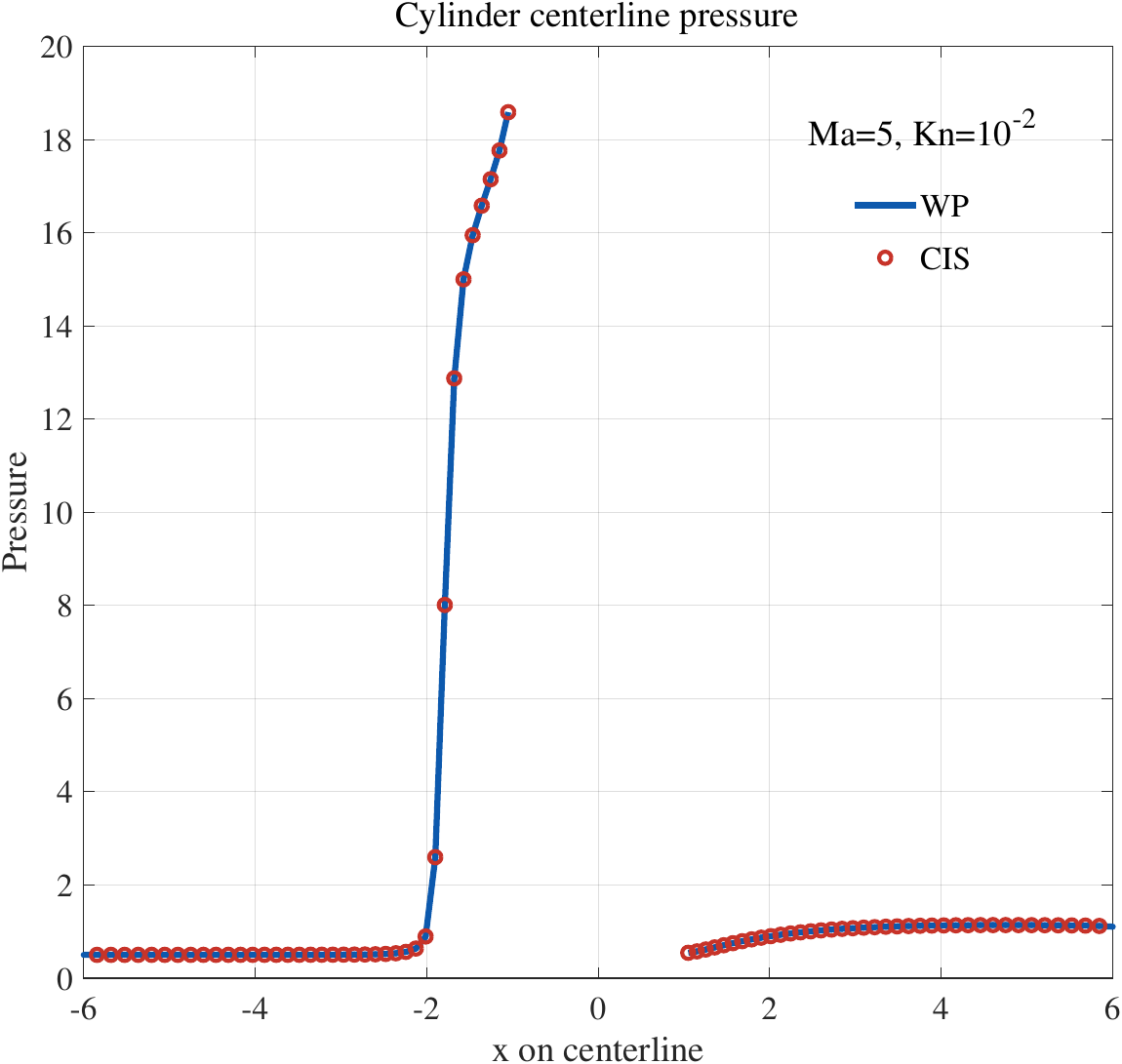}}
\caption{Centreline pressure}
\end{subfigure}
\par\smallskip
\begin{subfigure}[t]{0.4892\figW}\centering
\raisebox{0.0000\figW}{\includegraphics[width=0.4892\figW]{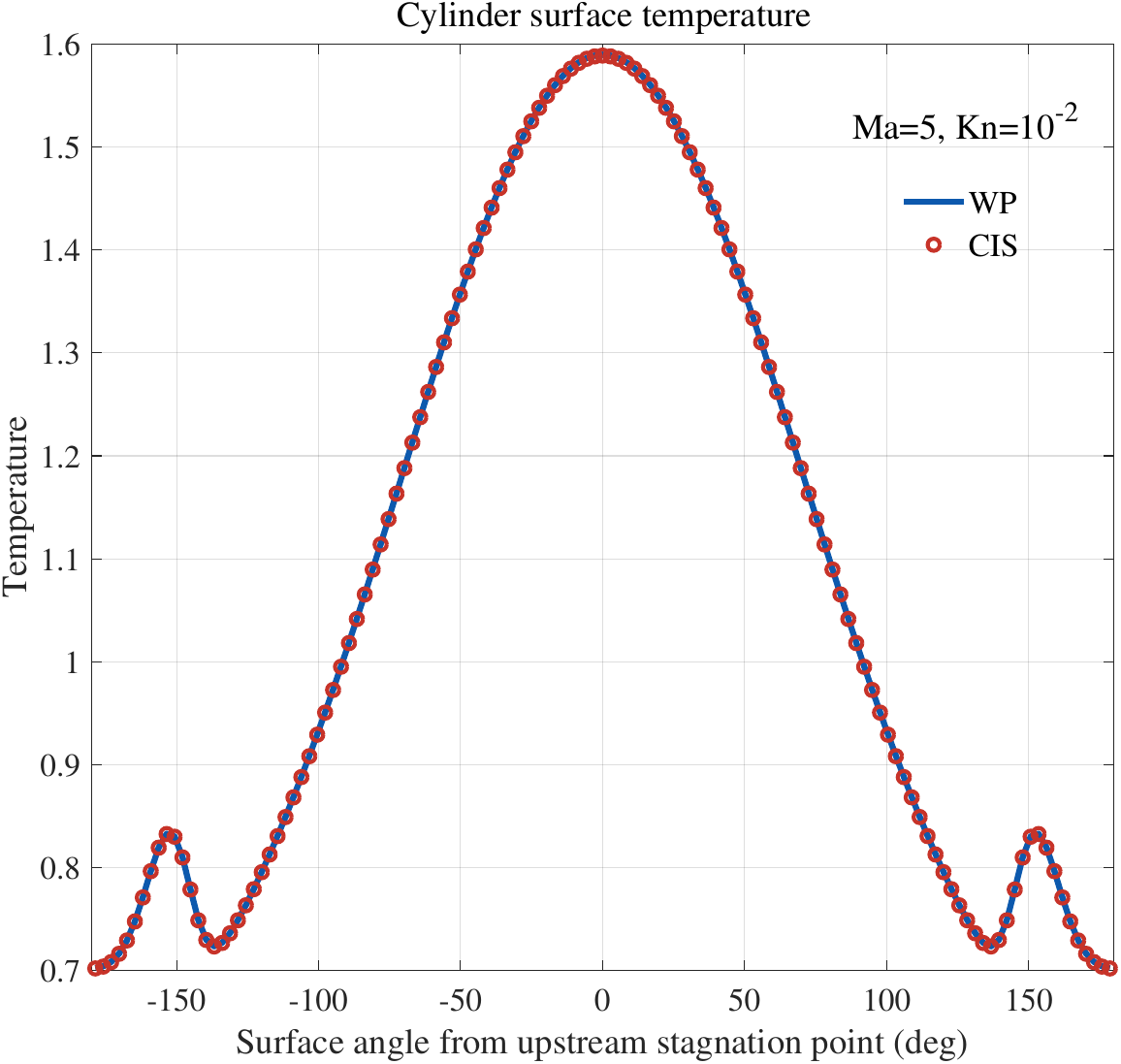}}
\caption{Surface temperature}
\end{subfigure}\hspace{0.0200\figW}%
\begin{subfigure}[t]{0.4858\figW}\centering
\raisebox{0.0000\figW}{\includegraphics[width=0.4858\figW]{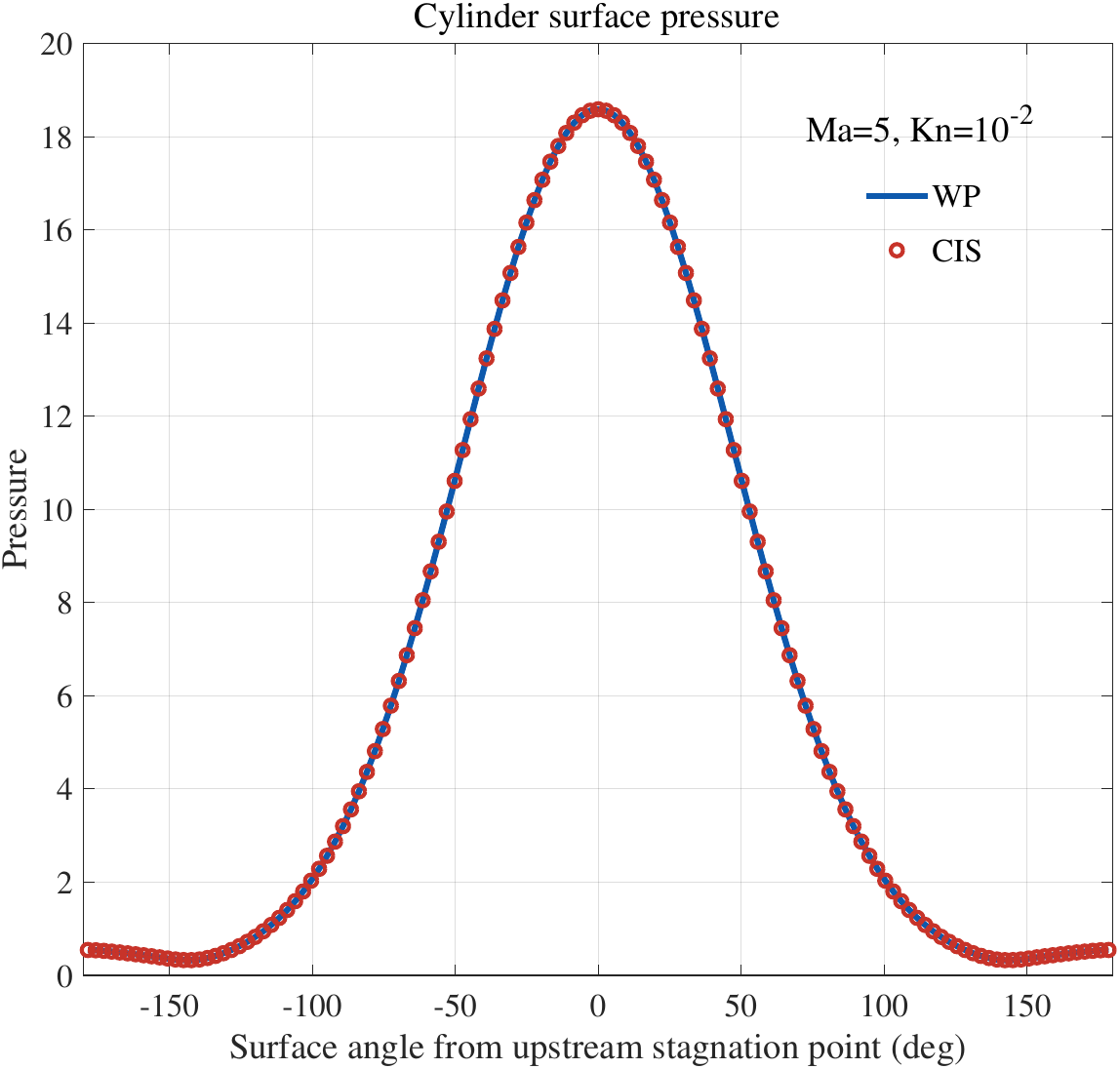}}
\caption{Surface pressure}
\end{subfigure}
\caption{Mach-5 flow past a cylinder, $\varepsilon=10^{-2}$: temperature and pressure of WP and CIS along the centreline (a, b) and along the body surface (c, d). WP is shown by lines and the reference by symbols.}\label{fig:cyl-kn1em2-profiles}
\end{figure}
\begin{figure}[tbp]
\centering\singlespacing\setlength{\figW}{\linewidth}
\begin{subfigure}[t]{0.4783\figW}\centering
\raisebox{0.0128\figW}{\includegraphics[width=0.4783\figW]{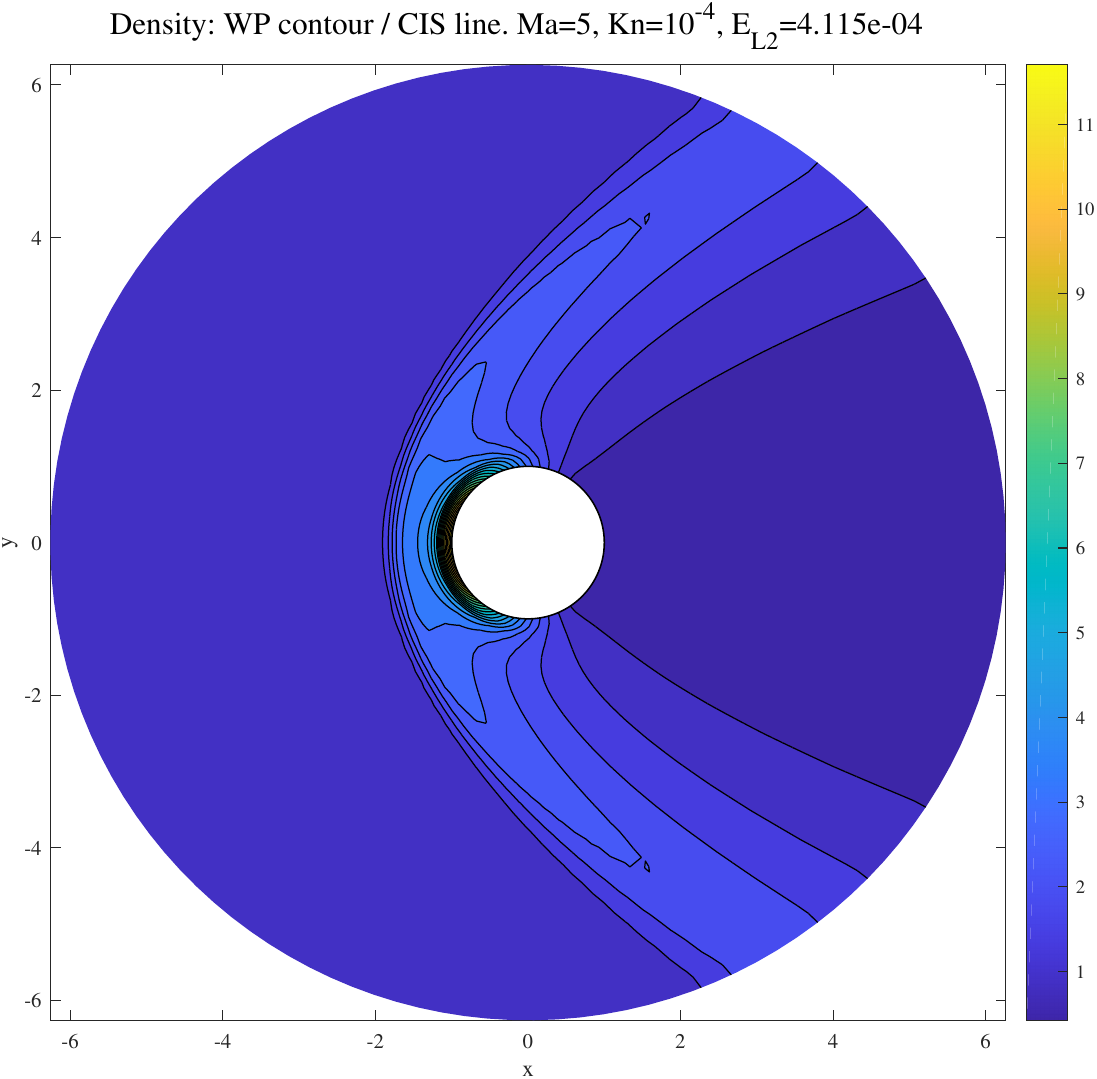}}
\caption{Density}
\end{subfigure}\hspace{0.0200\figW}%
\begin{subfigure}[t]{0.4967\figW}\centering
\raisebox{0.0000\figW}{\includegraphics[width=0.4967\figW]{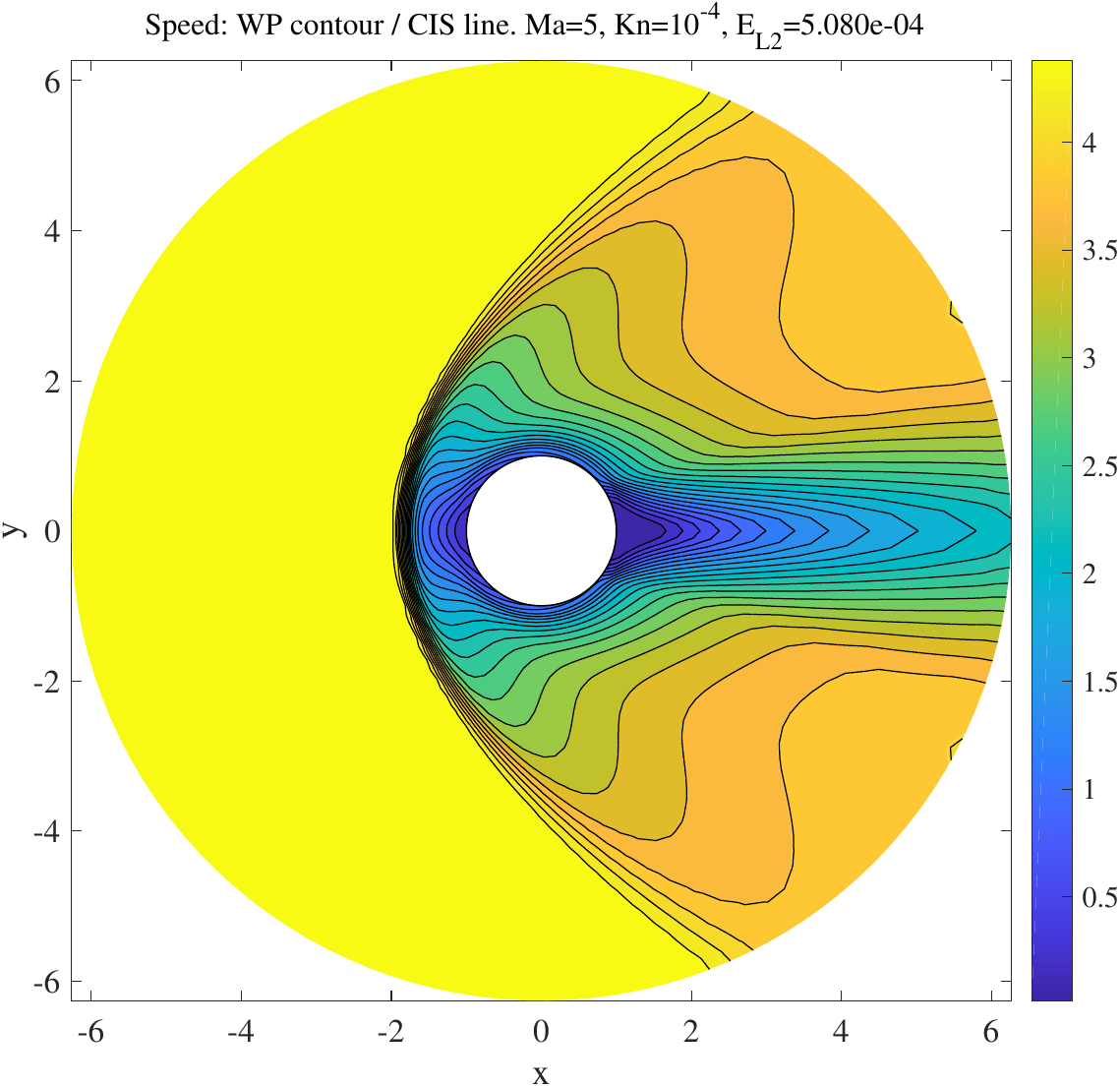}}
\caption{Speed}
\end{subfigure}
\par\smallskip
\begin{subfigure}[t]{0.4967\figW}\centering
\raisebox{0.0000\figW}{\includegraphics[width=0.4967\figW]{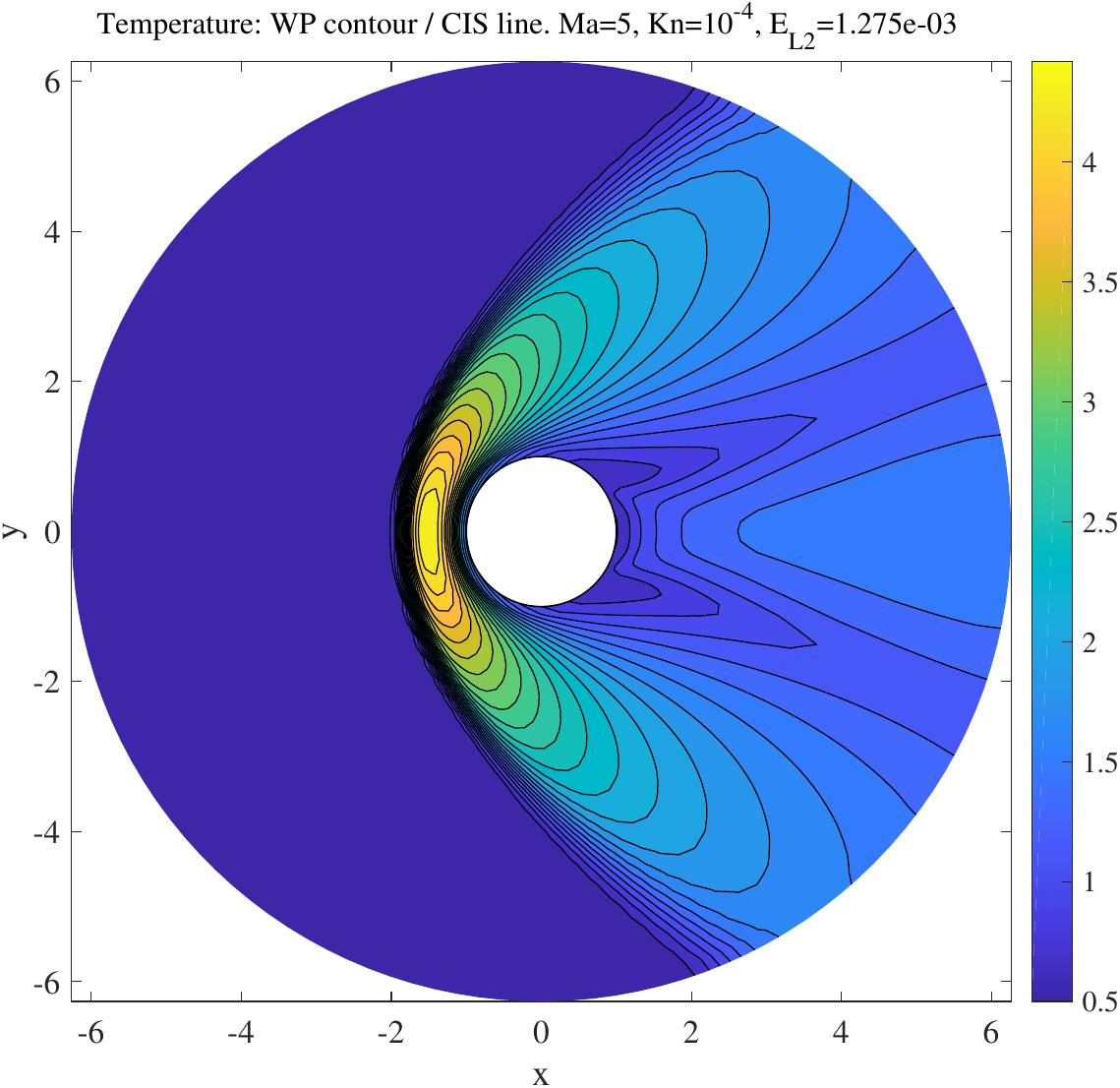}}
\caption{Temperature}
\end{subfigure}\hspace{0.0200\figW}%
\begin{subfigure}[t]{0.4783\figW}\centering
\raisebox{0.0128\figW}{\includegraphics[width=0.4783\figW]{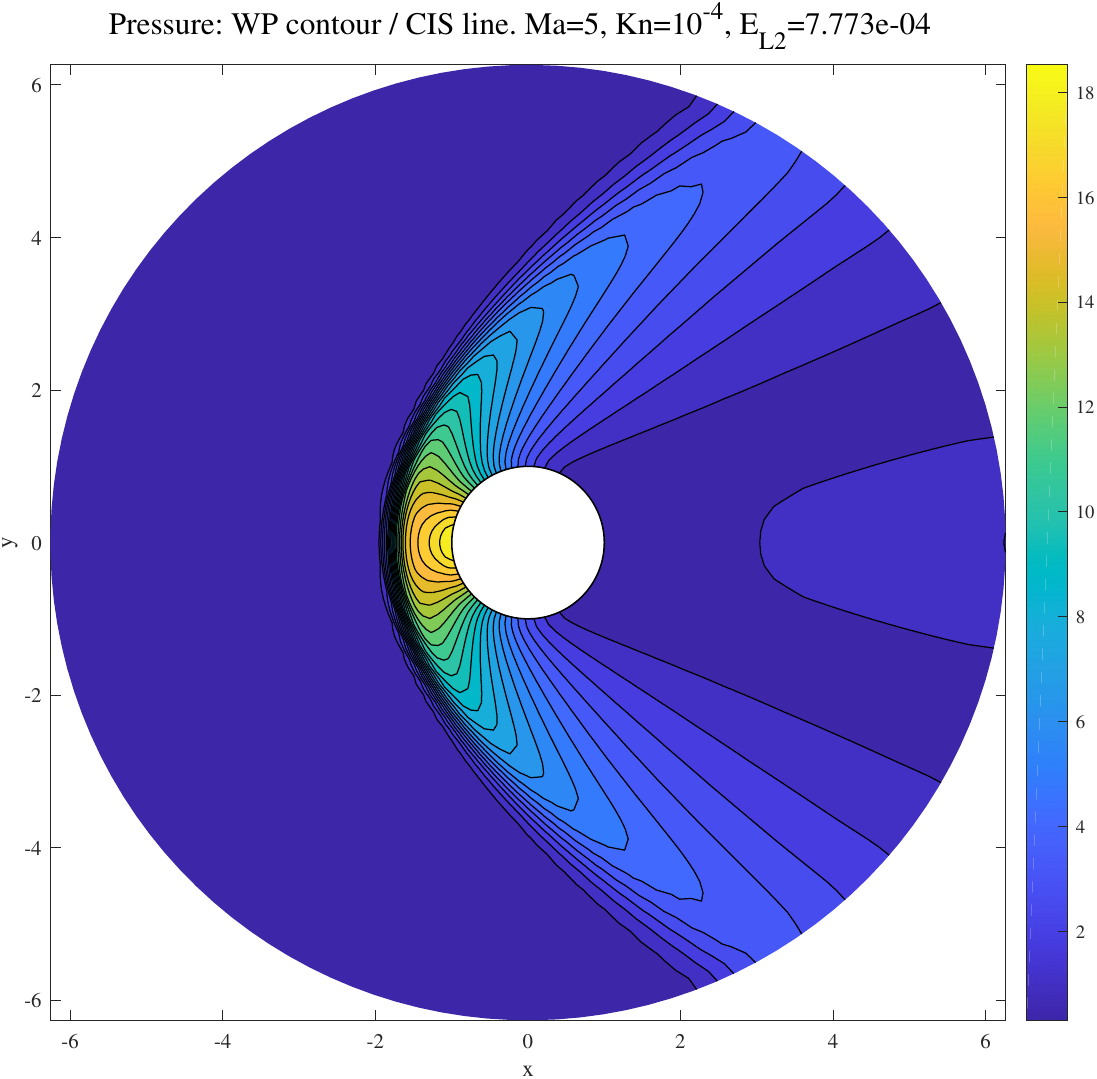}}
\caption{Pressure}
\end{subfigure}
\caption{Mach-5 flow past a cylinder, $\varepsilon=10^{-4}$: density, speed, temperature and pressure fields of WP compared with CIS of the Shakhov model~\cite{liu_wpd_shakhov}. WP is shown by filled contours and the reference by lines at the same levels, and $E_{L_2}$ is printed above each panel.}\label{fig:cyl-kn1em4-fields}
\end{figure}
\begin{figure}[tbp]
\centering\singlespacing\setlength{\figW}{\linewidth}
\begin{subfigure}[t]{0.9706\figW}\centering
\makebox[\linewidth]{\raisebox{0.0027\figW}{\includegraphics[width=0.4770\figW]{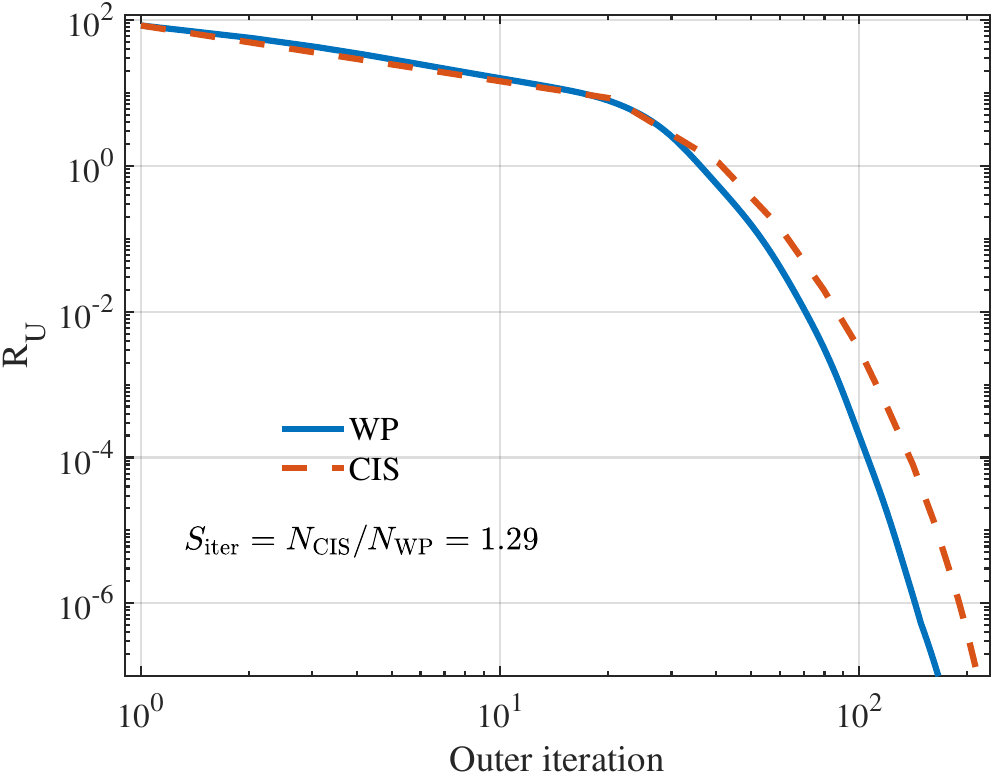}}\hspace{0.0200\figW}\raisebox{0.0000\figW}{\includegraphics[width=0.4736\figW]{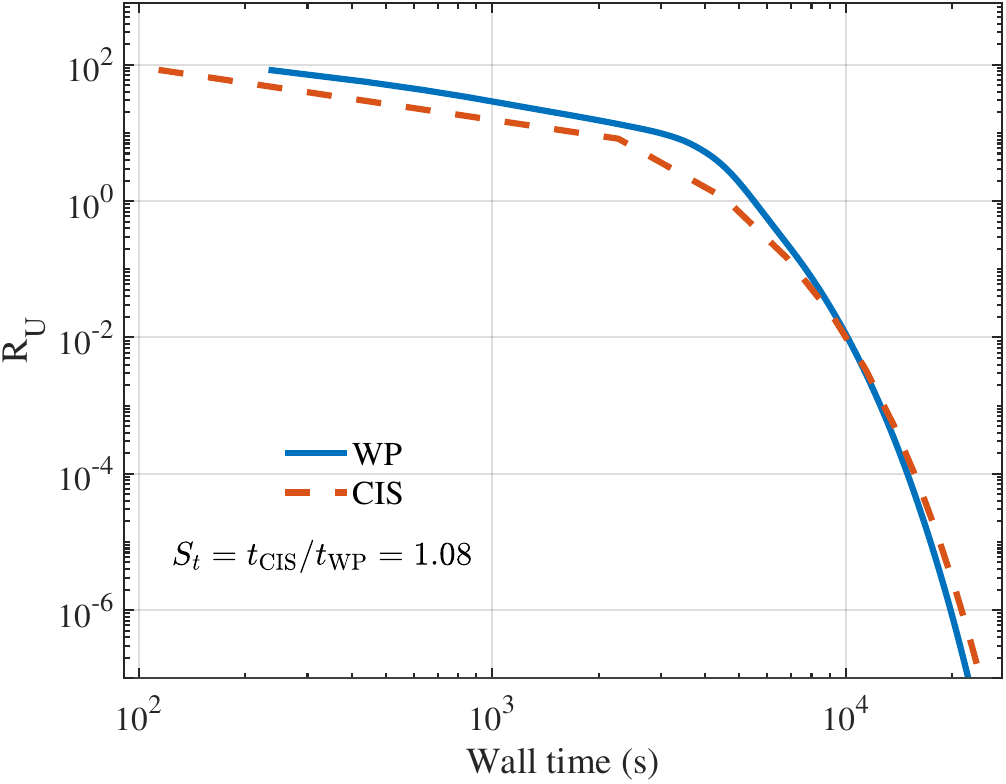}}}
\caption{$\varepsilon=1$}\label{fig:cyl-res-kn1}
\end{subfigure}
\par\smallskip
\begin{subfigure}[t]{0.9706\figW}\centering
\makebox[\linewidth]{\raisebox{0.0000\figW}{\includegraphics[width=0.4941\figW]{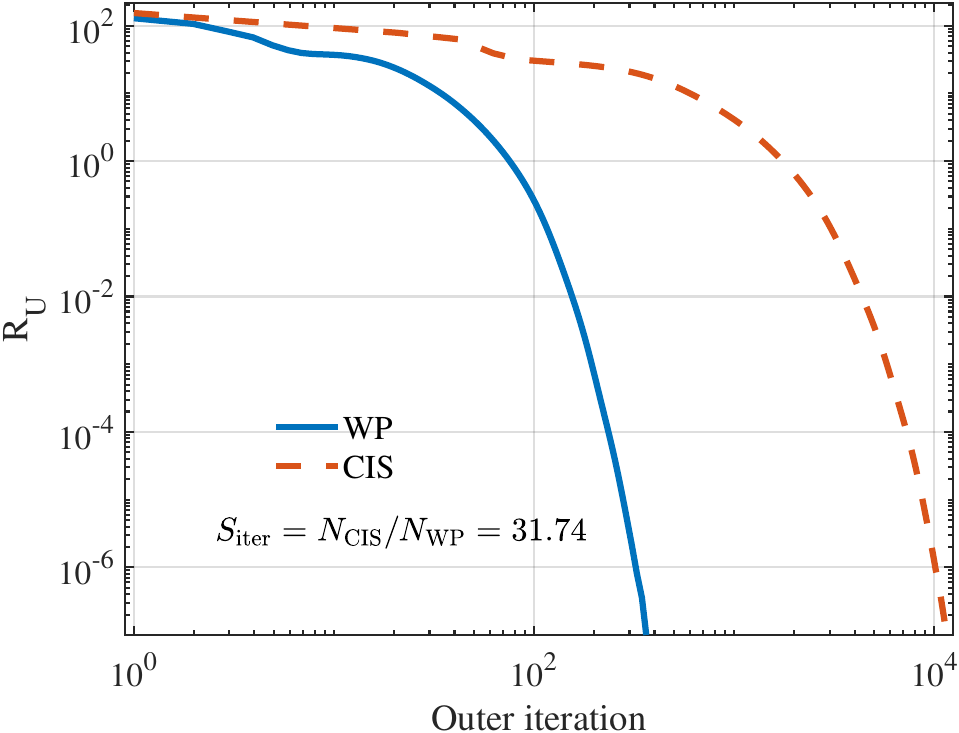}}\hspace{0.0200\figW}\raisebox{0.0053\figW}{\includegraphics[width=0.4564\figW]{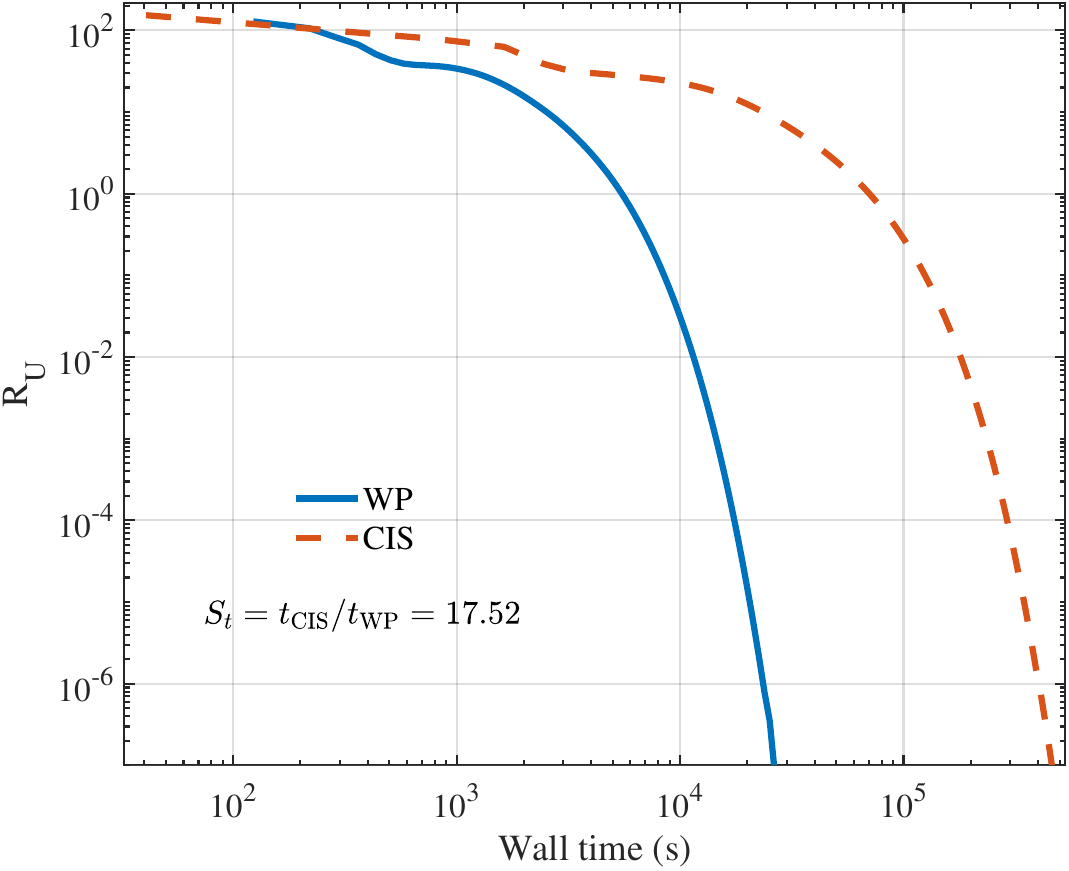}}}
\caption{$\varepsilon=10^{-2}$}\label{fig:cyl-res-kn1em2}
\end{subfigure}
\par\smallskip
\begin{subfigure}[t]{0.9706\figW}\centering
\makebox[\linewidth]{\raisebox{0.0000\figW}{\includegraphics[width=0.4997\figW]{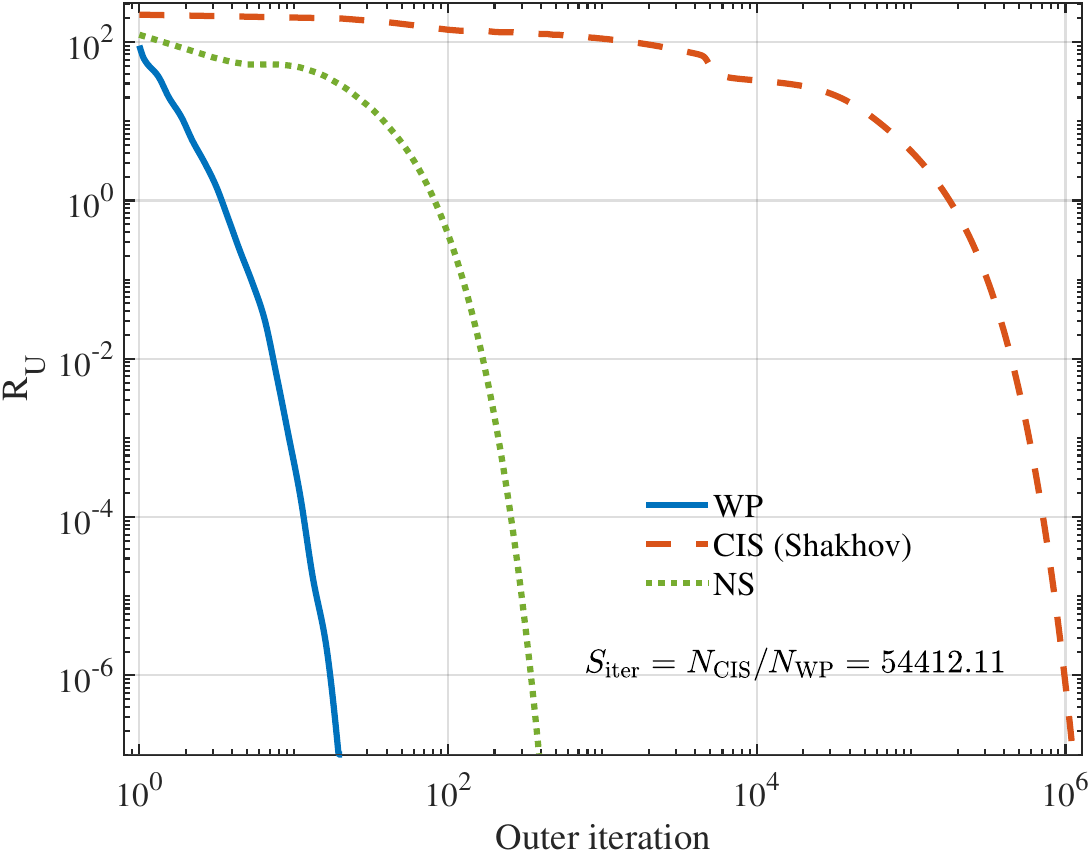}}\hspace{0.0200\figW}\raisebox{0.0056\figW}{\includegraphics[width=0.4509\figW]{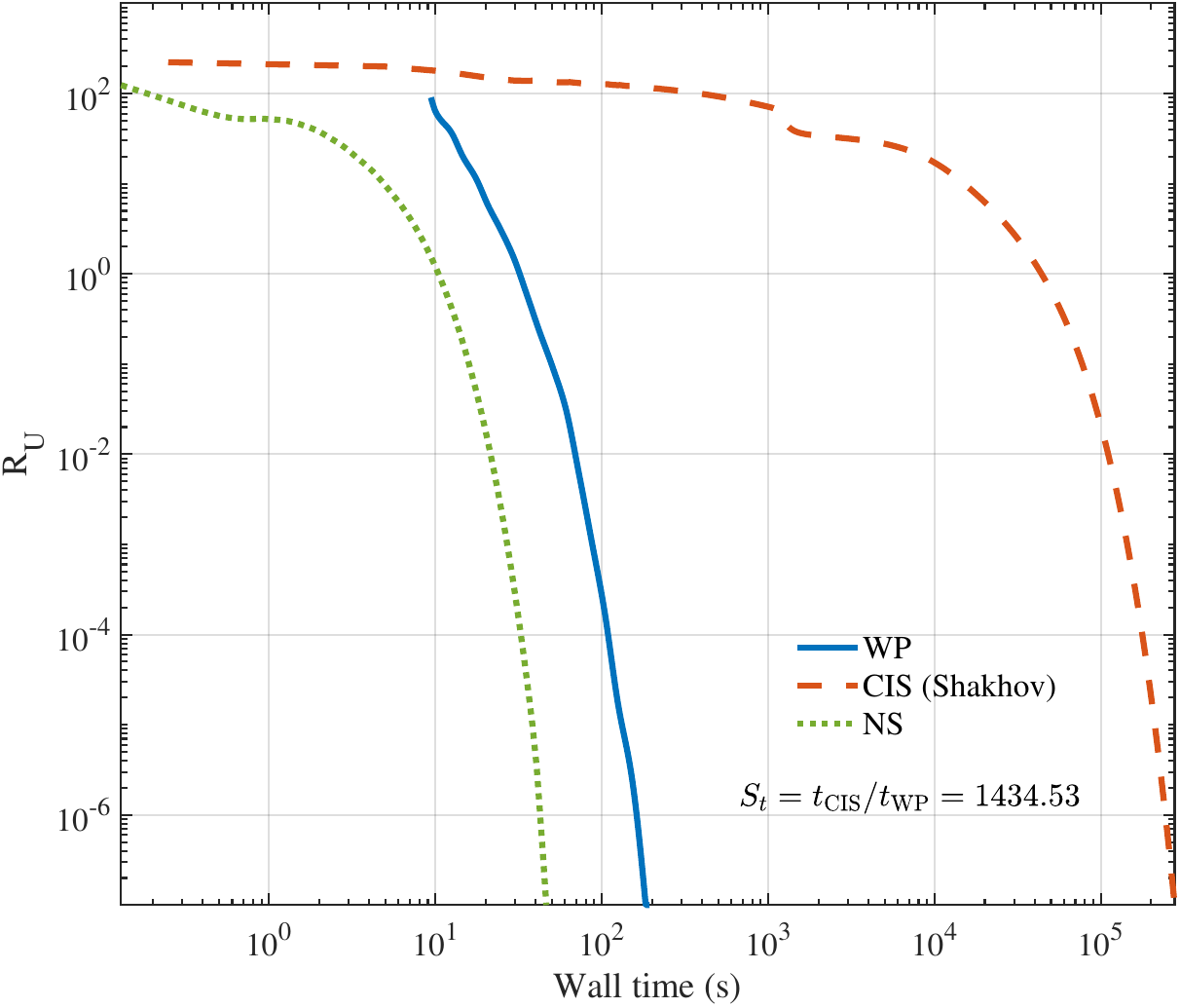}}}
\caption{$\varepsilon=10^{-4}$}\label{fig:cyl-res-kn1em4}
\end{subfigure}
\caption{Mach-5 flow past a cylinder: histories of the normalized macroscopic residual of WP and the reference solvers against outer iteration (left) and wall time (right). At $\varepsilon=10^{-4}$ both panels show CIS of the Shakhov model~\cite{liu_wpd_shakhov} and the NS solver.}\label{fig:cyl-residuals}
\end{figure}

\FloatBarrier
\subsection{Apollo flow}
The last test is Mach-5 flow around a two-dimensional section of the Apollo
capsule, with a stationary isothermal diffuse wall and a uniform far field.
The body-fitted O-grid has $64\times38$ cells at $\varepsilon=1$, where the
far-field boundary is placed farther out, and $64\times24$ cells at
$\varepsilon=10^{-2}$ and $10^{-4}$. The velocity grid covers $[-9,9]^3$ with
$96\times96\times24$ nodes at the two larger Knudsen numbers and
$32\times32\times12$ nodes at $\varepsilon=10^{-4}$. The controls and configurations of the
coupled WP iteration are those of the cylinder, and
the reference is again CIS at $\varepsilon=1$ and $10^{-2}$; at
$\varepsilon=10^{-4}$ the fields are compared with CIS of the Shakhov model and
the profiles with the NS solution.

At $\varepsilon=1$ the WP and CIS fields in Fig.~\ref{fig:apollo-kn1-fields}
differ by $3.08\times10^{-5}$ in density, $1.01\times10^{-5}$ in temperature,
$4.69\times10^{-6}$ in speed and $9.95\times10^{-6}$ in pressure, and the
pressure-drag coefficients of WP and CIS, both $1.38$, differ by
$2.19\times10^{-6}$. The centreline and surface profiles in
Fig.~\ref{fig:apollo-kn1-profiles} follow the compression ahead of the
forebody and the relaxation around the capsule, and the four molecular
distributions sampled on the centreline in Fig.~\ref{fig:apollo-kn1-dist}
show an asymmetric non-Maxwellian structure.

\begin{figure}[tbp]
\centering\singlespacing\setlength{\figW}{\linewidth}
\begin{subfigure}[t]{0.4856\figW}\centering
\raisebox{0.0000\figW}{\includegraphics[width=0.4856\figW]{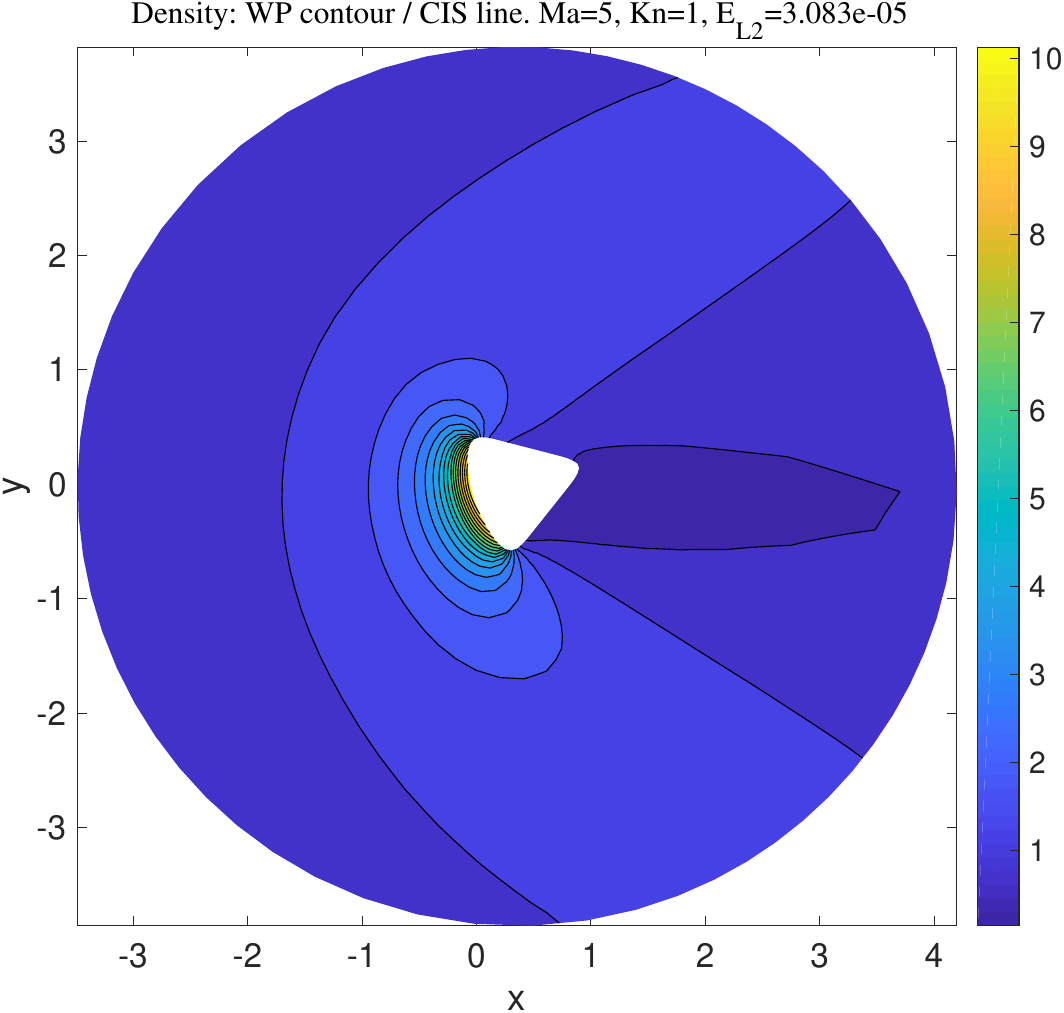}}
\caption{Density}
\end{subfigure}\hspace{0.0200\figW}%
\begin{subfigure}[t]{0.4894\figW}\centering
\raisebox{0.0000\figW}{\includegraphics[width=0.4894\figW]{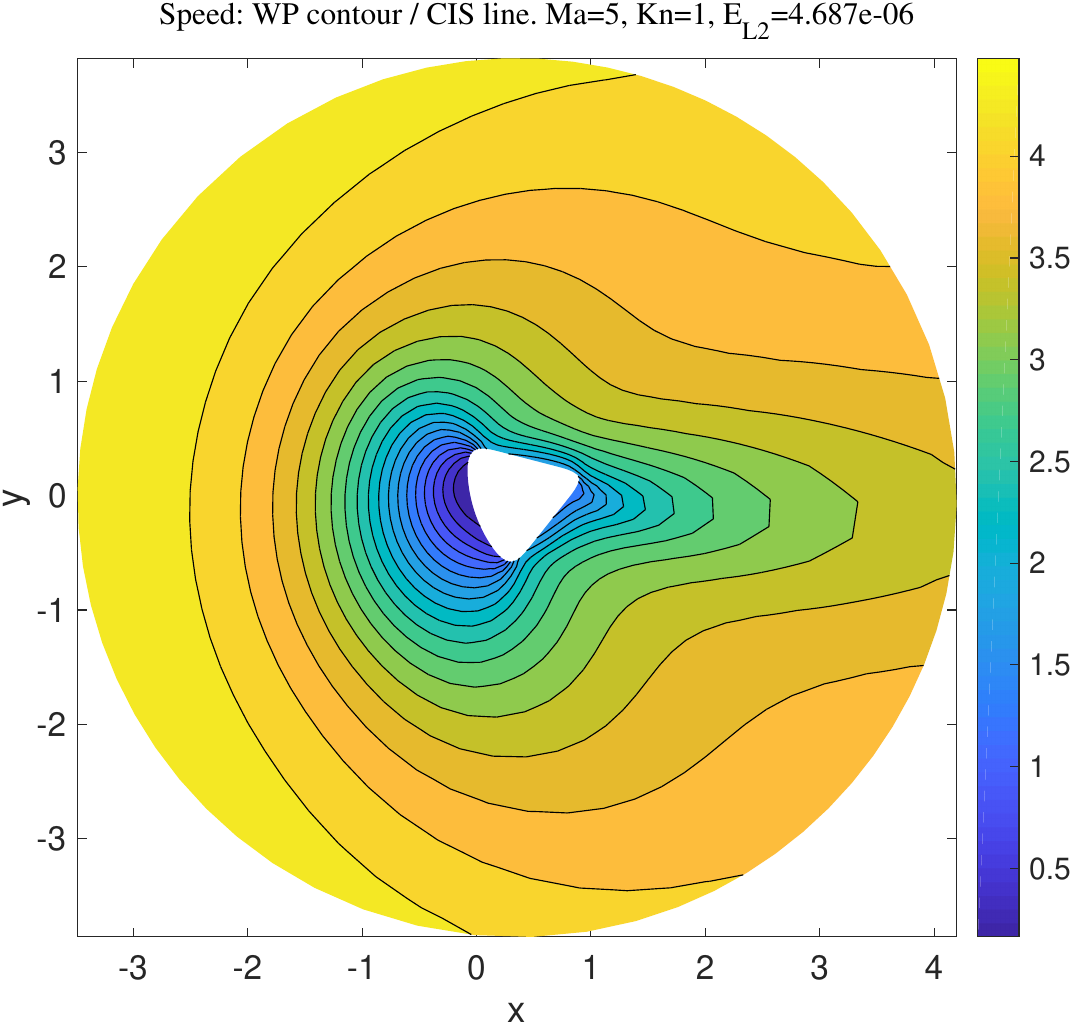}}
\caption{Speed}
\end{subfigure}
\par\smallskip
\begin{subfigure}[t]{0.4894\figW}\centering
\raisebox{0.0000\figW}{\includegraphics[width=0.4894\figW]{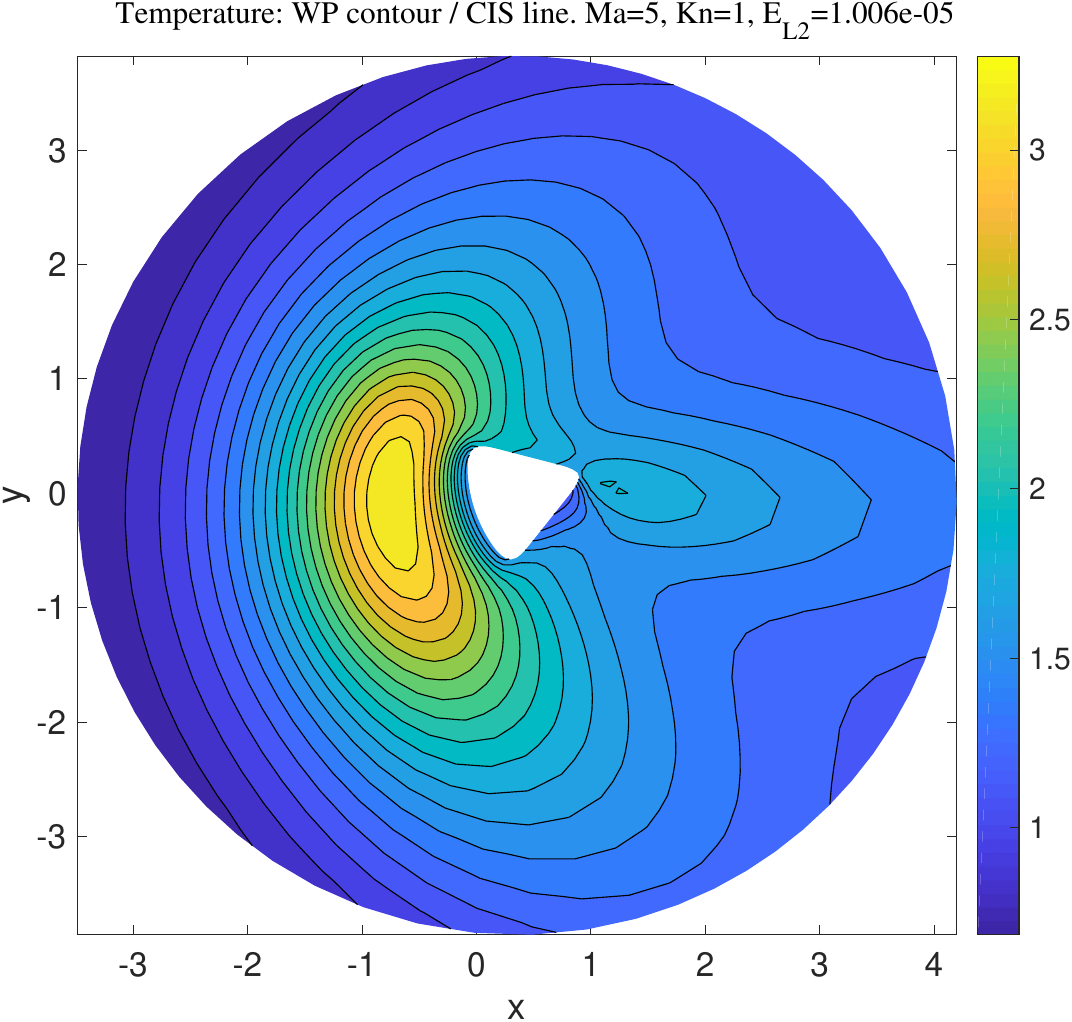}}
\caption{Temperature}
\end{subfigure}\hspace{0.0200\figW}%
\begin{subfigure}[t]{0.4856\figW}\centering
\raisebox{0.0000\figW}{\includegraphics[width=0.4856\figW]{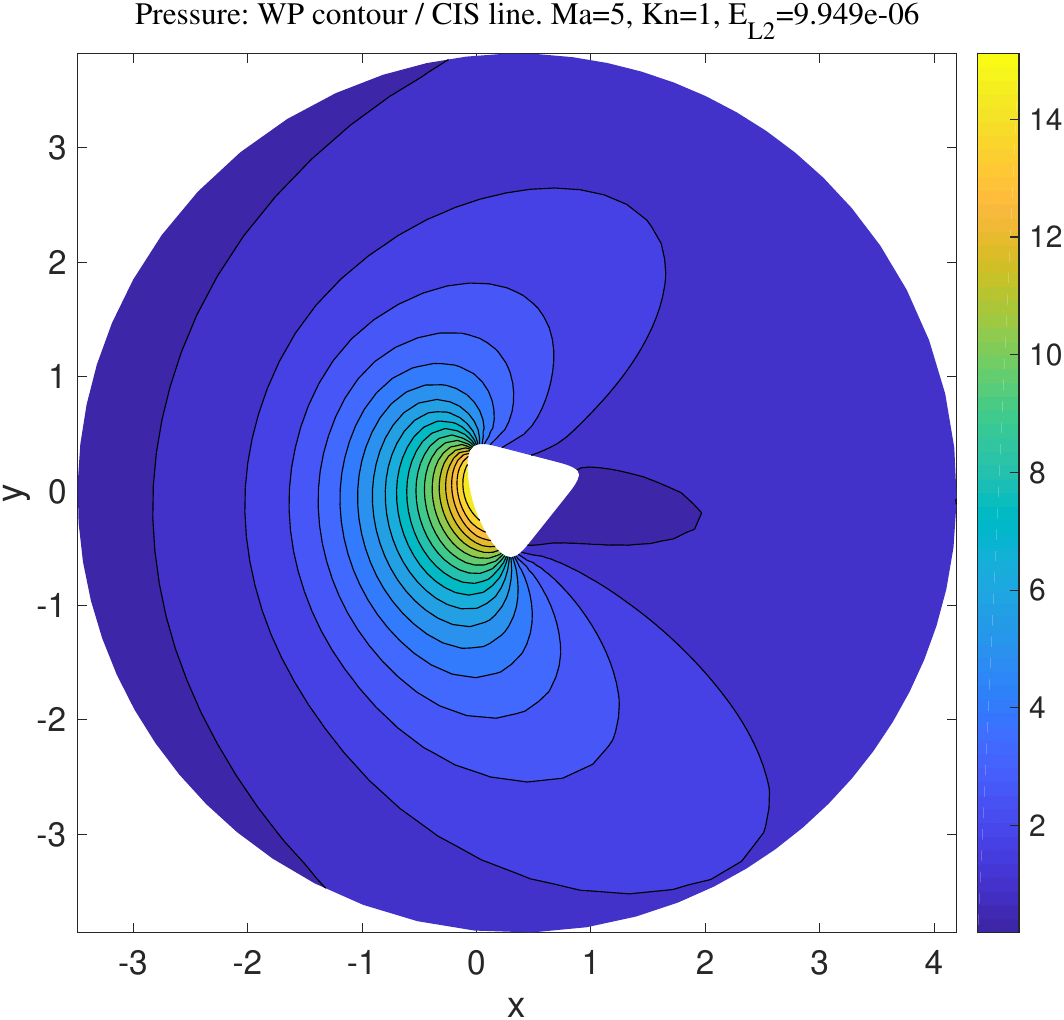}}
\caption{Pressure}
\end{subfigure}
\caption{Mach-5 flow around the Apollo section, $\varepsilon=1$: density, speed, temperature and pressure fields of WP compared with CIS. WP is shown by filled contours and the reference by lines at the same levels, and $E_{L_2}$ is printed above each panel.}\label{fig:apollo-kn1-fields}
\end{figure}
\begin{figure}[tbp]
\centering\singlespacing\setlength{\figW}{\linewidth}
\begin{subfigure}[t]{0.4908\figW}\centering
\raisebox{0.0000\figW}{\includegraphics[width=0.4908\figW]{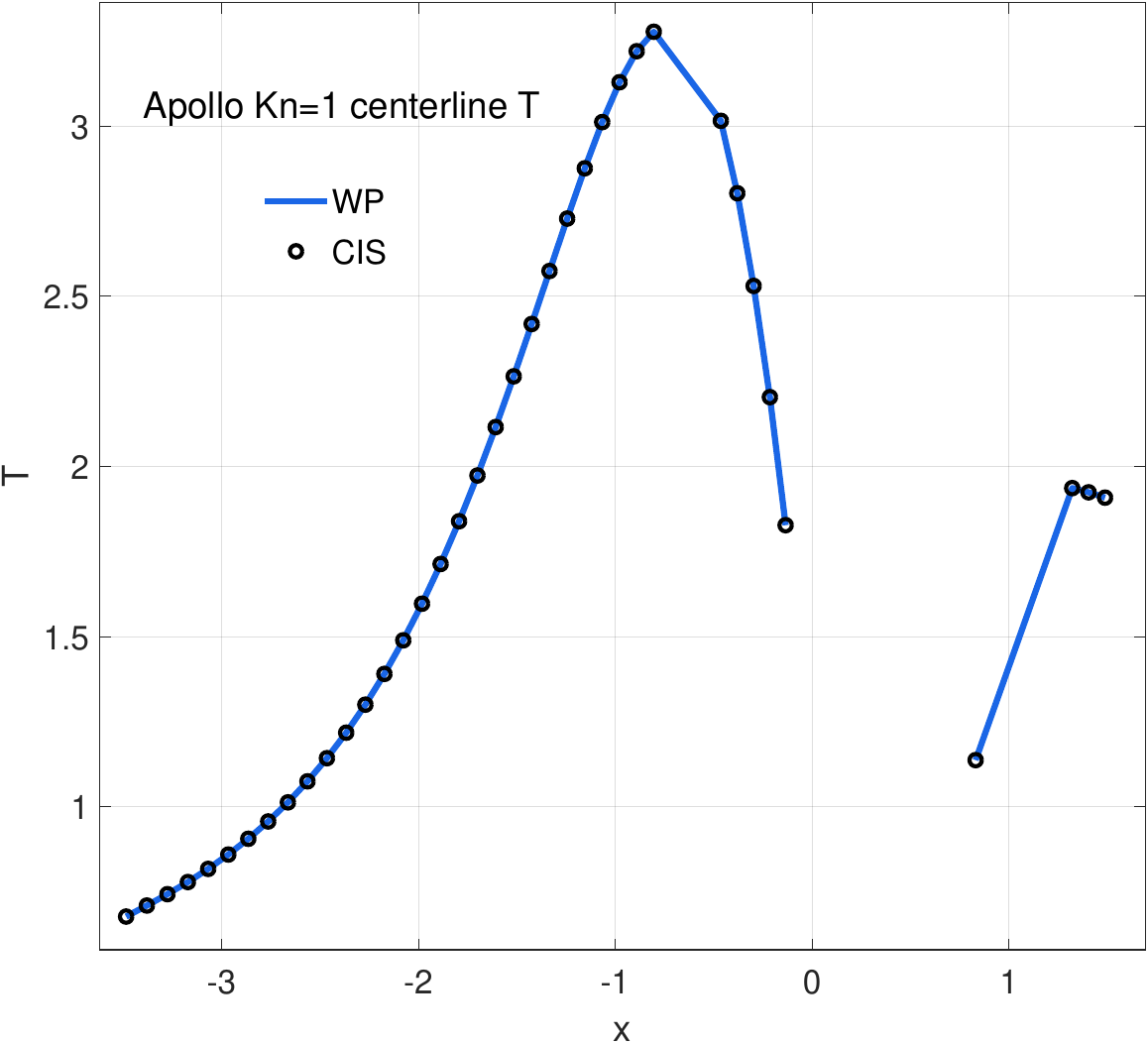}}
\caption{Centreline temperature}
\end{subfigure}\hspace{0.0200\figW}%
\begin{subfigure}[t]{0.4842\figW}\centering
\raisebox{0.0000\figW}{\includegraphics[width=0.4842\figW]{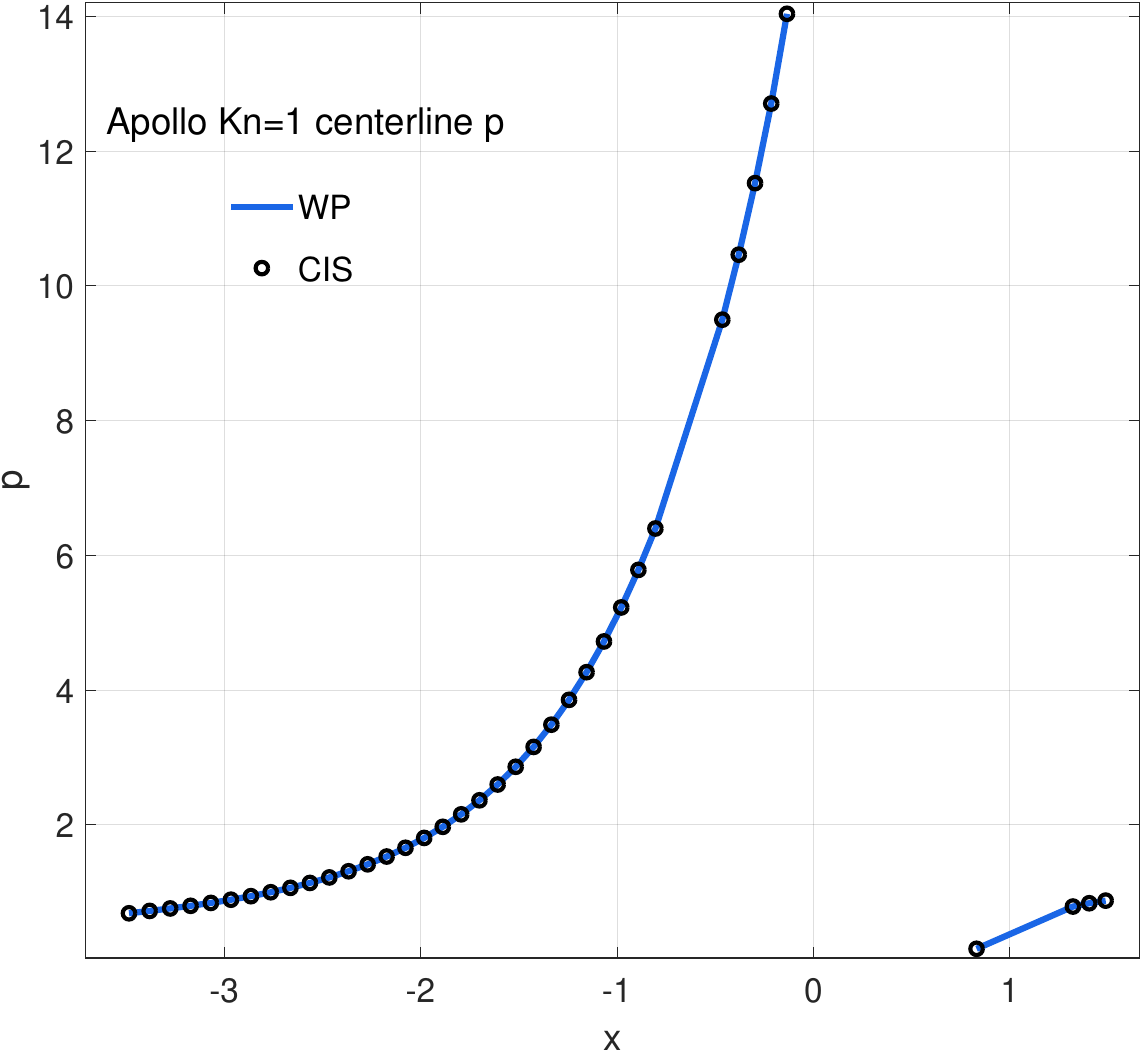}}
\caption{Centreline pressure}
\end{subfigure}
\par\smallskip
\begin{subfigure}[t]{0.4908\figW}\centering
\raisebox{0.0000\figW}{\includegraphics[width=0.4908\figW]{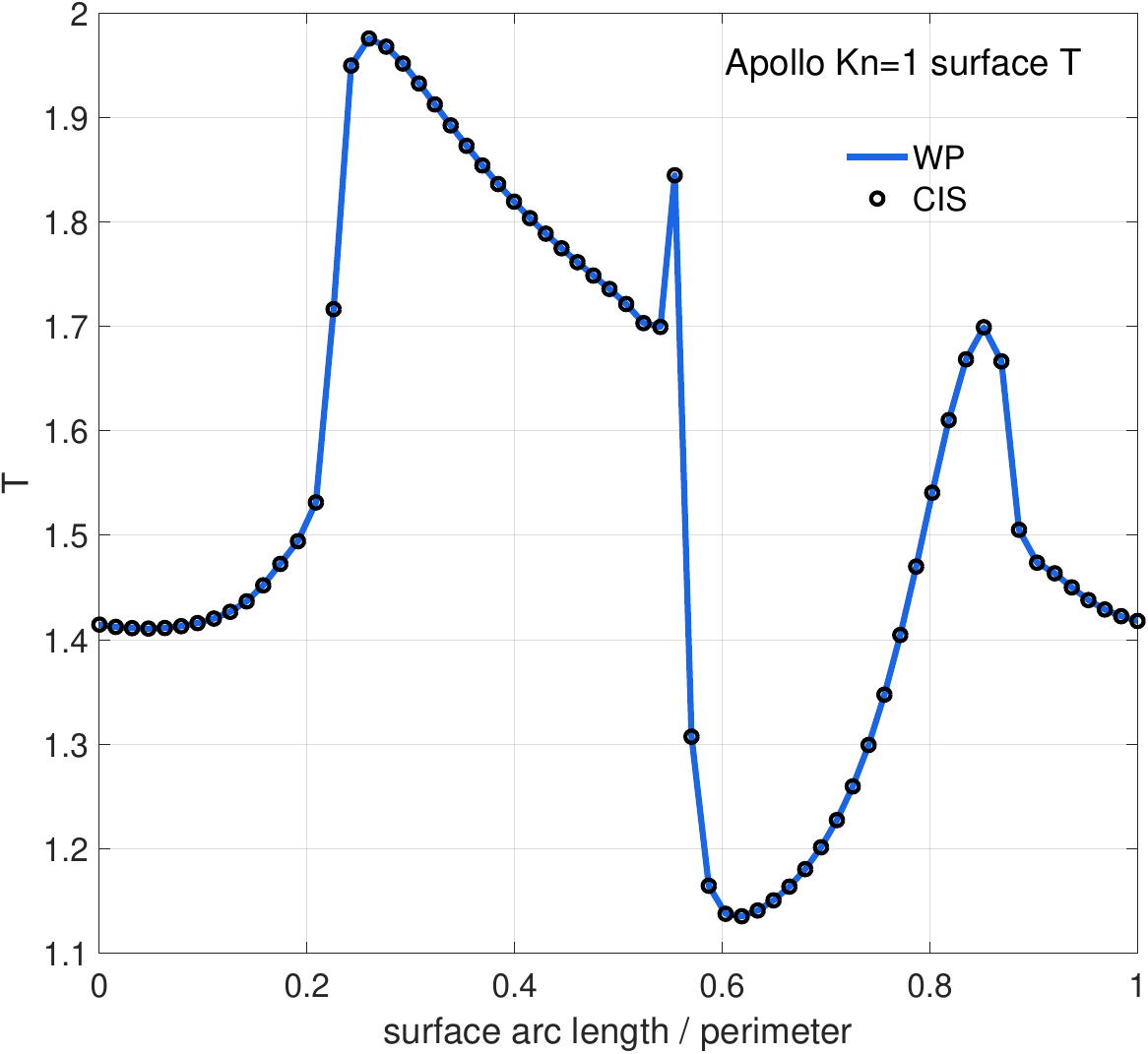}}
\caption{Surface temperature}
\end{subfigure}\hspace{0.0200\figW}%
\begin{subfigure}[t]{0.4842\figW}\centering
\raisebox{0.0000\figW}{\includegraphics[width=0.4842\figW]{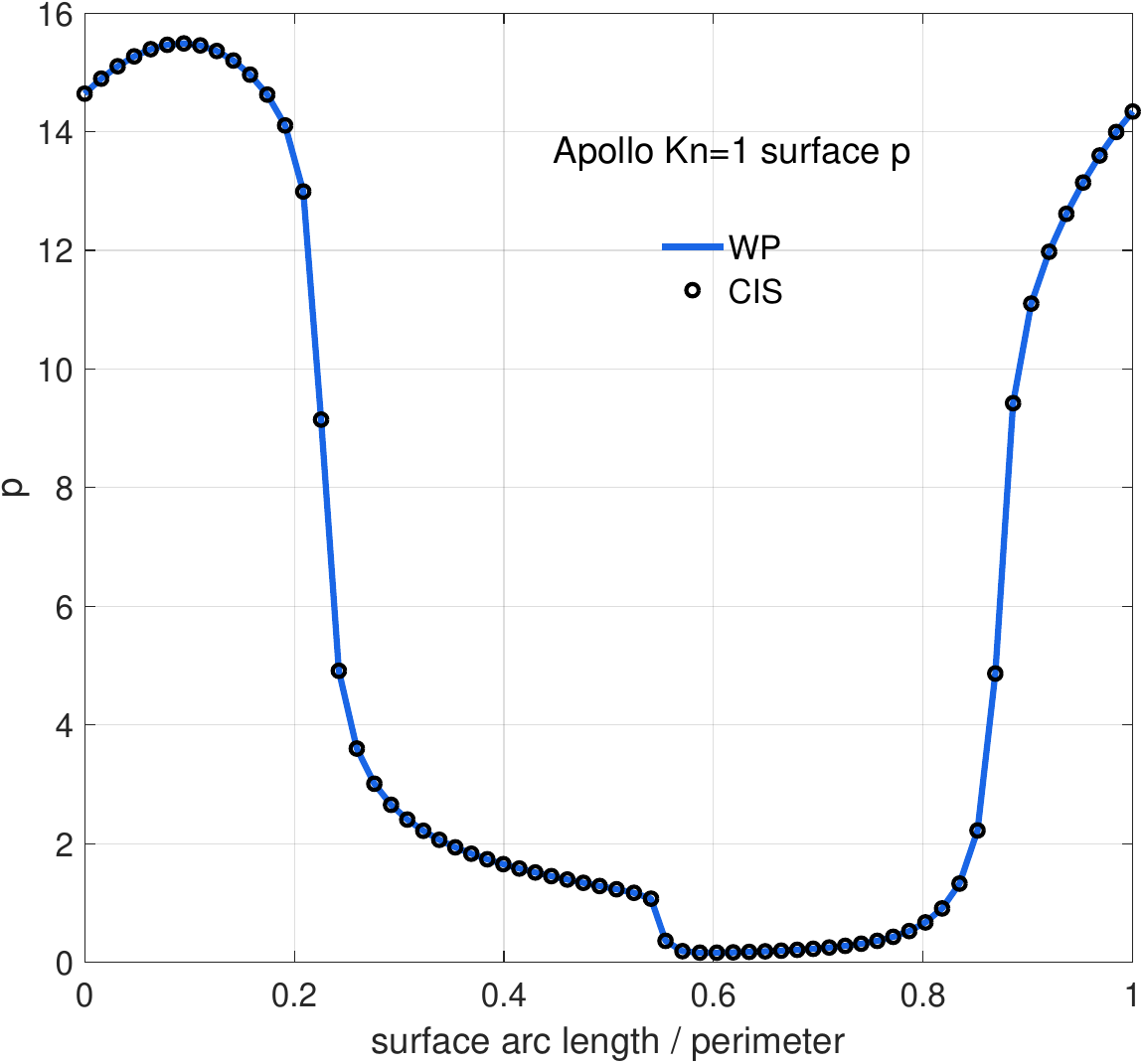}}
\caption{Surface pressure}
\end{subfigure}
\caption{Mach-5 flow around the Apollo section, $\varepsilon=1$: temperature and pressure of WP and CIS along the centreline (a, b) and along the body surface (c, d). WP is shown by lines and the reference by symbols.}\label{fig:apollo-kn1-profiles}
\end{figure}
\begin{figure}[tbp]
\centering\singlespacing\setlength{\figW}{\linewidth}
\begin{subfigure}[t]{0.4913\figW}\centering
\raisebox{0.0000\figW}{\includegraphics[width=0.4913\figW]{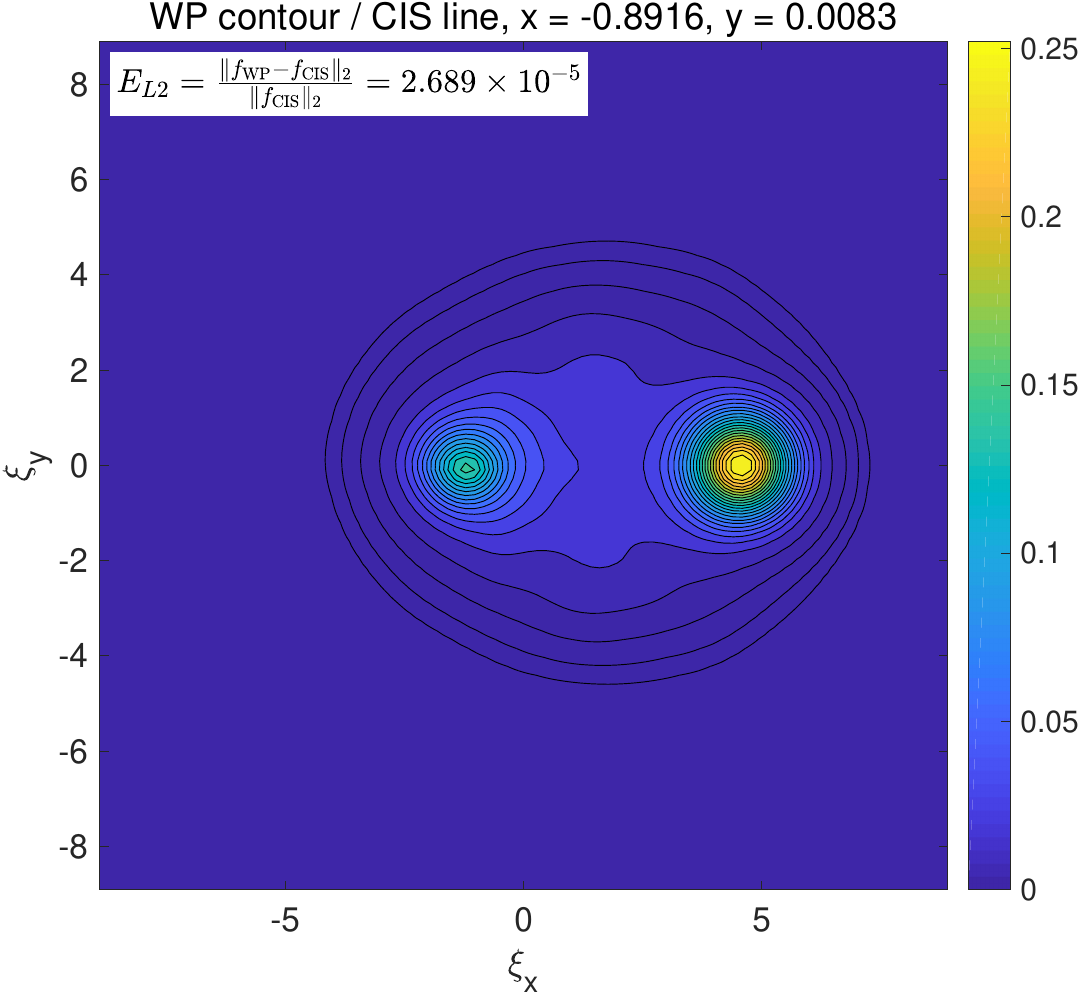}}
\caption{$(x,y)=(-0.89,\,0.01)$}
\end{subfigure}\hspace{0.0200\figW}%
\begin{subfigure}[t]{0.4837\figW}\centering
\raisebox{0.0000\figW}{\includegraphics[width=0.4837\figW]{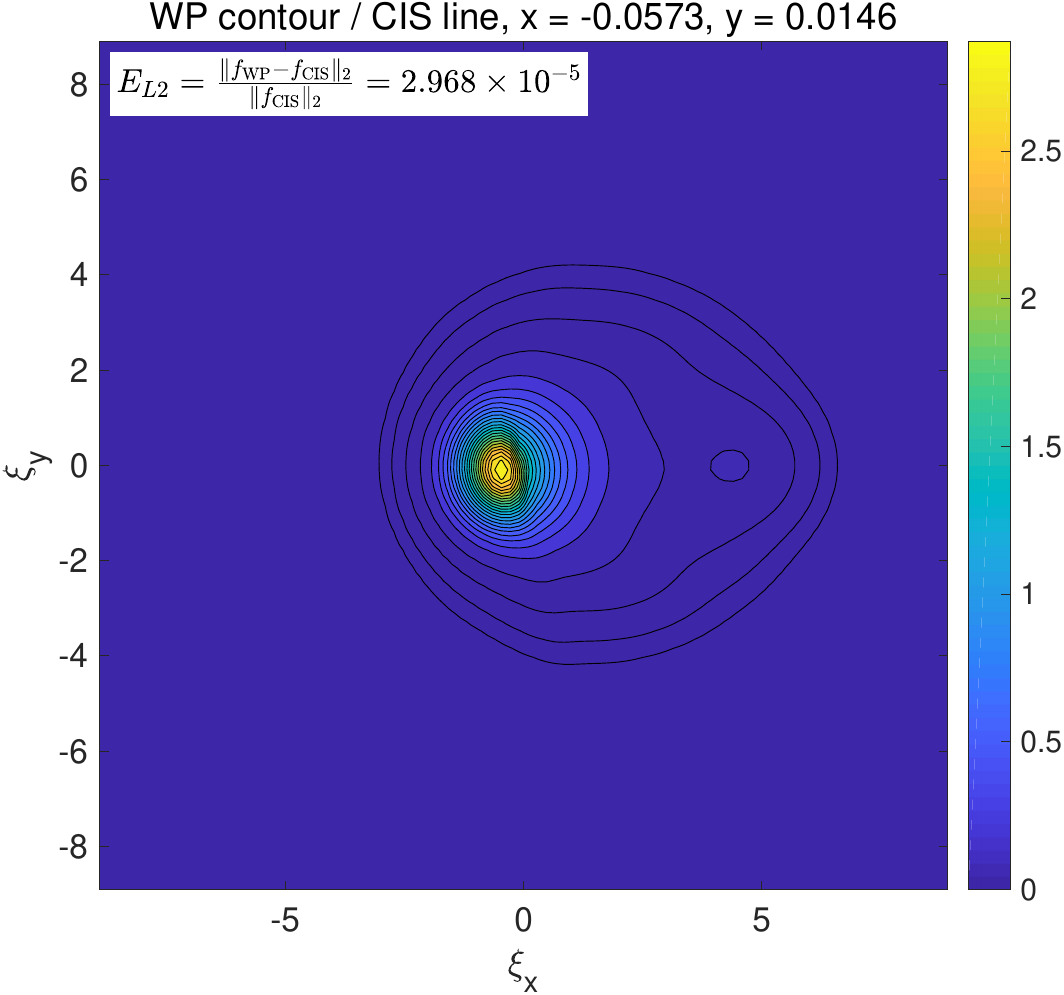}}
\caption{$(x,y)=(-0.06,\,0.01)$}
\end{subfigure}
\par\smallskip
\begin{subfigure}[t]{0.4912\figW}\centering
\raisebox{0.0000\figW}{\includegraphics[width=0.4912\figW]{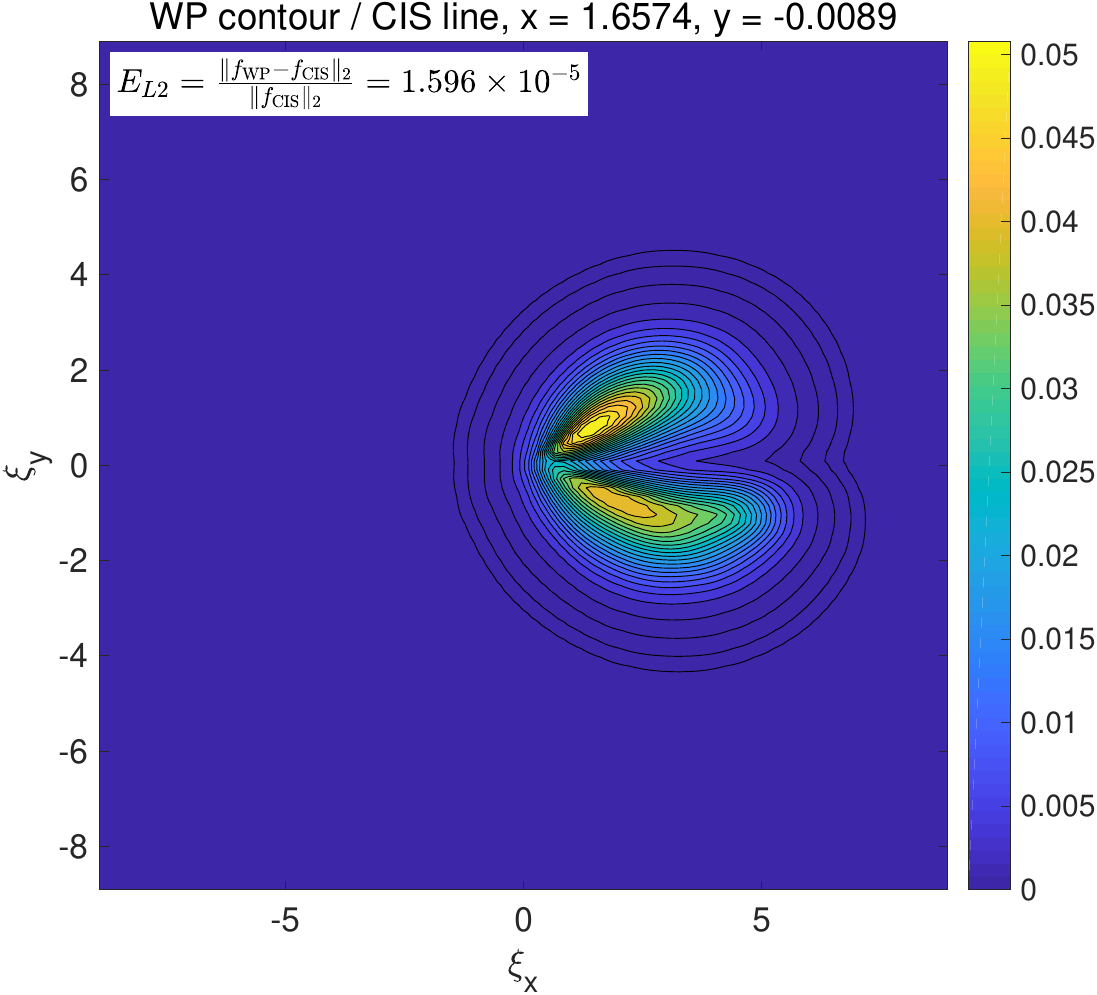}}
\caption{$(x,y)=(1.66,\,-0.01)$}
\end{subfigure}\hspace{0.0200\figW}%
\begin{subfigure}[t]{0.4838\figW}\centering
\raisebox{0.0000\figW}{\includegraphics[width=0.4838\figW]{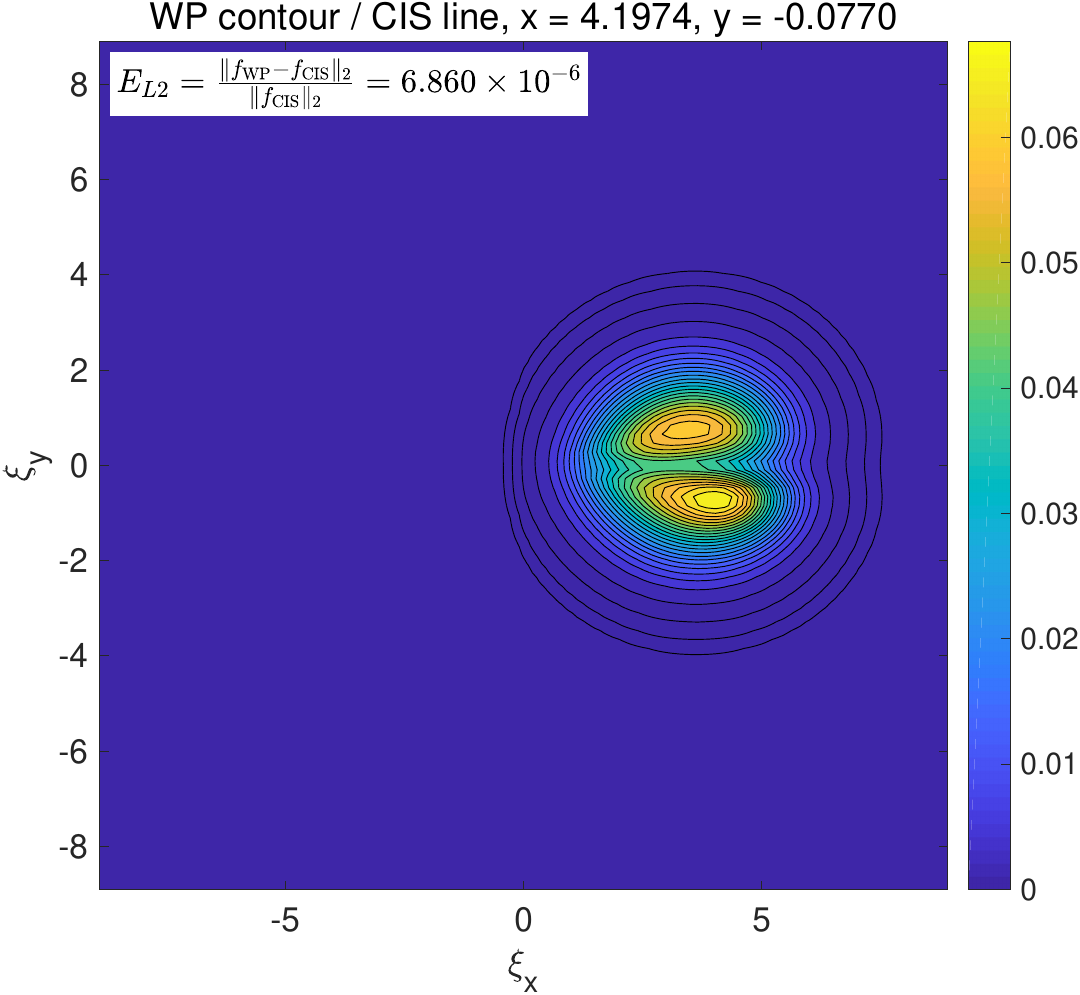}}
\caption{$(x,y)=(4.20,\,-0.08)$}
\end{subfigure}
\caption{Mach-5 flow around the Apollo section, $\varepsilon=1$: reduced molecular distribution in the $(\xi_x,\xi_y)$ velocity plane at the marked sampling points, linear scale. WP is shown by filled contours and CIS by lines at the same levels, and $E_{L_2}$ is printed in each panel.}\label{fig:apollo-kn1-dist}
\end{figure}

At $\varepsilon=10^{-2}$ the field differences in
Fig.~\ref{fig:apollo-kn1em2-fields} are $4.00\times10^{-4}$,
$6.97\times10^{-4}$, $1.81\times10^{-4}$ and $1.96\times10^{-4}$ in density,
temperature, speed and pressure, and the pressure-drag coefficients of WP and
CIS, both $1.51$, differ by $2.49\times10^{-5}$; the centreline and surface
profiles are shown in Fig.~\ref{fig:apollo-kn1em2-profiles}. At
$\varepsilon=10^{-4}$ the WP fields in Fig.~\ref{fig:apollo-kn1em4-fields}
differ from CIS of the Shakhov model by $4.40\times10^{-4}$,
$1.21\times10^{-3}$, $4.02\times10^{-4}$ and $8.47\times10^{-4}$ in the same
four quantities, and the centreline and surface profiles in
Fig.~\ref{fig:apollo-kn1em4-profiles} are compared with the NS solution.

The residual histories of the three calculations are collected in
Fig.~\ref{fig:apollo-residuals}. At $\varepsilon=1$ the speedups are
$S_{\rm iter}=1.11$ and $S_t=0.90$, so that the two methods again cost about
the same, as for the cylinder, and at $\varepsilon=10^{-2}$
they grow to $27.46$ and $16.10$. At $\varepsilon=10^{-4}$ the ratios with
respect to CIS of the Shakhov model, computed with the solver
of~\cite{liu_wpd_shakhov}, are $S_{\rm iter}=36488.23$ and $S_t=1470.93$, and
WP reaches $R_U=10^{-7}$ in $35.7$~s, $4.02$ times the $8.89$~s of the NS
solver.

\begin{figure}[tbp]
\centering\singlespacing\setlength{\figW}{\linewidth}
\begin{subfigure}[t]{0.4857\figW}\centering
\raisebox{0.0000\figW}{\includegraphics[width=0.4857\figW]{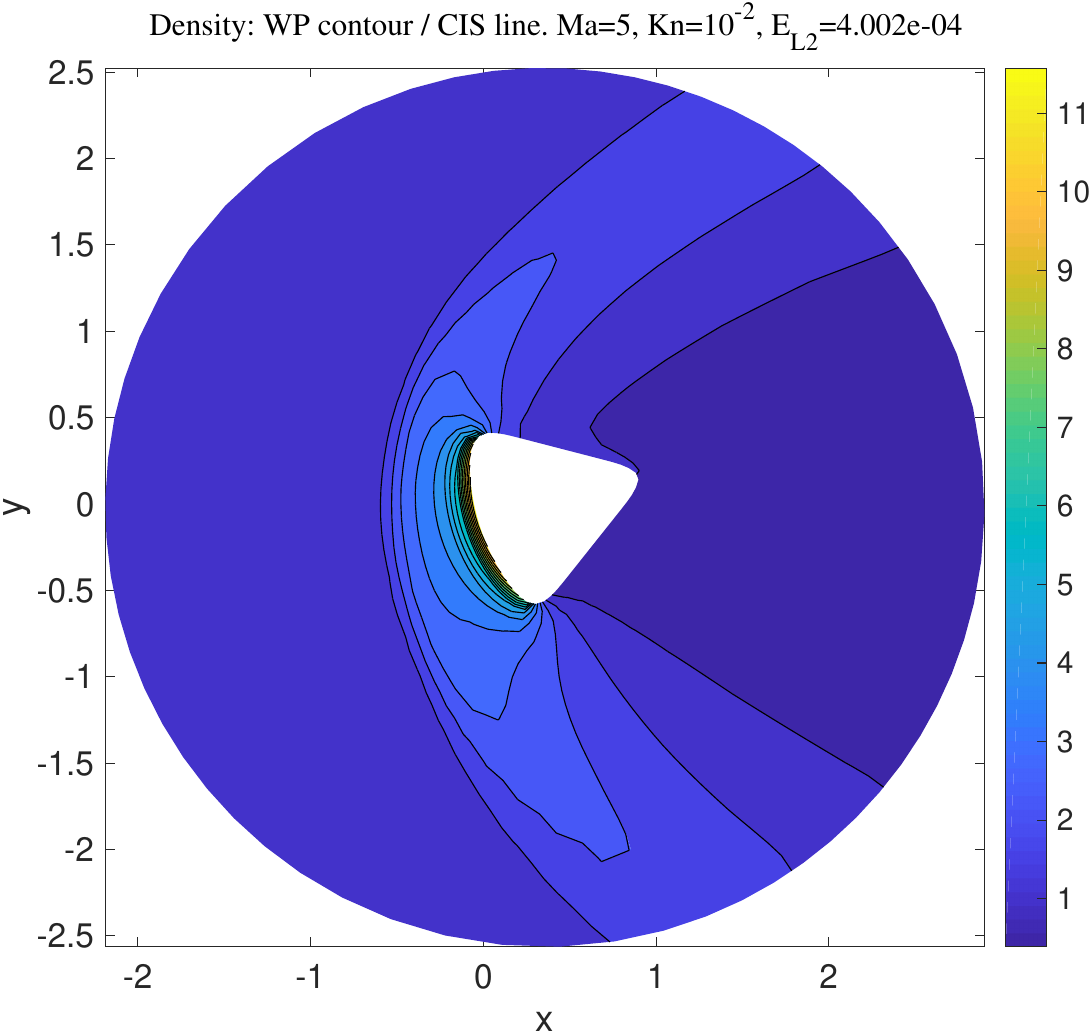}}
\caption{Density}
\end{subfigure}\hspace{0.0200\figW}%
\begin{subfigure}[t]{0.4893\figW}\centering
\raisebox{0.0000\figW}{\includegraphics[width=0.4893\figW]{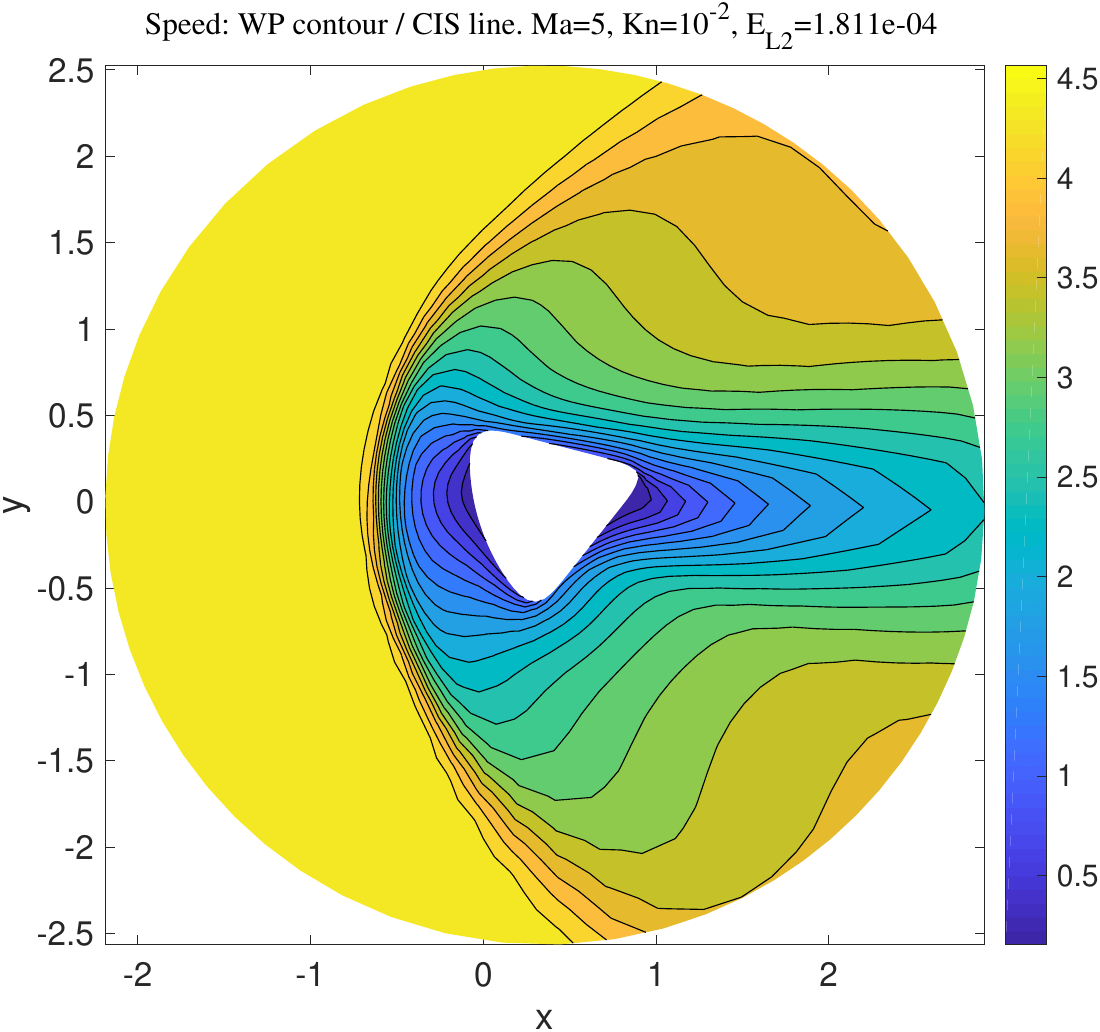}}
\caption{Speed}
\end{subfigure}
\par\smallskip
\begin{subfigure}[t]{0.4893\figW}\centering
\raisebox{0.0000\figW}{\includegraphics[width=0.4893\figW]{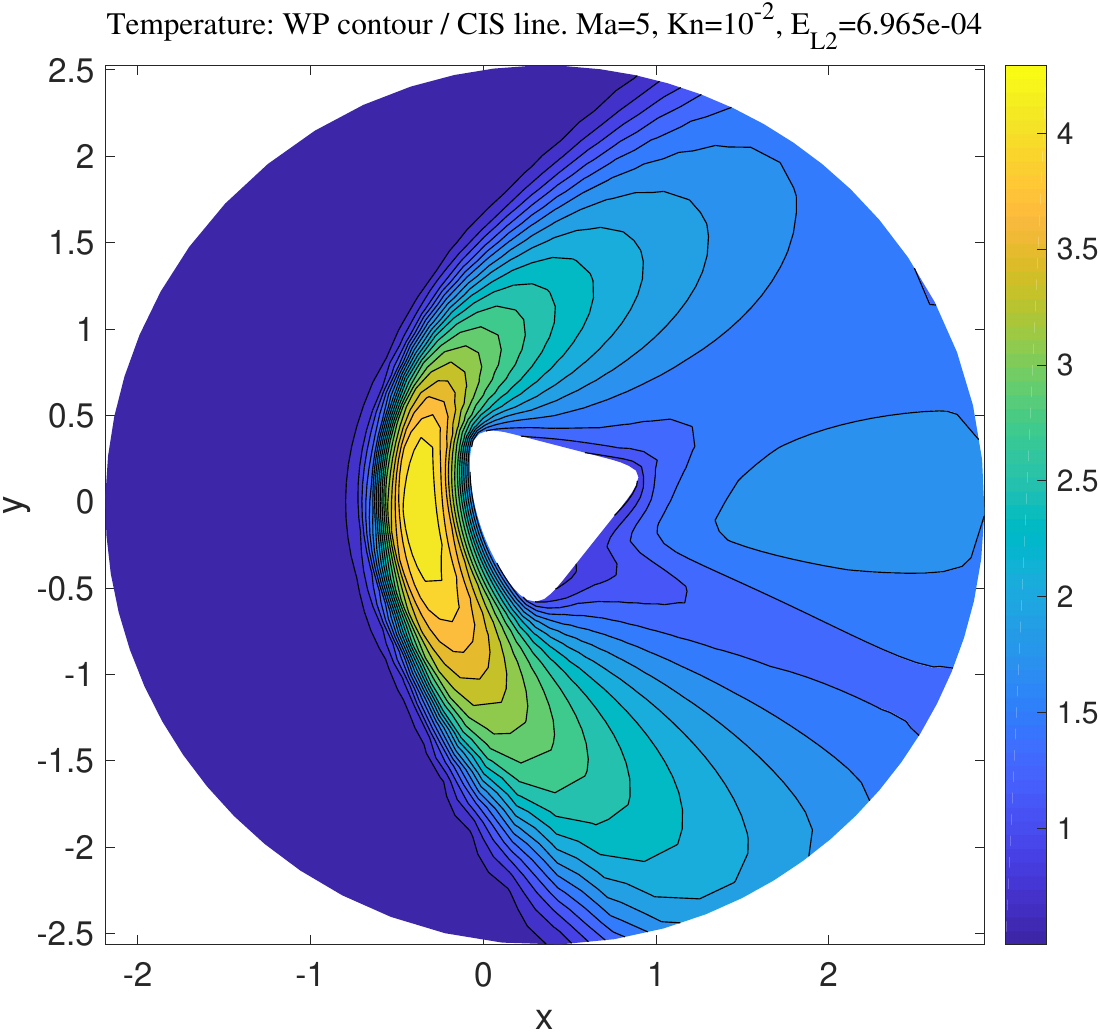}}
\caption{Temperature}
\end{subfigure}\hspace{0.0200\figW}%
\begin{subfigure}[t]{0.4857\figW}\centering
\raisebox{0.0000\figW}{\includegraphics[width=0.4857\figW]{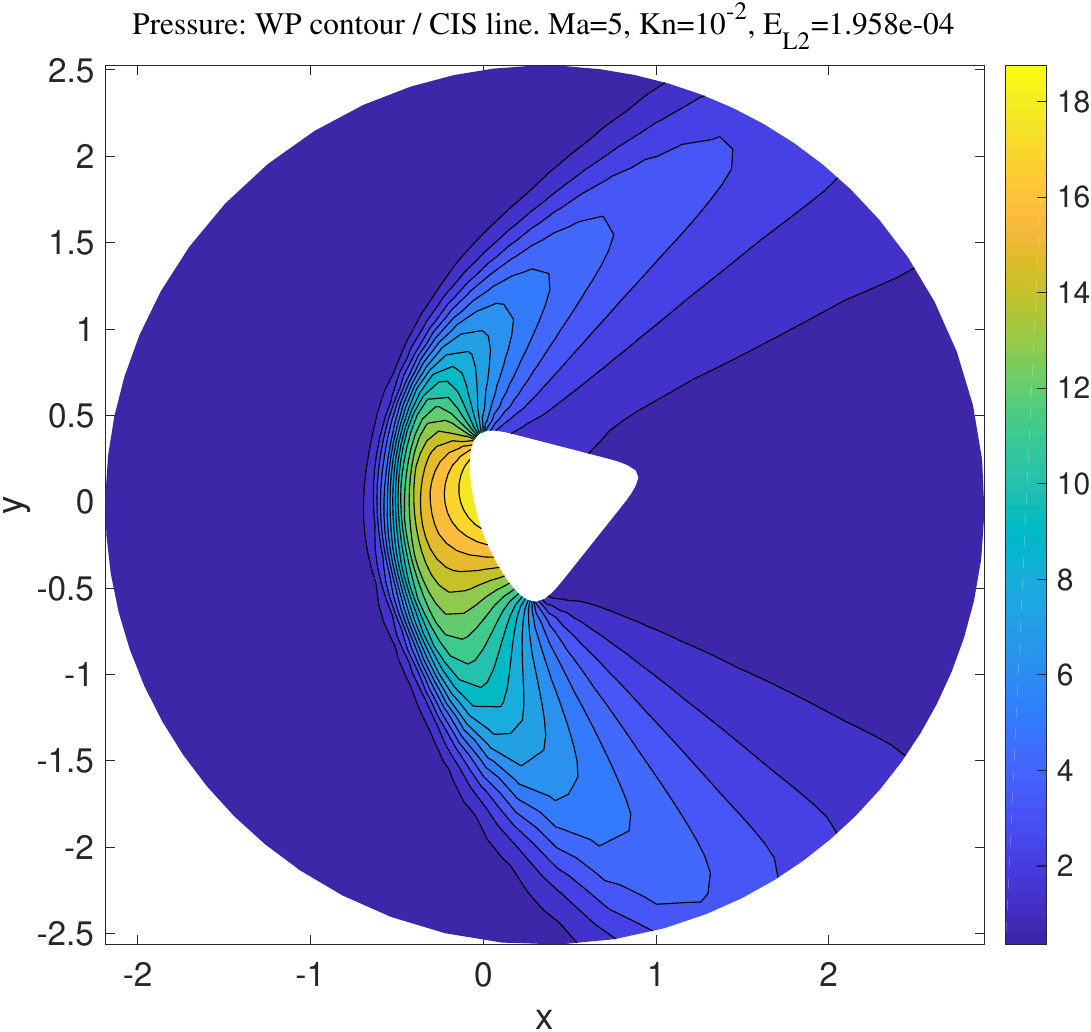}}
\caption{Pressure}
\end{subfigure}
\caption{Mach-5 flow around the Apollo section, $\varepsilon=10^{-2}$: density, speed, temperature and pressure fields of WP compared with CIS. WP is shown by filled contours and the reference by lines at the same levels, and $E_{L_2}$ is printed above each panel.}\label{fig:apollo-kn1em2-fields}
\end{figure}
\begin{figure}[tbp]
\centering\singlespacing\setlength{\figW}{\linewidth}
\begin{subfigure}[t]{0.4952\figW}\centering
\raisebox{0.0000\figW}{\includegraphics[width=0.4952\figW]{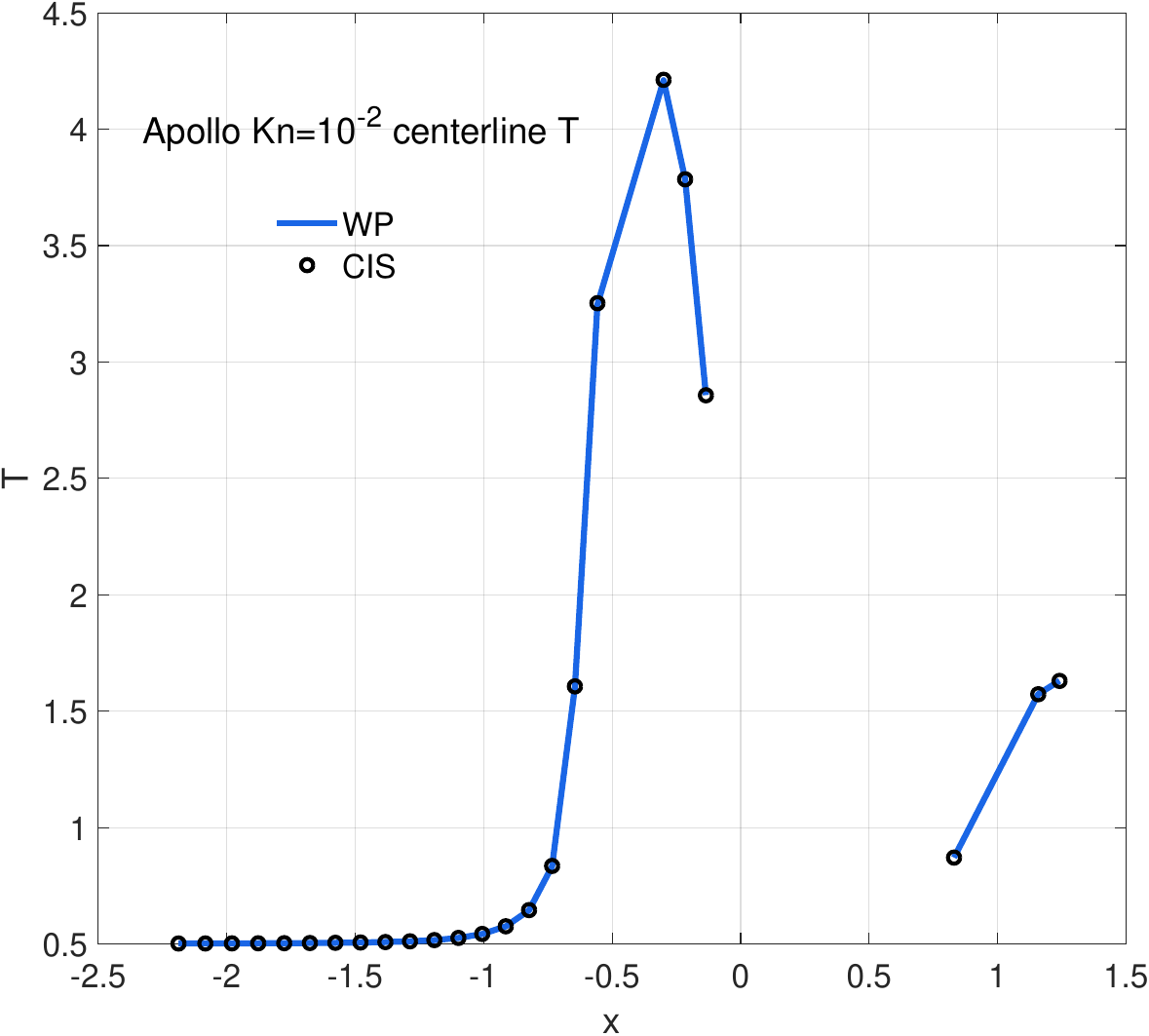}}
\caption{Centreline temperature}
\end{subfigure}\hspace{0.0200\figW}%
\begin{subfigure}[t]{0.4798\figW}\centering
\raisebox{0.0000\figW}{\includegraphics[width=0.4798\figW]{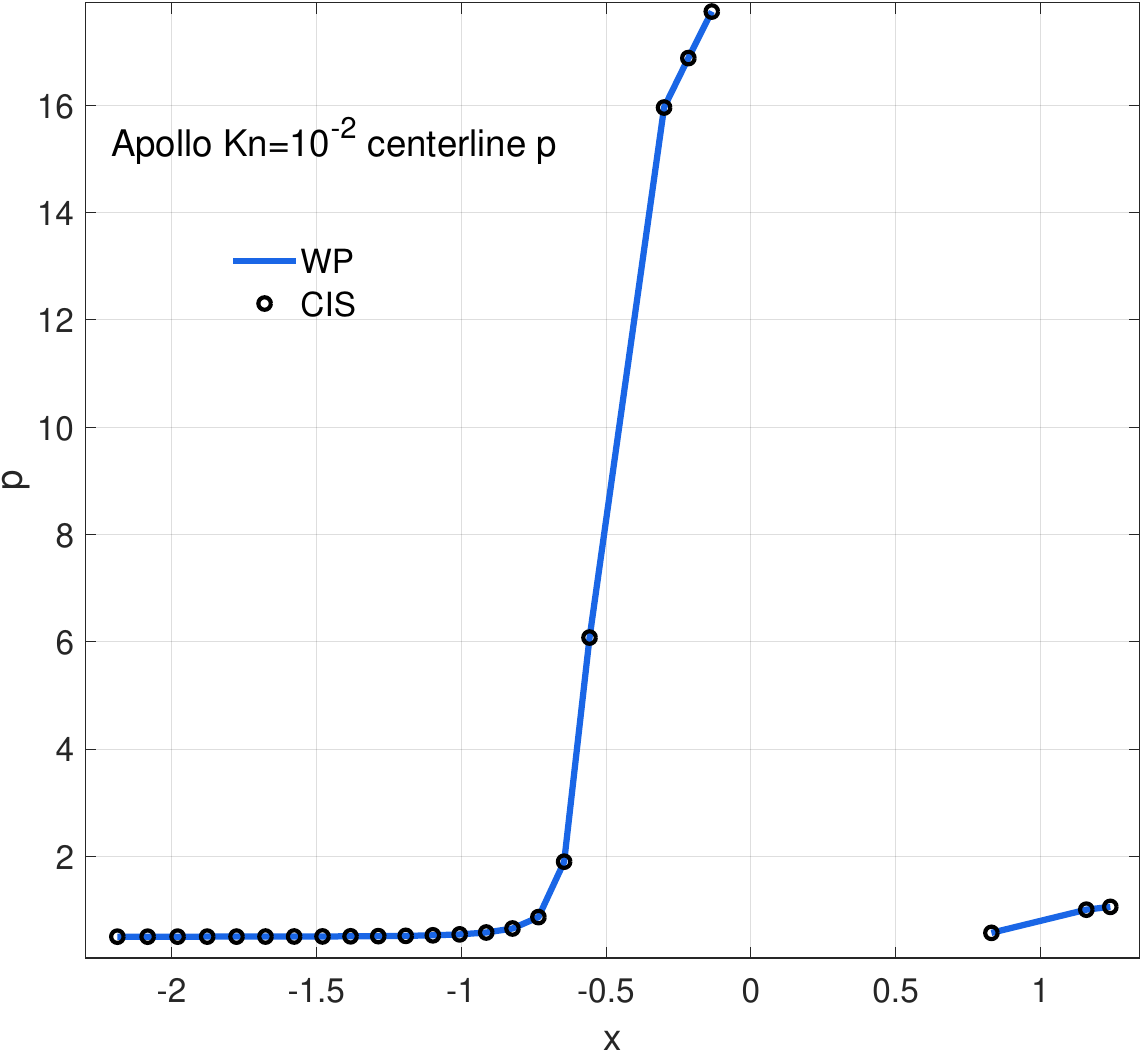}}
\caption{Centreline pressure}
\end{subfigure}
\par\smallskip
\begin{subfigure}[t]{0.4908\figW}\centering
\raisebox{0.0000\figW}{\includegraphics[width=0.4908\figW]{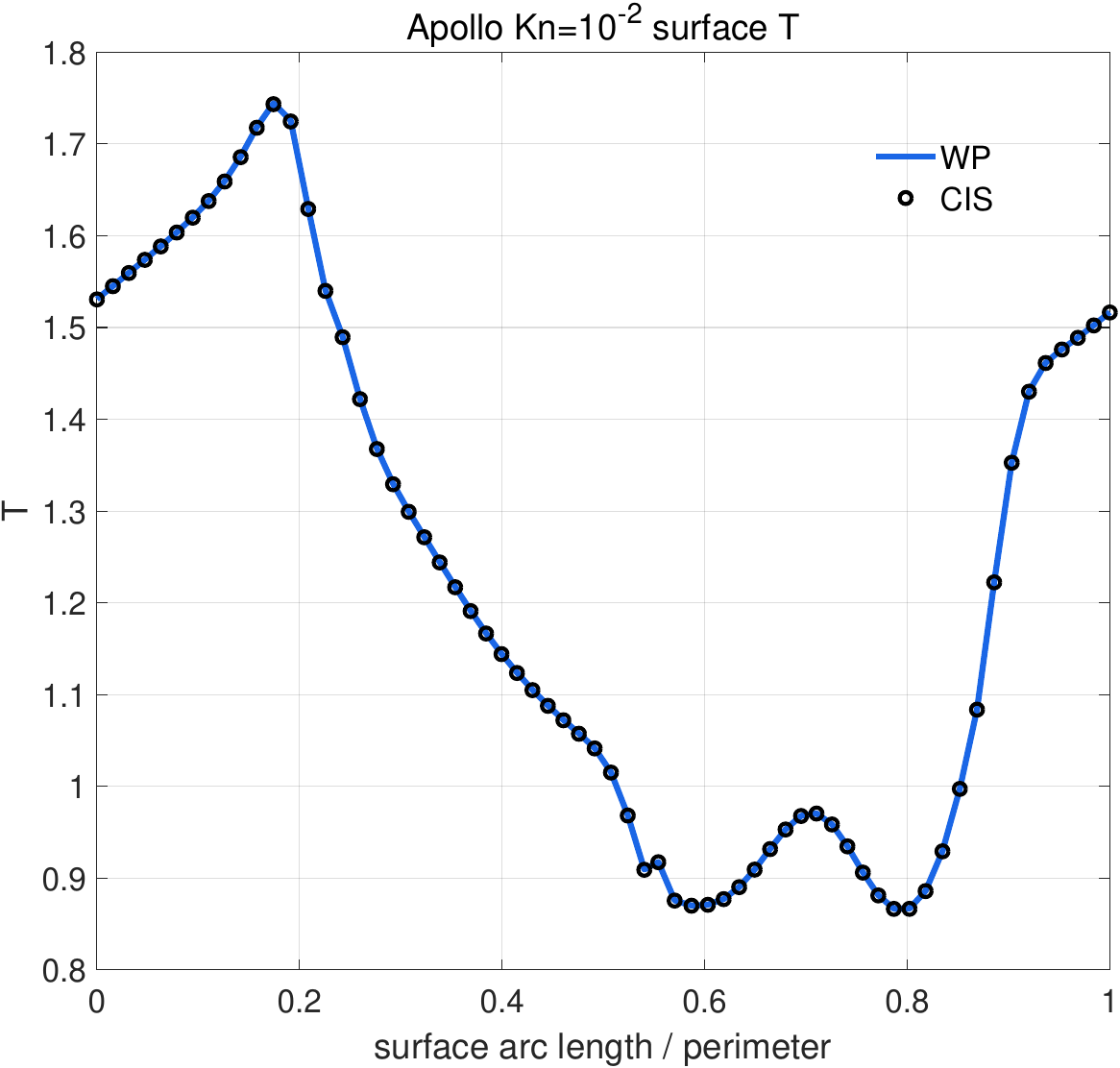}}
\caption{Surface temperature}
\end{subfigure}\hspace{0.0200\figW}%
\begin{subfigure}[t]{0.4842\figW}\centering
\raisebox{0.0000\figW}{\includegraphics[width=0.4842\figW]{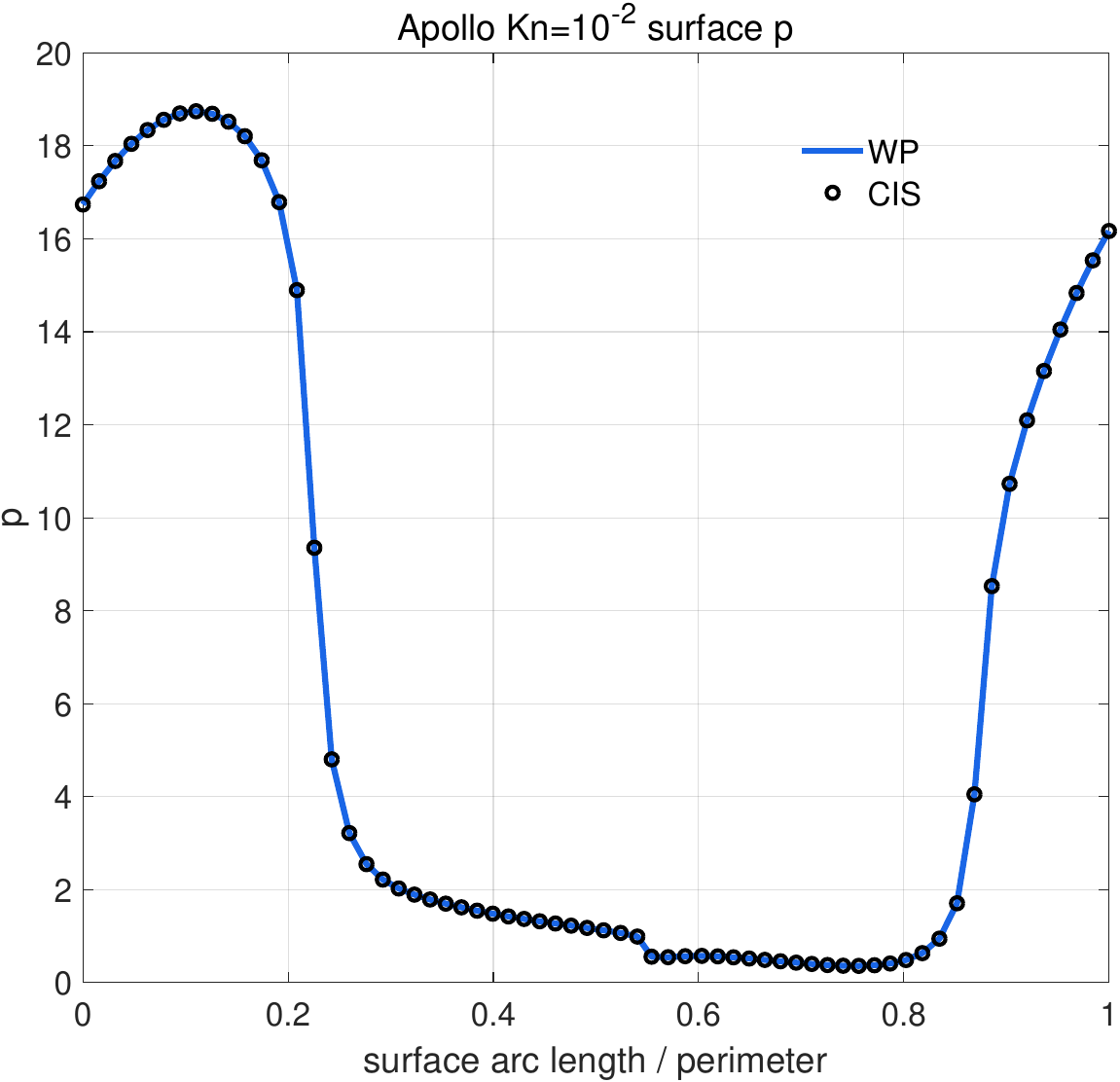}}
\caption{Surface pressure}
\end{subfigure}
\caption{Mach-5 flow around the Apollo section, $\varepsilon=10^{-2}$: temperature and pressure of WP and CIS along the centreline (a, b) and along the body surface (c, d). WP is shown by lines and the reference by symbols.}\label{fig:apollo-kn1em2-profiles}
\end{figure}
\begin{figure}[tbp]
\centering\singlespacing\setlength{\figW}{\linewidth}
\begin{subfigure}[t]{0.4857\figW}\centering
\raisebox{0.0000\figW}{\includegraphics[width=0.4857\figW]{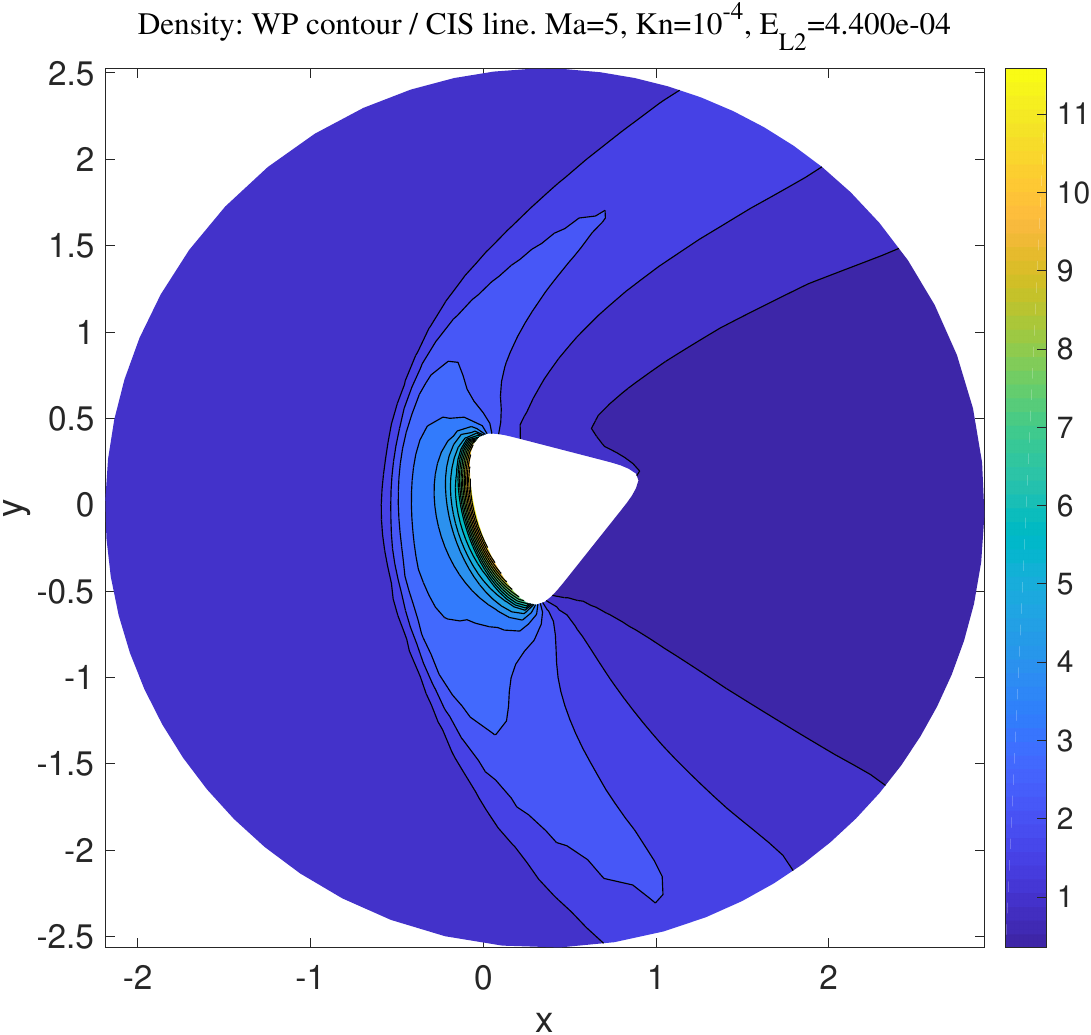}}
\caption{Density}
\end{subfigure}\hspace{0.0200\figW}%
\begin{subfigure}[t]{0.4893\figW}\centering
\raisebox{0.0000\figW}{\includegraphics[width=0.4893\figW]{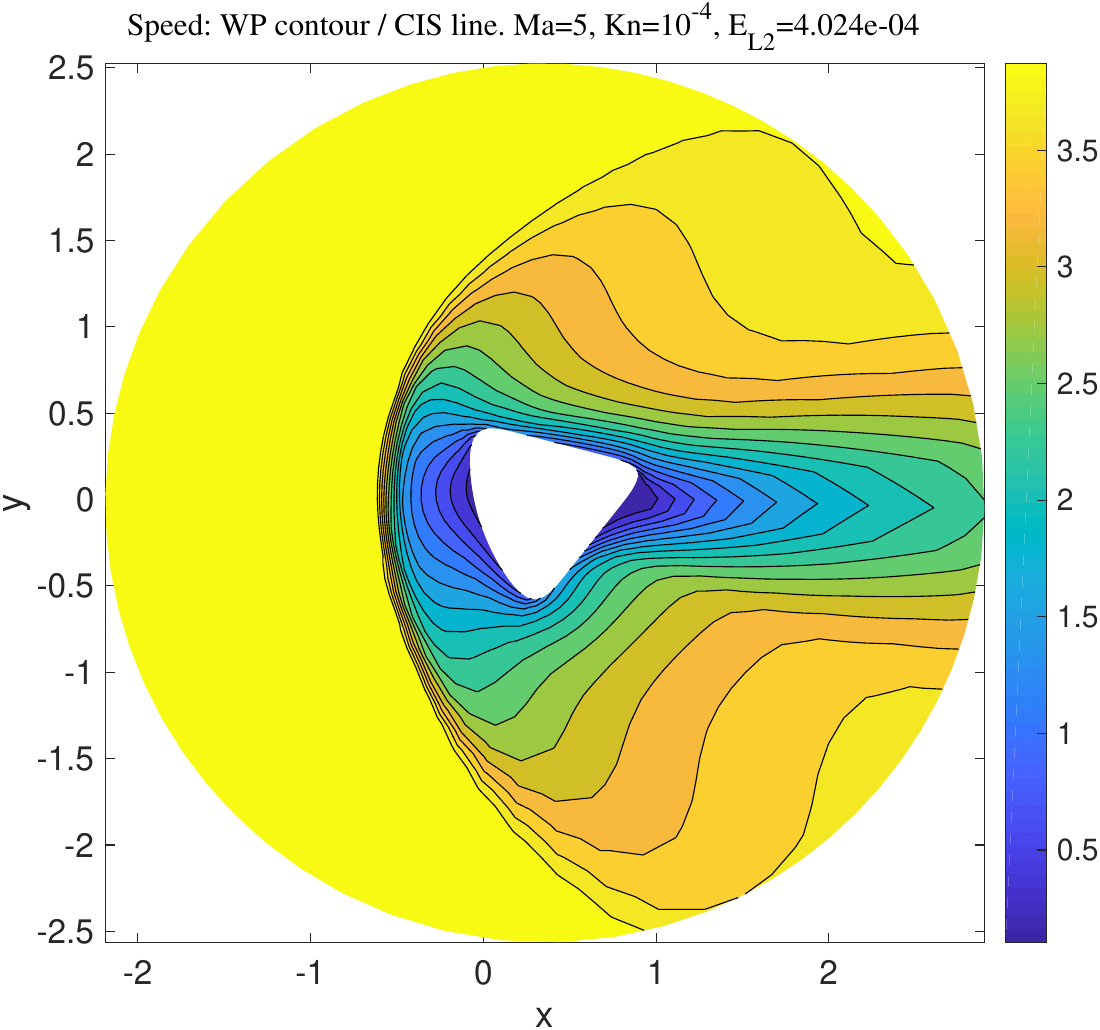}}
\caption{Speed}
\end{subfigure}
\par\smallskip
\begin{subfigure}[t]{0.4893\figW}\centering
\raisebox{0.0000\figW}{\includegraphics[width=0.4893\figW]{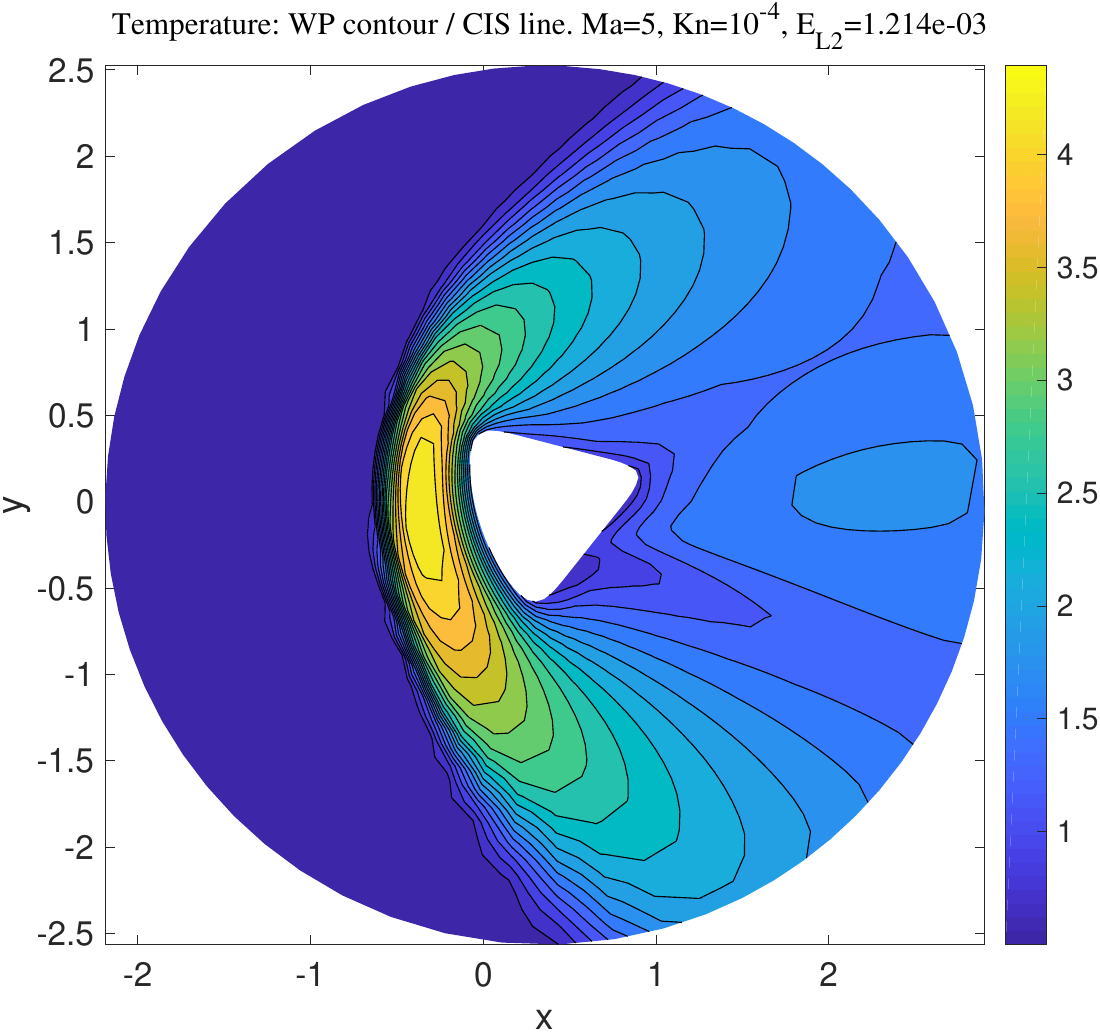}}
\caption{Temperature}
\end{subfigure}\hspace{0.0200\figW}%
\begin{subfigure}[t]{0.4857\figW}\centering
\raisebox{0.0000\figW}{\includegraphics[width=0.4857\figW]{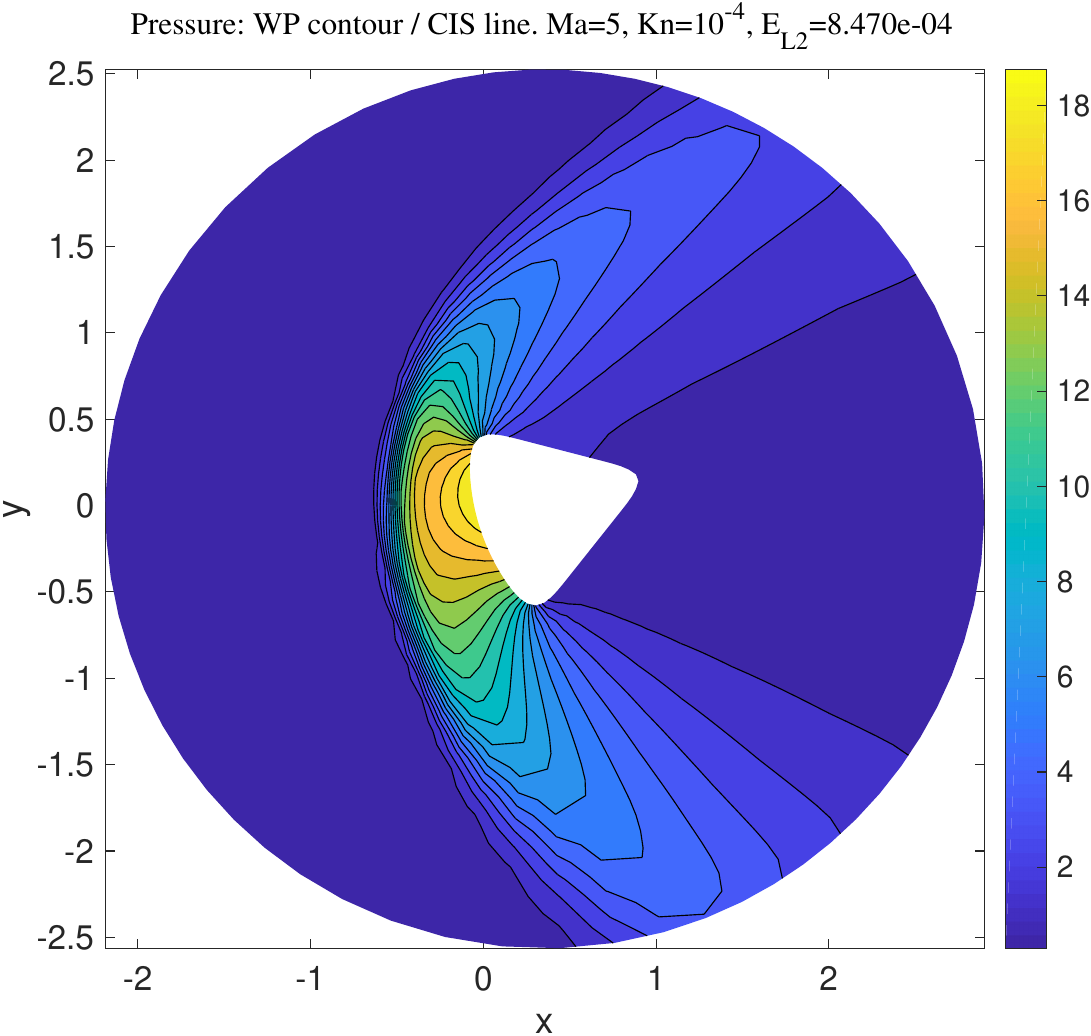}}
\caption{Pressure}
\end{subfigure}
\caption{Mach-5 flow around the Apollo section, $\varepsilon=10^{-4}$: density, speed, temperature and pressure fields of WP compared with CIS of the Shakhov model. WP is shown by filled contours and the reference by lines at the same levels, and $E_{L_2}$ is printed above each panel.}\label{fig:apollo-kn1em4-fields}
\end{figure}
\begin{figure}[tbp]
\centering\singlespacing\setlength{\figW}{\linewidth}
\begin{subfigure}[t]{0.4908\figW}\centering
\raisebox{0.0000\figW}{\includegraphics[width=0.4908\figW]{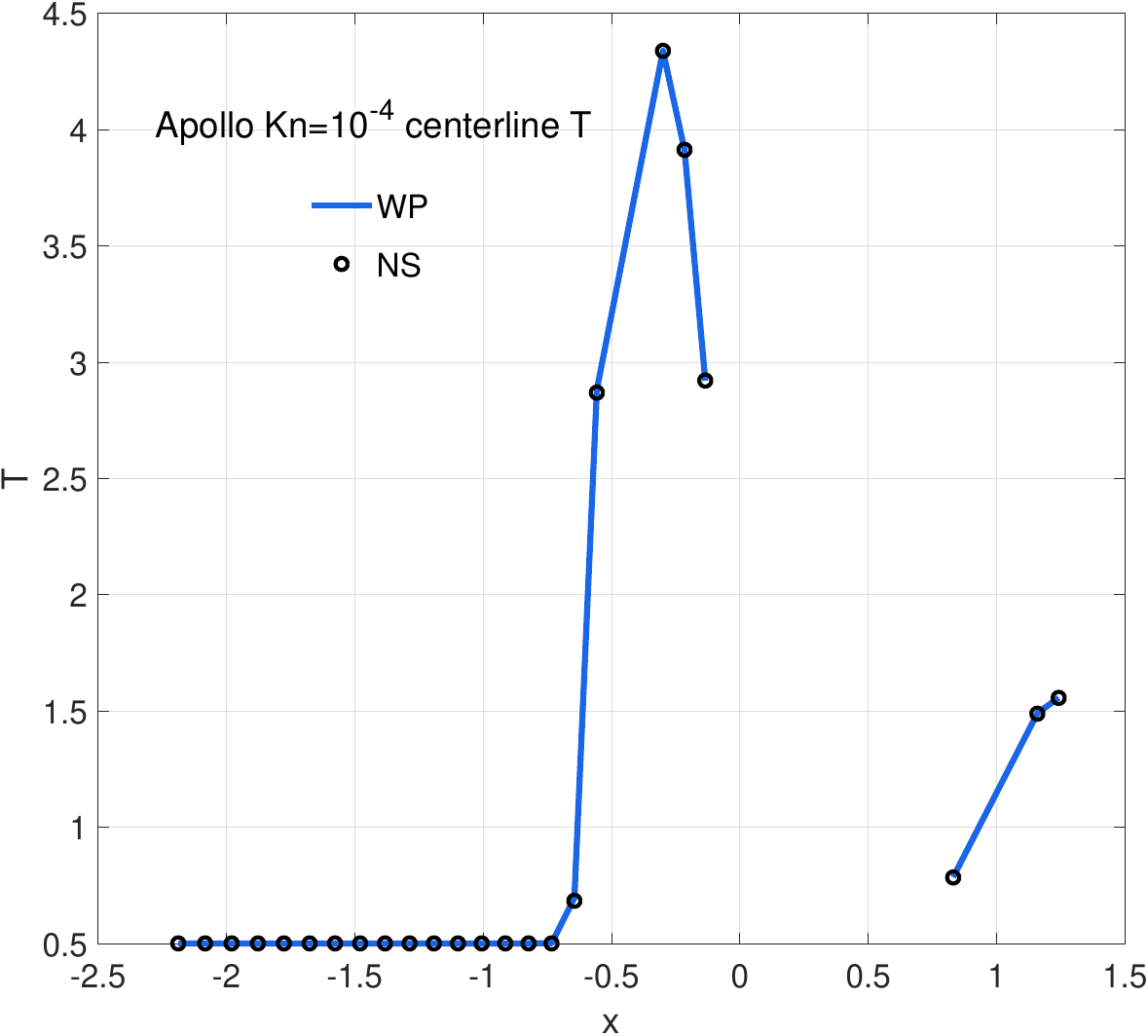}}
\caption{Centreline temperature}
\end{subfigure}\hspace{0.0200\figW}%
\begin{subfigure}[t]{0.4842\figW}\centering
\raisebox{0.0000\figW}{\includegraphics[width=0.4842\figW]{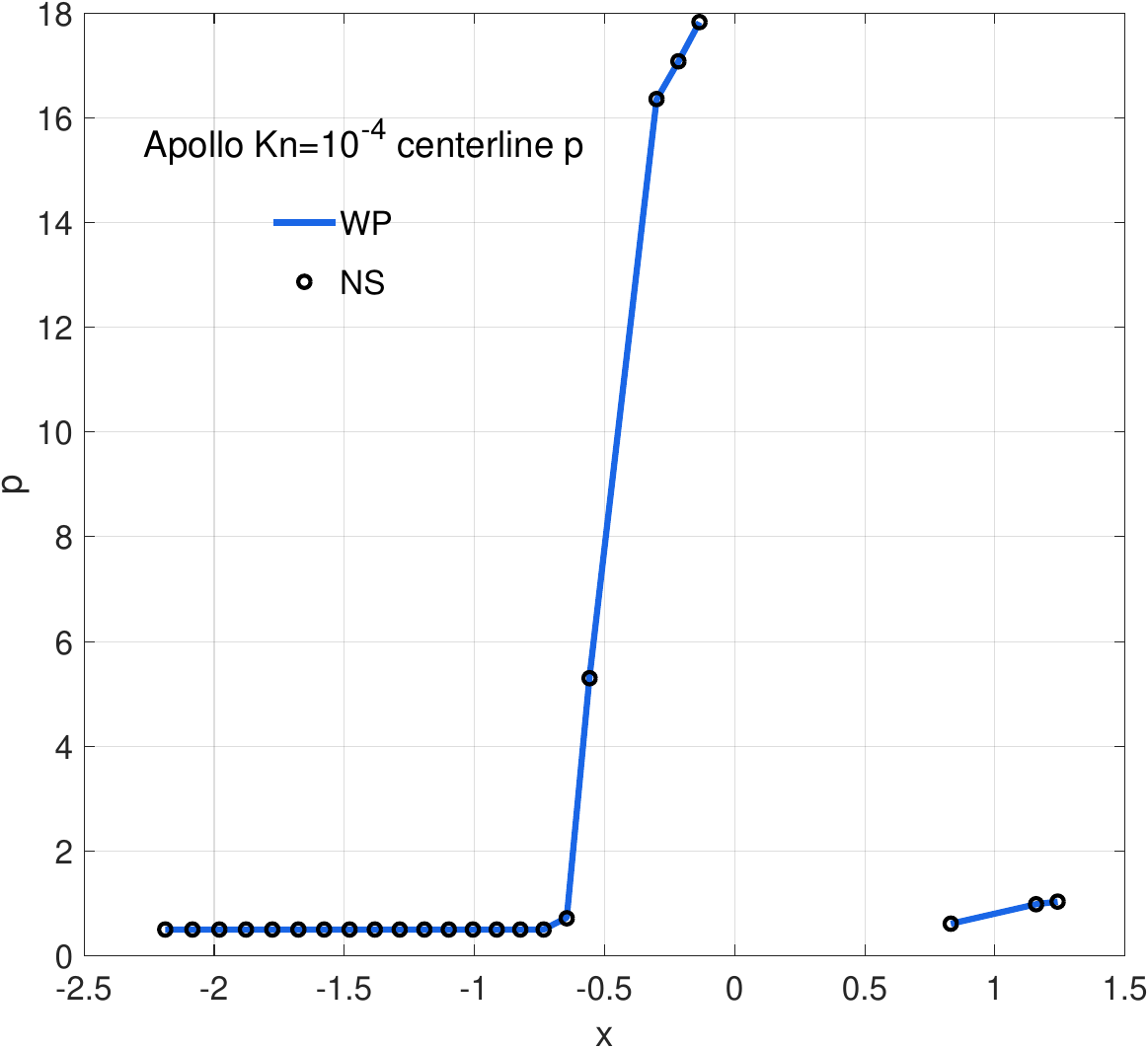}}
\caption{Centreline pressure}
\end{subfigure}
\par\smallskip
\begin{subfigure}[t]{0.4908\figW}\centering
\raisebox{0.0000\figW}{\includegraphics[width=0.4908\figW]{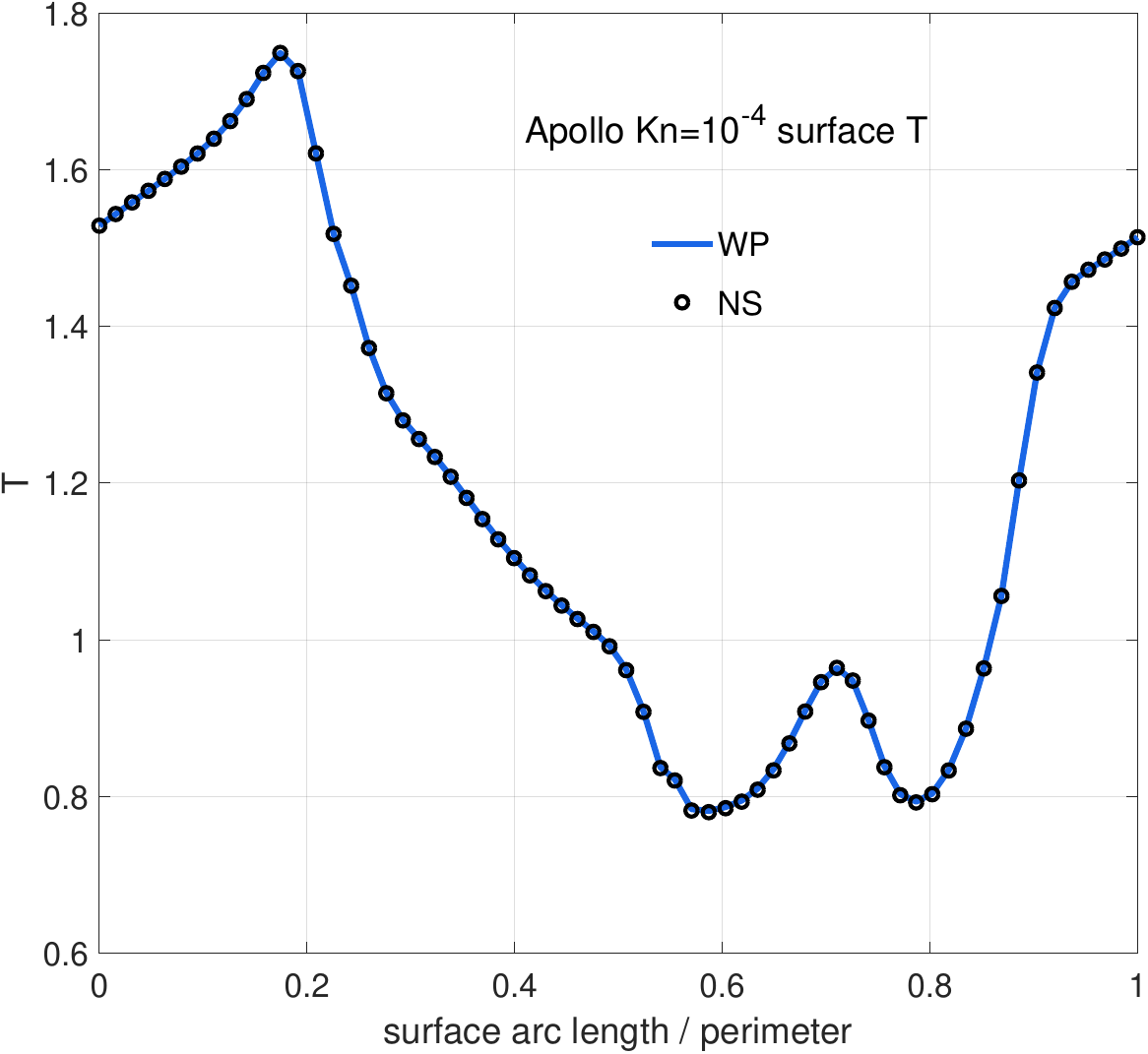}}
\caption{Surface temperature}
\end{subfigure}\hspace{0.0200\figW}%
\begin{subfigure}[t]{0.4842\figW}\centering
\raisebox{0.0000\figW}{\includegraphics[width=0.4842\figW]{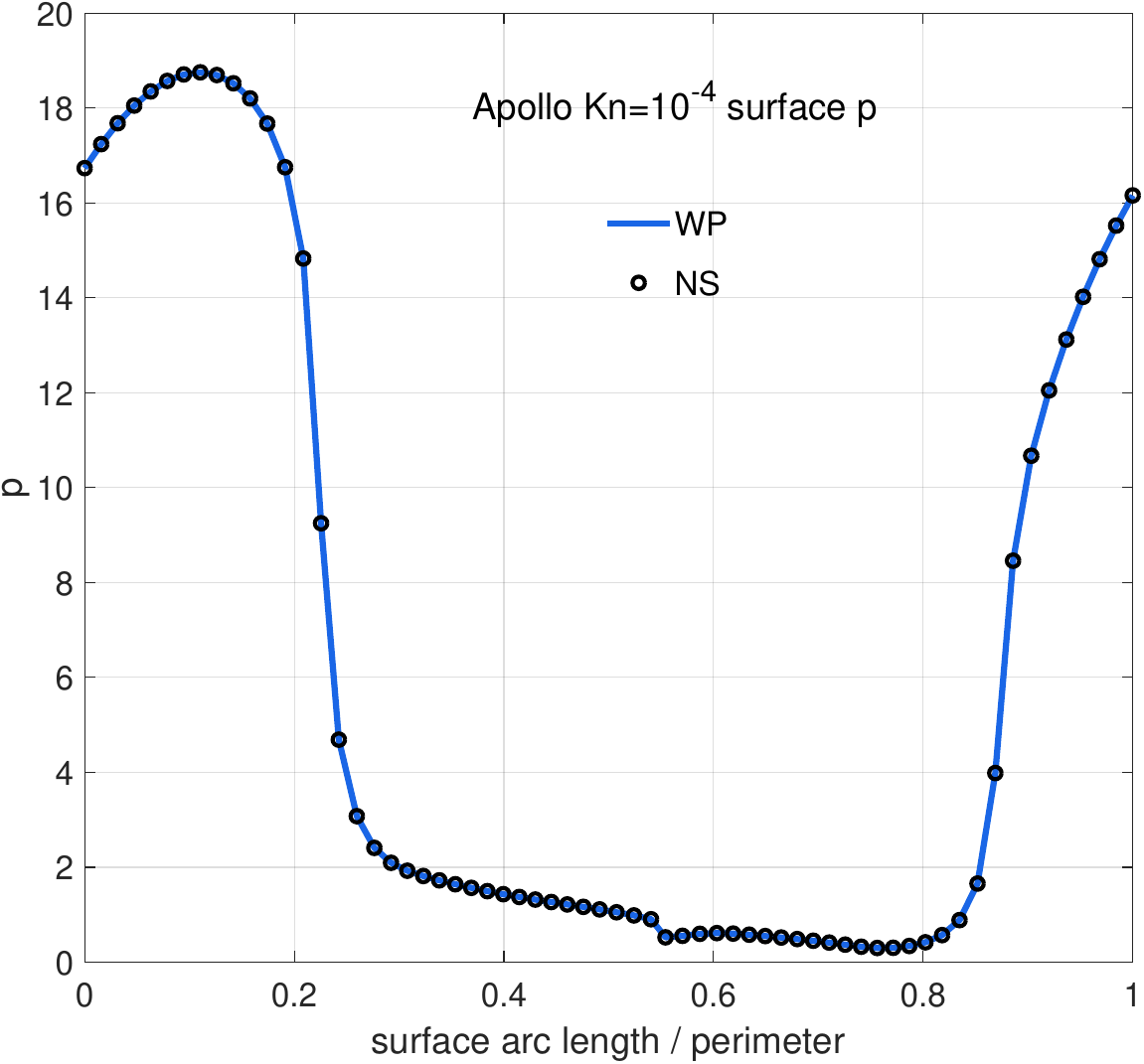}}
\caption{Surface pressure}
\end{subfigure}
\caption{Mach-5 flow around the Apollo section, $\varepsilon=10^{-4}$: temperature and pressure of WP and the NS solution along the centreline (a, b) and along the body surface (c, d). WP is shown by lines and the reference by symbols.}\label{fig:apollo-kn1em4-profiles}
\end{figure}
\begin{figure}[tbp]
\centering\singlespacing\setlength{\figW}{\linewidth}
\begin{subfigure}[t]{0.9820\figW}\centering
\makebox[\linewidth]{\raisebox{0.0000\figW}{\includegraphics[width=0.4962\figW]{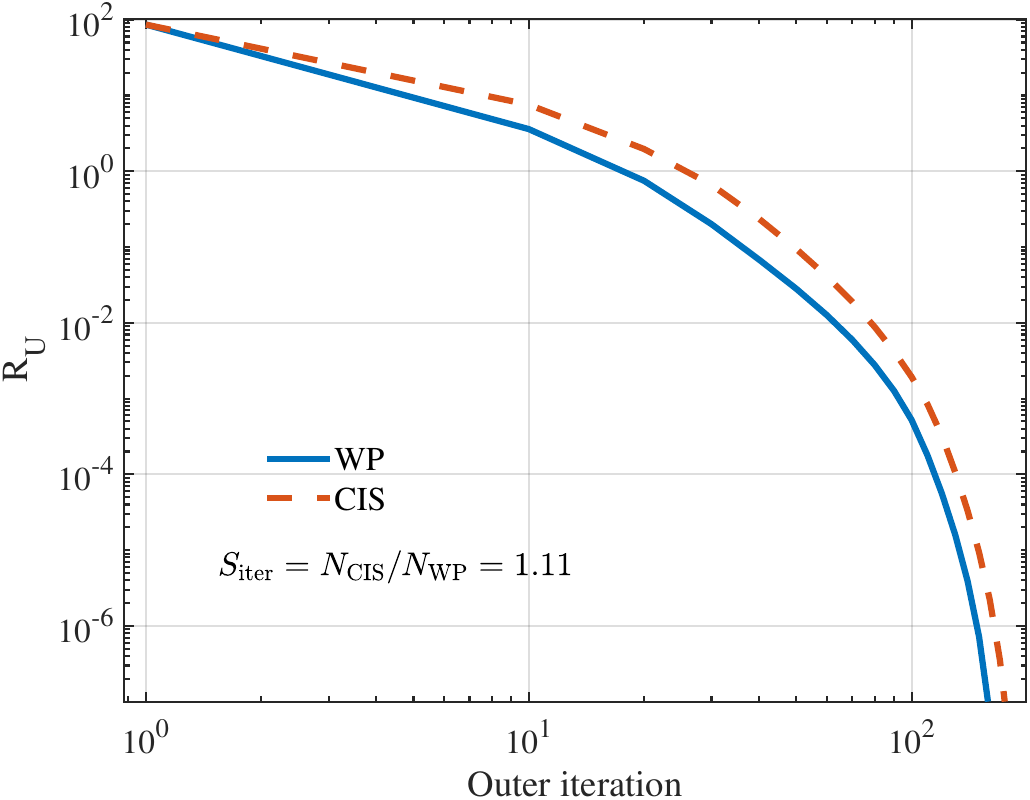}}\hspace{0.0200\figW}\raisebox{0.0026\figW}{\includegraphics[width=0.4658\figW]{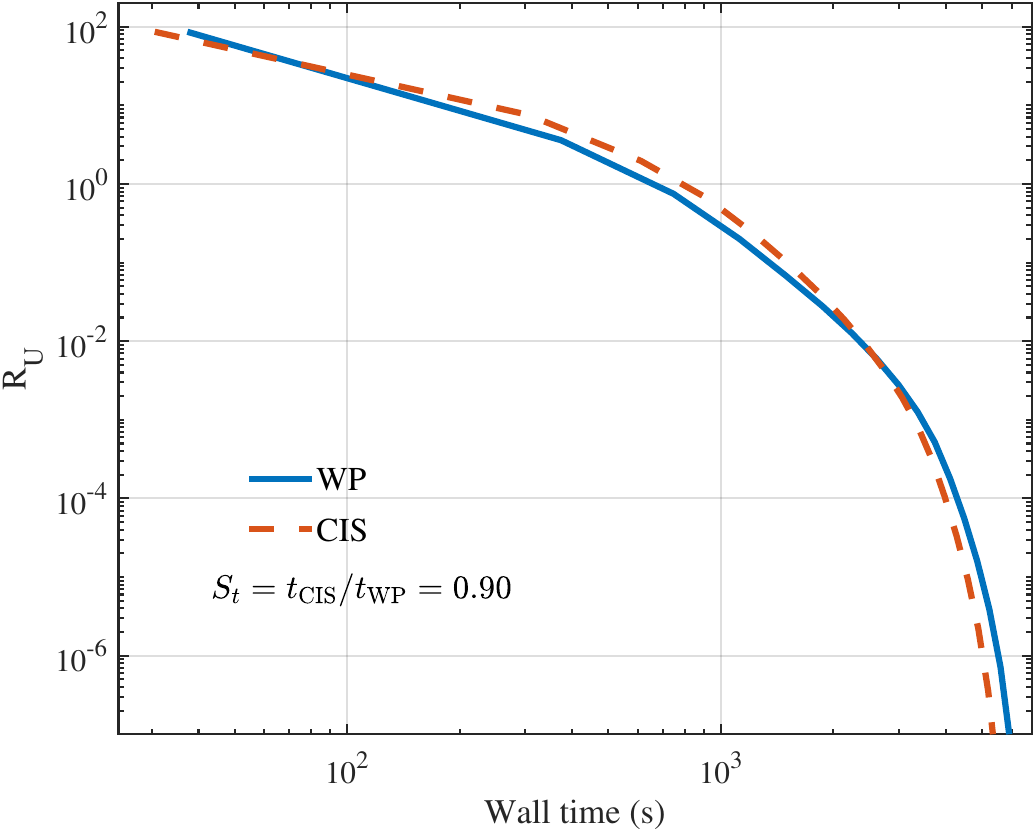}}}
\caption{$\varepsilon=1$}\label{fig:apollo-res-kn1}
\end{subfigure}
\par\smallskip
\begin{subfigure}[t]{0.9820\figW}\centering
\makebox[\linewidth]{\raisebox{0.0000\figW}{\includegraphics[width=0.4944\figW]{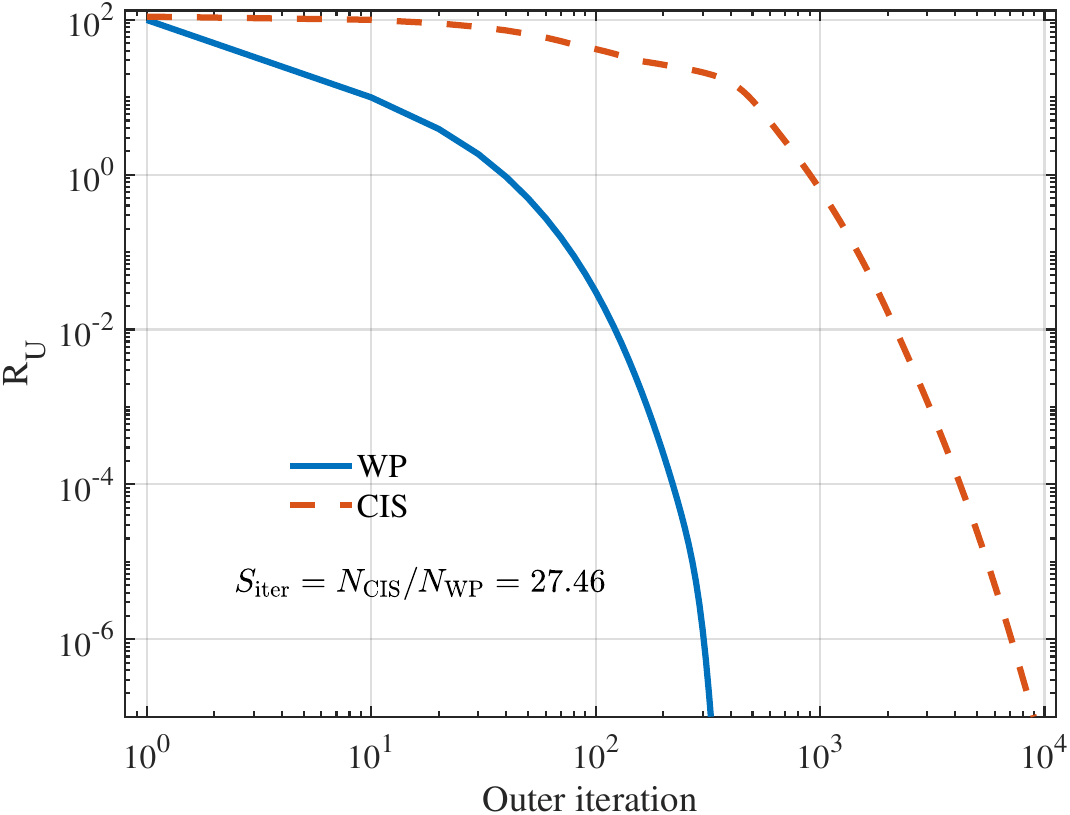}}\hspace{0.0200\figW}\raisebox{0.0005\figW}{\includegraphics[width=0.4676\figW]{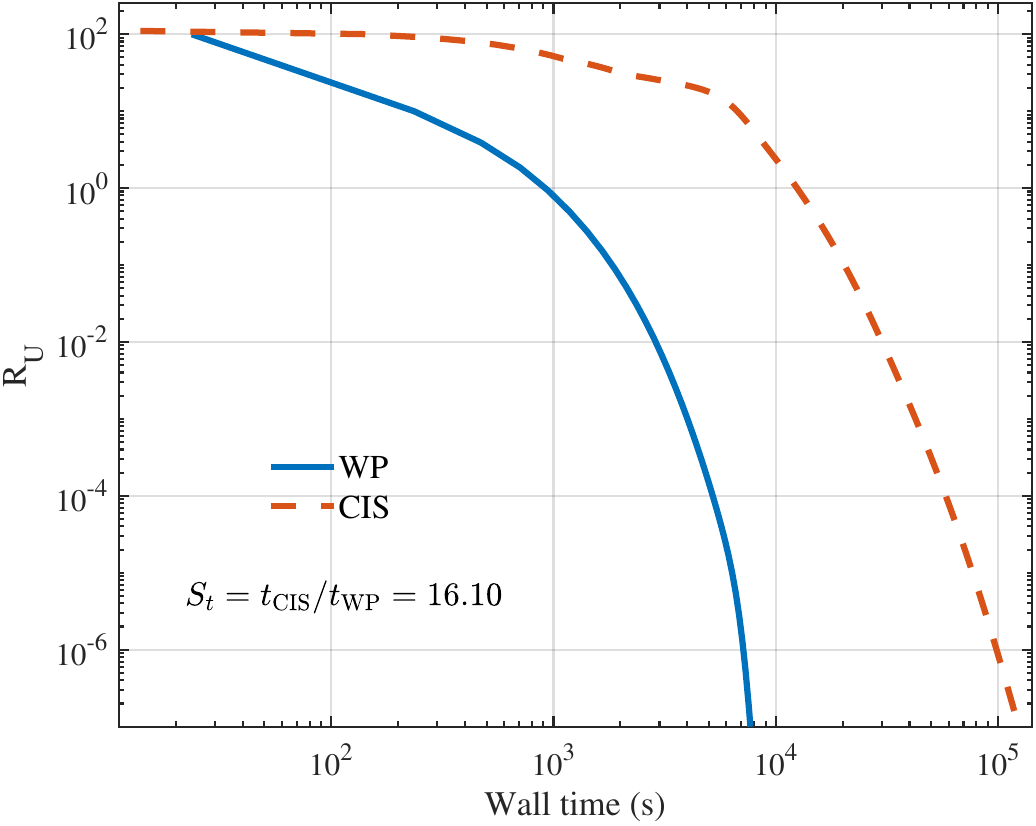}}}
\caption{$\varepsilon=10^{-2}$}\label{fig:apollo-res-kn1em2}
\end{subfigure}
\par\smallskip
\begin{subfigure}[t]{0.9820\figW}\centering
\makebox[\linewidth]{\raisebox{0.0036\figW}{\includegraphics[width=0.4721\figW]{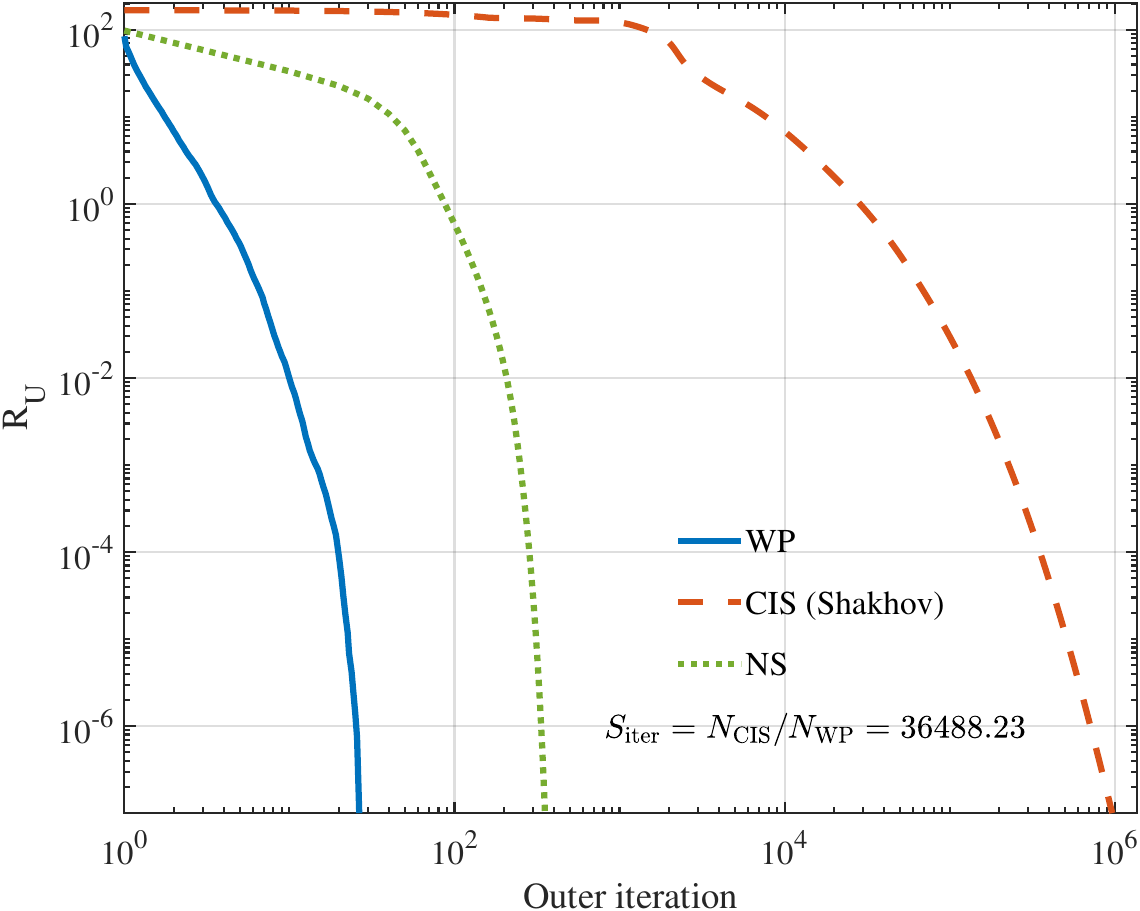}}\hspace{0.0200\figW}\raisebox{0.0000\figW}{\includegraphics[width=0.4900\figW]{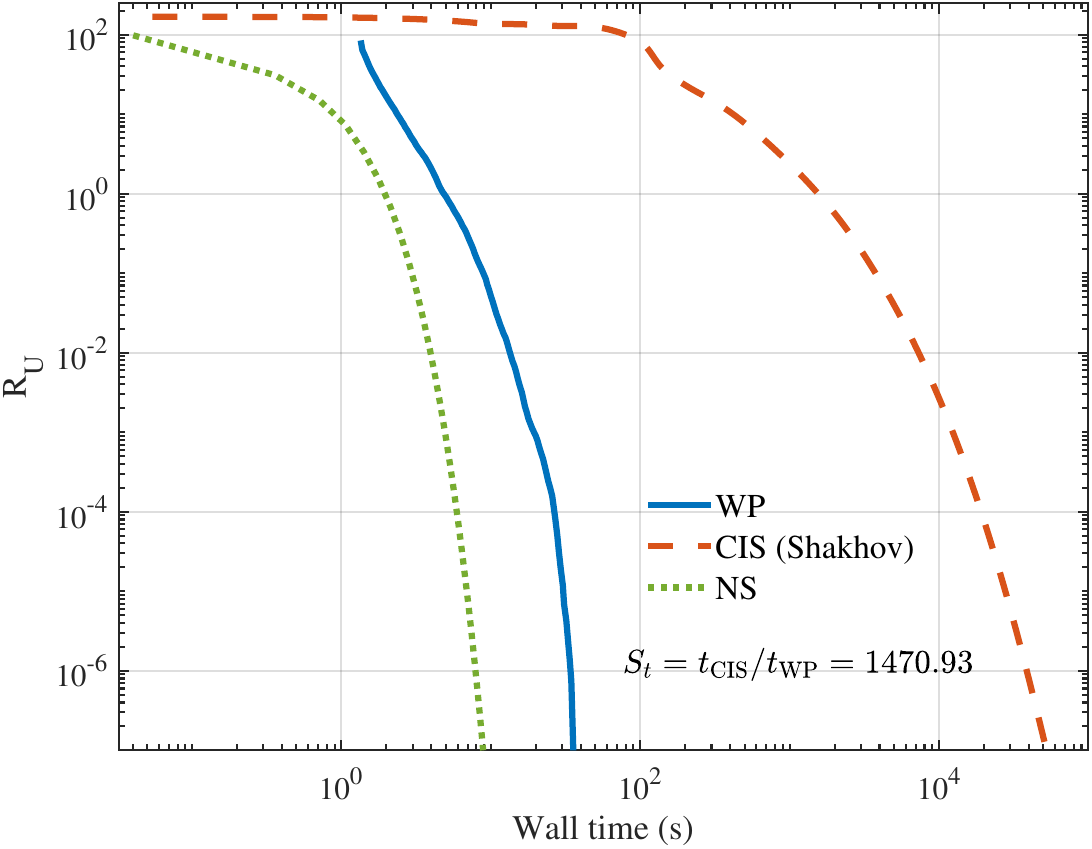}}}
\caption{$\varepsilon=10^{-4}$}\label{fig:apollo-res-kn1em4}
\end{subfigure}
\caption{Mach-5 flow around the Apollo section: histories of the normalized macroscopic residual of WP and the reference solvers against outer iteration (left) and wall time (right). At $\varepsilon=10^{-4}$ both panels show CIS of the Shakhov model and the NS solver.}\label{fig:apollo-residuals}
\end{figure}

\FloatBarrier
\subsection{Summary of the tests}\label{sec:test_summary}
Table~\ref{tab:efficiency} collects the iteration and wall-time ratios. At
$\varepsilon=1$ they are of order one, namely $1.56$ in iterations and $1.26$
in wall time for Couette flow, $1.78$ and $1.36$ for the cavity, $1.29$ and
$1.08$ for the cylinder and $1.11$ and $0.90$ for the Apollo section. In this
regime the particle carries nearly the whole solution, and the W updates add
20 to 30 per cent to the cost of an outer iteration, which the smaller number
of iterations largely pays for. As the
Knudsen number decreases the ratios grow
to $9.47$ and $7.55$ for the cavity at $\varepsilon=0.075$ and, at
$\varepsilon=10^{-2}$, to $50.83$ and $35.27$ for Couette flow, $31.74$ and
$17.52$ for the cylinder and $27.46$ and $16.10$ for the Apollo section. In
the continuum regime the speedups for the cavity at $\mathrm{Re}=100$ reach
$758.11$ and $534.84$. The iteration ratios with respect to CIS of the
Shakhov model reach $7383.90$ for Couette flow at $\varepsilon=10^{-4}$,
$3797.50$ for the cavity at $\mathrm{Re}=1000$, and $54412.11$ and $36488.23$
for the cylinder and the Apollo section at $\varepsilon=10^{-4}$, with
wall-time ratios of $18.14$, $33.18$, $1434.53$ and $1470.93$ against the same
solver. This trend agrees with Theorem~\ref{thm:acceleration}, according to
which CIS slows down as the resolved modes become long compared with the mean
free path, whereas the WP convergence factor is bounded by the particle share
of the transport.

Measured against the NS solver with the same macroscopic discretization and
iteration (Table~\ref{tab:ns_cost}), WP obtains the full Boltzmann solution
at $1.05$ to $2.09$ times the NS wall time in Couette and cavity flows and at
$4.17$ and $4.02$ times for the cylinder and the Apollo section, a cost
comparable to that of the Navier--Stokes solution.

The accuracy results are consistent with Theorem~\ref{thm:ap}. Where WP is
compared with CIS, the field differences range from $10^{-10}$ in the
rarefied cavity to about $10^{-3}$ in the normal shock and the hypersonic
external flows, and the molecular distributions of the shock and the external
flows reproduce the same non-equilibrium structure.
In the Reynolds-number cavity flows WP approaches the Navier--Stokes solution,
with speed differences of $8.15\times10^{-3}$ and $1.58\times10^{-2}$.

\begin{table}[!htb]
\centering
\small
\caption{CIS-to-WP ratios of iteration count and wall time, evaluated
independently at the common residual threshold of each comparison. Values
marked with an asterisk use the CIS iteration count and wall time of the Shakhov
model~\cite{liu_wpd_shakhov};
the wall time of all continuum cases is compared with the NS solver in
Table~\ref{tab:ns_cost}.}
\label{tab:efficiency}
\begin{tabular}{lllcc}
\toprule
Case & Regime & WP setting & Iteration speedup $S_{\rm iter}$ & Wall-time speedup $S_t$\\
\midrule
Couette & $\varepsilon=1$          & W1/PR6/P1   & $1.56$     & $1.26$\\
Couette & $\varepsilon=10^{-2}$    & W1/PR6/P1   & $50.83$    & $35.27$\\
Couette & $\varepsilon=10^{-4}$    & W256/PR6/P1 & $7383.90^{*}$ & $18.14^{*}$\\
Cavity  & $\varepsilon=1$          & W1/PR6/P1   & $1.78$     & $1.36$\\
Cavity  & $\varepsilon=0.075$      & W1/PR6/P1   & $9.47$     & $7.55$\\
Cavity  & $\mathrm{Re}=100$        & W24/PR6/P1  & $758.11$       & $534.84$\\
Cavity  & $\mathrm{Re}=1000$       & W24/PR6/P1  & $3797.50^{*}$  & $33.18^{*}$\\
Cylinder & $\varepsilon=1$         & W1/PR2/P1   & $1.29$     & $1.08$\\
Cylinder & $\varepsilon=10^{-2}$   & W1/PR2/P1   & $31.74$    & $17.52$\\
Cylinder & $\varepsilon=10^{-4}$   & W120/PR2/P1 & $54412.11^{*}$ & $1434.53^{*}$\\
Apollo  & $\varepsilon=1$          & W1/PR2/P1   & $1.11$     & $0.90$\\
Apollo  & $\varepsilon=10^{-2}$    & W1/PR2/P1   & $27.46$    & $16.10$\\
Apollo  & $\varepsilon=10^{-4}$    & W120/PR2/P1 & $36488.23^{*}$ & $1470.93^{*}$\\
\bottomrule
\end{tabular}
\end{table}

\begin{table}[!htb]
\centering
\small
\caption{Wall time to reach $R_U=10^{-7}$ for WP and for the NS solver with the
same mesh, reconstruction, implicit macroscopic iteration and CFL numbers in
the continuum cases. The NS solver performs no kinetic work; its wall time is
the reference cost of the macroscopic solution on the same mesh.}
\label{tab:ns_cost}
\begin{tabular}{lcccc}
\toprule
Case & Mesh & NS wall time (s) & WP wall time (s) & $t_{\rm WP}/t_{\rm NS}$\\
\midrule
Couette, $\varepsilon=10^{-4}$  & $4\times100$  & $163.0$  & $320.5$   & $1.97$\\
Cavity, $\mathrm{Re}=100$       & $64^2$        & $48.0$   & $100.1$   & $2.09$\\
Cavity, $\mathrm{Re}=1000$      & $100^2$       & $1275.0$ & $1342.7$  & $1.05$\\
Cylinder, $\varepsilon=10^{-4}$ & $129\times40$ & $46.3$   & $193.0$   & $4.17$\\
Apollo, $\varepsilon=10^{-4}$   & $64\times24$  & $8.89$   & $35.7$    & $4.02$\\
\bottomrule
\end{tabular}
\end{table}

\section{Conclusion}\label{sec:conclusion}
This paper has extended the wave--particle decomposition to the full
Boltzmann equation at both the continuous and the discrete level. At the
continuous level, an identity transformation splits the collision term into a
relaxation toward the local Maxwellian and a conservative remainder, leaving
the Boltzmann equation unchanged. The first term of the resulting integral
solution, the equilibrium integral over a local kinetic horizon, defines the
wave, and its complement defines a particle that carries the collision
remainder. The microscopic wave and particle equations add up to the
Boltzmann equation, and their moments are extended Navier--Stokes equations
whose generalized stresses and heat fluxes sum to those of the Boltzmann
distribution. The total conservation law and the particle equation form the
wave--particle multiscale equations, a closed system coupled in both
directions, in which the particle closes the stress and heat flux of the
conservation law and the macroscopic state supplies the Maxwellian, the wave
and the endpoint of the particle equation. These equations are exact for
every horizon, and the horizon-to-relaxation ratio parametrizes a continuous
spectrum from the Boltzmann equation to Navier--Stokes hydrodynamics, along
which the combined flux preserves the Euler and Navier--Stokes limits. As in
Part~I, the split is made in the equation, through a local kinetic horizon,
rather than in the numerical update, through the global time step of the
UGKWP method. For the Boltzmann equation the particle adds the signed
collision remainder to the collisionless fraction of Part~I and reduces to
that fraction when the remainder vanishes.

At the discrete level, the system is solved by the coupled WP iteration, a
two-way macroscopic--microscopic iteration. The W iteration advances the
conservative variables with the particle fixed, and the P iteration advances
the particle with the predicted wave fixed and with the endpoint traction
supplied by the macroscopic prediction. The wave flux is the horizon-weighted
second-order gas-kinetic flux, the particle flux is the second-order upwind
discrete-ordinate flux of MUSCL type, and acceptance reconstructs both
components from the updated total distribution, so that the two residuals
converge together. Under the stated assumptions of conservation, flux
compatibility, admissibility and solvability, the fixed-point equation is the
discrete Boltzmann equation and contains no preconditioning parameters. The linear model identifies the mechanism of the near-continuum
acceleration, in which the direct treatment of the conservative modes and the
endpoint coupling suppress the transport error passed to the P iteration
increasingly as the horizon-to-relaxation ratio grows.

In the normal shock, Couette, cavity, cylinder and Apollo calculations the
converged WP solutions agree with those of the conventional iterative scheme
in the rarefied and transition regimes and approach the Navier--Stokes
solutions in the continuum regime. As the Knudsen number decreases, the
wall-time speedup over the conventional iterative scheme for the Boltzmann
equation grows to about 35 for Couette flow at $\varepsilon=10^{-2}$ and to
more than 500 for the cavity at $\mathrm{Re}=100$, and in the continuum
regime the iteration count is reduced by about three to more than four orders
of magnitude. Where the Boltzmann scheme could not be carried to convergence
on the three-dimensional velocity grid, the reference is the Shakhov model,
whose iteration is the cheaper and faster of the two, so that the
corresponding wall-time ratios, from about 20 to about 1500, are conservative
lower bounds on the gain. Against the NS solver that shares its macroscopic
discretization, WP obtains the converged full Boltzmann solution at between
one and about four times the wall time.

\section*{Declaration of competing interest}
The authors declare that they have no known competing financial interests or
personal relationships that could have appeared to influence the work reported
in this paper.

\section*{Acknowledgements}
\begin{sloppypar}
Chang Liu is partially supported by the National Natural Science Foundation of
China (12031001, 12102061), the Beijing Natural Science Foundation (Z230003), and
the Presidential Foundation of the China Academy of Engineering Physics
(YZJJZQ2022017). The authors are partially supported by the National Key R\&D
Program of China (2022YFA1004500). Kun Xu is partially supported by the National
Natural Science Foundation of China (12172316, 92371107) and the Hong Kong
Research Grants Council (16301222, 16208324).
\end{sloppypar}


\clearpage
\begin{thebibliography}{99}
\bibitem{boltzmann1872} L. Boltzmann, Weitere Studien \"uber das W\"armegleichgewicht unter Gasmolek\"ulen, \textit{Sitzungsber. Kais. Akad. Wiss. Wien} 66 (1872) 275--370.
\bibitem{cercignani1988} C. Cercignani, \textit{The Boltzmann Equation and Its Applications}, Springer, New York, 1988.
\bibitem{sone2007} Y. Sone, \textit{Molecular Gas Dynamics: Theory, Techniques, and Applications}, Birkh\"auser, Boston, 2007.
\bibitem{chapman1970} S. Chapman, T.G. Cowling, \textit{The Mathematical Theory of Non-Uniform Gases}, 3rd ed., Cambridge University Press, Cambridge, 1970.
\bibitem{bird1994} G.A. Bird, \textit{Molecular Gas Dynamics and the Direct Simulation of Gas Flows}, Oxford University Press, Oxford, 1994.
\bibitem{yang1995} J.Y. Yang, J.C. Huang, Rarefied flow computations using nonlinear model Boltzmann equations, \textit{J. Comput. Phys.} 120 (1995) 323--339.
\bibitem{mieussens2000} L. Mieussens, Discrete-velocity models and numerical schemes for the Boltzmann-BGK equation in plane and axisymmetric geometries, \textit{J. Comput. Phys.} 162 (2000) 429--466.
\bibitem{bgk1954} P.L. Bhatnagar, E.P. Gross, M. Krook, A model for collision processes in gases. I. Small amplitude processes in charged and neutral one-component systems, \textit{Phys. Rev.} 94 (1954) 511--525.
\bibitem{shakhov1968} E.M. Shakhov, Generalization of the Krook kinetic relaxation equation, \textit{Fluid Dyn.} 3 (1968) 95--96, doi:10.1007/BF01029546.
\bibitem{mouhot2006} C. Mouhot, L. Pareschi, Fast algorithms for computing the Boltzmann collision operator, \textit{Math. Comp.} 75 (2006) 1833--1852.
\bibitem{wu2013} L. Wu, C. White, T.J. Scanlon, J.M. Reese, Y. Zhang, Deterministic numerical solutions of the Boltzmann equation using the fast spectral method, \textit{J. Comput. Phys.} 250 (2013) 27--52.
\bibitem{adams2002} M.L. Adams, E.W. Larsen, Fast iterative methods for discrete-ordinates particle transport calculations, \textit{Prog. Nucl. Energy} 40 (2002) 3--159.
\bibitem{jin1999} S. Jin, Efficient asymptotic-preserving (AP) schemes for some multiscale kinetic equations, \textit{SIAM J. Sci. Comput.} 21 (1999) 441--454.
\bibitem{pareschi2005} L. Pareschi, G. Russo, Implicit--explicit Runge--Kutta schemes and applications to hyperbolic systems with relaxation, \textit{J. Sci. Comput.} 25 (2005) 129--155.
\bibitem{bennoune2008} M. Bennoune, M. Lemou, L. Mieussens, Uniformly stable numerical schemes for the Boltzmann equation preserving the compressible Navier--Stokes asymptotics, \textit{J. Comput. Phys.} 227 (2008) 3781--3803.
\bibitem{filbet2010} F. Filbet, S. Jin, A class of asymptotic-preserving schemes for kinetic equations and related problems with stiff sources, \textit{J. Comput. Phys.} 229 (2010) 7625--7648.
\bibitem{dimarco2014} G. Dimarco, L. Pareschi, Numerical methods for kinetic equations, \textit{Acta Numer.} 23 (2014) 369--520.
\bibitem{su2020} W. Su, L. Zhu, P. Wang, Y. Zhang, L. Wu, Can we find steady-state solutions to multiscale rarefied gas flows within dozens of iterations?, \textit{J. Comput. Phys.} 407 (2020) 109245.
\bibitem{xu2010} K. Xu, J.-C. Huang, A unified gas-kinetic scheme for continuum and rarefied flows, \textit{J. Comput. Phys.} 229 (2010) 7747--7764.
\bibitem{huang2012ugks} J.-C. Huang, K. Xu, P. Yu, A unified gas-kinetic scheme for continuum and rarefied flows II: multi-dimensional cases, \textit{Commun. Comput. Phys.} 12 (2012) 662--690.
\bibitem{xu2014direct} K. Xu, \textit{Direct Modeling for Computational Fluid Dynamics: Construction and Application of Unified Gas-Kinetic Schemes}, World Scientific, Singapore, 2015.
\bibitem{liu2016boltzmann} C. Liu, K. Xu, Q. Sun, Q. Cai, A unified gas-kinetic scheme for continuum and rarefied flows IV: Full Boltzmann and model equations, \textit{J. Comput. Phys.} 314 (2016) 305--340, doi:10.1016/j.jcp.2016.03.014.
\bibitem{zhu2016implicit} Y. Zhu, C. Zhong, K. Xu, Implicit unified gas-kinetic scheme for steady state solutions in all flow regimes, \textit{J. Comput. Phys.} 315 (2016) 16--38, doi:10.1016/j.jcp.2016.03.038.
\bibitem{zhu2017multigrid} Y. Zhu, C. Zhong, K. Xu, Unified gas-kinetic scheme with multigrid convergence for rarefied flow study, \textit{Phys. Fluids} 29 (2017) 096102.
\bibitem{xu2022implicit} X. Xu, Y. Zhu, C. Liu, K. Xu, UGKS-based implicit iterative method for multiscale nonequilibrium flow simulations, \textit{SIAM J. Sci. Comput.} 44 (2022) B996--B1017.
\bibitem{guo2013} Z. Guo, K. Xu, R. Wang, Discrete unified gas kinetic scheme for all Knudsen number flows: low-speed isothermal case, \textit{Phys. Rev. E} 88 (2013) 033305.
\bibitem{sun2015ugks} W. Sun, S. Jiang, K. Xu, An asymptotic preserving unified gas kinetic scheme for gray radiative transfer equations, \textit{J. Comput. Phys.} 285 (2015) 265--279.
\bibitem{liu2017plasma} C. Liu, K. Xu, A unified gas kinetic scheme for continuum and rarefied flows V: Multiscale and multi-component plasma transport, \textit{Commun. Comput. Phys.} 22 (2017) 1175--1223.
\bibitem{liu2019multiphase} C. Liu, Z. Wang, K. Xu, A unified gas-kinetic scheme for continuum and rarefied flows VI: Dilute disperse gas-particle multiphase system, \textit{J. Comput. Phys.} 386 (2019) 264--295.
\bibitem{zhu2021decade} Y. Zhu, K. Xu, The first decade of unified gas kinetic scheme, arXiv:2102.01261 (2021).
\bibitem{liu2020ugkwp} C. Liu, Y. Zhu, K. Xu, Unified gas-kinetic wave-particle methods I: Continuum and rarefied gas flow, \textit{J. Comput. Phys.} 401 (2020) 108977.
\bibitem{zhu2019ugkwp} Y. Zhu, C. Liu, C. Zhong, K. Xu, Unified gas-kinetic wave-particle methods II: Multiscale simulation on unstructured mesh, \textit{Phys. Fluids} 31 (2019) 067105.
\bibitem{li2020ugkwp} W. Li, C. Liu, Y. Zhu, J. Zhang, K. Xu, Unified gas-kinetic wave-particle methods III: Multiscale photon transport, \textit{J. Comput. Phys.} 408 (2020) 109280.
\bibitem{liu2021plasma} C. Liu, K. Xu, Unified gas-kinetic wave-particle methods IV: Multi-species gas mixture and plasma transport, \textit{Adv. Aerodyn.} 3 (2021) 9, doi:10.1186/s42774-021-00062-1.
\bibitem{xu2021diatomic} X. Xu, Y. Chen, C. Liu, Z. Li, K. Xu, Unified gas-kinetic wave-particle methods V: Diatomic molecular flow, \textit{J. Comput. Phys.} 442 (2021) 110496, doi:10.1016/j.jcp.2021.110496.
\bibitem{wei2024ugkwp} Y. Wei, Y. Zhu, K. Xu, Unified gas-kinetic wave-particle methods VII: Diatomic gas with rotational and vibrational nonequilibrium, \textit{J. Comput. Phys.} 497 (2024) 112610, doi:10.1016/j.jcp.2023.112610.
\bibitem{yang2022multiphase} X. Yang, C. Liu, X. Ji, W. Shyy, K. Xu, Unified gas-kinetic wave-particle methods VI: Disperse dilute gas-particle multiphase flow, \textit{Commun. Comput. Phys.} 31 (2022) 669--706, doi:10.4208/cicp.OA-2021-0153.
\bibitem{long2024nonequilibrium} W. Long, Y. Wei, K. Xu, Nonequilibrium flow simulations using unified gas-kinetic wave-particle method, \textit{AIAA J.} 62 (2024) 1411--1433, doi:10.2514/1.J063641.
\bibitem{yang2026turbulence} X. Yang, K. Xu, Wave-particle based multiscale modeling and simulation of non-equilibrium turbulent flows, \textit{Comput. Fluids} (2026) 107162.
\bibitem{guo2026representation} Z. Guo, Y. Zhu, K. Xu, Kinetic representation of the unified gas-kinetic wave-particle method and beyond, \textit{Commun. Comput. Phys.} 39 (2026) 1512--1535, doi:10.4208/cicp.OA-2024-0192.
\bibitem{yang2026lts} S. Yang, C. Zhong, H. Jin, S. Liu, C. Zhuo, Convergence-accelerated simplified unified wave-particle method via local time-stepping for 3D hypersonic non-equilibrium flow simulation, \textit{Theor. Appl. Mech. Lett.} (2026) 100668, doi:10.1016/j.taml.2026.100668.
\bibitem{guo2026lts} W. Guo, J. Cao, W. Long, K. Xu, Rigorously justified local time stepping in the unified gas-kinetic wave-particle method for steady multiscale flow simulation, arXiv:2607.09552 (2026).
\bibitem{liu_wpd_I} C. Liu, K. Xu, Wave-particle decomposition for kinetic equations I: Theory and numerics, arXiv:2606.26915 (2026).
\bibitem{vanleer1974} B. van Leer, Towards the ultimate conservative difference scheme. II. Monotonicity and conservation combined in a second-order scheme, \textit{J. Comput. Phys.} 14 (1974) 361--370.
\bibitem{xu2001} K. Xu, A gas-kinetic BGK scheme for the Navier--Stokes equations and its connection with artificial dissipation and Godunov method, \textit{J. Comput. Phys.} 171 (2001) 289--335.
\bibitem{vanleer1979} B. van Leer, Towards the ultimate conservative difference scheme. V. A second-order sequel to Godunov's method, \textit{J. Comput. Phys.} 32 (1979) 101--136.
\bibitem{yoon1988} S. Yoon, A. Jameson, Lower--upper symmetric-Gauss--Seidel method for the Euler and Navier--Stokes equations, \textit{AIAA J.} 26 (1988) 1025--1026.
\bibitem{liu_wpd_shakhov} C. Liu, K. Xu, An implicit scheme for the wave--particle decomposition of the Shakhov kinetic model, preprint (2026).
\bibitem{ghia1982} U. Ghia, K.N. Ghia, C.T. Shin, High-Re solutions for incompressible flow using the Navier--Stokes equations and a multigrid method, \textit{J. Comput. Phys.} 48 (1982) 387--411.
\end{thebibliography}
\end{document}